\documentclass[
    nonacm,
    acmsmall,
    screen,
]{acmart}

\AtBeginDocument{}

\usepackage{amsmath}
\usepackage{mathpartir}
\usepackage{xcolor}
\usepackage{enumitem}
\usepackage{listings}
\usepackage{tikz}
\usetikzlibrary{
  arrows.meta,
  tikzmark,
  calc
}
\usepackage{forest}
\usepackage{tikz-cd}
\AtBeginDocument{\usepackage{thm-restate}}
\usepackage{wrapfig}
\usepackage{pifont}

\newcommand{\mf}[1]{\mathsf{#1}}
\newcommand{\tsc}[1]{\textsc{\small{#1}}}
\newcommand{\mgrey}[1]{{\setlength{\fboxsep}{0pt}\colorbox{gray!20}{$#1$}}}
\newcommand{\grey}[1]{{\setlength{\fboxsep}{0pt}\colorbox{gray!20}{#1}}}
\newcommand{\blue}[1]{{\color{blue} #1}}

\newcommand{\semi}{\, ;\, }
\newcommand{\bang}{\ !\ }
\newcommand{\concat}{\mathbin{+\mkern-8mu+}}
\newcommand{\Ell}{\mathcal{L}}
\newcommand{\R}{\varrho}
\newcommand{\heap}{\mathcal{H}}
\newcommand{\dom}{\mf{dom}}
\newcommand{\step}{\rightsquigarrow}
\newcommand{\pcase}[1]{\textbf{Case} \textsc{#1}}
\newcommand{\bsq}[1]{\llbracket #1 \rrbracket}
\newcommand{\steprt}{
  \,\mathop{\step}\limits^{
    \vbox to -1ex{\kern-1ex\hbox{$\scriptstyle{\hskip-0.75pt *}$}\vss}
  }{\!}
}

\newcommand{\ba}{\begin{array}}
\newcommand{\ea}{\end{array}}

\newcommand{\bl}{\ba[t]{@{}l@{}}}
\newcommand{\el}{\ea}

\newenvironment{reductions}{\[\ba{@{}r@{~}c@{~}l@{\qquad}l@{}}}{\ea\]\ignorespacesafterend}
\newenvironment{reds}{\ba[t]{@{}r@{~}c@{~}l@{\qquad}l@{}}}{\ea}

\newenvironment{equations}{\[\ba{@{}r@{\;\;}c@{\;\;}l@{}}}{\ea\]\ignorespacesafterend}

\newenvironment{syntax}{\[\ba{@{}r@{\;\;}c@{\;\;}l@{}}}{\ea\]\ignorespacesafterend}
\newenvironment{syn}{\ba[t]{@{}r@{~}c@{~}l@{}}}{\ea}

\newcommand{\trans}{\bsq}

\newcommand{\prefixtrans}[2]{\mathcal{#1}\trans{#2}}

\newcommand{\strans}[1]{\prefixtrans{S}{#1}}
\newcommand{\stransp}[4]{\strans{#1}(#2,#3,#4)}
\newcommand{\itrans}[1]{\prefixtrans{I}{#1}}
\newcommand{\itransp}[5]{\itrans{#1}(#2,#3,#4,#5)}

\newcommand{\stept}{
  \,\mathop{\step}\limits^{
    \vbox to -0.75ex{\kern-1ex\hbox{$\scriptstyle{\hskip-0.75pt +}$}\vss}
  }{\!}
}
\newcommand{\val}[1]{\mf{val}\ #1}
\newcommand{\valp}[1]{\mf{val}\ (#1)}

\lstdefinelanguage{LambdaReact}{
  keywords={val, let, in, if, then, else, true, false, onClick, class, attr,
  tag, text, component, const, return},
  keywordstyle=\color{blue}\bfseries,
  morestring=[b]",
  stringstyle=\color{orange},
  literate=
    {\\}{{$\lambda$}}1
}

\lstdefinestyle{lrstyle}{
  language=LambdaReact,
  basicstyle=\linespread{0.8}\footnotesize\ttfamily,
  basewidth=0.5em,
  aboveskip=2pt,
  belowskip=2pt,
  xleftmargin=0pt,
  xrightmargin=0pt,
  showstringspaces=false,
}

\newcommand{\lampure}{\lambda^\mf{react}_\mf{pure}}
\newcommand{\lamhook}{\lambda^\mf{react}_\mf{hook}}
\newcommand{\lamref}{\lambda^\mf{react}_\mf{ref}}

\newcommand{\If}{\mf{if}}
\newcommand{\Then}{\mf{then}}
\newcommand{\Else}{\mf{else}}
\newcommand{\Let}{\mf{let}}
\newcommand{\In}{\mf{in}}
\newcommand{\nil}{[]}

\newcommand{\setter}{\mf{set}}
\newcommand{\tone}{\mathbf{1}}
\newcommand{\elet}[3]{\Let\ #1 = #2\ \mf{in}\ #3}
\newcommand{\eif}[3]{\If\ #1\ \mf{then}\ #2\ \mf{else}\ #3}
\newcommand{\eval}[4]{\lceil #1 \rceil \downarrow^#2_#3 #4}
\newcommand{\compInst}[3]{\mf{compInst}_{#1}(\lceil #2 \rceil, #3)}
\newcommand{\estate}[4]{
  \mf{let}\ #1, {#1}_\mf{set} = \mf{useState}^{#2}\ #3\ \mf{in}\ #4
}

\newcommand{\header}[1]{
    \begin{flushleft}
        \textbf{#1}\hfill
    \end{flushleft}
}
\newcommand{\subheader}[1]{
    \begin{flushleft}
        \textit{#1}\hfill
    \end{flushleft}
}
\newcommand{\headersig}[2]{
    \begin{flushleft}
        \textbf{#1}\hfill \framebox{#2}
    \end{flushleft}
}

\newcommand{\headertwo}[2]{
    \begin{flushleft}
        \textbf{#1}\hfill #2
    \end{flushleft}
}

\begin{document}

\title{On the Expressive Power of Hook-Indexed State}
\subtitle{Comprehending State in the React UI Library}

\author{Toby Ueno}
\email{toby.ueno@ed.ac.uk}
\orcid{0009-0004-5848-5327}
\affiliation{
  \institution{The University of Edinburgh}
  \city{Edinburgh}
  \country{UK}
}

\author{Ezra Cooper}
\orcid{0009-0003-7700-5547}
\email{ezra.cooper@ed.ac.uk}
\affiliation{%
  \institution{The University of Edinburgh}
  \city{Edinburgh}
  \country{UK}
}

\author{Simon Fowler}
\orcid{0000-0001-5143-5475}
\affiliation{%
  \institution{University of Glasgow}
  \city{Glasgow}
  \country{UK}
}

\author{Sam Lindley}
\orcid{0000-0002-1360-4714}
\email{sam.lindley@ed.ac.uk}
\affiliation{%
  \institution{The University of Edinburgh}
  \city{Edinburgh}
  \country{UK}
}

\begin{abstract}
  React is currently the most popular JavaScript library
  for front-end web development.
  It endows UI components with state
  via the \texttt{useState} construct, through a mechanism
  that we name \emph{hook-indexed state}.
  However, React's implementation of hook-indexed state
  is both fragile and poorly understood.
  Programmers frequently experience unexpected behavior,
  and the benefits of programming in this paradigm are wholly unclear.

  We define $\lamhook$,
  a core calculus that captures the essence of hook-indexed state,
  and give it a type system to statically rule out unexpected behavior.
  We then translate $\lamhook$ into a calculus $\lamref$ based on ML-style
  references, and prove that this translation preserves typability and
  semantics, thus showing that hook-indexed state is expressible in terms of
  standard references.
  The translation is local and admits a compositional definition,
  but involves significant code duplication.
  As well as illuminating the essence of state in React, our
  observations can inform the design of future reactive programming
  systems.


\end{abstract}




\received{20 February 2007}
\received[revised]{12 March 2009}
\received[accepted]{5 June 2009}

\maketitle

\section{Introduction}\label{sec:introduction}

React is a JavaScript library that enables programmers to
declaratively specify dynamic user interfaces (UIs). It is currently
the most popular library for front-end web development
\cite{so:survey,stateofjs:survey}, and has influenced the development
of other widespread libraries and frameworks, such as Preact and
NextJS. Through the React Native project, React-style programming is
also now common in mobile app development.

React programmers structure UIs in terms of \emph{components}, which
are functions that compute and return HTML-like \emph{elements}.
Although basic components are pure functional specifications, React
provides the \texttt{useState} function to persist some internal state
across rerenders (re-evaluations) of components.  The
\texttt{useState} construct is one of a number of primitive
\emph{hooks}, which ``hook into'' React's runtime: in this case,
\texttt{useState} coordinates with the runtime to inject the most
recent state value, and to rerender the necessary components when that
value is updated.

We develop a formal understanding of this widely-used but
under-studied state mechanism, which we name \emph{hook-indexed
state}. In prior work, \citet{lee:25:react} give a semantics for
React, intended as a guide for \emph{React developers}. In contrast,
we study \texttt{useState} from a \emph{language designer's}
perspective. We examine \emph{how} a language designer may integrate a
version of hook-indexed state with static soundness guarantees into a
reactive language. We also explore \emph{why} a designer would
consider hook-indexed state by relating it to standard ML-style state.

%

\paragraph{Soundness}

To the novice React developer,
\texttt{useState} may appear magically simple,
but they will inevitably discover that what underlies
this magic is in fact a fragile contraption,
held together only by duct-tape invariants.
If \texttt{useState} is placed inside a conditional
statement, the program will crash.
Updating state during a render may lead to an infinite rerender.
To avoid such problems as these, programmers must obey
the ``Rules of React'' \cite{react:rules}, a set of React-specific
maxims (e.g. ``only call hooks at the top level'')
that lead to bugs and unexpected behavior when
broken, but which are not statically enforceable on the language
level.\footnote{
  React has a linter plugin \cite{react:linter} that can \emph{help} catch
  Rule violations but which, as a linter, cannot guarantee soundness.
}

We introduce a calculus, $\lamhook$, that captures the essence of
hook-indexed state while abstracting away from the chaotic real-world
semantics of JavaScript \citep{guha:10:js} and complicating orthogonal
features/performance concerns of React
\citep{lee:25:react,madsen:20:react}. We endow $\lamhook$ with a novel
type system, featuring \emph{lightweight effect typing} and
\emph{pseudo-linearity} of pointers.
%
%
Our type system captures and enforces the informal Rules of
React that pertain to state.

\paragraph{Expressivity}
The \texttt{useState} construct has an unorthodox dynamic
context-dependent semantics.
%
%
During the \emph{initial} render of the component, it uses the
provided argument as an initial value, but on \emph{successive}
renders, it must instead access some location on the heap
\cite{react:state}. This observation was first formalized by
\citet{lee:25:react}, and is likewise modeled in our calculus
$\lamhook$.

To elucidate the pros and cons of programming with \texttt{useState},
we consider what it takes to \emph{express}
hook-indexed state in terms of a more well-understood canonical model of
state.
We show by translation that the behavior of
\texttt{useState} is expressible in terms of ML-style
\emph{first-class references}, which have a well-studied formalization
\citep{harper:94:ref, wright:94:sound, pottier:11:store,
  ahmed:04:thesis} and are built-in to ML dialects such as OCaml
\citep{leroy:25:ocaml}, F\# \cite{syme:16:fsharp}, and SML
\citep{milner:97:sml}. We define a calculus of first-class references,
$\lamref$, and an encoding of $\lamhook$ into $\lamref$
which satisfies a tight operational correspondence
\cite{wadler:22:plfa,gorla:10:encode,griesemer:20:go,
  tang:26:metx,minamide:96:closure,plotkin:75:lc}.
The benefits and drawbacks of our translation are somewhat
nuanced. Our encoding admits a \emph{compositional} definition,
and is \emph{local} in that unrelated constructs are unaffected.
However, it leads to code duplication and breaks abstraction,
which shows an advantage of \texttt{useState} as a feature.
%

\paragraph{Contributions}
Our primary contributions are summarized as follows.
\begin{itemize}[leftmargin=*]
  \item We propose three reactive calculi:
    $\lampure$ (Section \ref{sec:lampure}),
    $\lamhook$ (Section \ref{sec:lamhook}), and
    $\lamref$ (Section \ref{sec:lamref}).
%
    Each has a sound type system and a small-step operational
    semantics.
  \begin{itemize}
    \item $\lampure$ is \emph{stateless}, and as such can only represent static
      elements, but is equipped with the machinery necessary to support
      stateful components.
    \item $\lamhook$ extends $\lampure$ with \emph{hook-based}
      stateful components, in the style of React.
      Its type system notably features a lightweight effect system
      and a pseudo-linear pointer context.
    \item $\lamref$ extends $\lampure$ with \emph{reference-based}
      stateful components, in the style of ML.
  \end{itemize}
  \item We define a \emph{translation} from $\lamhook$ to $\lamref$
    (Section \ref{sec:trans}), prove that the translation is
    type-preserving, and show a tight operational
    correspondence between the two calculi.
%
  \item We demonstrate that $\lamhook$ offers a scalable foundation by sketching three
      extensions: reconciliation, the $\mf{useEffect}$ hook, and custom hooks (Section
    \ref{sec:ext}).
\end{itemize}

Section~\ref{sec:overview} contains a more detailed exposition of
React and \texttt{useState} and an overview of our findings.
Section~\ref{sec:relwork} surveys related work and
Section~\ref{sec:conc} concludes.

\section{Overview}
\label{sec:overview}

%


\subsection{Programming with React}
React programmers structure UIs in terms of \emph{components}, which
are ordinary JavaScript functions that return \emph{elements}.
Elements are datatypes that specify HTML: they consist of text
literals, tagged nodes (e.g. \lstinline|<div>|, \lstinline|<p>|), and
components themselves.  A complete program consists of a single
\emph{root} element that is rendered by the React runtime.

The \texttt{Hello} component below takes one parameter, \texttt{name},
and formats it into an element. We write this element as pseudo-HTML
directly in JavaScript via the JSX syntax extension \cite{react:jsx}.
We may then use \texttt{Hello} in other components: the \texttt{Main}
component below uses \texttt{Hello} twice with different parameters.
If we designate \texttt{Main} as the root of our program, React's
runtime will \emph{instantiate} the component by evaluating it and its
subcomponents, yielding the HTML on the right.

\begin{minipage}{0.32\linewidth}
\begin{lstlisting}
const Hello = ({name}) => {
  const s = "Hello, ";
  return (
    <p>{s + name + "!"}</p>
  );
}
\end{lstlisting}
\end{minipage}
\hfill
\begin{minipage}{0.32\linewidth}
\begin{lstlisting}
const Main = () => {
  return (<div class="foo">
    <Hello name="reader" />
    <Hello name="world" />
  </div>);
}
\end{lstlisting}
\end{minipage}
\hfill
\begin{minipage}{0.32\linewidth}
\centering\small\textit{(Rendered HTML)}
\begin{lstlisting}
<div class="foo">
  <p>"Hello, reader!"</p>
  <p>"Hello, world!"</p>
</div>
\end{lstlisting}
\end{minipage}

The above program is entirely \emph{static} as it is unchanged by user
input. The power of React lies in how the system \emph{dynamically}
responds to events.
In this paper, we primarily focus on a single kind of event,
\emph{clicks}, although our system easily extends to others~(following
e.g.~\cite{Fowler20}).

React supports primitive event handlers that can be attached to tags.
We can provide the \texttt{onClick} attribute with an event handler
function, which the runtime executes whenever the \lstinline|<div>| is
clicked.  However, for these event handlers to affect the output of a
React application, components must maintain \emph{state} that event
handlers can update. How do we model this state?
JavaScript features mutable variables, introduced via the
\lstinline|let| keyword, but these cannot concisely capture the
desired semantics.
As a running example, consider
a counter that displays its current value and,
when clicked, increments that value.
If we model this value as a \emph{local} mutable variable, as in
\texttt{CounterLocal} below, then \texttt{x} remains 0 after a
click. This is because React rerenders the component by rerunning the
function, rebinding \texttt{x} to its initial value.
\begin{lstlisting}
  const CounterLocal = () => { let x = 0; return (<div onClick={() => x = x + 1}>{x}</div>); }
\end{lstlisting}
%
%
%
We can persist the state across rerenders by hoisting \texttt{x}
outside the function body, making it global.
\begin{lstlisting}
  let x = 0; const CounterGlobal = () => { return (<div onClick={() => x = x + 1}>{x}</div>); }
\end{lstlisting}
Alas, this style breaks when we instantiate the component multiple
times.
Updating \texttt{x} affects \emph{all instances} of
\texttt{CounterGlobal}, but each counter instance should maintain its
own independent state.

React's solution is the primitive \texttt{useState} hook, which
coordinates with the runtime to maintain a local state cell across
rerenders of the same component instance. The hook is parameterized
by the initial state and returns both a value and a setter, the
latter of which triggers a rerender when invoked.
Upon rendering, \texttt{useState} returns the most recent state value.
A correct implementation of a counter using \texttt{useState}
is given by component \texttt{Counter} below:
we can instantiate \texttt{Counter} in multiple places,
and each will have its own local state value that persists across rerenders.

\begin{lstlisting}
  const Counter = () => { const [x, setX] = useState(0);
                          return (<div onClick={() => setX (x + 1)}>{x}</div>); }
\end{lstlisting}

\subsection{Local State, by Hook or by Crook}
\label{sec:overview-problems}

Hooks such as \texttt{useState} are \emph{superficially} simple, but
beneath the veneer lie complex and fragile internals. Programmers must
adhere to the statically unenforceable Rules of React
\cite{react:rules}, lest they encounter one of many program-breaking
pitfalls.  We examine two of these pitfalls---\emph{inconsistent hook
evaluation order} and \emph{top-level setter calls}---to see how
\texttt{useState} leads to unexpected behavior.

\paragraph{Evaluation Order}
Consider the following \texttt{OutOfOrder} code.
\begin{center}
\begin{lstlisting}[numbers=left,numbersep=4pt]
  const OutOfOrder = () => { const [x, setX] = useState(true);
                             if (x) { const [y, setY] = useState(0); }
                             const [z, setZ] = useState("foo"); ...
\end{lstlisting}
\end{center}
On the first render, all three \texttt{useState} calls bind initial values
\lstinline|true|, \lstinline|0|, and \lstinline|"foo"| to \texttt{x},
\texttt{y}, and \texttt{z}, respectively.
However, if \texttt{x} is set to \lstinline|false| and the component
is rerendered, then \texttt{z} will be bound to \lstinline|0|.
%
%
%
%
%
In React, \texttt{useState} calls are identified \emph{by their evaluation
order} with respect to a component instance, e.g. the first \texttt{useState}
call sets/retrieves from the same place on the heap \cite{react:state}.
We refer to React's state mechanism as \emph{hook-indexed state}
because the location of the state cell is \emph{indexed} by the
evaluation order of the corresponding hook.
In \texttt{OutOfOrder}, the \emph{second} such call is made at line 2
on the first render, but at line 3 on the second render, leading to
the described bug.

To avoid this bug, React programmers must ensure that their hooks
always respect the \emph{same evaluation order} during each run of
their components.
This evaluation-order invariant is captured by the \emph{top-level
restriction} of the Rules of React, which states that hooks should
only be called at the top level of components, i.e. not nested inside
of conditionals, loops, and so on\footnote{ Hooks may also be called
at the top-level of ``custom hooks'', which are programmer-defined
compositions of primitive hooks. We defer discussion of custom hooks
to Section \ref{sec:custom}. }. This restriction ensures that the
execution order of all \texttt{useState} calls is consistent, thereby
avoiding the above bug.

\paragraph{Top-Level Setter Calls}
The state values returned from \texttt{useState} are immutable, and so
must be set using the provided setter. This setter, however, does not
work like a standard reference update, as evidenced by the following component.
\begin{center}
\begin{lstlisting}
  const CounterTwo = () => { const [x, setX] = useState(0); setX (x + 2);
                             return (<div onClick=() => setX (x + 1)> {x} </div>); }
\end{lstlisting}
\end{center}
\texttt{CounterTwo} only differs from \texttt{Counter} by the addition
of the call to \lstinline{setX (x + 2)}. One might quite naturally
assume that \texttt{x} is rendered with value 2, whereas in fact
\texttt{CounterTwo} \emph{will not terminate}.

If React detects that a setter is invoked during rendering, it will
immediately attempt to rerender the component. As the setter call in
\lstinline{CounterTwo} is evaluated unconditionally, React will
attempt this rerender indefinitely. Intuitively, React does this
because it wants to \emph{avoid rendering stale state}, but this
behavior may easily catch the novice React programmer off guard.

Another Rule of React states that components should be \emph{pure}, in
that they should not perform side effects (but may perform, for
instance, local mutation).  Interestingly, React's notion of purity
does \emph{not} preclude calling a state setter: the guidelines only
caution against doing so \emph{unconditionally} \cite{react:set}. As
with the other Rules, this invariant must be manually verified by the
programmer.

\subsection{A Core Calculus of Hook-Indexed State}
\label{sec:overview-lamhook}

In Section \ref{sec:lamhook}, we propose a calculus, $\lamhook$, which captures
the essence of hook-indexed state.
We draw inspiration from the work of \citet{lee:25:react},
but where their system prioritizes semantic faithfulness to React
in all its glory,
we design our language around safety and (relative) simplicity.

We introduce \emph{component bodies} $D$, which are wrapped into
values as $\mf{component}\ \{D\}$.  The grammar of component bodies is
as follows, where $V$ are values and $M, N$ are computations:
{\small
\[D ::= \mf{val}\ V \mid \elet{x}{M}{D} \mid \estate{x}{}{V}{D}\]
}
Bodies extend computations with the $\mf{useState}$ construct. This
body-computation distinction syntactically enforces the top-level
restriction: the prior \texttt{OutOfOrder} example, for instance, is
\emph{syntactically invalid}, as $\mf{if}$-statements can only take
the form $\eif{V}{M}{N}$ and \emph{not} $\eif{V}{D_1}{D_2}$.

To give $\mf{useState}$ a small-step operational semantics,
we must provide two key pieces of context.
\begin{itemize}
  \item We provide a \emph{pointer} $p$ to a component instance's memory,
  which contains its state. We need this pointer to distinguish
  the state of \emph{different instantiations} of the \emph{same component}.
  \item We provide a \emph{phase} $\phi$, which is either
  $\mf{Init}$ (initial phase) or $\mf{Succ}$ (successor phase).
  When $\phi=\mf{Init}$, $\mf{useState}$ \emph{initializes} state,
  and when $\phi=\mf{Succ}$, $\mf{useState}$ \emph{retrieves} state.
\end{itemize}
We define the body semantics via step judgment $\step^\phi_p$,
which captures this context.
Consider the first $\mf{useState}$ call in \texttt{Counter}. In
$\lamhook$, that line evaluates as follows in the $\mf{Init}$ phase:
{\small
\[
\heap_0 \semi \estate{x}{\ell}{0}{D} \step^\mf{Init}_p D[0/x, \setter^\ell_p/x_\mf{set}] \semi \heap_1\]
}
where $\heap_0$ is the initial heap and $\setter^\ell_p$ is a
\emph{setter} for the $\ell$th evaluated hook at pointer $p$. (For
convenience, we annotate $\mf{useState}$ constructs with unique labels
$\ell$ in an elaboration pass: our syntax guarantees that these labels
correspond to evaluation order.)
The initial value $0$ is bound to $x$, and is \emph{also} saved on the
updated heap $\heap_1$. We also bind a setter to $x_\mf{set}$.

Suppose later we evaluate a setter, perhaps in the event handler
for \texttt{Counter}. Doing so is an ordinary computation,
and does not require the extra context given to bodies in the step judgment.
{\small
\[\heap_1 \semi \setter^\ell_p\ 1 \step
\mf{val}\ () \semi \heap_2\]
}
Now, say we rerender \texttt{Counter}, which requires us to evaluate
its body in the $\mf{Succ}$ phase. We disregard the initial value and
look up the most recent value in $\heap_2$, which is now $1$, and bind
it to $x$.
{\small
\[\heap_2 \semi \estate{x}{\ell}{0}{D} \step^\mf{Succ}_p
D[1/x, \setter^\ell_p/x_\mf{set}] \semi \heap_2\]
}

\subsection{The Typing Rules of React}
\label{sec:overview-types}

The syntactic restrictions of $\lamhook$ get us only so far towards
safety: for instance, the $\mf{Succ}$ phase becomes stuck
at a $\mf{useState}$ call if the corresponding location on the heap
has not yet been initialized.
To statically eliminate such problems,
we introduce a \emph{type system} to $\lamhook$ (Section
\ref{sec:lamhook-typing}) and prove it sound.
Here we briefly survey its salient
features.

\paragraph{Every Body Has a Type}
The type of a body is a \emph{record} $\Ell$
that maps $\mf{useState}$ labels $\ell$ to value types $A$.
For example, in the typing derivation for \texttt{Counter},
an instance of rule \textsc{TB-State} records that
the state at label $\ell$ will have type $\mf{Int}$ once $\mf{useState}$
is executed (where $\emptyset$ is the empty record):
{\small
\begin{mathpar}
  \inferrule*[right=TB-State] {
    \Gamma \vdash 0: \mf{Int}
    \and
    \Gamma, x: \mf{Int}, x_\mf{set}: (\mf{Int} \xrightarrow{\mf{Mut}} \tone)
    \vdash D: \emptyset
  } {
    \Gamma \vdash \estate{x}{\ell}{0}{D}: (\ell: \mf{Int})
  }
\end{mathpar}
}
When a $\mf{useState}$ construct is evaluated, it disappears from the body,
and so a body's type \emph{changes} during evaluation.
Once a body reduces to a value, it has type $\emptyset$.
(We defer discussion of the $\mf{Mut}$ annotation on the function arrow
until later in this section.)

\paragraph{The Two-Body Problem}
The type of a component instance's memory is also given by $\Ell$.
Each instance memory contains a state store $\sigma$ that maps labels $\ell$
to values $V$. If $\ell: A \in \Ell$, then $V$ has type $A$, as expected.
For a body with type $\Ell_1$ and an instance memory with type $\Ell_2$,
we compare $\Ell_1$ and $\Ell_2$
to obtain an invariant that depends on $\phi$ as follows:
\begin{itemize}[leftmargin=*]
  \item When $\phi=\mf{Init}$, we maintain that
    $\Ell_1 \cup \Ell_2$ is constant. All state labels $\ell$ in
    $\dom(\Ell_1 \cup \Ell_2)$ are either \emph{on the heap} (if
    $\ell \in \dom\ \Ell_2$) or \emph{to be initialized} (if
    $\ell \in \dom\ \Ell_1$).  Once the body has reduced to a value, we have
    $\Ell_1=\emptyset$, and the component instance is \emph{fully initialized}.
  \item When $\phi=\mf{Succ}$, we maintain that $\Ell_1 \subseteq \Ell_2$.
  All labels in the body must already be on the heap
  so that we can safely retrieve them via $\mf{useState}$.
\end{itemize}

We illustrate the $\mf{Init}$-phase invariant by returning to our
\texttt{Counter} example. We zoom in from the entire heap $\heap$ to
just the relevant store $\sigma$, which begins empty, and so has type
$\emptyset$.
\[\emptyset \semi \estate{x}{\ell}{0}{D} \step^\mf{Init}_p D[0/x, \setter^\ell_p/x_\mf{set}] \semi \ell \mapsto 0\]
After a step, the store $\ell \mapsto 0$ now has type $(\ell:
\mf{Int})$, but the body $D$ has type $\emptyset$. Although the types
have changed, their \emph{union}, which was our invariant, has not.
Now that all state cells in \texttt{Counter} have been initialized, it
is safe to rerender \texttt{Counter}, i.e., to evaluate its body in
the $\mf{Succ}$ phase. This also follows from our typing invariants.
To begin $\mf{Succ}$-phase evaluation, we require that $\Ell_1
\subseteq \Ell_2$, which is true in this instance:
both are equal to $(\ell: \mf{Int})$.

\paragraph{Every Instance Has a Phase}
We embed body evaluation into normal computation via the runtime
construct $\eval{D_0}{\phi}{p}{D}$, read as ``initial body $D_0$ has
evaluated to $D$ at phase $\phi$ and pointer $p$''.  When
$\phi=\mf{Init}$, evaluating this computation may set some memory at
$p$, as described above. Now, consider the following computation,
which initializes $D$ \emph{twice} at $p$.
\[
\bl
  \mf{let}\;y = \eval{D_0}{\mf{Init}}{p}{(\estate{x}{\ell}{0}{\ldots})} \;\In \\
  \eval{D_0}{\mf{Init}}{p}{(\estate{x}{\ell}{1}{\ldots})} \\
\el
\]
Once the first initialization is complete, we expect that the state store
contains value $0$ at label $\ell$, and that this
remains unchanged unless a setter is called. However, the second initialization
will \emph{overwrite} the first, leading to unexpected behavior.
We thus choose to forbid reinitialization and
guarantee that every component instance is initialized \emph{exactly once}.

To enforce this restriction, we introduce a \emph{phase context} $\Phi$
that maps pointers $p$ to phases $\phi$.
We define a non-standard context merging relation $\Join$ on $\Phi$:
we treat pointers with $\phi=\mf{Init}$ \emph{linearly},
whereas pointers with $\phi=\mf{Succ}$ are \emph{unrestricted}.
The typing rule \textsc{TC-EvalI} for construct $\eval{D_0}{\phi}{p}{D}$
guarantees that $p: \mf{Init}$ occurs in $\Phi$,
meaning that no other computation in the derivation
may refer to $p$ in the phase context. The above expression is therefore
untypeable (denoted by an \ding{55} right of the conclusion),
since we cannot merge the phase contexts $\Phi_1$ and $\Phi_2$:
{\small\begin{mathpar}
\inferrule*[right=TC-Let] {
  \inferrule*[right=TC-EvalI] {\cdots} {
    \Gamma \semi \Pi \semi \Phi_1 \vdash \eval{D_0}{\phi}{p}{D_1}: \mf{ElemInst}
  }
  \inferrule*[Right=TC-EvalI] {\cdots} {
    \Gamma, y: \mf{ElemInst} \semi \Pi \semi \Phi_2
    \vdash \eval{D_0}{\phi}{p}{D_2}: \mf{ElemInst}
  }
} {
  \Gamma \semi \Pi \semi \textcolor{red}{\Phi_1 \Join \Phi_2} \vdash
  \mf{let}\ y = \eval{D_0}{\mf{Init}}{p}{D_1}\
  \mf{in}\ \eval{D_0}{\mf{Init}}{p}{D_2}: \mf{ElemInst}
  \quad \text{\ding{55}}
}
\end{mathpar}}

\paragraph{Every Body is (Almost) Pure}

Modeling the retrying behavior that results from top-level setter
calls, as mentioned in Section \ref{sec:overview-problems}, would
significantly complicate our semantics for no real gain.
We regard such programs as ill-behaved\footnote{
  The React documentation \cite{react:usestate}
  mentions a case where top-level setter calls
  can be used to store information from previous renders,
  but they admit this pattern is
  ``hard to understand and is usually best avoided.''
  In any case, it can be easily avoided at a performance cost
  by wrapping the relevant code with React's $\mf{useEffect}$ construct,
  which defers its contents until rendering completes and
  setter calls are permitted.
  We discuss $\mf{useEffect}$ as an extension in Section \ref{sec:ext}.
}, and we introduce a small \emph{effect system} to rule them out.
We have only two effects in $\lamhook$: $\mf{Pure}$ for pure
computations and $\mf{Mut}$ for those that may modify the heap. In
particular, setters $\setter^\ell_p$ have type $A
\xrightarrow{\mf{Mut}} \tone$. The typing rule \textsc{TB-Let} for
bodies, which allows us to introduce arbitrary computations, requires
these computations to be $\mf{Pure}$, which prevents us from invoking
the setter inside of a body.
For instance, the following derivation fails:
{\small\begin{mathpar}
  \inferrule*[right=TB-Let] { \inferrule*[right=TC-App] {
      \inferrule*[right=TV-Setter] { } { \Gamma \semi \Pi \semi
        \emptyset \vdash \setter^\ell_p: \mf{Int}
        \xrightarrow{\mf{Mut}} \tone } \cdots} { \Gamma \semi \Pi \semi
      \emptyset \vdash \setter^\ell_p\ x: \tone \bang
      \mf{\textcolor{red}{Mut}} } \cdots} { \Gamma \semi \Pi \vdash
    \elet{()}{\setter^\ell_p\ x}{\cdots}: \Ell
  \quad \text{\ding{55}}}
\end{mathpar}}
Setters may, as expected,
be used in event handlers, which have no such effect restriction.


\subsection{Putting useState in Other Terms}

Having outlined a well-behaved fragment of $\lamhook$, we are now
ready to consider the \emph{expressive power} of hook-indexed
state. To this end, we seek to define a \emph{translation} from
$\lamhook$ into some other calculus that lacks $\mf{useState}$.  We
argue that this problem is non-trivial in that, to the best of our
knowledge, the behavior of $\mf{useState}$ is not well-captured by any
known model of persistent state in UI programming.  We identify two
particular features of $\mf{useState}$ that lead to this difficulty:
\begin{itemize}[leftmargin=*]
    \item React components \emph{interleave} initialization with ordinary
    computation. This prevents us from straightforwardly ``decomposing''
    a component into its state and pure functions that update and render
    the state, in the style of e.g.
    $\lambda_\mf{MVU}$~\citep{Fowler20}
    or the reactive document calculus~\cite{crichton:24:doc}.
  \item During component evaluation, some $\mf{useState}$ constructs
    will have been initialized while others will not. Representing the
    state of a component as a \emph{single value} is therefore awkward
    to do in a type-safe way: we could, for instance, declare every
    field in the value as optional during initialization, and then (somehow)
    guarantee that all fields are non-null once initialization is
    complete. (This also precludes adapting the React implementation,
    which models state as a dynamically typed list and
    only ``enforces'' safety via The Rules~\cite{react:state}.)
\end{itemize}

In light of these considerations, we demonstrate that $\mf{useState}$
can be expressed in terms of a language $\lamref$ with
\emph{first-class references}, which are well-understood and widely
used in dialects of ML \cite{milner:97:sml, leroy:25:ocaml,
  syme:16:fsharp}. Drawing inspiration from the two phases of
evaluation, our translation \emph{duplicates component bodies}, once
under each phase. For example, the component

\begin{minipage}{0.9\linewidth}
\begin{lstlisting}[numbers=left,numbersep=2pt]
  let CounterSquare = component { let x, setX = useState 0 in let y = x * x in V }
\end{lstlisting}
\end{minipage}

\noindent
is translated as follows:

\begin{minipage}{0.465\linewidth}
\begin{lstlisting}[numbers=left,numbersep=2pt,firstnumber=1]
  let Counter = component \().{
    let rr = ref(false) in
    let xLoc = ref(0) in
    let x = !xLoc in
    let setX = \y.(xLoc := y; rr := true) in
    let y = x * x in
\end{lstlisting}
\end{minipage}
\quad
\begin{minipage}{0.465\linewidth}
\begin{lstlisting}[numbers=left,numbersep=2pt,firstnumber=7]
    (\().(
      let x = !xLoc in
      let setX = \y.(xLoc := y; rr := true) in
      let y = x * x in
      W
    ), rr, W) }
\end{lstlisting}
\end{minipage}



\noindent
We leave the return value \lstinline|V|
and its translation \lstinline|W| abstract for conciseness.
We use monospaced font for example $\lamhook$ and $\lamref$ code.

The outermost function corresponds to the \emph{initial}
phase, whereas the inner function (returned by the outer one)
represents the \emph{successor} phase. Lines 3--5 \emph{and}
8--9 of the target are the translations of $\mf{useState}$ in each
phase. In the initial phase, we \emph{initialize} a reference, which
is crucially captured by the successor phase, wherein we \emph{access
that same reference}.  Setters are translated as functions that set
their associated reference, but they also flip a ``rerender flag''
\lstinline|rr|, which is declared at the top of a component and is returned by
the initial phase.

We prove that the translation between $\lamhook$ and $\lamref$
preserves both \emph{types} and \emph{operational semantics}.
Furthermore, the translation is \emph{local} in the sense that it only
performs nontrivial work on components, bodies, and setters.
For readability, we do not formulate it \emph{compositionally},
but it easily admits a compositional definition.

Not all is perfect with the translation.  Writing components in the
style of the translation is cumbersome due to the duplication of code:
in the above example, lines 6 and 10 both mirror the same expression
\lstinline|let y = x * x| from the
source, but must be in both phases to maintain the semantics.
Furthermore, the translation \emph{leaks the abstraction}: for
example, a programmer must now manually control the rerender flag,
whereas in $\lamhook$, the flag is hidden neatly behind the
declarative interface of $\mf{useState}$.

\section[Lambda-React-Pure: A Core Calculus for Stateless Reactivity]
  {$\lampure$: A Core Calculus for Stateless Reactivity}
\label{sec:lampure}

We present $\lampure$, an extension of the simply-typed lambda calculus
with \emph{elements} and \emph{lightweight reactivity}.
$\lampure$ has no state, and so can only model static elements:
nevertheless, we include state-related machinery (such as events and handlers)
as $\lampure$ serves as a basis for $\lamhook$ and $\lamref$.

\subsection{Syntax}
\label{sec:lampure-syntax}

We present $\lampure$'s syntax below.
\emph{Runtime-only} syntax that is not intended to be written directly
by the programmer is \grey{highlighted in gray}:
we call the remainder \emph{programmer-level}.
We adopt Barendregt's variable convention \cite{barendregt:85:lam}
and identify terms up to $\alpha$-equality.
We formulate $\lampure$ in a \emph{fine-grain call-by-value} (FGCBV) style
\cite{levy:03:cbv},
with a syntactic distinction between \emph{values} and \emph{computations},
and where intermediate computations must be explicitly $\mf{let}$-bound:
this style simplifies the metatheory by restricting where nested computations
can occur. We can straightforwardly elaborate a more standard
nested calculus into FGCBV through $\mf{let}$-insertion.

    {\small
\begin{syntax}
\text{Types } A, B
  &::=& \tone \mid \mf{Str} \mid \mf{Bool} \mid A \rightarrow B
         \mid A \times B \mid \mf{List}\ A \mid \mf{Element} \mid \mf{Attr}
         \mid \mgrey{\mf{ElemInst}} \\
\text{Type context } \Gamma
  &::=& \emptyset \mid \Gamma, x: A \\
\text{Values } V, W
  &::=& x \mid () \mid s \mid \mf{true} \mid \mf{false}
        \mid \lambda x.M \mid (V, W) \mid \nil \mid V :: W
        \mid \mf{attr}(V, W)\\
  &\mid& \mf{onClick}\ V
         \mid \mf{text}\ V \mid \mf{tag}(V, W) \mid \mgrey{\mf{textInst}\ V}
         \mid \mgrey{\mf{tagInst}_{id}(V, W)} \\
\text{Computations } M, N
  &::=& \mf{val}\ V \mid \elet{x}{M}{N} \mid V\ W
        \mid \eif{V}{M}{N} \mid \mf{fst}\ V \mid \mf{snd}\ V \\
  &\mid& \mf{fold}\ V_1\ V_2\ V_3 \mid
  \mgrey{\mf{init}\ V} \mid \mgrey{\mf{click}_{id}\ V}
         \mid \mgrey{\mf{runClick}\ V} \mid \mgrey{\mf{update}\ V} \\
\end{syntax}
\begin{center}
$\begin{array}{c@{\qquad}c@{\qquad}c}
\text{Variables }x & \text{Strings }s & \mgrey{\text{Identifiers } {id}}
\end{array}$
\end{center}
}

\paragraph{Types}
Standard types include the unit type $\tone$, strings $\mf{Str}$,
booleans $\mf{Bool}$, functions $A_1 \rightarrow A_2$,
products $A_1 \times A_2$, and lists $\mf{List}\ A$.
The $\mf{Element}$ type denotes decorated trees of strings (e.g. HTML, XML),
and $\mf{Attr}$ is the type of those decorations.
$\mf{ElemInst}$ represents the \emph{realization} of an element,
to be generated by the runtime from an element:
alternatively, we may think of elements as \emph{specifications} of instances.
A single $\mf{Element}$
may be realized into multiple distinguishable $\mf{ElemInst}$s.

\paragraph{Values}
Values include variables (ranged over by $x$), the unit value,
strings (ranged over by $s$), booleans, $\lambda$-abstractions, pairs,
and list constructors, all of which are entirely standard.
$\mf{attr}(V, W)$ denotes a purely cosmetic attribute
(given by a key-value pair of strings),
whereas $\mf{onClick}\ V$ attaches an \emph{event handler} $V$
to its parent element.
The constructors $\mf{text}$ and $\mf{tag}$ produce elements:
$\mf{text}\ V$ embeds a string into an element,
whereas $\mf{tag}(V, W)$ contains a list of subelements $W$ annotated
with some attributes $V$ (the tag \emph{name} can be understood as a special attribute).
The corresponding instance constructors are
$\mf{textInst}$ and $\mf{tagInst}$ respectively, and are analogous
except that $\mf{tagInst}$ carries an \emph{identifier},
ranged over by ${id}$ and generated at runtime.

\paragraph{Computations}
The $\mf{val}$ construct treats a value as a pure computation, and
computations are sequenced using $\mf{let}$-binders.
Constructs for
application, conditionals, and pair projection are standard.
The remaining computations model fragments of the runtime.
$\mf{init}\ V$ realizes an element $V$ into an instance.
$\mf{click}_{id}\ V$ searches an instance $V$ for a node
with identifier ${id}$, and if found, runs any $\mf{onClick}$ event handlers.
$\mf{runClick}\ V$ runs an attribute $V$'s handler if it exists.
$\mf{update}\ V$ reinitializes any instances with outdated state.
(In $\lampure$, no instances have state, and so $\mf{update}\ V$
is a no-op, but including it here makes our lives easier when
considering different implementations of stateful components in sections
\ref{sec:lamhook-opsem} and \ref{sec:lamref-opsem}.)

\paragraph{Canonicity}
These ``runtime constructs'' are \emph{prima facie} ad-hoc,
but they can easily be expressed in terms of more standard constructs.
Elements, attributes, and instances are all inductive datatypes,
so we could define our runtime in terms of folds/destructors on these datatypes.
$\mf{init}$ is a fold on elements, $\mf{runClick}$ destructs an attribute, and
$\mf{click}_{id}$ and $\mf{update}$ are folds on instances.
For simplicity, we choose to explicitly define only the
machinery necessary for our purposes.

\paragraph{Syntactic sugar}
For brevity, we use the following abbreviations in presentation
(for all languages):

\begin{itemize}
  \item We use CBV syntax instead of strict FGCBV,
    i.e. we allow for nested computations and often elide $\mf{val}$,
    assuming a (standard) elaboration pass that performs
    $\mf{let}$-insertion appropriately;
  \item we allow for n-ary tuple-deconstruction
    (e.g. $\elet{x, y, z}{V}{M}$),
    assuming a standard desugaring into multiple pair destructurings,
  \item and we use the following abbreviations:
\end{itemize}
\begin{center}{\small
  \begin{minipage}{0.5\textwidth}
    $\ba{r@{\;\triangleq\;}l@{\quad}l}
      M \semi N & \elet{x}{M}{N} &(\text{ for } x \notin \mf{fv}(N))\\
      \lambda ().M & \lambda x.M &(\text{ for } x \notin \mf{fv}(M))\\
      \mf{map}\ V\ W & \multicolumn{2}{l}{
        \mf{fold}\ (\lambda x.\lambda acc. V x :: acc)\ []\ W}\\
      \mf{iter}\ V\ W & \multicolumn{2}{l}{
        \mf{fold}\ (\lambda x.\lambda (). V x)\ ()\ W}\\
    \ea$
  \end{minipage}
  \begin{minipage}{0.4\textwidth}
    $\ba{r@{\;\triangleq\;}l}
      \mf{inits}\ V & \mf{map}\ (\lambda x.\mf{init}\ x)\ V\\
      \mf{clicks}_{id}\ V & \mf{iter}\ (\lambda x.\mf{click}_{id}\ x)\ V\\
      \mf{runClicks}\ V & \mf{iter}\ (\lambda x.\mf{runClick}\ x)\ V\\
      \mf{updates}\ V & \mf{map}\ (\lambda x.\mf{update}\ x)\ V\\
    \ea$
  \end{minipage}}
\end{center}

Furthermore, following JSX, we allow for direct embedding of HTML-like
structures into source code, and desugar these structures into elements and
attributes.  $\{V\}$ embeds the value $V$ into such a structure.
For instance, the computation on the left desugars to
the computation on the right:

\noindent
\begin{minipage}[t]{0.375\linewidth}
\begin{lstlisting}
  let s = "Hello, world!" in
  <div onClick={\().()} class="foo">
    <b>{x}</b>
    {s}
  </div>
\end{lstlisting}
\end{minipage}
\hfill
\begin{minipage}[t]{0.61\linewidth}
\begin{lstlisting}
  let s = val "Hello, world!" in
  val tag(
    attr("tag", "div") :: onClick (\y.val ())
      :: attr("class", "foo") :: [],
    tag(attr("tag", "b") :: [], x :: []) :: text s :: [])
\end{lstlisting}
\end{minipage}

We elide a formal definition of this desugaring process,
although we note that it is \emph{type-directed},
since embedded strings (\texttt{s} in the above example)
should be wrapped in the $\mf{text}$ constructor
while embedded tags (\texttt{x} above) should not.

\subsection{Type System}

We have separate but mutually inductive typing judgments for values and
computations (Figure \ref{fig:lampure-typing}).
Most rules are entirely unsurprising.
Rule \textsc{TV-OnClick} specifies that an event handler
should have type $\tone \rightarrow \tone$:
since $\lampure$ has no side effects,
the contents of the handler will only become significant
in the presence of effectful extensions.

\begin{figure}[t]
    {\footnotesize
    \headersig{Typing Rules for Values}{$\Gamma \vdash V: A$}
\begin{mathpar}
\inferrule[TV-Var] {
  x: A \in \Gamma
} {
  \Gamma \vdash x: A
}

\inferrule[TV-Unit] { } {
  \Gamma \vdash (): \tone
}

\inferrule[TV-Bool] {V \in \{\mf{true}, \mf{false}\}} {
  \Gamma \vdash V: \mf{Bool}
}

\inferrule[TV-Str] { } {
  \Gamma \vdash s: \mf{Str}
}

\inferrule[TV-Lam] {
  \Gamma, x: A \vdash M: B
} {
  \Gamma \vdash \lambda x.M: A \rightarrow B
}

\inferrule[TV-Nil] { } {
  \Gamma \vdash \nil: \mf{List}\ A
}

\inferrule[TV-Cons] {
  \Gamma \vdash V: A
  \quad
  \Gamma \vdash W: \mf{List}\ A
} {
  \Gamma \vdash V :: W: \mf{List}\ A
}
\quad
\inferrule[TV-Pair] {
  \Gamma \vdash V: A
  \quad
  \Gamma \vdash W: B
} {
  \Gamma \vdash (V, W): A \times B
}
\quad
\inferrule[TV-Attr] {
  \Gamma \vdash V: \mf{Str}
  \quad
  \Gamma \vdash W: \mf{Str}
} {
  \Gamma \vdash \mf{attr}(V, W): \mf{Attr}
}
\quad
\inferrule[TV-OnClick] {
  \Gamma \vdash V: \tone \rightarrow \tone
} {
  \Gamma \vdash \mf{onClick}\ V: \mf{Attr}
}

\inferrule[TV-Text] {
  \Gamma \vdash V: \mf{Str}
} {
  \Gamma \vdash \mf{text}\ V: \mf{Element}
}
\;\;
\inferrule[TV-Tag] {
  \Gamma \vdash V: \mf{List\ Attr}
  \\\\
  \Gamma \vdash W: \mf{List\ Element}
} {
  \Gamma \vdash \mf{tag}(V, W): \mf{Element}
}
\;\;
\inferrule[TV-TextI] {
  \Gamma \vdash V: \mf{Str}
} {
  \Gamma \vdash \mf{textInst}\ V: \mf{ElemInst}
}
\;\;
\inferrule[TV-TagI] {
  \Gamma \vdash V: \mf{List\ Attr}
  \\\\
  \Gamma \vdash W: \mf{List\ ElemInst}
} {
  \Gamma \vdash \mf{tagInst}_{id}(V, W): \mf{ElemInst}
}
\end{mathpar}

\headersig{Typing Rules for Computations}{$\Gamma \vdash M: A$}
\begin{mathpar}
\inferrule[TC-Val] {
  \Gamma \vdash V: A
} {
  \Gamma \vdash \mf{val}\ V: A
}
\quad
\inferrule[TC-Let] {
  \Gamma \vdash M: A
  \\\\
  \Gamma, x: A \vdash N: B
} {
  \Gamma \vdash \elet{x}{M}{N}: B
}
\quad
\inferrule[TC-App] {
  \Gamma \vdash V: A \rightarrow B
  \\\\
  \Gamma \vdash W: A
} {
  \Gamma \vdash V\ W: B
}
\quad
\inferrule[TC-If] {
  \Gamma \vdash V: \mf{Bool}
  \\\\
  \Gamma \vdash M: A
  \and
  \Gamma \vdash N: A
} {
  \Gamma \vdash \eif{V}{M}{N}: A
}
\quad
\inferrule[TC-Fst] {
  \Gamma \vdash V: A \times B
} {
  \Gamma \vdash \mf{fst}\ V: A
}
\quad
\inferrule[TC-Snd] {
  \Gamma \vdash V: A \times B
} {
  \Gamma \vdash \mf{snd}\ V: B
}

\inferrule[TC-Init] {
  \Gamma \vdash V: \mf{Element}
} {
  \Gamma \vdash \mf{init}\ V: \mf{ElemInst}
}
\quad
\inferrule[TC-Click] {
  \Gamma \vdash V: \mf{ElemInst}
} {
  \Gamma \vdash \mf{click}_{id}\ V: \tone
}
\quad
\inferrule[TC-Run] {
  \Gamma \vdash V: \mf{Attr}
} {
  \Gamma \vdash \mf{runClick}\ V: \tone
}
\quad
\inferrule[TC-Update] {
  \Gamma \vdash V: \mf{ElemInst}
} {
  \Gamma \vdash \mf{update}\ V: \mf{ElemInst}
}
\quad
\inferrule[TC-Fold] {
  \Gamma \vdash V_1: A \rightarrow B \rightarrow B
  \\\\
  \Gamma \vdash V_2: B
  \quad
  \Gamma \vdash V_3: \mf{List}\ A
} {
  \Gamma \vdash \mf{fold}\ V_1\ V_2\ V_3: B
}
\end{mathpar}
}
\caption{Typing rules for $\lampure$ values and computations.}
\Description{TODO:DESC}
\label{fig:lampure-typing}
\end{figure}

\paragraph{Programs}
We require that a (programmer-level) program $M$ satisfies
$\emptyset \vdash M: \mf{Element}$.
The initial configuration $M_0$ is then $\elet{x}{M}{\mf{init}\ x}$,
and satisfies $\emptyset \vdash M_0: \mf{ElemInst}$.


\begin{figure}[t]
{\small
    \headersig{Reduction Rules for Computations}{$M \step N$}
    \subheader{$\beta$-reductions}
    \begin{minipage}{0.45\textwidth}
    \[\ba{@{}r@{~}c@{~}l@{\quad}l@{}}
    \elet{x}{M}{N}           &\step& \elet{x}{M'}{N} & \tsc{EC-Let} \\
                            & & \quad (\text{if } M \step M') \\
    \elet{x}{\mf{val}\ V}{M} &\step&   M[V/x]        & \tsc{EC-Val} \\
    (\lambda x.M)\ V         &\step&   M[V/x]        &\tsc{EC-App} \\
    \mf{fold}\ V\ W\ []        &\step&   \mf{val}\ W &\tsc{EC-FdN} \\
    \ea\]\ignorespacesafterend
\end{minipage}
\hfill
\begin{minipage}{0.53\textwidth}
    \[\ba{@{}r@{~}c@{~}l@{\quad}l@{}}
\eif{\mf{true}}{M}{N}    &\step&   M                                   &\tsc{EC-IfT} \\
\eif{\mf{false}}{M}{N}   &\step&   N                                   &\tsc{EC-IfF} \\
\mf{fst}\; (V, W)         &\step&   \mf{val}\ V                         &\tsc{EC-Fst} \\
\mf{snd}\; (V, W)         &\step&   \mf{val}\ W                         &\tsc{EC-Snd} \\
\mf{fold}\; V_f\; V_0\; (V_h :: V_t) &\step&   V_f\; V_h\; (\mf{fold}\; V_f\; V_0\; V_t)  &\tsc{EC-FdC} \\
    \ea\]
\end{minipage}

\subheader{Runtime reductions}
\begin{reductions}
\mf{init}\ (\mf{text}\ V) &\step& \mf{val}\ (\mf{textInst}\ V) &\tsc{EC-InitText} \\
\mf{init}\ (\mf{tag}(V, W))
&\step& \elet{x}{\mf{inits}\ W}{\mf{val}\ (\mf{tagInst}_{id}(V, x))},
  \quad {id}\ \text{fresh}\!\!\!\!
    &\tsc{EC-InitTag} \\
\mf{click}_{id}\ (\mf{textInst}\ V) &\step& \mf{val}\ ()
    &\tsc{EC-ClickText} \\
\mf{click}_{id}\ (\mf{tagInst}_{id}(V, W))
  &\step& \mf{runClicks}\ V \semi \mf{clicks}_{id}\ W
    &\tsc{EC-ClickTagT} \\
\mf{click}_{id}\ (\mf{tagInst}_{{id}'}(V, W))
  &\step& \mf{clicks}_{id}\ W,
    \quad {id} \neq {id}'
    &\tsc{EC-ClickTagF} \\
\mf{runClick}\ \mf{attr}(V, W) &\step& \mf{val}\ () &\tsc{EC-RunAttr} \\
\mf{runClick}\ (\mf{onClick}\ V) &\step& V\ () &\tsc{EC-RunClick} \\
\mf{update}\ (\mf{textInst}\ V) &\step& \mf{val}\ \mf{textInst}\ V
    &\tsc{EC-UpdateText} \\
\mf{update}\ (\mf{tagInst}_{id}(V, W))
  &\step& \mf{let}\ x = \mf{updates}\ W\ \mf{in}\
          \mf{val}\ (\mf{tagInst}_{id}(V, x))
    &\tsc{EC-UpdateTag}
\end{reductions}

\headersig
    {Top-level Eventful Reduction Rules}
    {$M \step_\top N$}
    \[
    \begin{array}{lcl}
    M \step_\top N \:\: (\text{if } M \step N) \quad \tsc{ET-Lift}
    & \qquad &
    \mf{val}\ V \step_\top
  \mf{click}_{id}\ V \semi \mf{update}\ V
      \quad \tsc{ET-Event}
\end{array}
\]
}

\caption{Small-step operational semantics and top-level semantics for $\lampure$.}
\Description{TODO:DESC}
\label{fig:lampure-opsem}
\end{figure}

\subsection{Operational Semantics}
Figure \ref{fig:lampure-opsem} defines a small-step reduction on computations;
substitution is defined on values and computations in the expected way.
The $\beta$-reduction rules on expressions are standard; the remaining rules
show how to evaluate the runtime computations.
Rules \textsc{EC-InitText} and \textsc{EC-InitTag} evaluate
static $\mf{text}$ and $\mf{tag}$ expressions $\mf{Element}$
into their runtime counterparts of type $\mf{ElemInst}$. In the
case of tags, this involves creating a fresh runtime identifier.

Rules \textsc{EC-ClickText}, \textsc{EC-ClickTagT}, and
\textsc{EC-ClickTagF} handle the $\mf{click}_{id}$ construct,
which invokes all of the event handlers contained within the node
with identifier ${id}$.
Specifically in the case that the
identifier of a tag instance $\mf{tagInst}_{id}(V, W)$
matches ${id}$, \textsc{EC-ClickTagT}
evaluates $\mf{runClicks}\ V$ to invoke all handlers contained
within the node's attributes, and recursively searches $W$.
Correspondingly, rule \textsc{EC-RunAttr} ignores key-value
attributes, and \textsc{EC-RunClick} invokes any $\mf{onClick}$
handlers.

To simulate ``reactive behavior''
(i.e. how the system responds to user input),
we introduce an additional step judgment $\step_\top$
with two rules: a congruence rule \textsc{ET-Lift},
and a rule \textsc{ET-Event} that represents event reception.
This latter rule \emph{nondeterministically} chooses a node identifier
whose handlers should be run. The right-hand-side involves finding and invoking
the $\mf{onClick}$ handler, but also performing node updates, which are more
meaningful in $\lamhook$ and $\lamref$.

\subsection{Metatheory}

$\lampure$ enjoys syntactic type soundness, given as usual by progress and
preservation.

\begin{theorem}[Type Soundness for $\step$]
\label{thm:lampure-sound}
\hfill

\begin{itemize}
    \item \textbf{Preservation:} If $\Gamma \vdash M: A$ and $M \step N$,
then $\Gamma \vdash N: A$.
    \item \textbf{Progress:} If $\emptyset \vdash M: A$, then either $\exists V. M=\mf{val}\ V$
or $\exists N. M \step N$.
\end{itemize}
\begin{proof}
  Preservation by induction on $M \step N$,
  progress by induction on $\emptyset \vdash M: A$.
\end{proof}
\end{theorem}

We also have a form of type soundness
for the top-level step judgment $\step_\top$,
wherein the computation type must be $\mf{ElemInst}$.
Since we can always receive an event, $\step_\top$ never terminates.

\begin{theorem}[Type Soundness for $\step_\top$]
\label{thm:lampure-sound-top}
\hfill

\begin{itemize}
    \item \textbf{Preservation:} If $\Gamma \vdash M: \mf{ElemInst}$ and $M \step_\top N$,
then $\Gamma \vdash N: \mf{ElemInst}$.
    \item \textbf{Progress:} If $\emptyset \vdash M: A$, then
$\exists N. M \step_{\top} N$.
\end{itemize}
\end{theorem}
\begin{proof}
  By case-analysis on $M \step N$ and $\emptyset \vdash M: A$ respectively,
  with appeal to Theorem~\ref{thm:lampure-sound}.
\end{proof}

\section[Lambda-React-Hook: A Calculus of Hook-Indexed State]
  {$\lamhook$: A Calculus of Hook-Indexed State}
\label{sec:lamhook}
We extend $\lampure$ with components based on \emph{hook-indexed state}
in the style of React hooks.

\subsection{Syntax}

\paragraph{Term-Level}
Below we present the new $\lamhook$ syntax; runtime-only syntax is
again highlighted in \grey{gray}.
To enforce the top-level restriction, we add \emph{component bodies}.
%
%

\newcommand{\slab}[1]{\text{#1 }}

{\small
\[
\ba{@{}c@{}}
\begin{syn}
\slab{Effects} E &::=& \mf{Pure} \mid \mf{Mut} \\
\slab{Types} A, B &::=& \cdots \mid A \xrightarrow{E} B \\
\end{syn}
\qquad
\begin{syn}
\slab{Label record} \Ell &::=& \emptyset \mid \Ell, \ell: A \\
\slab{Heap context} \Pi &::=& \emptyset \mid \Pi, p: \Ell \\
\slab{Phase context} \Phi &::=& \emptyset \mid \Phi, p: \phi \\
\end{syn}
\qquad
\begin{syn}
\slab{Pointers} p \\
\slab{Labels} \ell \\
\end{syn}
\medskip \\
\begin{syn}
\slab{Values} V, W &::= & \cdots \mid \mf{component}\;\{D\}
                    \mid   \mgrey{\compInst{p}{D}{V}}
                           \mid \mgrey{\setter^\ell_p} \smallskip \\
\slab{Computations} M, N &::=& \cdots \mid \mgrey{\eval{D_0}{\phi}{p}{D}} \\
\slab{Bodies} D &::= & \mf{val}\ V \mid \elet{x}{M}{D}
                  \mid  \estate{x}{\ell}{V}{D} \\
\end{syn}
\medskip \\
\begin{syn}
\slab{Phase} \phi &::=& \mf{Init} \mid \mf{Succ} \\
\slab{Heap} \heap &::=& \emptyset \mid \heap, p \mapsto \pi \\
\end{syn}
\qquad
\begin{syn}
\slab{Instance memory} \pi &::=& \{\mf{rerender}: b, \mf{st}: \sigma\} \\
\slab{Boolean flag} b &::=& \mf{tt} \mid \mf{ff} \\
\end{syn}
\qquad
\begin{syn}
\slab{State store} \sigma &::=& \emptyset \mid \sigma, \ell \mapsto V \\
\end{syn}
\\
\ea
\]
}


Components are made \emph{stateful} by the $\mf{useState}$ construct,
written $\estate{x}{\ell}{V}{D}$.
Upon evaluation, the \emph{current} state value is bound to $x$,
and a \emph{setter} for the state cell is bound to $x_\mf{set}$.
$V$ is the \emph{initial} state value.
$\ell$ ranges over a set of \emph{labels},
and identifies the location of the state cell
(with respect to a given component instance).
Instead of the programmer writing these labels,
we expect that they are generated via a straightforward elaboration pass.

Component instances are crucially differentiated
from components: a single component may be instantiated multiple times.
In such cases, $\ell$ alone is insufficient to distinguish
a state cell on the heap, so we use runtime-generated \emph{pointers} $p$ to
disambiguate state cells at runtime.

We introduce three new values, principal among which are components
$\mf{component}\ \{D\}$
that wrap bodies $D$ into elements
and can then be composed with other elements
(e.g. as the child of a $\mf{tag}$).
A \emph{component instance} $\compInst{p}{D}{V}$
has memory at pointer $p$ on the heap,
initial body $D$ (which is used for rerendering),
and child element $V$.
Conceptually, component instances are a stateful
``shell'' around their sub-elements, with no content that would manifest
in e.g. their reification into HTML.
\emph{Setters}, written $\setter^\ell_p$,
support mutation of a particular state cell $(p, \ell)$.

The sole new computation, $\eval{D_0}{\phi}{p}{D}$,
embeds body evaluation into computation evaluation:
as explained in Section \ref{sec:overview-lamhook},
evaluating bodies requires a \emph{phase} $\phi$
and occurs at a pointer $p$.
Evaluation occurs on the right-hand body $D$,
while the left-hand body $D_0$ remains fixed:
in practice, $D$ is a \emph{reduction} of $D_0$.
This computation is read,
``at phase $\phi$ and pointer $p$, $D_0$
has evaluated to $D$''.

The \emph{heap} $\heap$ stores the contents of the state
cells owned by each instance, along with a flag to track whether
the instance should be rerendered.
Specifically, $\heap$ contains some number of pointers $p$
associated with \emph{instance memories} $\pi$.
Each $\pi$ is a record with two fields:
$\mf{rerender}$ contains a meta-boolean flag $b$
tracking whether the instance should be rerendered,
and $\mf{st}$ contains the state $\sigma$ of the instance,
where $\sigma$ is a map from labels $\ell$ to values $V$.
We identify heaps and state stores up to reordering of their entries.

\paragraph{Type-Level}
To prohibit top-level setter calls as explained
in Section \ref{sec:overview-types},
we introduce a simple \emph{effect system} to the language
with only two effects, $\mf{Pure}$ (for pure computations)
and $\mf{Mut}$ for state-altering computations.
We add effect annotations to all function arrows.

Label records $\Ell$ map labels $\ell$ to types $A$:
as mentioned in Section \ref{sec:overview-types},
they serve as the type of both bodies $D$
\emph{and} of state stores $\sigma$.
The heap type $\Pi$ is a map
from pointers $p$ to label records $\Ell$.
We also introduce a phase context $\Phi$
that tracks the current phase of each component instance
to prohibit reinitialization,
as described in Section \ref{sec:overview-types}.
Both $\Pi$ and $\Phi$ are \emph{runtime typing contexts},
in that they ``grow'' as a configuration steps,
and do not grow throughout a single typing derivation (read bottom-to-top).

\paragraph{Notation}

We use array indexing notation to refer to particular values in maps
(such as $\Pi$ and $\heap$),
and we use dot-notation to refer to specific fields of record $\pi$.
Furthermore, we use bar notation, e.g. $\heap \mid x \mapsto V$,
to concisely denote (potentially multi-level) record updates.
For instance, expression $\heap \mid \heap[p].\mf{st}[\ell] \mapsto V$
would be read as ``heap $\heap$ with the state cell at pointer $p$
and label $\ell$ set to $V$.''

\subsection{Type System}
\label{sec:lamhook-typing}
We introduce three significant changes from $\lampure$
throughout the type system, which is shown in Figure~\ref{fig:lamhook-typing}.
First, $\Pi$ and $\Phi$ serve as additional contexts in all judgments.
Second, we extend the computation typing judgment to capture
the computational \emph{effect}.
For brevity, we re-present only the programmer-level $\lampure$ rules
in Figure \ref{fig:lamhook-typing} with their effects:
all runtime-only computations in $\lampure$ have the $\mf{Mut}$ effect.
We also introduce a new typing judgment for bodies.

\paragraph{Phase Contexts}
A component instance is initialized \emph{once}, but may be rerun
at will after initialization. Reflecting this,
in $\Phi$, we treat pointers in phase $\mf{Init}$ \emph{linearly},
whereas we allow pointers in phase $\mf{Succ}$ to be used multiple times.
We define disjoint context merging accordingly as a partial operation
$\Phi = \Phi_1 \Join \Phi_2$.
When rules are read bottom-to-top, this acts as a split.
In particular, we have that if $p: \mf{Init}$ in $\Phi_1$,
then for the result to be well-defined, we must have $p \notin \dom(\Phi_2)$.
We write $\Phi_1 \Join \Phi_2$ in the conclusion of a typing rule to
abbreviate a premise $\Phi = \Phi_1 \Join \Phi_2$.
{\footnotesize
\begin{mathpar}
    \inferrule
    { }
    { \cdot = \cdot \Join \cdot }
    \quad
    \inferrule
    { \Phi = \Phi_1 \Join \Phi_2 \\ p \notin \dom(\Phi)}
    { \Phi, p {:}\; \mf{Succ} = (\Phi_1, p {:}\; \mf{Succ}) \Join (\Phi_2, p {:}\; \mf{Succ}) }
    \quad
    \inferrule
    { \Phi = \Phi_1 \Join \Phi_2 \\ p \notin \dom(\Phi)}
    { \Phi, p {:}\; \phi = (\Phi_1, p {:}\; \phi) \Join \Phi_2 }
    \quad
    \inferrule
    { \Phi = \Phi_1 \Join \Phi_2 \\ p \notin \dom(\Phi)}
    { \Phi, p {:}\; \phi = \Phi_1 \Join (\Phi_2, p {:}\; \phi) }
\end{mathpar}
}

\begin{definition}
  Say $\mf{Succ}(\Phi)$ iff all $p \in \dom(\Phi)$ have $\Phi[p] = \mf{Succ}$.
\end{definition}

To respect linearity, we always have $\mf{Succ}(\Phi)$
for values, bodies, and \emph{pure} computations.
This fact is captured in the following lemma;
we defer discussions of the typing judgments to below.

\begin{restatable}{lemma}{lamhookunrest}
\label{thm:lamhook-unrest}
  If $\Gamma \semi \Pi \semi \Phi \vdash V: A$
  or $\Gamma \semi \Pi \semi \Phi \vdash D: \Ell$
  or $\Gamma \semi \Pi \semi \Phi \vdash M: A \bang \mf{Pure}$,
  then $\mf{Succ}(\Phi)$.
\end{restatable}
\begin{proof}
  By mutual induction on the typing derivation.
  The conclusion is either given directly
  (for lambdas and terms with no subterms),
  or by inductive hypothesis for other terms with subterms.
\end{proof}

\begin{figure}[t]
{\small
    \headersig
    {Typing Rules for Values}
    {$\Gamma \semi \Pi \semi \Phi \vdash V: A$}
\begin{mathpar}
\inferrule[TV-Lam] {
  \Gamma, x: A \semi \Pi \semi \Phi \vdash M: B \bang E
  \and \mf{Succ}(\Phi)
} {
  \Gamma \semi \Pi \semi \Phi \vdash \lambda x.M: A \xrightarrow{E} B
}
\quad
\inferrule[TV-Comp] {
  \Gamma \semi \Pi \semi \Phi \vdash D: \Ell
} {
  \Gamma \semi \Pi \semi \Phi \vdash \mf{component}\ \{D\}: \mf{Element}
}
\quad
\inferrule[TV-Setter] {
  \ell: A \in \Pi[p]
  \and \mf{Succ}(\Phi)
} {
  \Gamma \semi \Pi \semi \Phi
  \vdash \setter^\ell_p: A \xrightarrow{\mf{Mut}} \tone
}

\inferrule[TV-CompI] {
  \Gamma \semi \Pi \semi \Phi_1 \vdash D: \Pi[p]
  \\
  (\Phi_1 \Join \Phi_2)[p] = \mf{Succ}
  \\
  \Gamma \semi \Pi \semi \Phi_2 \vdash V: \mf{ElemInst}
} {
  \Gamma \semi \Pi \semi \Phi_1 \Join \Phi_2
  \vdash \compInst{p}{D}{V}: \mf{ElemInst}
}
\end{mathpar}

    \headersig
        {Typing Rules for Computations}
        {$\Gamma \semi \Pi \semi \Phi \vdash M: A \bang E$}
\begin{mathpar}
\inferrule[TC-Val] {
  \Gamma \semi \Pi \semi \Phi \vdash V: A
} {
  \Gamma \semi \Pi \semi \Phi \vdash \mf{val}\ V: A \bang \mf{Pure}
}

\inferrule[TC-Sub] {
  \Gamma \semi \Pi \semi \Phi \vdash M: A \bang \mf{Pure}
} {
  \Gamma \semi \Pi \semi \Phi \vdash M: A \bang \mf{Mut}
}

\inferrule[TC-Let] {
  \Gamma \semi \Pi \semi \Phi_1 \vdash M: A \bang E
  \\\\
  \Gamma, x: A \semi \Pi \semi \Phi_2 \vdash N: B \bang E
} {
  \Gamma \semi \Pi \semi \Phi_1 \Join \Phi_2
  \vdash \elet{x}{M}{N}: B \bang E
}

\inferrule[TC-App] {
  \Gamma \semi \Pi \semi \Phi_1 \vdash V: A \xrightarrow{E} B
  \\\\
  \Gamma \semi \Pi \semi \Phi_2 \vdash W: A
} {
  \Gamma \semi \Pi \semi \Phi_1 \Join \Phi_2 \vdash V\ W: B \bang E
}
\quad
\inferrule[TC-If] {
  \Gamma \semi \Pi \semi \Phi_1 \vdash V: \mf{Bool} \\\\
  \Gamma \semi \Pi \semi \Phi_2 \vdash M: A \bang E \\
  \Gamma \semi \Pi \semi \Phi_2 \vdash N: A \bang E
} {
  \Gamma \semi \Pi \semi \Phi_1 \Join \Phi_2 \vdash \eif{V}{M}{N}: A \bang E
}
\quad
\inferrule[TC-Fst] {
  \Gamma \semi \Pi \semi \Phi \vdash V: A \times B
} {
  \Gamma \semi \Pi \semi \Phi \vdash \mf{fst}\ V: A \bang \mf{Pure}
}


\inferrule[TC-Fold] {
  \Gamma \semi \Pi \semi \Phi_1 \vdash
  V_1: A \xrightarrow{E} B \xrightarrow{E} B
  \\\\
  \Gamma \semi \Pi \semi \Phi_2 \vdash
  V_2: B
  \\
  \Gamma \semi \Pi \semi \Phi_3 \vdash
  V_3: \mf{List}\ A
} {
  \Gamma \semi \Pi \semi \Phi_1 \Join \Phi_2 \Join \Phi_3 \vdash
  \mf{fold}\ V_1\ V_2\ V_3: B \bang \mf{Mut}
}

\inferrule[TC-EvalI] {
  \Gamma \semi \Pi \semi \Phi_1 \vdash D_0: \Ell \uplus \Pi[p]
  \\
  \Gamma \semi \Pi \semi \Phi_2 \vdash D: \Ell
} {
  \Gamma \semi \Pi \semi \Phi_1 \Join \Phi_2,p: \mf{Init} \vdash
  \eval{D_0}{\mf{Init}}{p}{D}: \mf{ElemInst} \bang \mf{Mut}
}

\inferrule[TC-EvalS] {
  \Gamma \semi \Pi \semi \Phi_1 \vdash D_0: \Pi[p]
  \and
  \Gamma \semi \Pi \semi \Phi_2 \vdash D: \Ell
  \and
  \Ell \subseteq \Pi[p]
  \and
  (\Phi_1 \Join \Phi_2)[p] = \mf{Succ}
} {
  \Gamma \semi \Pi \semi \Phi_1 \Join \Phi_2 \vdash
  \eval{D_0}{\mf{Succ}}{p}{D}: \mf{ElemInst} \bang \mf{Mut}
}
\end{mathpar}

  \headersig
      {Typing Rules for Component Bodies}
      {$\Gamma \semi \Pi \semi \Phi \vdash D: \Ell$}
\begin{mathpar}
\inferrule[TB-Val] {
  \Gamma \semi \Pi \semi \Phi \vdash V: \mf{Element}
} {
  \Gamma \semi \Pi \semi \Phi \vdash \mf{val}\ V: \emptyset
}
\quad
\inferrule[TB-Let] {
  \Gamma \semi \Pi \semi \Phi_1 \vdash M: A \bang \mf{Pure}
  \\
  \Gamma, x: A \semi \Pi \semi \Phi_2 \vdash D: \Ell
} {
  \Gamma \semi \Pi \semi \Phi_1 \Join \Phi_2 \vdash \elet{x}{M}{D}: \Ell
}

\inferrule[TB-State] {
  \Gamma \semi \Pi \semi \Phi_1 \vdash V: A
  \and
  \Gamma, x: A, x_\mf{set}: A \xrightarrow{\mf{Mut}} \tone \semi \Pi
  \semi \Phi_2 \vdash D: \Ell
  \and
  \ell \notin \dom(\Ell)
} {
  \Gamma \semi \Pi \semi \Phi_1 \Join \Phi_2
  \vdash \estate{x}{\ell}{V}{D}: \Ell, \ell: A
}
\end{mathpar}
}
\caption{Selected value, computation, and body typing rules in $\lamhook$.}
\Description{TODO:DESC}
\label{fig:lamhook-typing}
\end{figure}

\paragraph{Values}
We elide rules that feature only the additions of $\Pi$ and $\Phi$.
The updated \textsc{TV-Lam} rule uses the updated computation typing judgment,
explained in more detail below. For now, we note that the computation performs
an effect $E$, which is captured by the annotation on the function type arrow.
Rule \textsc{TV-Comp} makes use of the body typing judgement
$\Gamma \semi \Pi \semi \Phi \vdash D: \Ell$. We defer discussion of
this judgement but it ensures that component body $D$ is
well typed with state labels $\Ell$; the labels are only
necessary for runtime typing guarantees, and so $\Ell$ does not
appear in the conclusion.
The body typing judgment guarantees that the
computation type underlying $D$ is $\mf{Element}$.
The \textsc{TV-Setter} rule assigns a type to a setter
based on the type of state cell $(p, \ell)$ according to $\Pi$.
The first premise implies that $p \in \dom(\Pi)$ and $\ell \in \dom(\Pi[p])$.
Rule \textsc{TV-CompI} types a component instantiation.  Its first
premise ensures that all state cells in the instance memory,
represented by $\Pi[p]$, are consistent with the type of the
component instance's initial body $D$. The second premise
$(\Phi_1 \Join \Phi_2)[p] = \mf{Succ}$ guarantees that initialization of $p$
is complete, and thus that the component is safe to rerender. The third premise
ensures that the child of a component instance is an $\mf{ElemInst}$.

\paragraph{Computations}
The $\lamhook$ computation typing judgement has the form
$\Gamma \semi \Pi \semi \Phi \vdash M : A \bang E$, which can be read
``under typing environment $\Gamma$, heap context $\Pi$, and phase context
$\Phi$, computation $M$ has type $A$ and produces effects $E$''.
Our effect system classes expressions as $\mf{Pure}$ (for
expressions that do not change the heap)
or $\mf{Mut}$. The effect system is deliberately simple:
rule \textsc{TC-Sub} allows us to treat $\mf{Pure}$ expressions as
$\mf{Mut}$, and rules like \textsc{TC-Let} and \textsc{TC-If}
require that their subcomputations share the same effect type.
We omit \textsc{TC-Snd} due to its similarity to \textsc{TC-Fst}.


The construct $\eval{D_0}{\phi}{p}{D}$ is typed
by either \textsc{TC-EvalI} and \textsc{TC-EvalS},
depending on the value of the phase $\phi$.
In the former, $\Ell$ can be read as
``the labels left to be initialized'',
and then the first premise asserts that
the initial body $D_0$ contains the labels
(and appropriately associated types)
in the instance memory at $p$
\emph{and} those still to be initialized.
Consequently, when $\Ell=\emptyset$,
then $D_0$ has type $\Pi[p]$, satisfying the first premise
of \textsc{TV-CompI}.
In \textsc{TC-EvalS}, we guarantee
that all labels left in $D$ \emph{are already in}
the instance memory.
Set operations $\cup$ and $\subseteq$ are defined on label records
in the obvious way; $\Ell \uplus \Ell'$ resolves to $\Ell \cup \Ell'$
alongside the premise that $\Ell$ and $\Ell'$ are disjoint.

\paragraph{Bodies}
As previously mentioned, the type of a body is a label record $\Ell$,
which denotes the labels (and their associated types) present in the body.

Rules \textsc{TB-Val} and \textsc{TB-Let} are analogous to
their computation-level counterparts \textsc{TC-Val} and \textsc{TC-Let},
respectively. \textsc{TB-Val} enforces that the ultimate type of any body is
$\mf{Element}$, and assigns as its state-labels the empty record $\emptyset$.
\textsc{TB-Let} allows us to embed computations
into bodies, and leaves the label record of the continuation unchanged.
These embedded computations must be $\mf{Pure}$,
which guarantees that no setters, which have effect $\mf{Mut}$,
are invoked inside of bodies.

Rule \textsc{TB-State} types the $\mf{useState}$ construct.
We require that the initial value has type $A$,
and in the continuation body, we add a binding $x : A$ to $\Gamma$
to record the type of the value. We also add a binding
$x_{\mf{set}} : A \xrightarrow{\mf{Mut}} \tone$ for the setter.
The conclusion of the rule includes $\ell : A$ in the label set to
indicate that state cell $\ell$ has been used.

\subsection{Operational Semantics}
\label{sec:lamhook-opsem}

The semantics of $\lamhook$ extend those of $\lampure$ by including a heap
$\heap$. We elide the augmented $\lampure$ $\beta$-reduction rules as none
except \textsc{EC-Let} and \textsc{ET-Lift} interact with $\heap$.
We also introduce a series of new rules for computations
and a new step judgement $\step^\phi_p$ on \emph{bodies}.

\begin{figure}[ht]
    {\small
\headersig{Reduction Rules for Computations}{$\heap \semi M \step N \semi \heap'$}
\begin{reductions}
\heap \semi \mf{init\ component}\ \{D\}
  &\step& \eval{D}{\mf{Init}}{p}{D} \semi
    \heap, p \mapsto \{\mf{st}: \emptyset, \mf{rerender}: \mf{ff}\},
    & \tsc{EC-InitComp} \\
&&\multicolumn{2}{l}
{\quad \text{if } p\ \text{fresh}} \\

\heap \semi \mf{update}\ (\compInst{p}{D}{V})
  &\step& \eval{D}{\mf{Succ}}{p}{D} \semi
    \heap \mid \heap[p].\mf{rerender} \mapsto \mf{ff},
    & \tsc{EC-UpdateCompT} \\
&&\multicolumn{2}{l}
{\quad \text{if } \heap[p].\mf{rerender} = \mf{tt}} \\

\heap \semi \mf{update}\ (\compInst{p}{D}{V})
  &\step& \compInst{p}{D}{\mf{update}\ V} \semi \heap,
    & \tsc{EC-UpdateCompF} \\
&&\multicolumn{2}{l}
{\quad \text{if } \heap[p].\mf{rerender} = \mf{ff}} \\

\heap \semi \eval{D_0}{\phi}{p}{D}
  &\step& \eval{D_0}{\phi}{p}{D'} \semi \heap',
    & \tsc{EC-Eval} \\
&&\multicolumn{2}{l}
{\quad \text{if } \heap \semi D \step^\phi_p D' \semi \heap'} \\
\heap \semi \eval{D}{\phi}{p}{\mf{val}\ V}
  &\step& \compInst{p}{D}{\mf{init}\ V} \semi \heap
    & \tsc{EC-EvalDone} \\
\heap \semi \mf{click}_{id}\ (\compInst{p}{D}{V})
  &\step& \mf{click}_{id}\ V \semi \heap
    & \tsc{EC-ClickComp} \\
\heap \semi \setter^\ell_p\ V
  &\step& \mf{val}\ () \semi \heap',
   & \tsc{EC-AppSet} \\
&&\multicolumn{2}{l}
{\quad \text{where } \heap' = \heap \mid \heap[p].\mf{st}[\ell] \mapsto V\ \text{and}\ \heap[p].\mf{rerender} \mapsto \mf{tt}} \\
\end{reductions}

\headersig{Reduction Rules for Bodies}{$\heap \semi D \step^\phi_p D' \semi \heap'$}
\[
\renewcommand{\arraystretch}{1.3}
\begin{reds}
    \heap \semi \elet{x}{M}{D}
  &\step^\phi_p& \elet{x}{N}{D} \semi \heap',
    {\quad \text{if } \heap \semi M \step N \semi \heap'}
    & \tsc{EB-Let} \\
  \heap \semi \elet{x}{\mf{val}\ V}{D}
    &\step^\phi_p& D[V/x] \semi \heap
      & \tsc{EB-Val} \\
  \heap \semi \estate{x}{\ell}{V}{D}
    &\step^\mf{Init}_p& D[V/x, \setter^\ell_p/x_\mf{set}] \semi
                       \heap \mid \heap[p].\mf{st}[\ell] \mapsto V
      & \tsc{EB-StateI} \\
  \heap \semi \estate{x}{\ell}{V}{D}
    &\step^\mf{Succ}_p&
      D[\heap[p].\mf{st}[\ell]/x, \setter^\ell_p/x_\mf{set}] \semi \heap
      & \tsc{EB-StateS} \\
\end{reds}
\]
}
\vspace{-1em}
\caption{Operational semantics for $\lamhook$ computations and bodies,
extending $\lampure$.}
\Description{TODO:DESC}
\label{fig:lamhook-opsem}
\end{figure}

\paragraph{Computations}
Figure \ref{fig:lamhook-opsem} details all new reduction rules in $\lamhook$.
\textsc{EC-AppSet} reduces a setter application to a unit value
and updates the heap accordingly.
\textsc{EC-InitComp} creates a fresh pointer $p$ to a new instance memory
on the heap and begins the evaluation of $D$.
\textsc{EC-ClickComp} triggers click handlers in the body of a component
instance.

The purpose of the $\mf{update}$ construct is to rerender any component bodies
whose state has changed, and in $\lamhook$ has a nontrivial semantics.
To update a component instance at pointer $p$, if $\heap[p].\mf{rerender} =
\mf{tt}$, \textsc{EC-UpdateCompT} begins the \emph{re}-evaluation of $D$.
If $\heap[p].\mf{rerender} = \mf{ff}$, \textsc{EC-UpdateCompF} instead
propagates the update trigger to the child of the instance.

\textsc{EC-Eval} is a congruence rule that allows evaluation
in the body $D$ of $\eval{D_0}{\phi}{p}{D}$, and
\textsc{EC-EvalDone} embeds the result of evaluating a component into a
component instance.

\paragraph{Bodies}
The step judgment for a $\lamhook$ body is parameterized by a phase $\phi$ and a
pointer $p$, i.e. we evaluate bodies \emph{in a particular phase} and \emph{with
respect to a particular heap location}.  Rules \textsc{EB-Let} and
\textsc{EB-Val} ignore these extra parameters and use the standard
let-binding semantics.

The semantics of $\mf{useState}$ is determined by $\phi$. If $\phi=\mf{Init}$,
then \textsc{EC-StateI} applies, which binds the provided initial value to $x$
and simultaneously adds it to the heap at $(p, \ell)$.
If instead $\phi=\mf{Succ}$, then \textsc{EC-StateS} binds to $x$ the current
value on the heap at $(p, \ell)$, and leaves the heap unchanged.
In either case, we also bind a setter for location $(p, \ell)$ to $x_\mf{set}$.

\subsection{Metatheory}

\paragraph{Configuration Typing}
To reason about runtime configurations, we must describe
what it means for a heap to be well-typed.
Judgement $\Gamma \semi \Pi \semi \Phi \vDash \heap$
ensures that heap $\heap$ is consistent
with heap context $\Pi$: namely that the pointers in the
heap coincide with the pointers in the heap context, and that
the values contained in the state store for each instance have
types consistent with the heap context. Specifically,
$\Gamma \semi \Pi \semi \Phi \vDash \heap$
abbreviates $\Gamma \semi \Pi \semi \Phi \vDash \heap: \Pi$,
which is defined as follows:
    {\footnotesize
\begin{mathpar}
\inferrule[TS-Nil] {\mf{Succ}(\Phi)} {
  \Gamma \semi \Pi \semi \Phi \vDash \emptyset {:}\; \emptyset
}
\quad
\inferrule[TS-Cons] {
  \Gamma \semi \Pi \semi \Phi_1 \vDash \sigma: \Sigma
  \quad
  \Gamma \semi \Pi \semi \Phi_2 \vdash V: A
} {
  \Gamma \semi \Pi \semi \Phi_1 \Join \Phi_2
  \vDash \sigma, \ell \mapsto V: (\Sigma, \ell: A)
}
\quad
\inferrule[TH-Nil] {\mf{Succ}(\Phi)} {
  \Gamma \semi \Pi \semi \Phi \vDash \emptyset {:}\; \emptyset
}
\quad
\inferrule[TH-Cons] {
  \Gamma \semi \Pi_0 \semi \Phi_1 \vDash \heap: \Pi
  \quad
  \Gamma \semi \Pi_0 \semi \Phi_2 \vDash \pi.\mf{st}: \Ell
} {\Gamma \semi \Pi_0 \semi \Phi_1 \Join \Phi_2
  \vDash \heap, p \mapsto \pi: (\Pi, p: \Ell)
}
\end{mathpar}
}

We then define \emph{configurations} $\heap \semi M$
to be well-typed via
$\Gamma \semi \Pi \semi \Phi_H \semi \Phi_M \vdash \heap \semi M: A \bang E$,
which abbreviates that $\Phi_H \Join \Phi_M$ is defined and we also have
$\Gamma \semi \Pi \semi \Phi_M \vdash M: A \bang E$
and $\Gamma \semi \Pi \semi \Phi_H \vDash \heap$.
We extend a similar notion to bodies via judgment
$\Gamma \semi \Pi \semi \Phi_D \semi \Phi_M \vdash \heap \semi D: \Ell$.

\paragraph{Type Soundness}
$\lamhook$ enjoys syntactic type soundness, although the statements
and proofs of such are more involved than those of $\lampure$.

We prove progress for computations, for each phase of body evaluation,
and for the top-level step judgment.
When $\phi=\mf{Succ}$, we require the premise that all labels in the
body are also in the instance memory: this guarantees that
we only ever access values from the heap that have already been initialized.
Below we include only our top-level result.
As with $\lampure$, we can always receive an event, and so never terminate.
Appendix \ref{appdx:prog} contains all other statements and full proofs.

\begin{restatable}[Progress for $\step_\top$]{theorem}{lamhookprogtop}
\label{thm:lamhook-progtop}
If
$\emptyset \semi \Pi \semi \Phi_H \semi \Phi_M \vdash \heap \semi M: A \bang E$,
then $\heap \semi M \step_\top N \semi \heap'$ for some $N$ and $\heap'$.
\end{restatable}
\begin{proof}
  By case-analysis on the typing derivation
  with appeal to Theorems \ref{thm:lamhook-progc} (computations),
  \ref{thm:lamhook-progbi} (bodies for $\phi=\mf{Init}$),
  and \ref{thm:lamhook-progbs} (bodies for $\phi=\mf{Succ}$),
  which are shown by mutual induction.
\end{proof}

We similarly must prove preservation for all three main modes of evaluation
and the top-level step judgment.
Since the heap and phases evolve as we step, we state preservation with respect
to partial orders on both $\Pi$ and $\Phi$, which are given as follows.

\begin{definition}[Phase and phase context orders]\hfill
    \begin{itemize}
        \item \textbf{Phase order:} $\phi \sqsubseteq \phi'$ is the least partial order
satisfying $\mf{Init} \sqsubseteq \mf{Succ}$.
        \item \textbf{Phase context order:} $\Phi \sqsubseteq \Phi'$ holds iff
        $\Phi(p) \sqsubseteq \Phi'(p)$ for all $p \in \dom(\Phi)$.
        \item \textbf{Phase context order (fresh):} $\Phi \sqsubseteq_\mf{fresh}
            \Phi'$ holds
iff $\Phi \sqsubseteq \Phi'$ and all $p \in \dom(\Phi') \setminus \dom(\Phi)$
are fresh (with respect to a particular step).
    \end{itemize}
\end{definition}

The $\sqsubseteq_\mf{fresh}$ notion is necessary to show the following lemma,
since otherwise, a non-fresh extension to $\Phi_1$
could conflict with pointers in $\Phi_2$.

\begin{restatable}{lemma}{freshphi}
\label{thm:lamhook-freshphi}
If $\Phi = \Phi_1 \Join \Phi_2$ and $\Phi_1 \sqsubseteq_\mf{fresh} \Phi_1'$,
then $\Phi \sqsubseteq_\mf{fresh} \Phi_1' \Join \Phi_2$
(and $\Phi_1' \Join \Phi_2$ is defined).
\end{restatable}
\begin{proof}
  By case-analysis on an arbitrary $p \in \Phi$.
  We elide a full proof to Appendix \ref{appdx:prez}.
\end{proof}

Intuitively, the \emph{heap context order} only permits $\Pi$ to grow at a
pointer $p$ if $p$ is in initialization phase according to $\Phi$. Otherwise,
we require that $\Pi[p] = \Pi'[p]$.

\begin{definition}[Heap context order]
Define $\Pi \sqsubseteq_\Phi \Pi'$ to hold iff
\begin{enumerate}[label=(\roman*)]
  \item $\dom(\Pi) \subseteq \dom(\Pi')$,
  \item for all $p \in \dom(\Pi')$,
    either $p \in \dom(\Pi)$ or $\Phi[p] = \mf{Init}$, and
  \item for all $p \in \dom(\Pi)$,
    if $\Phi[p] = \mf{Init}$ then $\Pi[p] \subseteq \Pi'[p]$,
    and otherwise, $\Pi[p] = \Pi'[p]$.
\end{enumerate}

The order $\subseteq$ on label records in condition (iii) invokes type equality:
both are defined as expected.
\end{definition}

We can now state preservation.
For computations, we make separate statements for $\mf{Pure}$ and $\mf{Mut}$,
since we also want to show that pure computations do not change the heap,
i.e. typed purity implies \emph{operational} purity.
Here we again include only our top-level result: all other statements
and full proofs can be found in Appendix \ref{appdx:prez}.

\begin{restatable}[Preservation for $\step_\top$]{theorem}{lamhookpreztop}
\label{thm:lamhook-preztop}
If
$\Gamma \semi \Pi \semi \Phi_H \semi \Phi_M \vdash
\heap \semi M: \mf{ElemInst} \bang \mf{Mut}$ and
$\heap \semi M \step_\top N \semi \heap'$,
then for some $\Phi_N$, $\Phi_H'$, and $\Pi'$
such that $\Phi_M \Join \Phi_H \sqsubseteq_\mf{fresh} \Phi_N \Join \Phi_H'$
and $\Pi \sqsubseteq_{\Phi_N \Join \Phi_H'} \Pi'$, we have
$\Gamma \semi \Pi' \semi \Phi_H' \semi \Phi_N \vdash \heap' \semi N:
\mf{ElemInst} \bang \mf{Mut}$.
\end{restatable}
\begin{proof}
  By case-analysis on the step,
  appealing to Lemma \ref{thm:lamhook-prezcm}
  (computations for $E=\mf{Mut}$).
\end{proof}

\subsubsection{Consistency}

Our type system allows us to statically rule out a large swath of ``bad''
behavior present in React. For instance, the infinite rerender cycle
mentioned in Section \ref{sec:overview-problems} can never arise:
due to effect typing,
setters can never be called at the top level of a component body.

To formally characterize the problem, we introduce the notion of
\emph{consistency}.
Intuitively, a component $D$ with instance pointer $p$
is consistent with a heap $\heap$ if the evaluation sequence from $D$
to a value $V$ in $\mf{Succ}$-phase does not modify the heap.

\begin{definition}
A configuration $\heap \semi D$ is consistent with $V$ at $p$ iff
$\heap \semi D \steprt^\mf{Succ}_p\ \mf{val}\ V \semi \heap$.
\end{definition}

Importantly, if $\heap \semi D$ is \emph{inconsistent} but reduces to a value,
then the heap has changed, at which point React will retry evaluation
in order to avoid rendering stale state.
Consistency guarantees that all rendered state is ``consistent'' with the heap,
i.e. we never render stale state.
We need not model React's retry behavior in our semantics because our effect
system statically enforces consistency.
We prove that, if a well-typed configuration reduces to a value,
then it is consistent with that value.

\begin{restatable}[Component consistency]{theorem}{consistency}
\label{thm:lamhook-consistency}
If
$\Gamma \semi \Pi \semi \Phi_H \semi \Phi_D \vdash \heap \semi D: \Ell$ with
$\Ell \cap \Pi[p] = \emptyset$, $p \notin \dom(\Phi_H \Join \Phi_D)$, and
$\heap \semi D \steprt^\phi_p\ \mf{val}\ V \semi \heap'$,
then $\heap' \semi D$ is consistent with $V$ at $p$.
\end{restatable}
\begin{proof}
  By induction on the length of the step sequence,
  relying on type soundness and the fact that
  subcomputations in bodies are pure.
  Appendix \ref{appdx:consistency} contains a full proof.
\end{proof}

\begin{figure}[b]
    {\small
\header{Syntax}
\begin{minipage}{0.45\textwidth}
\begin{syntax}
\text{Types } A           &::=& \cdots \mid \mf{Ref}\ A \\
\text{Values } V, W       &::=& \cdots \mid \mf{component}\ V \\
                          & & \mid \mgrey{\mf{compInst}(V_1, V_2, V_3)} \mid \mgrey{\ell} \\
\end{syntax}
\end{minipage}
\begin{minipage}{0.45\textwidth}
\begin{syntax}
\text{Computations } e    &::=& \cdots \mid \mf{ref}\ V \mid \bang V \mid V := W \\
\text{Store } \sigma      &::=& \emptyset \mid \sigma, \ell \mapsto V \\
\text{Store type } \Sigma &::=& \emptyset \mid \Sigma, \ell \mapsto A
\end{syntax}
\end{minipage}

\headertwo
    {Typing Rules for Values and Computations}
    {\framebox{$\Gamma \semi \Sigma \vdash V: A$}~\framebox{$\Gamma \semi \Sigma \vdash M: A$}}
\begin{mathpar}
\inferrule[TV-Comp] {
  \Gamma \semi \Sigma \vdash V: \tone \rightarrow A_\mf{comp}
} {
  \Gamma \semi \Sigma \vdash \mf{component}\ V: \mf{Element}
}
\quad
\inferrule[TV-Loc] {
  \ell: A \in \Sigma
} {
  \Gamma \semi \Sigma \vdash \ell: \mf{Ref}\ A
}
\quad
\inferrule[TV-CompI] {
  \Gamma \semi \Sigma \vdash V_1: \tone \rightarrow \mf{Element}
  \\\\
  \Gamma \semi \Sigma \vdash V_2: \mf{Ref\ Bool}
  \and
  \Gamma \semi \Sigma \vdash V_3: \mf{ElemInst}
} {
  \Gamma \semi \Sigma \vdash \mf{compInst}(V_1, V_2, V_3): \mf{ElemInst}
}

\inferrule[TC-Ref] {
  \Gamma \semi \Sigma \vdash V: A
} {
  \Gamma \semi \Sigma \vdash \mf{ref}\ V: \mf{Ref}\ A
}

\inferrule[TC-Get] {
  \Gamma \semi \Sigma \vdash V: \mf{Ref}\ A
} {
  \Gamma \semi \Sigma \vdash \bang V: A
}

\inferrule[TC-Set] {
  \Gamma \semi \Sigma \vdash V: \mf{Ref}\ A
  \and
  \Gamma \semi \Sigma \vdash W: A
} {
  \Gamma \semi \Sigma \vdash V := W: \tone
}
\end{mathpar}

\headersig{Reduction Rules}{$\sigma \semi M \step N \semi \sigma'$}
\begin{reductions}
\sigma \semi \mf{ref}\;V
  &\step& \val{\ell} \semi \sigma, \ell \mapsto V, \quad (\ell\;\text{fresh})
   & \tsc{EC-Ref} \\
\sigma \semi \bang \ell
  &\step& \val{(\sigma[\ell])} \semi \sigma
    & \tsc{EC-Get} \\
\sigma \semi \ell := V
  &\step& \mf{val}\;() \semi \sigma[\ell \mapsto V]
    & \tsc{EC-Set} \\
\sigma \semi \mf{init}\;(\mf{component}\;V)
  &\step& \bl
          \Let\; rr, f, x = V\;() \;\In \\
          (\val{\mf{compInst}(rr, f, \mf{init}\;x)}) \semi \sigma \\
          \el
    & \tsc{EC-InitComp} \\
\sigma \semi \mf{update}\;(\mf{compInst}(V_1, V_2, V_3))
  &\step& \bl
          \Let\; x = \bl
                     \If\; \bang V_2 \\
                     \Then\; V_2 := \mf{false} \semi \mf{init}\;(V_1\;()) \\
                     \Else\; \mf{update}\;V_3 \;\In \\
                     \el \\
          (\val{\mf{compInst}(V_1, V_2, x)}) \semi \sigma \\
          \el
    &\tsc{EC-UpdateComp}\\
\sigma \semi \mf{click}_{id}\;(\mf{compInst}(V_1, V_2, V_3))
  &\step& \mf{click}_{id}\;V_3 \semi \sigma
   & \tsc{EC-ClickComp} \\
\end{reductions}
}
\caption{Additional syntax, typing rules, and reduction rules for $\lamref$.}
\label{fig:lamref}
\Description{TODO:DESC}
\end{figure}
\section[Lambda-React-Ref: A Canonical Stateful Calculus]
  {$\lamref$: A Canonical Stateful Calculus}
\label{sec:lamref}

Our second reactive calculus $\lamref$ extends $\lampure$ with components based
on \emph{ML-style references}.

\subsection{Syntax and Typing Rules}
Figure~\ref{fig:lamref} (top) shows the syntax and typing rules of $\lamref$.
$\mf{Ref}\ A$ is the type of references containing values of type $A$.
All previous typing judgments are straightforwardly extended with
the additional context $\Sigma$, so we concentrate on new constructs.

Stateful components are encoded in $\lamref$ via the $\mf{component}\ V$
constructor (with typing rule \textsc{TV-Comp}), where $V$ is a function of type
$\tone \rightarrow A_\mf{comp}$,
and $A_\mf{comp}$ abbreviates $(\tone \rightarrow \mf{Element})
\times \mf{Ref\ Bool} \times \mf{Element}$.
This function will close over any references describing state required by the
component and will return a triple. The three constituent types are as follows:
\begin{itemize}
    \item $(\tone \rightarrow \mf{Element})$ is a \emph{rendering function}
        that produces an $\mf{Element}$ based on the current
        state.
    \item $\mf{Ref\ Bool}$ is a \emph{rerender flag} that contains
        $\mf{true}$ if the component needs rerendered.
    \item $\mf{Element}$ is the initial element returned by the component.
\end{itemize}

Like in $\lamhook$, instances $\mf{compInst}$ (with typing rule
\textsc{TV-CompI}) are the realizations of components, and their three arguments
correspond to the three constituents of type $A_\mf{comp}$, \emph{except} that
the $\mf{Element}$ has been instantiated into an $\mf{ElemInst}$.
From left to right, the arguments represent the rerendering function,
the rerender flag, and the current sub-instance.

All other syntactic extensions concern references: $\ell$ ranges over
\emph{locations} that are generated fresh at runtime, and our three new
computations, $\mf{ref}\ V$, $\bang V$, and $V := W$, create, retrieve, and
update locations, respectively. Their typing rules are standard.
Rule \textsc{TV-Loc} types a location and is the only place we use $\Sigma$ in
the term typing judgment.
A store $\sigma$ is a map from locations to values,
and a store type $\Sigma$ is a map from locations to types.

\paragraph{Canonicity}
As with $\lampure$ (Section \ref{sec:lampure-syntax}),
elements and their instances---specifically, the $\mf{component}$
and $\mf{compInst}$ constructors---can be expressed as inductive datatypes.
$\mf{init}$ is still a fold on elements, and $\mf{update}$ is still a fold on
instances. Thus, although these new constructs and their rules may appear quite
ad-hoc, we maintain that they are straightforwardly expressible in terms of
more canonical constructs, and that the only \emph{truly} new feature
in $\lamref$ is references.

\subsection{Operational Semantics}
\label{sec:lamref-opsem}

Figure~\ref{fig:lamref} (bottom) shows the reduction rules for $\lamref$.
We extend both previous step judgments $\step$ and $\step_\top$
with a store $\sigma$, and amend all previous rules accordingly.
The first three rules (\textsc{EC-Ref}, \textsc{EC-Get}, and \textsc{EC-Set})
implement the standard semantics for general references.
The remaining rules add support for components and instances to the runtime:
due to the size of the reduction rules, we present them using syntactic sugar.
Rule \textsc{EC-InitComp} initializes a component $\mf{component}\ V$.
It evaluates the component's function to return $rr$ (the rendering function),
$f$ (the render flag), and $x$ (the $\mf{Element}$ returned by the component)
before instantiating $x$ and creating a $\mf{compInst}$.
Rule \textsc{EC-UpdateComp} handles updating a component instance.
Given a component instance $\mf{compInst}(V_1, V_2, V_3)$, the rule firstly
dereferences the render flag $V_2$. If this is set to $\mf{true}$ then the
component must be rerendered, so the rule will reset the flag and
reinitialize the element resulting from reinvoking the rendering function.
If the component does not need rerendering, then the rule will propagate the
update to the current realized element.
Finally, rule \textsc{EC-ClickComp} propagates a click past a component instance
and to its child instance.

\subsection{Metatheory}

$\lamref$ enjoys syntactic type soundness with respect to
both of its step judgments. As usual, we include only the statements
for $\step_\top$ here.
The store typing judgment $\Gamma \semi \Sigma \vDash \sigma$ holds iff
$\dom(\Sigma) = \dom(\sigma)$ and, for all $\ell \in \dom(\sigma)$,
we have $\Gamma \semi \Sigma \vDash \sigma[\ell]: \Sigma[\ell]$.

\begin{theorem}[Type Soundness for $\step_\top$]
\label{thm:lamref-sound-top}
\hfill
\begin{itemize}
    \item \textbf{Preservation:}
    If
      $\Gamma \semi \Sigma \vdash M: \mf{ElemInst}$ with
      $\Gamma \semi \Sigma \vDash \sigma$, and
      $\sigma \semi M \step N \semi \sigma'$,
    then for some $\Sigma' \supseteq \Sigma$, we have that
      $\Gamma \semi \Sigma' \vdash N: \mf{ElemInst}$ and
      $\Gamma \semi \Sigma' \vDash \sigma'$.
    \item \textbf{Progress:} If $\emptyset \semi \Sigma \vdash M : A$ and
    $\emptyset \semi \Sigma \vDash \sigma$,
    then $\sigma \semi M \step N \semi \sigma'$ for some $N$ and $\sigma'$.
\end{itemize}
\end{theorem}
\begin{proof}
  By case-analysis on $M \step N$ and $\emptyset \vdash M: A$ respectively,
  with appeal to Theorem~\ref{thm:lampure-sound}.
\end{proof}

\section{A Translation from Hooks to References}
\label{sec:trans}

In this section, we define a translation
from $\lamhook$ (the \emph{source} language)
to $\lamref$ (the \emph{target} language),
thereby demonstrating the expressibility of hook-indexed state in terms
of first-class references.
We first define our metatypes (\ref{sec:metatypes})
and the prerequisite machinery (\ref{sec:machinery})
before presenting the translation (\ref{sec:transmain}),
proving its formal correctness (\ref{sec:transcorrect}),
and then discussing its other properties (\ref{sec:transdisc}).

\subsection{Metatypes}
\label{sec:metatypes}

As a presentation aid, we often provide the \emph{metatypes}
of the functions presented in this section.

We construct metatypes via function arrows $\rightarrow$
and products $\times$, as usual.
We write $f(x)=\bot$ if a function result $f(x)$ is undefined
and we subscript a type with $\bot$ to indicate that it may be undefined.
We use $A \rightharpoonup B$ to abbreviate $A \rightarrow B_\bot$,
which describes a \emph{partial} function from $A$ to $B$.

Syntactic categories of the \emph{source} language
are all prefixed with $\mf{Hook}$.
They include $\mf{HookType}$ for types, $\mf{HookTypeCtx}$ for type contexts,
$\mf{HookValue}$ for values, $\mf{HookComp}$ for computations,
$\mf{HookBody}$ for bodies, $\mf{HookHeap}$ for heaps, $\mf{HookHeapCtx}$
for heap contexts, $\mf{HookPtr}$ for pointers, $\mf{HookLabel}$
for labels, and $\mf{HookID}$ for identifiers.
Syntactic categories of the \emph{target} language
are similarly prefixed with $\mf{Ref}$, and are analogous.
$\mf{RefLoc}$ refers to locations, $\mf{RefStore}$ refers to stores,
and $\mf{RefStoreCtx}$ refers to store contexts.
Notation $x : A$ means that $x$ has metatype $A$, and we use set inclusion
$X \subseteq A$ to say that all $x \in X$ have $x : A$.

\subsection{Machinery}
\label{sec:machinery}

\subsubsection{Runtime Environment}

Both the source and target languages dynamically generate
information at runtime, namely identifiers, pointers, and locations.
To show a semantic correspondence, we must not only be able to
translate this dynamic information: we must be able to assert the equality
of a fresh source-language construct with its fresh target-language counterpart.
We achieve this by parameterizing our translations
by a runtime environment $\R$ that witnesses the relation between dynamic
constructs.
The environment $\R$ consists of three partial functions: $\R_{id}$ to
relate identifiers, $\R_p$ to relate rerender flags, and $\R_\ell$ to
relate state cells.  We overload $\R$ to refer to each of these
functions unambiguously.  It suffices to represent $\R$ as an
inductively defined structure, where each function is defined with
respect to the contents of the structure.  {\small
\[\R ::= \emptyset \mid \R, {id} \mapsto {id}' \mid
\R, p \mapsto \ell \mid \R, (p, \ell) \mapsto \ell'\]

\begin{center}
$\begin{array}{c@{\qquad}c@{\qquad}c}
  \R_{id} : \mf{HookID} \rightharpoonup \mf{RefID} &
  \R_p : \mf{HookPtr} \rightharpoonup \mf{RefLoc} &
  \R_\ell : \mf{HookPtr} \times \mf{HookLabel} \rightharpoonup \mf{RefLoc}
\end{array}$

$\begin{array}{c@{\quad}c@{\quad}c}
  \R_{id}({id}) \triangleq
  \begin{cases}
    {id}' & {id} \mapsto {id}' \in \R\\
    \bot & \text{otherwise}
  \end{cases}
  &
  \R_p(p) \triangleq
  \begin{cases}
    \ell & p \mapsto \ell \in \R\\
    \bot & \text{otherwise}
  \end{cases}
  &
  \R_\ell(p, \ell) \triangleq
  \begin{cases}
    \ell' & (p, \ell) \mapsto \ell' \in \R\\
    \bot & \text{otherwise}
  \end{cases}
\end{array}$
\end{center}
\smallskip
}
\noindent
Environment $\R$ has metatype $\mf{RunEnv}$, an alias for the product
of the metatypes of its three parts.

\paragraph{Sufficiency}
If an invocation of $\R$ within a translation results in $\bot$,
we say the entire translation result is undefined,
i.e. our translation functions are likewise partial.
We therefore require a means of specifying when a runtime environment
is \emph{sufficient} to define a particular translation result.

\begin{definition}
For a set of identifiers $\mathbb{I} \subseteq \mf{HookID}$,
define $\Pi \semi \mathbb{I} \vDash \R$ to hold iff $\R$ is injective and:
\begin{center}
$\begin{array}{c@{\qquad}c@{\qquad}c}
  {id} \in \mathbb{I} \Rightarrow \R({id}) \neq \bot&
  p \in \dom(\Pi) \Rightarrow \R(p) \neq \bot &
  \ell \in \Pi[p] \Leftrightarrow \R(p, \ell) \neq \bot
\end{array}$
\end{center}
\end{definition}

Only the rightmost condition above requires bi-implication,
since the freshness of identifiers and pointers is always given directly
by the operational semantics, whereas that of pointer/label pairs
must be deduced from the typing judgments.

We also define the function $\mf{ids}$ on source-language terms and heaps
to extract the set of identifiers in the term or heap.
We elide the definition, which is standard,
and trivial except at $\mf{tagInst}_{id}(V, W)$.

\subsubsection{Location Variable Environment}
Our body translation frequently \emph{generates} fresh variables $x_\ell$ that
are let-bound to references and occur again later in the translation process.
The translation of a \emph{static} body
can essentially be achieved by a label-variable mapping.
However, this simple strategy is unable to reason about
\emph{intermediate evaluation stages} wherein $x_\ell$ has already been bound
and so should \emph{no longer} occur free in the rest of the translation.

To reason about these variables, we define another environment
$\delta$, which maps labels (and the distinguished rerender flag
$\mf{flag}$) to variables, without mention of a particular pointer
$p$. We use $\delta$ to specify two partial functions
\emph{parameterized by} $\R$, such that if there is no variable in
$\delta$ corresponding to the input, then we invoke $\R$
appropriately.
Intuitively, if a variable is in the range of $\delta$, it was bound
previously in the translation; if not, then we have some raw location
(given by $\R$) instead.
Environment $\delta$ has metatype $\mf{VarEnv}$, an alias for the
product of the metatypes of its two parts.  {\small
\[\delta ::= \emptyset \mid \delta, \mf{flag} \mapsto x
  \mid \delta, \ell \mapsto x_\ell\]

\begin{center}
$\begin{array}{c@{\quad}c}
  \delta(p, \R) : \mf{HookPtr}_\bot \times \mf{RunEnv}
    \rightharpoonup \mf{RefValue} &
  \delta(p, \ell, \R) : \mf{HookPtr}_\bot \times \mf{HookLab}
    \times \mf{RunEnv} \rightharpoonup \mf{RefValue}
\end{array}$

$\begin{array}{c@{\qquad}c}
  \delta(p, \R) \triangleq
  \begin{cases}
    x & \text{if } \mf{flag} \mapsto x \in \delta\\
    \R(p) & \text{otherwise}
  \end{cases}
  &
  \delta(p, \ell, \R) \triangleq
  \begin{cases}
    x & \text{if } \ell \mapsto x \in \delta\\
    \R(p, \ell) & \text{otherwise}
  \end{cases}
\end{array}$
\end{center}
}

\paragraph{Sufficiency}
As with $\R$, we require a notion of a \emph{sufficient} $\delta$
that leads to a well-defined translation.
We also concern ourselves here with (target-language) well-typedness,
and so we define $\delta$-sufficiency with respect to a label record $\Ell$.
We can achieve well-definedness through either the contents of $\delta$ or $\R$:
in the former case, we have a variable
(which should be typed by target-language context $\Gamma$),
and in the latter, we have a location
(which will be typed by the translation of $\Pi$).

\begin{definition}
Define $\Gamma \semi \Pi \vDash_p \delta: \Ell$ to hold iff:
\noindent
\begin{itemize}
\item for all $\ell: A \in \Ell$, \emph{either}
  $\ell \mapsto x_\ell \in \delta$
  and $x_\ell: \mf{Ref}\ \trans{A} \in \Gamma$, \emph{or}
  $\ell \notin \dom(\delta)$ and $\ell: A \in \Pi[p]$,
\item and \emph{either}
  $\mf{flag} \mapsto x_f \in \delta$
  and $x_f: \mf{Ref\ Bool} \in \Gamma$, \emph{or}
  $\mf{flag} \notin \dom(\delta)$ and $p \in \dom(\Pi)$.
\end{itemize}

If $p=\bot$, then the judgment may still hold
if the contents of $\delta$ are sufficient.
\end{definition}

\subsection{Translation}
\label{sec:transmain}
We can now present our translation in full. We
overload the translation brackets $\trans{\cdot}$ on each syntactic category of
$\lamhook$. We omit the standard translation cases.

\paragraph{Types and Contexts}
The translation on types is entirely standard, and is not parameterized by $\R$.
The only non-trivial case strips effect annotations from function arrows,
since $\lamref$ has no effect system.
The type context translation on $\Gamma$ is defined pointwise on type contexts,
and the phase context $\Phi$ has no translation because it has no
analog in $\lamref$.

\paragraph{Values, Computations, Bodies}
The translation on values is given by the homomorphic extension of the
following equations.
{\small
\begin{equations}
  \trans{\cdot} &:& \mf{HookValue} \to \mf{RunEnv}
    \rightharpoonup \mf{RefValue} \\
  \trans{\mf{tagInst}_{id}(V, W)}\R
    &\triangleq& \mf{tagInst}_{\R({id})}(\trans{V}\R, \trans{W}\R) \\
  \trans{\setter^\ell_p}\R
    &\triangleq& \lambda x.(\R(p, \ell) := x \semi \R(p) := \mf{true}) \\
  \trans{\mf{component}\ \{D\}}\R
    &\triangleq& \mf{component}\left(
      \lambda ().\elet{x}{\mf{ref}(\mf{false})}{
        \itransp{D}{D}{\bot}{\R}{\mf{flag} \mapsto x}
      }
      \right) \\
  \trans{\compInst{p}{D}{V}}\R
    &\triangleq& \valp{\mf{compInst}(
                    \lambda ().\stransp{D}{p}{\R}{\emptyset},
                    \R(p), \trans{V}\R}
                 ) \\
\end{equations}
}%
For tags, we lookup ${id}$ in $\R$ to obtain the corresponding
target-language identifier.
Primitive setters become functions that set the appropriate location
to the given value, determined by $\R(p, \ell)$, and also indicate that
a rerender should be performed by setting location $\R(p)$.
Components invoke the \emph{initial}-phase body translation
$\itransp{\cdot}{D}{p}{\R}{\delta}$.
Notably, $D$ is passed as parameter \emph{twice} to this translation,
which results in the code duplication discussed in
Section~\ref{sec:introduction}:
the bracketed $D$ is translated in \emph{initial} phase,
wherein at base case, the parenthesized $D$
is then translated in the \emph{successor} phase.
Beforehand, we also create a reference to hold a rerender flag, bound to
variable $x$, and pass this to the body translation via $\delta$.
Instances, our final notable values, contain the rerender flag $\R(p)$
and invoke the \emph{successor}-phase
body translation $\stransp{\cdot}{p}{\R}{\delta}$
to obtain a rerender function.

For computations, the two most interesting cases concern body evaluation.
{\small
\begin{equations}
\trans{\cdot} &:& \mf{HookComp} \to \mf{RunEnv}
  \rightharpoonup \mf{RefComp} \\
\trans{\eval{D_0}{\mf{Init}}{p}{D}}\R
  &\triangleq&
    \elet{rr, f, x}{\itransp{D}{D_0}{p}{\R}{\emptyset}}{
      \valp{\mf{compInst }(rr, f, \mf{init}\ x)}
    }\\
\trans{\eval{D_0}{\mf{Succ}}{p}{D}}\R
  &\triangleq&
    \valp{\mf{compInst}(
      \lambda ().\stransp{D_0}{p}{\R}{\emptyset},
      \R(p),
      \mf{init}\ (\stransp{D}{p}{\R}{\emptyset})
    )}
\end{equations}
}%
Initial-phase body evaluation ${\eval{D_0}{\mf{Init}}{p}{D}}$
invokes the initial-phase translation on $D$ but \emph{parameterized}
by $D_0$, and plumbs the results into an instance.
Successor-phase body evaluation $\eval{D_0}{\mf{Succ}}{p}{D}$
invokes the successor-phase translation \emph{twice}: once on $D_0$,
representing the original rerender function, and once on $D$,
representing the current evaluation state.

Finally, bodies are translated with two functions, corresponding to the two phases.
{\small
\[
\strans{\cdot} : \mf{HookBody}
  \to (\mf{HookPtr}_\bot \times \mf{RunEnv} \times \mf{VarEnv})
  \rightharpoonup \mf{RefComp}
\]
\begin{equations}
\stransp{\mf{val}\;V}{p}{\R}{\delta} &\triangleq& \mf{val}\;\trans{V}\R\\
\stransp{\elet{x}{M}{D}}{p}{\R}{\delta}
  &\triangleq& \elet{x}{\trans{M}\R}{\stransp{D}{p}{\R}{\delta}}\\
\stransp{\estate{x}{\ell}{V}{D}}{p}{\R}{\delta}
  &\triangleq& \bl
               \Let\; x = \bang \delta(p, \ell, \R) \;\In \\
               \Let\; x_\mf{set} = \valp{\lambda y.(\delta(p, \ell, \R) := y \semi \delta(p, \R) := \mf{true})} \;\In \\
               \stransp{D}{p}{\R}{\delta} \\
               \el \\
\end{equations}

\[
\itrans{\cdot} : \mf{HookBody}
  \to (\mf{HookBody} \times \mf{HookPtr}_\bot
    \times \mf{RunEnv} \times \mf{VarEnv})
  \rightharpoonup \mf{RefComp}
\]
\begin{equations}
\itransp{\mf{val}\;V}{D_0}{p}{\R}{\delta}
  &\triangleq& \mf{val}\;(\lambda ().\stransp{D_0}{p}{\R}{\delta}, \delta(p, \R), \trans{V}\R) \\
\itransp{\elet{x}{M}{D}}{D_0}{p}{\R}{\delta}
  &\triangleq& \elet{x}{\trans{M}\R}{\itransp{D}{D_0}{p}{\R}{\delta}} \\
\itransp{\estate{x}{\ell}{V}{D}}{D_0}{p}{\R}{\delta}
  &\triangleq& \bl
               \Let\; x_\ell = \mf{ref}(\trans{V}\R) \;\In \\
               \Let\;x = \bang x_\ell \;\In \\
               \Let\;x_\mf{set} = \valp{\lambda y.x_\ell := y \semi \delta(p, \R) := \mf{true}} \;\In \\
               \itransp{D}{D_0}{p}{\R}{ (\delta,\ell \mapsto x_\ell) } \\
               \el \\
\end{equations}
}
Let-binding in both phases is standard, as is the value case
in the successor phase. However, in the initial phase, alongside the given value
we return two extra values. To obtain a rerender function, we invoke the
successor-phase translation on the additional parameter $D_0$ and thunk the
result.  We also directly return the rerender flag, given by $\delta(p, \R)$.

For $\mf{useState}$ in the successor phase, we \emph{retrieve} from the
location $\delta(p, \ell, \R)$ and bind it to $x$. By contrast,
in the initial phase, $x$ is (indirectly) bound to the initial value
$\trans{V}\R$ by way of the reference bound to fresh variable $x_\ell$.
As previously explained,
this fresh variable is added to $\delta$ in the continuation
of the body translation for future reference.
In either phase, we also define an appropriate setter function
via the $\delta$-function, and bind it to $x_\mf{set}$.

\subsection{Correctness}
\label{sec:transcorrect}


To demonstrate the correctness of our translation,
we prove that our translation \emph{preserves types}
and that we maintain a \emph{tight operational correspondence}
\cite{wadler:22:plfa,gorla:10:encode,griesemer:20:go,
  tang:26:metx,minamide:96:closure,plotkin:75:lc}.
This operational correspondence is a straightforward yet strong result:
we do not have complications such as administrative redexes to account for,
and so more semantic (and often more involved) notions of translation
correctness
\cite{patrignani:21:rec,ahmed:11:cps, sabry:92:cps,gorla:10:encode}
are unnecessary for our purposes.

\paragraph{Heaps and Heap Contexts.}
To reason about the metatheory of the translation we first need to relate
$\lamhook$ heaps / heap contexts and $\lamref$ store / store contexts.
The translation from $\lamhook$ heaps to $\lamref$ stores
is given below.  The expression $\trans{\heap[p].\mf{rerender}}$ invokes an
obvious translation from source-language meta-booleans to target-language
boolean values.
{\small
\begin{equations}
\trans{\cdot} &:& \mf{HookHeap} \to \mf{RunEnv} \rightharpoonup \mf{RefStore}\\
\trans{\heap}\R &\triangleq& \bigcup_{p \in \heap} \left(
  \{\R(p) \mapsto \trans{\heap[p].\mf{rerender}}\}
  \cup \bigcup_{\ell \in \heap[p].\mf{st}} \{
    \R(p, \ell) \mapsto \trans{\heap[p].\mf{st}[\ell]}\R
  \}
\right)
\end{equations}
}
The translation on heap contexts is parameterized by $\R$,
and is given as follows, wherein locations
are assigned either type $\mf{Bool}$ (for rerender flags)
or type $\trans{\Pi[p][\ell]}$ (for state cells):
{\small
\begin{equations}
\trans{\cdot} &:& \mf{HookHeapCtx} \to \mf{RunEnv}
  \rightharpoonup \mf{RefStoreCtx} \\
\trans{\Pi}\R &\triangleq& \bigcup_{p \in \Pi}\left(
  \{\R(p) \mapsto \mf{Bool}\}
  \cup \bigcup_{\ell \in \Pi[p]}
    \{\R(p, \ell) \mapsto \trans{\Pi[p][\ell]}\}
\right)
\end{equations}
}
\subsubsection{Type Preservation}
Well-typed constructs in $\lamhook$ translate to well-typed constructs (of the
same type) in $\lamref$.  We prove four mutually recursive theorems on
terms, as well as a fifth theorem on heaps.
The statements on values, computations, and heaps are fairly standard:
we require as an extra premise in each that our given $\R$ is sufficient.
For bodies, we prove that their translations have specific types,
independent of the source-language type $\Ell$. Each body theorem uses $\Ell$
to assert the sufficiency of $\delta$ as an additional premise.

Note that all type preservation theorems ignore
phase context $\Phi$ and computational effects $E$,
which both exist for specific source-language guarantees
and have no target-language analog.

We show the most interesting theorem statement here: the case for bodies
when $\phi=\mf{Init}$. All others and all proofs are located in Appendix
\ref{appdx:transprez}.

\begin{restatable}[Type preservation, bodies, initial]{theorem}{transprezbi}
\label{thm:trans-prezbi}
If

\begin{minipage}{0.45\textwidth}
\begin{enumerate}[label=(\roman*)]
\item $\Gamma \semi \Pi \semi \Phi_1 \vdash D: \Ell$,
\item $\Gamma \semi \Pi \semi \Phi_2 \vdash D_0: \Ell_0$
  (where $\Ell \subseteq \Ell_0$),
\end{enumerate}
\end{minipage}
\begin{minipage}{0.45\textwidth}
\begin{enumerate}[label=(\roman*)]
\setcounter{enumi}{2}
\item $\Pi \semi \mf{ids}(D) \cup \mf{ids}(D_0) \vDash \R$, and
\item $\Gamma' \semi \Pi \vDash_p \delta: (\Ell_0 \setminus \Ell)$,
\end{enumerate}
\end{minipage}

then $\trans{\Gamma}, \Gamma' \semi \trans{\Pi}\R \vdash
  \itransp{D}{D_0}{p}{\R}{\delta} : A_\mf{comp}$.
\end{restatable}
\begin{proof}
  By mutual induction on the typing derivation \emph{(i)}
  alongside Theorems \ref{thm:trans-prezv} (values),
  \ref{thm:trans-prezc} (computations),
  and \ref{thm:trans-prezbs} (bodies at successor phase).
\end{proof}

This theorem guarantees that the translation is
well-defined in the following sense.


\begin{definition}[Well-definedness]
Call the result of a translation function $\trans{\cdot}\R$ \emph{well-defined}
iff it is a syntactically well-formed term in $\lamref$.

By definition of our translation functions, for a result to be well-defined,
it suffices that all instances of $\R({id})$, $\R(p)$, and $\R(p, \ell)$
that appear in the result are well-defined, that is, $\neq \bot$.
\end{definition}

Well-definedness is an obvious corollary of
well-typedness in the target language. Therefore, for a translation
to be well-defined on e.g. a value $\trans{V}\R$,
by Theorem \ref{thm:trans-prezv} (type preservation), it suffices to have that
$\Gamma \semi \Pi \semi \Phi \vdash V: A$ and $\Pi \semi \mf{ids}(V) \vDash \R$.

\subsubsection{Simulation}
We must now show that the translation is semantics-preserving:
roughly, the property that if a term (e.g. $M$) takes a step to $N$ in
$\lamhook$, then its translation $\trans{M}$ steps to the translation
$\trans{N}$ in $\lamref$.
However, several technical details serve to obscure the essence of
this idea. As well-definedness of our translation functions is
contingent on a runtime relation $\R$, we must have as premise that
$\R$ is sufficient. Since this sufficiency is defined with respect to
$\Pi$, it must be $\Pi$ types the source-language term
(and heap) in question.
Moreover, we require that $\R$ grows as we step (according to a
partial order given by $\subseteq$), as does $\Pi$.
To apply these theorems transitively, we must ensure that our new
$\R'$ is sufficient with respect to our new $\Pi'$, and in turn that
$\Pi'$ types the new term and heap. Therefore, the structure of our
proofs mirrors that of $\lamhook$ type preservation.

We present only the main statement for $\step_\top$ here.
The premises and conclusions \emph{not} found in the corresponding
preservation statement are \blue{typeset in blue}.
As usual, all other statements and all proofs can be found in
Appendix \ref{appdx:transsim}.

\begin{restatable}[Simulation, computations]{theorem}{transsimcm}
\label{thm:trans-simc}
\;\\[2px]
\begin{minipage}{0.45\linewidth}
\setlength{\parindent}{1em}
If we have
\begin{enumerate}[label=(\roman*)]
\item $\Gamma \semi \Pi \semi \Phi_H \semi \Phi_M \vdash
  \heap \semi M: A \bang \mf{Mut}$,
\item \blue{$\Pi \semi \mf{ids}(M) \cup \mf{ids}(\heap) \vDash \R$}, and
\item $\heap \semi M \step N \semi \heap'$,
\end{enumerate}
then for some $\Phi_H'$, $\Phi_N$,
$\Pi'$, and $\R' \supseteq \R$, we get:
\end{minipage}
\hfill
\begin{minipage}{0.5\linewidth}
\setlength{\parindent}{1em}
\begin{enumerate}[label=(\roman*)]
\item $\Phi_M \Join \Phi_H \sqsubseteq_\mf{fresh} \Phi_N \Join \Phi_H'$,
\item $\Pi \sqsubseteq_{\Phi_N \Join \Phi_H'} \Pi'$,
\item $\Gamma \semi \Pi' \semi \Phi_H' \semi \Phi_N \vdash N: A \bang \mf{Mut}$,
\item \blue{$\Pi' \semi \mf{ids}(N) \cup \mf{ids}(\heap') \vDash \R'$}, and
\item \blue{$\bsq{\heap}\R \semi \bsq{M}\R
\stept\ \bsq{N}\R' \semi \bsq{\heap'}\R' $}.
\end{enumerate}
\end{minipage}
\end{restatable}
\begin{proof}
  By induction on $\heap \semi M \step N \semi \heap'$,
  with appeal to Lemmas
  \ref{thm:trans-simbi} (bodies for $\phi=\mf{Init}$)
  and \ref{thm:trans-simbs} (bodies for $\phi=\mf{Succ}$).
\end{proof}

As the semantics of the target language is deterministic we
immediately obtain a backwards simulation result too~\citep{Leroy09}
(a target-language step must lie along the unique path given
by the simulation of one step in the source language, thus showing that
translated terms do not ``go astray''). Hence there is
a tight operational correspondence between the source and target
languages.






\subsection{Discussion}
\label{sec:transdisc}

%
\paragraph{Compositionality}
The definition of our translation is entirely compositional
\cite{gorla:10:encode,Felleisen91}, except
for one small deviation: when we translate a component,
we parameterize the $\mf{Init}$-phase translation by the body \emph{twice}.
The first parameter serves as the structural basis for the translation
as normal.
The second parameter remains entirely
untouched until it is translated for a second time in $\mf{Succ}$-phase:
thus, the $\mf{val}\; V$ case in the $\mf{Init}$-phase translation violates
compositionality.

One can make the translation formally compositional using standard
techniques such as having the translation of bodies return a pair of
terms, one for $\mf{Init}$-phase and another for $\mf{Succ}$-phase, and
using a higher-order translation to avoid introducing administrative
redexes~\citep{DanvyMN07}.
We choose not to follow such a path here as it makes the definition of
the translation somewhat more complex.

\paragraph{Locality}
We distinguish two senses in which the word "local" has been used
in the literature to describe translations.
We call the weaker sense \emph{feature-local},
and it applies with respect to a  particular \emph{feature},
which we take to be a group of related constructs and their type formers.
For instance, the "product feature"
consists of the type $A \times B$, pairs $(V, W)$, and projections
$\mf{fst}$ and $\mf{snd}$.
Feature-locality then means that the translation is a homomorphism over
constructs which are not part of the feature in question
\cite{kawahara:20:aeco, lindley:25:products}.
We call the stronger sense \emph{macro-local},
which is synonymous with \emph{macro-expressiveness}
as originally formulated by Felleisen \cite{Felleisen91}.
Macro-locality implies feature-locality and additionally requires that
the translations on features are specified by a fixed context,
parameterized only by the translations of the subterms.
These notions coincide for unparameterized translations
but may diverge for parameterized ones,
wherein a translation context may involve the parameters.
For instance:
\begin{itemize}
  \item A CPS translation \cite{plotkin:75:lc,DanvyMN07} is not feature-local
    (and so not macro-local),
    since we must thread a continuation through every computation
    in the language.
  \item A translation from named variables to de Bruijn indices
    \cite{debruijn:72:nameless} is feature-local
    but not macro-local. Translating variables requires ``non-local''
    information about any surrounding binders, but unrelated features
    (e.g. application, products if we add them, etc.) are unaffected.
  \item Translating $\elet{x}{M}{N}$ into
    $(\lambda x.N)\; M$, as in \cite{Felleisen91},
    is macro-local (and so feature-local).
\end{itemize}


Our translation is clearly \emph{not} macro-local: for instance,
when translating components, $\delta$ is used to thread ``non-local''
information between phases.
However, the translation \emph{is} feature-local, so long
as we take our ``feature'' to consist of
component values (and their instances), setters,
component evaluation computations, and the syntactic category of bodies.

\paragraph{Duplication}
A drawback of the translation is the duplication of code.
The source language allows the same component body $D$ to be interpreted
in both phases, but to reflect this behavior in the target language,
we translate $D$ twice, once per phase.
A programmer seeking to replicate source-language behavior in the target
language would repeat themselves frequently.

It is \emph{in theory} possible to remove this duplication
by parameterizing target-language components appropriately
(e.g. having them take as argument a ``phase'' and the necessary context
to run under that phase). We choose not to do so in order to avoid
obscuring the essence of the translation with (even) more plumbing.

\paragraph{Abstraction}
The $\mf{useState}$ construct supports UI state in quite a general
way, but also abstracts over implementation details. For instance,
setters in $\lamhook$ encapsulate the rerender flag, whereas a programmer
working in the $\lamref$ has direct access to the flag, and so can
manipulate it ``incorrectly'', as in the following example:

\begin{lstlisting}
  let flag = ref(false) in ...
  let setX = \y. (xLoc := y ; flag := true) in ...
  <div onClick={\().setX (x+1); flag := false}>{x}</div>
\end{lstlisting}

By manually setting $\mf{flag} := \mf{false}$ in the event handler,
a subsequent $\mf{update}$ will not rerender
the component, despite its state having changed.
In contrast, the encapsulation in $\lamhook$ prevents us from
rendering stale state in this way:
the only time the rerender flag can be set to $\mf{ff}$
is in rule \textsc{EC-UpdateCompT}, once a rerender has already begun.
Our translation thus does not preserve \emph{full abstraction},
in the sense that translated terms can be observably distinguished
by target-language contexts (where these contexts need not
be in the translation image themselves)
\cite{kennedy:06:net,ahmed:11:cps,abadi:98:protect}.

\section{Extensions}
\label{sec:ext}

We show how $\lamhook$ can be extended to capture additional
features and semantic details from React beyond those of Section \ref{sec:lamhook}.
In particular, we address
reconciliation (\ref{sec:reconciliation}),
$\mf{useEffect}$ (\ref{sec:useeffect}),
and custom hooks (\ref{sec:custom}).
We only sketch the necessary adjustments, and briefly indicate
how our translation approach
might likewise be adapted to model the new behavior in $\lamref$.

\subsection{Reconciliation}
\label{sec:reconciliation}

React's runtime involves a process called \emph{reconciliation}, where
it will \emph{diff} the old and new trees and only reinitialize the
parts that have changed \citep{lee:25:react,react:calls}.
Reconciliation helps performance but also affects the preservation of
state.
When a component instance is updated in $\lamhook$, which lacks reconciliation,
its sub-elements are always freshly reinitialized, discarding their state.
Consider the following \texttt{NestedCounter} component that wraps the stateful
\texttt{Counter} example seen previously:
\begin{lstlisting}
  let NestedCounter = component { let x, setX = useState(0) in
  <div onClick={\().setX (x+1)}>{x}<Counter/></div> }
\end{lstlisting}
The inner counter will increment as normal on click,
but when the \emph{outer} counter is incremented,
the inner counter will reset to 0.
In contrast, a reconciliation semantics will preserve the inner counter's state,
even when the outer counter's state changes.

\paragraph{Extension Approach.}
Extending $\lamhook$ with reconciliation requires carrying
the prior element contents during rerendering so that, when finished,
the semantics can choose whether to reinitialize. Thus the body-evaluation
construct $\lceil D_0 \rceil \downarrow_p^\phi D$
becomes $\lceil D_0 \rceil \downarrow_p^\phi (V_{\mf{old}}, D)$,
with the $V_{\mf{old}}$ providing the prior contents.
When rerendering terminates at a value
$\lceil D_0 \rceil \downarrow_p^\phi (V_{\mf{old}}, \mf{val}\ V_{\mf{new}})$,
it triggers a $\mf{reconcile}$ process that compares $V_{\mf{old}}$
with $V_{\mf{new}}$, either intializing new instances or preserving the
identity of tags/components shallowly and recursing to their children.



\paragraph{Translation Approach.}
The translation poses no new problems: the extension has essentially only added
some data and some primitive-recursive functions.
Like the other runtime computations, we treat
$\mf{reconcile}$ as primitive in the target-language semantics,
since it can be encoded straightforwardly in a language with
recursion and appropriate destructors.

\subsection{The useEffect Hook}
\label{sec:useeffect}
React provides another primitive hook, $\mf{useEffect}$, to
synchronize an application with non-React code, e.g. the DOM
or an external API.
The $\mf{useEffect}$ hook wraps a computation $M$, which is (confusingly) named
an ``Effect'': this capitalized use of the word ``effect'' should \emph{not} be
confused with the effects tracked in our type system
($\mf{Pure}$ and $\mf{Mut}$).
The Effect is executed \emph{after} rendering.
A common use case is to dispatch an \emph{asynchronous}
function, so that the long-running task does not block rendering,
as in the following JavaScript snippet:
\begin{lstlisting}
  const DataView = ({ dataID }) => { const x, setX = useState(null);
    useEffect(() => { const startFetch = async () => { const data = await fetchData(dataID);
                                                       setX(data); };
                      startFetch(); }); }
\end{lstlisting}
When a setter is called in an Effect, React rerenders the component
to display the most recent state.
However, if the setter is called unconditionally, we risk an infinite
loop, wherein this rerender dispatches another Effect, which triggers
another rerender, and so on.
Nevertheless, to express useful code like the example above, we must permit
the use of setters inside Effects, and so our semantics must capture this
looping behavior.

\paragraph{Extension Approach.}
We introduce an \emph{effect queue} $Q ::= \overline{M}$,
which is initially empty, to our runtime configuration for bodies.
Our new programmer-level construct, $\mf{useEffect}\; \{M\}\; \mf{in}\; D$,
adds the effect $M : \tone \bang \mf{Mut}$
to the queue (where $\concat$ is meta-append):
\begin{center}
{\small
$\ba{c@{\;\step^\phi_p\;}c}
\langle \heap \semi Q \semi \mf{useEffect}\; \{M\}\; \mf{in}\; D\rangle
  &
  \langle \heap \semi Q \concat [M] \semi D\rangle
\ea$
}
\end{center}
We also extend our runtime constructs
for body-evaluation and component instances to capture this queue.
To traverse the element tree and flush all queues,
we add a new runtime computation
$\mf{commit} : \mf{ElemInst} \rightarrow \mf{ElemInst}$, with the following
relevant rules:
{\small
\begin{center}
$\ba{r@{\;\step\;}l}
  \heap \semi \mf{commit}\; \mf{compInst}_p(D, M :: Q, V)
  &
  (M \semi \mf{commit}\; \mf{compInst}_p(D, Q, V)) \semi \heap \\
  \heap \semi \mf{commit}\; \mf{compInst}_p(D, [], V)
  &
  \mf{compInst}_p(D, [], \mf{commit}\; V) \semi \heap \\
\ea$
\end{center}
}

Finally, to ensure committed state changes are rendered, we add
a $\mf{check}$ runtime construct which repeatedly rerenders its argument
until no flags have been flipped in the process.
\begin{center}
{\small
$\ba{r@{\;\step\;}l@{\quad}c}
  \heap \semi \mf{check}\; V
  &
  \mf{check}\; (\mf{commit} \; (\mf{update}\; V)) \semi \heap
  & (\text{if } \exists p \in \heap . \heap[p].\mf{rerender} = \mf{tt})\\
  \heap \semi \mf{check}\; V
  &
  \mf{val}\; V \semi \heap
  & (\text{if } \forall p \in \heap . \heap[p].\mf{rerender} = \mf{ff})\\
\ea$
}
\end{center}


\paragraph{Translation Approach.}
Effects $M$ translate as thunks, $\lambda ().M$, and bodies into
\emph{state-passing style}, the state being a list of effects. Then
$\mf{useEffect}$ translates as an enqueue on this state, and the
final queue is saved into an (extended) \textsf{compInst} constructor.
Runtime computations $\mf{commit}$ and $\mf{check}$ then behave analogously
to the source language.


\subsection{Custom Hooks}
\label{sec:custom}
To support modularity, React allows programmers to define
``custom hooks'', which amount to \emph{higher-order} hooks.
Custom hooks are functions which may invoke other primitive or custom
hooks in sequence, as long as those invocations obey the top-level restriction
with respect to the custom-hook definition.
For example, we can abstract over the stateful logic
in our \texttt{Counter} components by creating a custom
\texttt{useCounterState} hook that exposes a specialized setter to increment the
counter.

\vspace{-8pt}
\begin{minipage}[t]{0.5\textwidth}
\begin{lstlisting}
const useCounterState = () => {
  const [x, setX] = useState(0);
  const incrX = () => {setX (x + 1)};
  return [x, incrX]; }
\end{lstlisting}
\end{minipage}
\hfill
\begin{minipage}[t]{0.5\linewidth}
\begin{lstlisting}
const TwoCounters = () => {
  const [x, incrX] = useCounterState();
  const [y, incrY] = useCounterState();
  ... }
\end{lstlisting}
\end{minipage}
%
%
%

\paragraph{Extension Approach.}
In the untyped setting, custom hooks are merely functions
that may call other hooks \cite{lee:25:react}.
Ensuring that custom hooks are \emph{safe} is more complicated.

Our previous shortcut of identifying $\mf{useState}$ calls with static labels
no longer works: otherwise, the two calls to \texttt{useCounterState}
above would share state! We now let $\Ell$ be a type-level list,
where the list order reflects the state initialization order.
State labels are allocated dynamically via a counter $n$ that starts at $0$,
which we add to our body-evaluation configuration:
this counter corresponds to an index in $\Ell$,
so that a setter $\setter^n_p$ takes type $\Ell[n]$.
For instance, the $\mf{useState}$ rule when $\phi=\mf{Init}$
would be as follows, with the significant changes \grey{greyed}:
{\small
\[
  \langle \heap \semi \mgrey{n} \semi \mf{useState}\; V \rangle
  \step^\mf{Init}_p
  \langle
    \heap \mid \heap[p].\mf{st}[\mgrey{n}] = V
    \semi \mgrey{n+1} \semi (V, \setter^\mgrey{n}_p)
  \rangle
\]
}
We replace our syntactic category of component bodies with
the type $\mf{Hook}\; \Ell\; A$, which denotes a hook that
returns type $A$ and initializes state according to $\Ell$.
Equipped with a standard operational semantics for return/bind,
$\mf{Hook}$ forms a \emph{graded monad} \cite{katsumata:14:pmonad,
fujii:16:gmonad,gaboardi:16:grading,orchard:14:effhaskell},
graded by the list monoid $\Ell$.
We add $\mf{useState}$ as a primitive operation with type
$A \xrightarrow{\mf{Mut}}
\mf{Hook}\; [A]\; (A \times (A \xrightarrow{\mf{Mut}} \tone))$,
and we adapt our $\mf{component}$ constructor
to take a hook of underlying type $\mf{Element}$ (and any grade).
The grade of a hook decreases during evaluation, but we can state
preservation with respect to the heap type as before.
{\small
\begin{mathpar}
  \inferrule {
    \Gamma \semi \Pi \semi \Phi \vdash M :
    \mf{Hook}\; \Ell\; \mf{Element} \bang \mf{Pure}
  } {
    \Gamma \semi \Pi \semi \Phi \vdash \mf{component}\; \{M\} : \mf{Element}
  }

  \inferrule {
    \Gamma \semi \Pi \semi \Phi \vdash M : A \bang E
  } {
    \Gamma \semi \Pi \semi \Phi \vdash
    \mf{return}\; M : \mf{Hook}\; \epsilon\; A \bang E
  }

  \inferrule {
    \Gamma \semi \Pi \semi \Phi \vdash M : \mf{Hook}\; \Ell_1\; A \bang E
    \and
    \Gamma, x : A \semi \Pi \semi \Phi \vdash N : \mf{Hook}\; \Ell_2\; B \bang E
  } {
    \Gamma \semi \Pi \semi \Phi \vdash
    \mf{bind}\; x \leftarrow M\; \mf{in}\; N
    : \mf{Hook}\; (\Ell_1 \circ \Ell_2)\; B \bang E
  }
\end{mathpar}
}

\paragraph{Translation Approach.}
Our approach to translating custom hooks is similar to that of components.
Every hook is first translated in $\mf{Init}$-phase, and must return
its $\mf{Succ}$-phase interpretation alongside its return value.
A hook $M$ which invokes a custom hook $N$ (in a $\mf{bind}$)
in $\mf{Init}$-phase then receives $N$'s $\mf{Succ}$-phase interpretation,
which must in turn be invoked appropriately in $M$'s own
$\mf{Succ}$-phase interpretation.
The rerender flag is still only generated at the top of a component,
but we must pass it as a parameter to any custom hooks,
since it must be in scope of any $\mf{useState}$ calls.

\section{Related Work}
\label{sec:relwork}

\citet{madsen:20:react} introduced a small-step semantics to model an earlier
class-based version of React. Their semantics models component lifecycles and
React's reconciliation process, but they treat state through a \texttt{SetState}
primitive that predates hooks.

Closest to our work is React-tRace~\cite{lee:25:react}, a detailed semantics for
React hooks that aims to faithfully model React's behaviour. Although we reuse
some surface syntax and several technical ideas (e.g., separating the $\mf{Init}$ and $\mf{Succ}$
phases), our goal is different: to distil a well-behaved language that
captures the essence of hook-indexed state and to characterise its expressiveness
via a translation to first-class references.
There are therefore many technical differences between $\lamhook$ and
React-tRace. First, our stratified syntax and type-and-effect system
rule out many erroneous programs that React-tRace must model. This
enables significant simplifications, allowing us to use a substitution-based
(rather than environment-based) semantics and to omit React's retry behaviour.
Since our focus is the expressiveness of state hooks, we also de-emphasize
reconciliation and $\mf{useEffect}$, though we outline extensions that
incorporate both. Finally, our semantics uses small-step evaluation for
expressions, whereas React-tRace uses a big-step relation.

Flapjax~\cite{meyerovich2009flapjax} is an early reactive web language with
support for signals (time-varying values) and event streams from functional
reactive programming (FRP)~\cite{ElliottH97}.  Flapjax propagates signal changes
to the DOM, but has a more imperative flavor because it predates
reconciliation-based approaches.

\citet{Chlipala15} and~\citet{FowlerDG15} describe related approaches to
reactive web programming where users can create \emph{views} of a mutable
data source, compose views monadically, and embed views into a declarative
description of a web page.  Updates to the data source are propagated through a
dataflow graph and are eventually reflected in the web page.
This is a different model to hook-indexed state and neither work includes
a formal semantics.

Elm~\cite{elm-website} is a functional web programming language. Early versions
of Elm~\cite{czaplicki:13:elm} used an applicative form of FRP that disallowed higher-order signals in
order to avoid memory leaks.
Elm later abandoned FRP~\cite{farewell-frp} and adopted a simpler architecture
known as \emph{The Elm Architecture} or \emph{Model-View-Update} (MVU). MVU
involves a simple event loop: a model is rendered by a view function, and the
materialized view can produce messages that are processed by an update function
to produce an updated model.
In general, MVU is simpler than React, but can lead to larger and more
monolithic data models.
\citet{Fowler20} formalizes MVU as a core calculus, $\lambda_{\mathsf{MVU}}$,
and shows how it can be extended to incorporate linear resources.
The functional IO system presented by~\citet{FelleisenFFK09} also implements a
form of MVU that works with external events (e.g., mouse movements).

\citet{crichton:24:doc} introduce a taxonomy of \emph{document languages}:
programming languages designed to produce structured text.  Their article
calculus, based on System F, can represent hierarchically-structured documents
(including HTML) and is given a translational semantics to data types similar to
our $\mf{tag}$ and $\mf{text}$ constructs. The reactive extension to the
calculus supports MVU-style components that include initialization, update, and
rendering functions. Their design elegantly represents reactive documents in a
functional language but does not model state hooks: instead, state updates are
defined with respect to a pre-defined list of events.

When it comes to understanding a feature such as state
via translation, there exist multiple methodologies in the literature.
Pottier \cite{pottier:11:store} demonstrates that System F with
first-class references is expressible in a calculus with recursive kinds.
He defines both languages and their semantics separately, and then
proves that his translation preserves types and is a simulation.
By contrast, when O'Hearn and Reynolds \cite{ohearn:00:algol}
translate fragments of Algol into a linear $\lambda$-calculus,
they do not give their Algol fragments a direct semantics:
instead, their source languages are \emph{given semantics via translation}.
Both methodologies have their place; our work adopts the former approach.




\section{Conclusion}
\label{sec:conc}

We have proposed $\lamhook$, a core calculus that captures the essence
of hook-indexed state in the style of React's hooks and components.
By endowing $\lamhook$ with a type system and proving it sound, we
offer a blueprint for how language designers may soundly adapt React's
hook-indexed state into their own languages.  We have also defined a
local translation from $\lamhook$ into a more canonical language,
$\lamref$, with components based on ML-style first-class
references. We have thereby shown that hook-indexed state can be
expressed in terms of first-class references, although the translation
duplicates code and lowers the level of abstraction. Designers of
languages with state that is scoped to re-instantiable blocks, and
whose changes must be observed (e.g., to trigger rerendering or
propagate updates), now have a principled basis for choosing between
useState, ML-style references, and alternatives yet to be devised.

\section*{Disclosure of Use of AI in Research}
We used AI to brainstorm names for hook-indexed state,
and to suggest how to tighten a couple of weak sentences in the conclusion.
AI was not used during the research itself.

\section*{Data Availability Statement}
Our work has no associated artifacts.

\bibliographystyle{ACM-Reference-Format}
\bibliography{lambda-react}

\clearpage
\appendix

\section[Lambda-Hook-View Proofs]{$\lamhook$ Proofs}
\label{appdx:lamhook}

\paragraph{Presentation}
We denote \emph{deductions via inversion} (in a bottom-up derivation)
with a label \emph{above} the conclusion, e.g.:
\begin{mathpar}
\inferrule* {(1)} {\textnormal{some judgment}}
\end{mathpar}

Comparatively, we denote the \emph{use of deductions} (in a top-down derivation)
with a label \emph{besides} the use, e.g.:
\begin{mathpar}
\inferrule*[right=(1)] { } {\textnormal{some judgment}}
\end{mathpar}

To split up large derivations, we may also label subconclusions
on the left, e.g.:
\begin{mathpar}
  \inferrule* {\cdots} {(1)\ \textnormal{some judgment}}
\end{mathpar}
We then may use (1) as a premise in a larger derivation.

Oftentimes when deriving a subconclusion, it is briefer to use
an unspecified context, which we specify when \emph{using} the subconclusion.
For instance, we might derive the following (for some
unspecified/partially specified $\Gamma_1$):
\begin{mathpar}
  \inferrule* {\cdots} {
    (1)\ \Gamma_1 \vdash \textnormal{some judgment}
  }
\end{mathpar}

We can then invoke (1) while specifying $\Gamma_1$ as so:

\begin{mathpar}
  \inferrule* {[(1),\ \Gamma_1 = \Gamma, x: A]} {
    \Gamma \vdash \textnormal{some judgment}
  }
\end{mathpar}

\subsection{Derived Rules}

For the derived forms (often relating to $\mf{fold}$)
which appear in our semantics, it is often convenient to use derived rules
in lieu of their lengthier primitive counterparts.
Here, we present each such rule alongside its derivation.

\paragraph{Map} We have
$\mf{map}\ V\ W \triangleq
\mf{fold}\ (\lambda x.\lambda acc. V x :: acc)\ []\ W$,
where the right-hand lambda desugars to
$\lambda x.\mf{val}\ \lambda acc. \elet{y}{V\ x}{\mf{val}\ y :: acc}$.
The derived typing rule is given as follows:

\begin{mathpar}
  \inferrule*[right=TC-Map] {
    \Gamma \semi \Pi \semi \Phi_1 \vdash V: A \xrightarrow{E} B
    \and
    \Gamma \semi \Pi \semi \Phi_2 \vdash W: \mf{List}\ A
  } {
    \Gamma \semi \Pi \semi \Phi_1 \Join \Phi_2 \vdash \mf{map}\ V\ W:
    \mf{List}\ B \bang E
  }
\end{mathpar}

We derive this rule via intermediate derivations (1)-(3) as follows
for $E=\mf{Pure}$. A derivation for $E=\mf{Mut}$ can be straightforwardly
obtained by approprately invoking \textsc{TC-Sub} throughout.
We employ a (standard) weakening lemma on $\Gamma$.

\begin{mathpar}
  \inferrule*[right=TC-App] {
    \Gamma_1 \semi \Pi \semi \Phi_1 \vdash V: A \xrightarrow{\mf{Pure}} B
    \and
    \inferrule*[right=TV-Var] { } {
      \Gamma_1 \semi \Pi \semi \emptyset \vdash x: A
    }
  } {
    (1) \quad
    \Gamma_1 \semi \Pi \semi \Phi_1 \vdash V\ x: B \bang \mf{Pure}
  }

  \inferrule*[right=TC-Val] {
    \inferrule*[right=TV-Cons] {
      \inferrule*[right=TV-Var] { } {
        \Gamma_2 \semi \Pi \semi \emptyset \vdash y: B
      }
      \and
      \inferrule*[right=TV-Var] { } {
        \Gamma_2 \semi \Pi \semi \emptyset \vdash acc: \mf{List}\ B
      }
    } {
      \Gamma_2 \semi \Pi \semi \emptyset \vdash y :: acc: \mf{List}\ B
    }
  } {
    (2) \quad
    \Gamma_2 \semi \Pi \semi \emptyset \vdash \mf{val}\ y :: acc:
    \mf{List}\ B \bang \mf{Pure}
  }

  \inferrule*[right=TV-Lam] {
    \inferrule*[Right=TC-Val] {
      \inferrule*[Right=TV-Lam] {
        \inferrule*[Right=TC-Let] {
          [(1),\ \Gamma_1 = \Gamma, x: A, acc: \mf{List}\ B]
          \and
          [(2),\ \Gamma_2 = \Gamma, x: A, acc: \mf{List}\ B, y: B]
        } {
          \Gamma, x: A, acc: \mf{List}\ B \semi \Pi \semi \Phi_1 \vdash
          \elet{y}{V\ x}{\mf{val}\ y :: acc}:
          \mf{List}\ B \bang \mf{Pure}
        }
      } {
        \Gamma, x: A \semi \Pi \semi \Phi_1 \vdash
        \lambda acc. \elet{y}{V\ x}{\mf{val}\ y :: acc}:
        \mf{List}\ B \xrightarrow{\mf{Pure}} \mf{List}\ B
      }
    } {
      \Gamma, x: A \semi \Pi \semi \Phi_1 \vdash
      \mf{val}\ \lambda acc. \elet{y}{V\ x}{\mf{val}\ y :: acc}:
      \mf{List}\ B \xrightarrow{\mf{Pure}} \mf{List}\ B \bang \mf{Pure}
    }
  } {
    (1) \quad
    \Gamma \semi \Pi \semi \Phi_1 \vdash
    \lambda x.\mf{val}\ \lambda acc. \elet{y}{V\ x}{\mf{val}\ y :: acc}:
    A \xrightarrow{\mf{Pure}} \mf{List}\ B \xrightarrow{\mf{Pure}} \mf{List}\ B
  }

  \inferrule*[right=TC-Fold] {
    (1)
    \and
    \inferrule*[right=TV-Nil] { } {
      \Gamma \semi \Pi \semi \emptyset \vdash []: \mf{List}\ B
    }
    \and
    \Gamma \semi \Pi \semi \Phi_2 \vdash W: \mf{List}\ A
  } {
    \Gamma \semi \Pi \semi \Phi_1 \Join \Phi_2 \vdash
    \mf{fold}\ (\lambda x.\cdots)\ []\ W:
    \mf{List}\ B \bang \mf{Pure}
  }
\end{mathpar}

\paragraph{List Initialization} We have
$\mf{inits}\ V \triangleq
\mf{map}\ (\lambda x.\mf{init}\ x)\ V$.
The derived typing rule is given as follows:

\begin{mathpar}
  \inferrule*[right=TC-Inits] {
    \Gamma \semi \Pi \semi \Phi \vdash V: \mf{List\ Element}
  } {
    \Gamma \semi \Pi \semi \Phi \vdash \mf{inits}\ V: \mf{List\ ElemInst}
    \bang \mf{Mut}
  }
\end{mathpar}

We derive this rule using the previously derived rule \textsc{TC-Map}
as follows:

\begin{mathpar}
  \inferrule*[right=TC-Map] {
    \inferrule*[right=TV-Lam] {
      \inferrule*[Right=TC-Init] {
        \inferrule*[Right=TV-Var] { } {
          \Gamma, x: \mf{Element} \semi \Pi \semi \emptyset \vdash
          x: \mf{Element}
        }
      } {
        \Gamma, x: \mf{Element} \semi \Pi \semi \emptyset \vdash
        \mf{init}\ x: \mf{ElemInst} \bang \mf{Mut}
      }
    } {
      \Gamma \semi \Pi \semi \emptyset \vdash \lambda x.\mf{init}\ x:
      \mf{Element} \xrightarrow{\mf{Mut}}\ \mf{ElemInst}
    }
    \
    \Gamma \semi \Pi \semi \Phi \vdash V: \mf{List\ Element}
  } {
    \Gamma \semi \Pi \semi \Phi \vdash \mf{map}\ (\lambda x.\mf{init}\ x)\ V:
    \mf{List\ ElemInst} \bang \mf{Mut}
  }
\end{mathpar}

\paragraph{List Updating} We have
$\mf{updates}\ V \triangleq
\mf{map}\ (\lambda x.\mf{update}\ x)\ V$.
The derived typing rule is given as follows:

\begin{mathpar}
  \inferrule*[right=TC-Updates] {
    \Gamma \semi \Pi \semi \Phi \vdash V: \mf{List\ ElemInst}
  } {
    \Gamma \semi \Pi \semi \Phi \vdash \mf{updates}\ V: \mf{List\ ElemInst}
    \bang \mf{Mut}
  }
\end{mathpar}

We derive this rule using the previously derived rule \textsc{TC-Map}
as follows:

\begin{mathpar}
  \inferrule*[right=TC-Map] {
    \inferrule*[Right=TV-Lam] {
      \inferrule*[Right=TC-Init] {
        \inferrule*[Right=TV-Var] { } {
          \Gamma, x: \mf{ElemInst} \semi \Pi \semi \emptyset \vdash
          x: \mf{ElemInst}
        }
      } {
        \Gamma, x: \mf{ElemInst} \semi \Pi \semi \emptyset \vdash
        \mf{update}\ x: \mf{ElemInst} \bang \mf{Mut}
      }
    } {
      \Gamma \semi \Pi \semi \emptyset \vdash \lambda x.\mf{update}\ x:
      \mf{ElemInst} \xrightarrow{\mf{Mut}}\ \mf{ElemInst}
    }
    \
    \Gamma \semi \Pi \semi \Phi \vdash V: \mf{List\ ElemInst}
  } {
    \Gamma \semi \Pi \semi \Phi \vdash \mf{map}\ (\lambda x.\mf{update}\ x)\ V:
    \mf{List\ ElemInst} \bang \mf{Mut}
  }
\end{mathpar}

\paragraph{Iteration} We have
$\mf{iter}\ V\ W \triangleq
\mf{fold}\ (\lambda x.\lambda (). V x)\ ()\ W$,
where $\lambda ().M$ is shorthand for a function which does not
bind a variable in $M$.
The derived typing rule is given as follows:

\begin{mathpar}
  \inferrule*[right=TC-Iter] {
    \Gamma \semi \Pi \semi \Phi_1 \vdash V: A \xrightarrow{E} \tone
    \and
    \Gamma \semi \Pi \semi \Phi_2 \vdash W: \mf{List}\ A
  } {
    \Gamma \semi \Pi \semi \Phi_1 \Join \Phi_2 \vdash \mf{iter}\ V\ W:
    \tone \bang E
  }
\end{mathpar}

We derive this rule via intermediate derivations (1)-(3) as follows
for $E=\mf{Pure}$. A derivation for $E=\mf{Mut}$ can be straightforwardly
obtained by approprately invoking \textsc{TC-Sub} throughout.

\begin{mathpar}
  \inferrule*[right=TV-Lam] {
    \inferrule*[Right=TC-Val] {
      \inferrule*[Right=TV-Lam] {
        \inferrule*[Right=TC-Let] {
          \Gamma, x: A \semi \Pi \semi \Phi_1 \vdash
          V: A \xrightarrow{\mf{Pure}} \tone
          \and
          \inferrule*[Right=TV-Var] { } {
            \Gamma, x: A \semi \Pi \semi \emptyset \vdash x: A
          }
        } {
          \Gamma, x: A \semi \Pi \semi \Phi_1 \vdash
          V\ x: \tone \bang \mf{Pure}
        }
      } {
        \Gamma, x: A \semi \Pi \semi \Phi_1 \vdash
        \lambda (). V\ x:
        \tone \xrightarrow{\mf{Pure}} \tone
      }
    } {
      \Gamma, x: A \semi \Pi \semi \Phi_1 \vdash
      \mf{val}\ \lambda(). V\ x:
      \tone \xrightarrow{\mf{Pure}} \tone \bang \mf{Pure}
    }
  } {
    (1) \quad
    \Gamma \semi \Pi \semi \Phi_1 \vdash
    \lambda x.\mf{val}\ \lambda (). V\ x:
    A \xrightarrow{\mf{Pure}} \tone \xrightarrow{\mf{Pure}} \tone
  }

  \inferrule*[right=TC-Fold] {
    (1)
    \and
    \inferrule*[right=TV-Unit] { } {
      \Gamma \semi \Pi \semi \emptyset \vdash (): \tone
    }
    \and
    \Gamma \semi \Pi \semi \Phi_2 \vdash W: \mf{List}\ A
  } {
    \Gamma \semi \Pi \semi \Phi_1 \Join \Phi_2 \vdash
    \mf{fold}\ (\lambda x.\cdots)\ ()\ W:
    \tone \bang \mf{Pure}
  }
\end{mathpar}

\paragraph{List Click-Search} We have
$\mf{clicks}_{id}\ V \triangleq
\mf{iter}\ (\lambda x.\mf{click}_{id}\ x)\ V$.
The derived typing rule is given as follows:

\begin{mathpar}
  \inferrule*[right=TC-Clicks] {
    \Gamma \semi \Pi \semi \Phi \vdash V: \mf{List\ ElemInst}
  } {
    \Gamma \semi \Pi \semi \Phi \vdash \mf{clicks}_{id}\ V: \tone
    \bang \mf{Mut}
  }
\end{mathpar}

We derive this rule using the previously derived rule \textsc{TC-Iter}
as follows:

\begin{mathpar}
  \inferrule*[right=TC-Iter] {
    \inferrule*[right=TV-Lam] {
      \inferrule*[Right=TC-Init] {
        \inferrule*[Right=TV-Var] { } {
          \Gamma, x: \mf{ElemInst} \semi \Pi \semi \emptyset \vdash
          x: \mf{ElemInst}
        }
      } {
        \Gamma, x: \mf{ElemInst} \semi \Pi \semi \emptyset \vdash
        \mf{click}_{id}\ x: \tone \bang \mf{Mut}
      }
    } {
      \Gamma \semi \Pi \semi \emptyset \vdash \lambda x.\mf{click}_{id}\ x:
      \mf{ElemInst} \xrightarrow{\mf{Mut}}\ \tone
    }
    \
    \Gamma \semi \Pi \semi \Phi \vdash V: \mf{List\ ElemInst}
  } {
    \Gamma \semi \Pi \semi \Phi \vdash
    \mf{iter}\ (\lambda x.\mf{click}_{id}\ x)\ V:
    \tone \bang \mf{Mut}
  }
\end{mathpar}

\paragraph{List Click-Run} We have
$\mf{runClicks}\ V \triangleq
\mf{iter}\ (\lambda x.\mf{runClick}\ x)\ V$.
The derived typing rule is given as follows:

\begin{mathpar}
  \inferrule*[right=TC-RunClicks] {
    \Gamma \semi \Pi \semi \Phi \vdash V: \mf{List\ Attr}
  } {
    \Gamma \semi \Pi \semi \Phi \vdash \mf{runClicks}\ V: \tone
    \bang \mf{Mut}
  }
\end{mathpar}

We derive this rule using the previously derived rule \textsc{TC-Iter}
as follows:

\begin{mathpar}
  \inferrule*[right=TC-Iter] {
    \inferrule*[right=TV-Lam] {
      \inferrule*[Right=TC-Init] {
        \inferrule*[Right=TV-Var] { } {
          \Gamma, x: \mf{Attr} \semi \Pi \semi \emptyset \vdash
          x: \mf{Attr}
        }
      } {
        \Gamma, x: \mf{Attr} \semi \Pi \semi \emptyset \vdash
        \mf{runClick}\ x: \tone \bang \mf{Mut}
      }
    } {
      \Gamma \semi \Pi \semi \emptyset \vdash \lambda x.\mf{runClick}\ x:
      \mf{Attr} \xrightarrow{\mf{Mut}}\ \tone
    }
    \
    \Gamma \semi \Pi \semi \Phi \vdash V: \mf{List\ Attr}
  } {
    \Gamma \semi \Pi \semi \Phi \vdash
    \mf{iter}\ (\lambda x.\mf{runClick}\ x)\ V:
    \tone \bang \mf{Mut}
  }
\end{mathpar}

\subsection{Weakening}
\label{appdx:wkn}

We have standard weakening results on $\Gamma$ and
on $\Phi$ (for $\mf{Succ}$-phase pointers only).
Our weakening result for $\Pi$ is comparatively non-standard, since
arbitrarily weakening on $\Pi$ would be unsound:
for example, extending $\Pi[p]$ in rule \textsc{TC-EvalI}
might violate the premise that $\Pi[p] \cap \Ell = \emptyset$.
To sidestep this issue, we only permit weakening $\Pi$ at pointers
which are \emph{initial-phase} and \emph{unrelated to the term},
both of which can be determined by $\Phi$.

When weakening typing judgments involving some $\Phi$,
if we have $\Pi \sqsubseteq_{\Phi'} \Pi'$,
we require that $\Phi \Join \Phi'$ is defined,
implying that $\Phi$ and $\Phi'$ are \emph{sufficiently disjoint}:
in particular, we get that, if $p \in \dom(\Phi)$, then $\Pi[p] = \Pi'[p]$.

\begin{restatable}[$\Pi$-weakening, values]{lemma}{lamhookpiwknv}
\label{thm:lamhook-piwknv}
If
\begin{enumerate}[label=(\roman*)]
\item $\Gamma \semi \Pi \semi \Phi \vdash V: A$,
\item $\Pi \sqsubseteq_{\Phi'} \Pi'$, and
\item $\Phi \Join \Phi'$ is defined,
\end{enumerate}
then $\Gamma \semi \Pi' \semi \Phi \vdash V: A$.
\end{restatable}

\begin{proof}
  By mutual induction on the derivation of the typing judgment \emph{(i)}
  alongside Lemmas \ref{thm:lamhook-piwknc} and \ref{thm:lamhook-piwknb}.

  \pcase{TV-Var}: We have
  \begin{mathpar}
    \inferrule*[right=TV-Var] {
      \inferrule* {(1)} {x: A \in \Gamma}
      \and
      \inferrule* {(2)} {\mf{Succ}(\Phi)}
    } {
      \Gamma \semi \Pi \semi \Phi \vdash x: A
    }
  \end{mathpar}
  \begin{itemize}
    \item We can immediately derive the conclusion:
    \begin{mathpar}
      \inferrule*[right=TV-Var] {
        \inferrule*[right=(1)] { } {x: A \in \Gamma}
        \and
        \inferrule*[right=(2)] { } {\mf{Succ}(\Phi)}
      } {
        \Gamma \semi \Pi' \semi \Phi \vdash x: A
      }
    \end{mathpar}
  \end{itemize}

  \textbf{Cases}
  \textsc{TV-Unit},
  \textsc{TV-Str},
  \textsc{TV-True},
  \textsc{TV-False}, and
  \textsc{TV-Nil} are similarly immediately rederivable.

  \pcase{TV-Pair}: We have
  \begin{mathpar}
    \inferrule*[right=TV-Pair] {
      \inferrule* {(1)} {\Gamma \semi \Pi \semi \Phi_1 \vdash V: A}
      \and
      \inferrule* {(2)} {\Gamma \semi \Pi \semi \Phi_2 \vdash W: B}
    } {
      \Gamma \semi \Pi \semi \Phi_1 \Join \Phi_2 \vdash (V, W):
      A \times B
    }
  \end{mathpar}
  \begin{itemize}
    \item From premise \emph{(iii)} and definition of $\Join$,
    it follows that (3) $\Phi_1 \Join \Phi'$ is defined
    and (4) $\Phi_2 \Join \Phi'$ is defined.
    \item By inductive hypothesis on (1), (3), and premise \emph{(ii)},
    we get (5) that $\Gamma \semi \Pi' \semi \Phi_1 \vdash V: A$.
    \item Likewise, by inductive hypothesis on (2), (4),
    and premise \emph{(ii)},
    we get (6) that\\
    $\Gamma \semi \Pi' \semi \Phi_2 \vdash W: B$.
    \item We derive the conclusion as follows:
    \begin{mathpar}
      \inferrule*[right=TV-Pair] {
        \inferrule*[right=(5)] { } {
          \Gamma \semi \Pi' \semi \Phi_1 \vdash V: A
        }
        \and
        \inferrule*[right=(6)] { } {
          \Gamma \semi \Pi' \semi \Phi_2 \vdash W: B
        }
      } {
        \Gamma \semi \Pi' \semi \Phi_1 \Join \Phi_2 \vdash (V, W):
        A \times B
      }
    \end{mathpar}
  \end{itemize}

  \textbf{Cases}
  \textsc{TV-Cons},
  \textsc{TV-Attr},
  \textsc{TV-OnClick},
  \textsc{TV-Text},
  \textsc{TV-Tag},
  \textsc{TV-TextI}, and
  \textsc{TV-TagI} can be shown analogously by inductive hypothesis.

  \pcase{TV-Lam}: We have
  \begin{mathpar}
    \inferrule*[right=TV-Lam] {
      \inferrule* {(1)} {
        \Gamma, x: A \semi \Pi \semi \Phi \vdash M: B \bang E
      }
      \and
      \inferrule* {(2)} {\mf{Succ}(\Phi)}
    } {
      \Gamma \semi \Pi \semi \Phi \vdash \lambda x.M: A \xrightarrow{E} B
    }
  \end{mathpar}
  \begin{itemize}
    \item By mutual inductive hypothesis (Lemma \ref{thm:lamhook-piwknc}) on (1)
    and premises \emph{(ii)} and \emph{(iii)}, we get (3) that
    $\Gamma, x: A \semi \Pi' \semi \Phi \vdash M: B \bang E$.
    \item We derive the conclusion as follows:
    \begin{mathpar}
      \inferrule*[right=TV-Lam] {
        \inferrule*[right=(3)] { } {
          \Gamma, x: A \semi \Pi' \semi \Phi \vdash e: B \bang E
        }
      \and
      \inferrule*[right=(2)] { } {\mf{Succ}(\Phi)}
      } {
        \Gamma \semi \Pi' \semi \Phi \vdash \lambda x.M: A \xrightarrow{E} B
      }
    \end{mathpar}
  \end{itemize}

  \pcase{TV-Comp}: We have
  \begin{mathpar}
    \inferrule*[right=TV-Comp] {
      \inferrule* {(1)} {\Gamma \semi \Pi \semi \Phi \vdash D: \Ell}
    } {
      \Gamma \semi \Pi \semi \Phi \vdash \mf{component}\ \{D\}:
      \mf{Element}
    }
  \end{mathpar}
  \begin{itemize}
    \item By mutual inductive hypothesis (Lemma \ref{thm:lamhook-piwknb}) on (1)
    and premise \emph{(ii)}, we get
    (2) that\\$\Gamma \semi \Pi' \semi \Phi \vdash D: \Ell$.
    \item We derive the conclusion as follows:
    \begin{mathpar}
      \inferrule*[right=TV-Comp] {
        \inferrule*[right=(2)] { } {
          \Gamma \semi \Pi' \semi \Phi \vdash D: \Ell
        }
      } {
        \Gamma \semi \Pi' \semi \Phi \vdash \mf{component}\ \{D\}: \mf{Element}
      }
    \end{mathpar}
  \end{itemize}

  \pcase{TV-CompI}: We have
  \begin{mathpar}
    \inferrule*[right=TV-CompI] {
      \inferrule* {(1)} {\Gamma \semi \Pi \semi \Phi \vdash D: \Pi[p]}
      \and
      \inferrule* {(2)} {\Phi[p] = \mf{Succ}}
      \and
      \inferrule* {(3)} {\Gamma \semi \Pi \semi \Phi \vdash V: \mf{ElemInst}}
    } {
      \Gamma \semi \Pi \semi \Phi_1 \Join \Phi_2
      \vdash \compInst{p}{D}{V}: \mf{ElemInst}
    }
  \end{mathpar}
  \begin{itemize}
    \item By (2) and premise \emph{(iii)}, either $\Phi'[p]=\mf{Succ}$
    or $p \notin \dom(\Phi')$. In either case, premise \emph{(ii)}
    then gives us (4) that $\Pi'[p] = \Pi[p]$.
    \item By mutual inductive hypothesis (Lemma \ref{thm:lamhook-piwknb})
    on (1) and premise \emph{(ii)}, we get that\\
    $\Gamma \semi \Pi' \vdash D: \Pi[p]$.
    By (4), we equivalently have (5) that $\Gamma \semi \Pi' \vdash D: \Pi'[p]$.
    \item By inductive hypothesis on (3)
    and premises \emph{(ii)} and \emph{(iii)},
    we get (6) that\\$\Gamma \semi \Pi' \semi \Phi \vdash V: \mf{ElemInst}$.
    \item We derive the conclusion as follows:
    \begin{mathpar}
      \inferrule*[right=TV-CompI] {
        \inferrule*[right=(5)] { } {\Gamma \semi \Pi' \vdash D: \Pi'[p]}
        \and
        \inferrule*[right=(2)] { } {\Phi[p] = \mf{Succ}}
        \and
        \inferrule*[right=(6)] { } {
          \Gamma \semi \Pi' \semi \Phi \vdash V: \mf{ElemInst}
        }
      } {
        \Gamma \semi \Pi' \semi \Phi \vdash \compInst{p}{D}{V}:
        \mf{ElemInst}
      }
    \end{mathpar}
  \end{itemize}

  \pcase{TV-Setter}: We have
  \begin{mathpar}
    \inferrule*[right=TV-Setter] {
      \inferrule* {(1)} {\ell: A \in \Pi[p]}
      \and
      \inferrule* {(2)} {\mf{Succ}(\Phi)}
    } {
      \Gamma \semi \Pi \semi \Phi \vdash \setter^\ell_p:
      A \xrightarrow{\mf{Mut}} \tone
    }
  \end{mathpar}
  \begin{itemize}
    \item By (1) and premise \emph{(ii)}, we get (3) that
    $\ell: A \in \Pi'[p]$ (regardless of whether $p \in \dom(\Phi')$).
    \item We derive the conclusion as follows:
    \begin{mathpar}
      \inferrule*[right=TV-Setter] {
        \inferrule*[right=(3)] { } {\ell: A \in \Pi'[p]}
        \and
        \inferrule*[right=(2)] { } {\mf{Succ}(\Phi)}
      } {
        \Gamma \semi \Pi' \semi \Phi \vdash \setter^\ell_p:
        A \xrightarrow{\mf{Mut}} \tone
      }
    \end{mathpar}
  \end{itemize}
\end{proof}

\begin{restatable}[$\Pi$-weakening, computations]{lemma}{lamhookpiwknc}
\label{thm:lamhook-piwknc}
If
\begin{enumerate}[label=(\roman*)]
\item $\Gamma \semi \Pi \semi \Phi \vdash M: A \bang E$,
\item $\Pi \sqsubseteq_{\Phi'} \Pi'$, and
\item $\Phi \Join \Phi'$ is defined,
\end{enumerate}
then $\Gamma \semi \Pi' \semi \Phi \vdash M: A \bang E$.
\end{restatable}
\begin{proof}
  By mutual induction on the derivation of the typing judgment \emph{(i)}
  alongside Lemmas \ref{thm:lamhook-piwknv} and \ref{thm:lamhook-piwknb}.

  \pcase{TC-Val}: We have
  \begin{mathpar}
    \inferrule*[right=TC-Val] {
      \inferrule* {(1)} {\Gamma \semi \Pi \semi \Phi \vdash V: A}
    } {
      \Gamma \semi \Pi \semi \Phi \vdash \mf{val}\ V: A \bang \mf{Pure}
    }
  \end{mathpar}
  \begin{itemize}
    \item By mutual inductive hypothesis (Lemma \ref{thm:lamhook-piwknv}) on (1)
    and premises \emph{(ii)} and \emph{(iii)}, we get
    (2) that $\Gamma \semi \Pi' \semi \Phi \vdash V: A$.
    \item We derive the conclusion as follows:
    \begin{mathpar}
      \inferrule*[right=TC-Val] {
        \inferrule*[right=(3)] { } {
          \Gamma \semi \Pi' \semi \Phi \vdash V: A
        }
      } {
        \Gamma \semi \Pi' \semi \Phi \vdash \mf{val}\ V: A \bang \mf{Pure}
      }
    \end{mathpar}
  \end{itemize}

  \textbf{Cases}
  \textsc{TC-App},
  \textsc{TC-If},
  \textsc{TC-Fst},
  \textsc{TC-Snd},
  \textsc{TC-Fold},
  \textsc{TC-Init},
  \textsc{TC-Click},
  \textsc{TC-RunClick}, and
  \textsc{TC-Update} can be shown analogously by mutual inductive hypothesis.

  \pcase{TC-Let}: We have
  \begin{mathpar}
    \inferrule*[right=TC-Let] {
      \inferrule* {(1)} {
        \Gamma \semi \Pi \semi \Phi_1 \vdash M: A \bang E
      }
      \and
      \inferrule* {(2)} {
        \Gamma, x: A_1 \semi \Pi \semi \Phi_2 \vdash N: B \bang E
      }
    } {
      \Gamma \semi \Pi \semi \Phi_1 \Join \Phi_2 \vdash \elet{x}{M}{N}:
      B \bang E
    }
  \end{mathpar}
  \begin{itemize}
    \item From premise \emph{(iii)} and definition of $\Join$,
    it follows that (3) $\Phi_1 \Join \Phi'$ is defined
    and (4) $\Phi_2 \Join \Phi'$ is defined.
    \item By inductive hypothesis on (1), (3), and premise \emph{(ii)},
    we get (5) that $\Gamma \semi \Pi' \semi \Phi_1 \vdash M: A \bang E$.
    \item Likewise, by inductive hypothesis on (2), (4),
    and premise \emph{(ii)}, we get (6) that\\
    $\Gamma, x: A \semi \Pi' \semi \Phi_2 \vdash N: B \bang E$.
    \item We derive the conclusion as follows:
    \begin{mathpar}
    \inferrule*[right=TC-Let] {
      \inferrule*[right=(5)] { } {
        \Gamma \semi \Pi' \semi \Phi_1 \vdash M: A \bang E
      }
      \and
      \inferrule*[right=(6)] { } {
        \Gamma, x: A \semi \Pi' \semi \Phi_2 \vdash N: B \bang E
      }
    } {
      \Gamma \semi \Pi' \semi \Phi_1 \Join \Phi_2 \vdash \elet{x}{M}{N}:
      B \bang E
    }
    \end{mathpar}
  \end{itemize}

  \pcase{TC-Sub}: We have
  \begin{mathpar}
    \inferrule*[right=TC-Sub] {
      \inferrule* {(1)} {
        \Gamma \semi \Pi \semi \Phi \vdash M: A \bang \mf{Pure}
      }
    } {
      \Gamma \semi \Pi \semi \Phi \vdash M: A \bang \mf{Mut}
    }
  \end{mathpar}
  \begin{itemize}
    \item By inductive hypothesis on (1) and premises
    \emph{(ii)} and \emph{(iii)}, we get (2) that\\
    $\Gamma \semi \Pi' \semi \Phi \vdash M: A \bang \mf{Pure}$.
    \item We derive the conclusion as follows:
    \begin{mathpar}
      \inferrule*[right=TC-Sub] {
        \inferrule*[right=(2)] { } {
          \Gamma \semi \Pi' \semi \Phi \vdash M: A \bang \mf{Pure}
        }
      } {
        \Gamma \semi \Pi' \semi \Phi \vdash M: A \bang \mf{Mut}
      }
    \end{mathpar}
  \end{itemize}

  \pcase{TC-EvalI}: We have
  \begin{mathpar}
    \inferrule*[right=TC-EvalI] {
      \inferrule* {(1)} {\Gamma \semi \Pi \vdash D_0: \Ell \uplus \Pi[p]}
      \and
      \inferrule* {(2)} {\Gamma \semi \Pi \vdash D: \Ell}
      \and
      \inferrule* {(3)} {\mf{Succ}(\Phi)}
    } {
      \Gamma \semi \Pi \semi \Phi,p: \mf{Init} \vdash
      \eval{D_0}{\mf{Init}}{p}{D}: \mf{ElemInst} \bang \mf{Mut}
    }
  \end{mathpar}
  \begin{itemize}
    \item By premise \emph{(iii)}, we must have that $p \notin \dom(\Phi')$.
    Premise \emph{(ii)} then gives (4) $\Pi[p] = \Pi'[p]$.
    \item By mutual inductive hypothesis (Lemma \ref{thm:lamhook-piwknb}) on (1)
    and premise \emph{(ii)}, we get that\\
    $\Gamma \semi \Pi' \vdash D_0: \Ell \uplus \Pi[p]$.
    By (4), this is equivalent to (5) that
    $\Gamma \semi \Pi' \vdash D_0: \Ell \uplus \Pi'[p]$.
    \item By mutual inductive hypothesis (Lemma \ref{thm:lamhook-piwknb}) on (2)
    and premise \emph{(ii)}, we get (6) that\\
    $\Gamma \semi \Pi' \vdash D: \Ell$.
    \item We derive the conclusion as follows:
    \begin{mathpar}
      \inferrule*[right=TC-EvalI] {
        \inferrule*[right=(5)] { } {
          \Gamma \semi \Pi' \vdash D_0: \Ell \uplus \Pi'[p]
        }
        \and
        \inferrule*[right=(6)] { } {
          \Gamma \semi \Pi' \vdash D: \Ell
        }
        \and
        \inferrule*[right=(3)] { } {\mf{Succ}(\Phi)}
      } {
        \Gamma \semi \Pi' \semi \Phi,p: \mf{Init} \vdash
        \eval{D_0}{\phi}{p}{D}: \mf{ElemInst} \bang \mf{Mut}
      }
    \end{mathpar}
  \end{itemize}

  \pcase{TC-EvalS}: We have
  \begin{mathpar}
    \inferrule*[right=TC-EvalS] {
      \inferrule* {(1)} {\Gamma \semi \Pi \vdash D_0: \Pi[p]}
      \and
      \inferrule* {(2)} {\Gamma \semi \Pi \vdash D: \Ell}
      \and
      \inferrule* {(3)} {\Ell \subseteq \Pi[p]}
      \and
      \inferrule* {(4)} {\mf{Succ}(\Phi)}
    } {
      \Gamma \semi \Pi \semi \Phi,p: \mf{Succ} \vdash
      \eval{D_0}{\mf{Succ}}{p}{D}: \mf{ElemInst} \bang \mf{Mut}
    }
  \end{mathpar}
  \begin{itemize}
    \item Since we start with $\Phi, p: \mf{Succ}$, by premise \emph{(iii)},
    either $\Phi'[p]=\mf{Succ}$ or $p \notin \dom(\Phi')$. In either case,
    premise \emph{(ii)} then gives us (5) that $\Pi'[p] = \Pi[p]$.
    \item By mutual inductive hypothesis (Lemma \ref{thm:lamhook-piwknb}) on (1)
    and premise \emph{(ii)}, we get that\\
    $\Gamma \semi \Pi' \vdash D_0: \Pi[p]$.
    By (5), this is equivalent to (6) that
    $\Gamma \semi \Pi' \vdash D_0: \Pi'[p]$.
    \item Also by mutual inductive hypothesis (Lemma \ref{thm:lamhook-piwknb})
    on (2) and premise \emph{(ii)}, we get (7) that\\
    $\Gamma \semi \Pi' \vdash D: \Ell$.
    \item By (5), (3) is equivalent to (8) that $\Ell \subseteq \Pi'[p]$.
    \item We derive the conclusion as follows:
    \begin{mathpar}
      \inferrule*[right=TC-EvalS] {
        \inferrule*[right=(6)] { } {\Gamma \semi \Pi' \vdash D_0: \Pi'[p]}
        \
        \inferrule*[right=(7)] { } {\Gamma \semi \Pi' \vdash D: \Ell}
        \
        \inferrule*[right=(8)] { } {\Ell \subseteq \Pi'[p]}
        \
        \inferrule*[right=(4)] { } {\mf{Succ}(\Phi)}
      } {
        \Gamma \semi \Pi' \semi \Phi,p: \mf{Succ} \vdash
        \eval{D_0}{\mf{Succ}}{p}{D}: \mf{ElemInst} \bang \mf{Mut}
      }
    \end{mathpar}
  \end{itemize}
\end{proof}

\begin{restatable}[$\Pi$-weakening, bodies]{lemma}{lamhookpiwknb}
\label{thm:lamhook-piwknb}
If
\begin{enumerate}[label=(\roman*)]
\item $\Gamma \semi \Pi \semi \Phi \vdash D: \Ell$,
\item $\Pi \sqsubseteq_{\Phi'} \Pi'$, and
\item $\Phi \Join \Phi'$ is defined,
\end{enumerate}
then $\Gamma \semi \Pi' \semi \Phi \vdash D: \Ell$.
\end{restatable}

\begin{proof}
  By mutual induction on the derivation of the typing judgment \emph{(i)}
  alongside Lemmas \ref{thm:lamhook-piwknv} and \ref{thm:lamhook-piwknc}.

  \pcase{TB-Val}: We have
  \begin{mathpar}
    \inferrule*[right=TB-Val] {
      \inferrule* {(1)} {
        \Gamma \semi \Pi \semi \Phi \vdash V: \mf{Element}
      }
    } {
      \Gamma \semi \Pi \semi \Phi \vdash \mf{val}\ V: \emptyset
    }
  \end{mathpar}
  \begin{itemize}
    \item By mutual inductive hypothesis (Lemma \ref{thm:lamhook-piwknv})
    on (1) and premises \emph{(ii)} and \emph{(iii)}, we get (2) that
    $\Gamma \semi \Pi' \semi \Phi \vdash V: \mf{Element}$.
    \item We derive the conclusion as follows:
    \begin{mathpar}
      \inferrule*[right=TB-Val] {
        \inferrule*[right=(2)] { } {
          \Gamma \semi \Pi' \semi \Phi \vdash V: \mf{Element}
        }
      } {
        \Gamma \semi \Pi' \semi \Phi \vdash \mf{val}\ V: \emptyset
      }
    \end{mathpar}
  \end{itemize}

  \pcase{TB-Let}: We have
  \begin{mathpar}
    \inferrule*[right=TB-Let] {
      \inferrule* {(1)} {
        \Gamma \semi \Pi \semi \Phi_1 \vdash M: A \bang \mf{Pure}
      }
      \and
      \inferrule* {(2)} {
        \Gamma, x: A \semi \Pi \semi \Phi_2 \vdash D: \Ell
      }
    } {
      \Gamma \semi \Pi \semi \Phi_1 \Join \Phi_2 \vdash \elet{x}{M}{D}: \Ell
    }
  \end{mathpar}
  \begin{itemize}
    \item From premise \emph{(iii)} and definition of $\Join$, it follows that
    (3) $\Phi_1 \Join \Phi'$ is defined and
    (4) that $\Phi_2 \Join \Phi'$ is defined.
    \item By mutual inductive hypothesis
    on (1), (3), and premise \emph{(ii)}, we get (5) that\\
    $\Gamma \semi \Pi' \semi \Phi_1 \vdash M: A \bang \mf{Pure}$.
    \item By inductive hypothesis on (2), (4), and premise \emph{(ii)},
    we get (6) that $\Gamma, x: A \semi \Pi' \semi \Phi_2 \vdash D: \Ell$.
    \item We derive the conclusion as follows:
    \begin{mathpar}
      \inferrule*[right=TB-Let] {
        \inferrule*[right=(5)] { } {
          \Gamma \semi \Pi' \semi \Phi_1 \vdash M: A \bang \mf{Pure}
        }
        \and
        \inferrule*[right=(6)] { } {
          \Gamma, x: A \semi \Pi' \semi \Phi_2 \vdash D: \Ell
        }
      } {
        \Gamma \semi \Pi' \semi \Phi_1 \Join \Phi_2 \vdash \elet{x}{M}{D}: \Ell
      }
    \end{mathpar}
  \end{itemize}

  \pcase{TB-State}: We have
  \begin{mathpar}
    \inferrule*[right=TB-State] {
      \inferrule* {(1)} {\Gamma \semi \Pi \semi \Phi_1 \vdash V: A}
      \and
      \inferrule* {(2)} {
        \Gamma, x: A, x_\mf{set}: A \xrightarrow{\mf{Mut}} \tone \semi \Pi
        \semi \Phi_2 \vdash D: \Ell
      }
      \and
      \inferrule* {(3)} {\ell \notin \dom(\Ell)}
    } {
      \Gamma \semi \Pi \semi \Phi_1 \Join \Phi_2
      \vdash \estate{x}{\ell}{V}{D}: \Ell, \ell: A
    }
  \end{mathpar}
  \begin{itemize}
    \item From premise \emph{(iii)} and definition of $\Join$, it follows that
    (4) $\Phi_1 \Join \Phi'$ is defined and
    (5) that $\Phi_2 \Join \Phi'$ is defined.
    \item By mutual inductive hypothesis (Lemma \ref{thm:lamhook-piwknv})
    on (1), (4), and premise \emph{(ii)}, we get (6) that
    $\Gamma \semi \Pi' \semi \emptyset \vdash V: A$.
    \item By inductive hypothesis on (2), (5), and premise \emph{(ii)},
    we get (7) that\\
    $\Gamma, x: A, x_\mf{set}: A \xrightarrow{\mf{Mut}} \tone
    \semi \Pi' \vdash D: \Ell$.
    \item We derive the conclusion as follows:
  \end{itemize}
  \begin{mathpar}
    \inferrule*[right=TB-Let] {
      \inferrule*[right=(6)] { } {
        \Gamma \semi \Pi' \semi \Phi_1 \vdash V: A
      }
      \
      \inferrule*[right=(7)] { } {
        \Gamma, x: A, x_\mf{set}: A \xrightarrow{\mf{Mut}} \tone
        \semi \Pi' \semi \Phi_2 \vdash D: \Ell
      }
      \
      \inferrule*[right=(3)] { } {\ell \notin \dom(\Ell)}
    } {
      \Gamma \semi \Pi' \semi \Phi_1 \Join \Phi_2
      \vdash \estate{x}{\ell}{V}{D}: \Ell, \ell: A
    }
  \end{mathpar}
\end{proof}

We have some further lemmas regarding $\Phi$:

\begin{restatable}{lemma}{lamhookunrestv}
\label{thm:lamhook-unrestv}
If $\Gamma \semi \Pi \semi \Phi \vdash V: A$,
then $\mf{Succ}(\Phi)$.
\end{restatable}

\begin{proof}
  By induction on the derivation of the typing judgment.
  If the value in question has no subvalues, then we always
  have $\mf{Succ}(\Phi)$ explicitly. Otherwise,
  the conclusion follows by inductive hypothesis.
\end{proof}

\begin{restatable}{lemma}{lamhookunrestc}
\label{thm:lamhook-unrestc}
If $\Gamma \semi \Pi \semi \Phi \vdash M: A \bang \mf{Pure}$,
then $\mf{Succ}(\Phi)$.
\end{restatable}

\begin{proof}
  By induction on the derivation of the typing judgment.
  We appeal to some combination of the inductive hypothesis
  and Lemma \ref{thm:lamhook-unrestv}.

  In particular, since our premise requires a $\mf{Pure}$ effect,
  rule \textsc{TC-EvalI}, which introduces a $\mf{Init}$-phase pointer,
  can never occur since it also introduces a $\mf{Mut}$ effect.
\end{proof}

\begin{restatable}{lemma}{lamhookunrestb}
\label{thm:lamhook-unrestb}
If $\Gamma \semi \Pi \semi \Phi \vdash D: \Ell$,
then $\mf{Succ}(\Phi)$.
\end{restatable}

\begin{proof}
  By induction on the derivation of the typing judgment.
  In case \textsc{TB-Let}, we have explicitly that the subcomputation
  has $\mf{Succ}(\Phi)$.
  In the other cases, we appeal to Lemma \ref{thm:lamhook-unrestv}.
\end{proof}

\lamhookunrest*

\begin{proof}
  Corollary of Lemmas \ref{thm:lamhook-unrestv},
  \ref{thm:lamhook-unrestc}, and \ref{thm:lamhook-unrestb}.
\end{proof}

\begin{restatable}{lemma}{lamhookunresth}
\label{thm:lamhook-unresth}
If $\Gamma \semi \Pi \semi \Phi \vDash \heap$,
then $\mf{Succ}(\Phi)$.
\end{restatable}

\begin{proof}
  By definition, there exist some $\Phi_i$
  such that $\Phi = \Join_{i=1}^{n} \Phi_i$ and\\
  $\Gamma \semi \Pi \semi \Phi_i \vdash \heap[p].\mf{st}[\ell]: \Pi[p][\ell]$
  for all $n$ elements in $\heap$, each identified by a $(p, \ell)$ pair.
  Lemma \ref{thm:lamhook-unrestv} tells us that all such $\Phi_i$
  have $\mf{Succ}(\Phi_i)$, so it follows that $\mf{Succ}(\Phi)$.
\end{proof}

\begin{lemma}[$\Phi$-weakening, values]
\label{thm:lamhook-phiwknv}
If \begin{enumerate}[label=(\roman*)]
  \item $\Gamma \semi \Pi \semi \Phi \vdash V: A$ and
  \item $\mf{Succ}(\Phi')$,
\end{enumerate}
then $\Gamma \semi \Pi \semi \Phi \Join \Phi' \vdash V: A$.

\begin{proof}
  By induction on the derivation of the typing judgment \emph{(i)}.

\end{proof}
\end{lemma}

\begin{lemma}[$\Phi$-weakening, computations]
\label{thm:lamhook-phiwknc}
If \begin{enumerate}[label=(\roman*)]
  \item $\Gamma \semi \Pi \semi \Phi \vdash M: A \bang E$,
  \item $\mf{Succ}(\Phi')$, and
  \item $\Phi \Join \Phi'$ is defined,
\end{enumerate}
then $\Gamma \semi \Pi \semi \Phi \Join \Phi' \vdash M: A \bang E$.

\begin{proof}
  By induction on the derivation of the typing judgment \emph{(i)}.

\end{proof}
\end{lemma}

\begin{lemma}[$\Phi$-weakening, bodies]
\label{thm:lamhook-phiwknb}
If \begin{enumerate}[label=(\roman*)]
  \item $\Gamma \semi \Pi \semi \Phi \vdash D: \Ell$,
  \item $\mf{Succ}(\Phi')$, and
  \item $\Phi \Join \Phi'$ is defined,
\end{enumerate}
then $\Gamma \semi \Pi \semi \Phi \Join \Phi' \vdash D: \Ell$.

\begin{proof}
  By induction on the derivation of the typing judgment \emph{(i)}.

\end{proof}
\end{lemma}

\subsection{Progress}
\label{appdx:prog}

We begin with a standard canonical forms lemma:

\begin{lemma}[Canonical forms]
\label{thm:lamhook-canon}
If $\emptyset \semi \Pi \semi \Phi \vdash V: A$, then:
\begin{itemize}
  \item if $A = \tone$, then $V = ()$.
  \item if $A = \mf{Str}$, then $V = s$.
  \item if $A = \mf{Bool}$, then $V = \mf{true}$ or $V = \mf{false}$.
  \item if $A = A_1 \xrightarrow{E} A_2$, then $V = \lambda x.M$
    or $V = \setter^\ell_p$.
  \item if $A = A_1 \times A_2$, then $V = (V_1, V_2)$.
  \item if $A = \mf{List}\ B$, then $V = \nil$ or $V = V_1 :: V_2$.
  \item if $A = \mf{Element}$, then $V = \mf{text}\ W$
    or $V = \mf{tag}(V_1, V_2)$ or $V = \mf{component}\ \{D\}$.
  \item if $A = \mf{Attr}$, then $V = \mf{attr}(V_1, V_2)$
    or $v = \mf{onClick}\ W$.
  \item if $A = \mf{ElemInst}$, then $V = \mf{textInst}\ W$
    or $V = \mf{tagInst}_{id}(V_1, V_2)$ or $V = \compInst{p}{D}{W}$.
\end{itemize}
\begin{proof}
  By case-analysis on the derivation of the typing judgment.
  Since $\Gamma = \emptyset$, rule \textsc{TV-Var} is impossible.

\end{proof}
\end{lemma}

We then proceed to proofs of progress for the three mutually-defined
step judgments.

\begin{restatable}[Progress for $\step$]{theorem}{lamhookprogc}
\label{thm:lamhook-progc}
If $\emptyset \semi \Pi \semi \Phi_H \semi \Phi_D
  \vdash \heap \semi M: A \bang E$,
then either:
\begin{enumerate}[label=(\roman*)]
\item $M=\mf{val}\ V$ or
\item $\heap \semi M \step N \semi \heap'$ for some $N$ and $\heap'$.
\end{enumerate}
\end{restatable}

\begin{proof}
  By definition of the premise, we have
  \begin{enumerate}[label=(\roman*)]
  \item $\emptyset \semi \Pi \semi \Phi_M \vdash M: A \bang E$ and
  \item $\emptyset \semi \Pi \semi \Phi_H \vDash \heap$.
  \end{enumerate}

  We proceed by mutual induction on the derivation of the computation
  typing judgment (i) alongside Lemmas \ref{thm:lamhook-progbi}
  and \ref{thm:lamhook-progbs}.

  \pcase{TC-Val}: We trivially satisfy conclusion \emph{(i)}.
  \begin{mathpar}
    \inferrule*[right=TC-Val] {
      \emptyset \semi \Pi \semi \Phi_M \vdash V: A
    } {
      \emptyset \semi \Pi \semi \Phi_M \vdash \mf{val}\ V: A \bang \mf{Pure}
    }
  \end{mathpar}

  \pcase{TC-Let}: We have
  \begin{mathpar}
  \inferrule*[right=TC-Let] {
    \inferrule* {(1)} {
      \emptyset \semi \Pi \semi \Phi_1 \vdash M: A \bang E
    }
    \and
    x: A \semi \Pi \semi \Phi_2 \vdash N: B \bang E
  } {
    \emptyset \semi \Pi \semi \Phi_1 \Join \Phi_2 \vdash \elet{x}{M}{N}:
    B \bang E
  }
  \end{mathpar}
  \begin{itemize}
    \item By inductive hypothesis on (1) and premise (ii),
    we have two subcases: either (a) $M=\mf{val}\ V$
    or (b) $\heap \semi M \step M' \semi \heap'$.
    \item In subcase (a), rule \textsc{EC-Val} says that
    $\heap \semi \elet{x}{\mf{val}\ V}{N} \step N[V/x] \semi \heap$.
    \item In subcase (b), we derive a step as follows:
    \begin{mathpar}
      \inferrule*[right=EC-Let] {
        \inferrule*[right=\textnormal{(b)}] { } {
          \heap \semi M \step M' \semi \heap'
        }
      } {
        \heap \semi \elet{x}{M}{N} \step \elet{x}{M'}{N} \semi \heap'
      }
    \end{mathpar}
    \item Thus, in either subcase we take a step,
    satisfying conclusion \emph{(ii)}.
  \end{itemize}

  \pcase{TC-App}: We have
  \begin{mathpar}
  \inferrule*[right=TC-App] {
    \inferrule* {(1)} {
      \emptyset \semi \Pi \semi \Phi_1 \vdash V: A \xrightarrow{E} B
    }
    \and
    \emptyset \semi \Pi \semi \Phi_2 \vdash W: A_1
  } {
    \emptyset \semi \Pi \semi \Phi_1 \Join \Phi_2 \vdash V\ W: B \bang E
  }
  \end{mathpar}
  \begin{itemize}
    \item By Lemma \ref{thm:lamhook-canon} (canonical forms) on (1),
    we have two subcases: either (a) $V = \lambda x.N$
    or (b) $V = \setter^\ell_p$.
    \item In subcase (a), we apply \textsc{EC-App}
    to get $\heap \semi (\lambda x.N)\ V \step N[V/x] \semi \heap$.
    \item In subcase (b), we continue our bottom-up derivation:
      \begin{mathpar}
        \inferrule*[right=TV-Set] {
          \inferrule* {(2)} {\ell: A \in \Pi[p]}
          \and \mf{Succ}(\Phi_1)
        } {
          \emptyset \semi \Pi \semi \Phi_1 \vdash \setter^\ell_p:
          A \xrightarrow{\mf{Mut}} B
        }
      \end{mathpar}
    \item From (2) and premise (ii), we know that
      $\ell \in \dom(\heap[p].\mf{st})$.
    \item Thus, $\heap' = \heap \mid \heap[p].\mf{st}[\ell] = V \text{ and }
        \heap[p].\mf{rerender} = \mf{tt}$ is well-defined.
    \item We can therefore apply \textsc{EC-AppSet}
      to get $\heap \semi \setter^\ell_p\ V \step \mf{val}\ () \semi \heap'$.
    \item Thus, in either subcase we take a step,
    satisfying conclusion \emph{(ii)}.
  \end{itemize}

  \pcase{TC-If}: We have
  \begin{mathpar}
  \inferrule*[right=TC-If] {
    \inferrule* {(1)} {\emptyset \semi \Pi \semi \Phi_1 \vdash V: \mf{Bool}}
    \and
    \emptyset \semi \Pi \semi \Phi_2 \vdash M: A \bang E
    \and
    \emptyset \semi \Pi \semi \Phi_2 \vdash N: A \bang E
  } {
    \emptyset \semi \Pi \semi \Phi_1 \Join \Phi_2 \vdash \eif{V}{M}{N}:
    A \bang E
  }
  \end{mathpar}
  \begin{itemize}
    \item By Lemma \ref{thm:lamhook-canon} (canonical forms) on (1),
    we have two subcases: either (a) $V = \mf{true}$
    or (b) $V = \mf{false}$.
    \item In subcase (a), we apply \textsc{EC-IfT}
    to get $\heap \semi \eif{\mf{true}}{M}{N} \step M \semi \heap$.
    \item In subcase (b), we apply \textsc{EC-IfF}
    to get $\heap \semi \eif{\mf{false}}{M}{N} \step N \semi \heap$.
    \item Thus, in either subcase we take a step,
    satisfying conclusion \emph{(ii)}.
  \end{itemize}

  \pcase{TC-Fst}: We have
  \begin{mathpar}
    \inferrule*[right=TC-Fst] {
      \inferrule* {(1)} {
        \emptyset \semi \Pi \semi \Phi_M \vdash V: A \times B
      }
    } {
      \emptyset \semi \Pi \semi \Phi_M \vdash \mf{fst}\ V: A \bang \mf{Pure}
    }
  \end{mathpar}
  \begin{itemize}
    \item By Lemma \ref{thm:lamhook-canon} (canonical forms) on (1),
    we have that $V = (V_1, V_2)$.
    \item Rule \textsc{EC-Fst} then tells us that
    $\heap \semi \mf{fst}\ (V_1, V_2) \step \mf{val}\ V_1 \semi \heap$,
    satisfying conclusion \emph{(ii)}.
  \end{itemize}

  \pcase{TC-Snd}: Analogous to the above case, via \textsc{EC-Snd}.

  \pcase{TC-Fold}: We have
  \begin{mathpar}
    \inferrule*[right=TC-Fold] {
      \emptyset \semi \Pi \semi \Phi_1 \vdash
      V_1: A \xrightarrow{E} B \xrightarrow{E} B
      \and
      \emptyset \semi \Pi \semi \Phi_2 \vdash
      V_2: B
      \and
      \inferrule* {(1)} {
        \emptyset \semi \Pi \semi \Phi_3 \vdash
        V_3: \mf{List}\ A
      }
    } {
      \emptyset \semi \Pi \semi \Phi_1 \Join \Phi_2 \Join \Phi_3 \vdash
      \mf{fold}\ V_1\ V_2\ V_3: B \bang \mf{Mut}
    }
  \end{mathpar}
  \begin{itemize}
    \item By Lemma \ref{thm:lamhook-canon} (canonical forms) on (1),
      we have two subcases: either (a) $V_3 = []$ or (b) $V_3 = V_h :: V_t$.
    \item In subcase (a), we apply \textsc{EC-FdN} to get
      $\heap \semi \mf{fold}\ V_1\ V_2\ [] \step \mf{val}\ V_2 \semi \heap$.
    \item In subcase (v), we apply \textsc{EC-FdC} to get
      $\heap \semi \mf{fold}\ V_1\ V_2\ (V_h :: V_t)
      \step V_1\ V_h\ (\mf{fold}\ V_1\ V_2\ V_t) \semi \heap$.
    \item Thus, in either subcase we take a step,
      satisfying conclusion \emph{(ii)}.
  \end{itemize}

  \pcase{TC-Sub}: We have
  \begin{mathpar}
  \inferrule*[right=TC-Sub] {
    \inferrule* {(1)} {
      \emptyset \semi \Pi \semi \Phi_M \vdash M: A \bang \mf{Pure}
    }
  } {
    \emptyset \semi \Pi \semi \Phi_M \vdash M: A \bang \mf{Mut}
  }
  \end{mathpar}
  \begin{itemize}
    \item The result follows directly from the inductive hypothesis
    applied to (1) and premise \emph{(ii)}.
  \end{itemize}

  \pcase{TC-EvalI}: We have
  \begin{mathpar}
    \inferrule*[right=TC-EvalI] {
      \inferrule* {(1)} {
        \emptyset \semi \Pi \semi \Phi_1 \vdash D_0: \Ell \uplus \Pi[p]
      }
      \and
      \inferrule* {(2)} {
        \emptyset \semi \Pi \semi \Phi_2 \vdash D: \Ell
      }
    } {
      \emptyset \semi \Pi \semi \Phi_1 \Join \Phi_2,p: \mf{Init} \vdash
      \eval{D_0}{\mf{Init}}{p}{D}: \mf{ElemInst} \bang \mf{Mut}
    }
  \end{mathpar}
  \begin{itemize}
    \item By (1), we have (3) that $p \in \dom(\Pi)$.
    \item By mutual inductive hypothesis (Theorem \ref{thm:lamhook-progbi})
    on (2), (3), and premise \emph{(ii)}, we have two subcases:
    either (a) $D = \mf{val}\ V$
    or (b) $\heap \semi D \step^\mf{Init}_p D' \semi \heap'$.
    \item In subcase (a), we apply \textsc{EC-EvalDone} to get\\
    $\heap \semi \eval{D_0}{\mf{Init}}{p}{\mf{val}\ V} \step
    \elet{x}{\mf{init}\ V}{\mf{val}\ \compInst{p}{D_0}{x}} \semi \heap$.
    \item In subcase (b), we derive a step as follows:
    \begin{mathpar}
      \inferrule*[right=EC-Eval] {
        \inferrule*[right=\textnormal{(b)}] { } {
          \heap \semi D \step^\mf{Init}_p D' \semi \heap'
        }
      } {
        \heap \semi \eval{D_0}{\mf{Init}}{p}{D}
        \step
        \eval{D_0}{\mf{Init}}{p}{D'} \semi \heap'
      }
    \end{mathpar}
    \item Thus, in either subcase we take a step,
    satisfying conclusion \emph{(ii)}.
  \end{itemize}

  \pcase{TC-EvalS}: We have
  \begin{mathpar}
    \inferrule*[right=TC-EvalS] {
      \emptyset \semi \Pi \semi \Phi_1 \vdash D_0: \Pi[p]
      \quad
      \inferrule* {(1)} {\emptyset \semi \Pi \semi \Phi_2 \vdash D: \Ell}
      \
      \inferrule* {(2)} {\Ell \subseteq \Pi[p]}
      \quad
      (\Phi_1 \Join \Phi_2)[p] = \mf{Succ}
    } {
      \emptyset \semi \Pi \semi \Phi_1 \Join \Phi_2 \vdash
      \eval{D_0}{\mf{Succ}}{p}{D}: \mf{ElemInst} \bang \mf{Mut}
    }
  \end{mathpar}
  \begin{itemize}
    \item By mutual inductive hypothesis (Theorem \ref{thm:lamhook-progbs})
    on (1), (2), and premise \emph{(ii)}, we have two subcases:
    either (a) $D = \mf{val}\ V$
    or (b) $\heap \semi D \step^\mf{Succ}_p D' \semi \heap'$.
    \item In subcase (a), we apply \textsc{EC-EvalDone} to get\\
    $\heap \semi \eval{D_0}{\mf{Succ}}{p}{\mf{val}\ V} \step
    \elet{x}{\mf{init}\ V}{\mf{val}\ \compInst{p}{D_0}{x}} \semi \heap$.
    \item In subcase (b), we derive a step as follows:
    \begin{mathpar}
      \inferrule*[right=EC-Eval] {
        \inferrule*[right=\textnormal{(b)}] { } {
          \heap \semi D \step^\mf{Succ}_p D' \semi \heap'
        }
      } {
        \heap \semi \eval{D_0}{\mf{Succ}}{p}{D}
        \step
        \eval{D_0}{\mf{Succ}}{p}{D'} \semi \heap'
      }
    \end{mathpar}
    \item Thus, in either subcase we take a step,
    satisfying conclusion \emph{(ii)}.
  \end{itemize}

  \pcase{TC-Init}: We have
  \begin{mathpar}
    \inferrule*[right=TC-Init] {
      \inferrule* {(1)} {
        \emptyset \semi \Pi \semi \Phi_M \vdash V: \mf{Element}
      }
    } {
      \emptyset \semi \Pi \semi \Phi_M \vdash \mf{init}\ V:
      \mf{ElemInst} \bang \mf{Mut}
    }
  \end{mathpar}
  \begin{itemize}
    \item By Lemma \ref{thm:lamhook-canon} (canonical forms) on (1),
    we have three subcases:
    either (a) $V = \mf{text}\ W$,
    (b) $V = \mf{tag}(V_1, V_2)$,
    or (c) $V = \mf{component}\ D$.
    \item In subcase (a), we apply \textsc{EC-InitText}
    to get $\heap \semi \mf{init\ text}\ W
    \step \mf{val\ textInst}\ W \semi \heap$.
    \item In subcase (b), we apply \textsc{EC-InitTag} to get\\
    $\heap \semi \mf{init}\ \mf{tag}(V_1, V_2) \step
    \elet{x}{\mf{inits}\ V_2}{\mf{val}\ \mf{tagInst}_{id}(V_1, x)}
    \semi \heap$ for a fresh ${id}$.
    \item In subcase (c), we apply \textsc{EC-InitComp}
    to get $\heap \semi \mf{init}\ \mf{component}\ \{D\}
    \step \eval{D}{\mf{Init}}{p}{D} \semi \heap'$
    for some $\heap'$ and a fresh $p$.
    \item Thus, in any subcase we take a step,
    satisfying conclusion \emph{(ii)}.
  \end{itemize}

  \pcase{TC-Click}: We have
  \begin{mathpar}
    \inferrule*[right=TC-Click] {
      \inferrule* {(1)} {
        \emptyset \semi \Pi \semi \Phi_M \vdash V: \mf{ElemInst}
      }
    } {
      \emptyset \semi \Pi \semi \Phi_M \vdash \mf{click}_{id}\ V:
      \tone \bang \mf{Mut}
    }
  \end{mathpar}
  \begin{itemize}
    \item By Lemma \ref{thm:lamhook-canon} (canonical forms) on (1),
    we have three subcases:
    either (a) $V = \mf{textInst}\ W$,
    (b) $V = \mf{tagInst}_{id'}(V_1, V_2)$,
    or (c) $V = \compInst{p}{D}{W}$.
    \item In subcase (a), we apply \textsc{EC-ClickText}
    to get $\heap \semi \mf{click}_{id}\ \mf{textInst}\ W
    \step \mf{val}\ () \semi \heap$.
    \item In subcase (b), we further subcase on whether ${id}={id}'$:
    \begin{itemize}
      \item If so, then \textsc{EC-ClickTagT} gives that\\
      $\heap \semi \mf{click}_{id}\ \mf{tagInst}_{id}(V_1, V_2)
      \step \mf{runClicks}\ V_1 \semi \mf{clicks}_{id}\ V_2 \semi \heap$.
      \item If not, then \textsc{EC-ClickTagF} gives that
      $\heap \semi \mf{click}_{id}\ \mf{tagInst}_{{id}'}(V_1, V_2)
      \step \mf{clicks}_{id}\ V_2 \semi \heap$.
    \end{itemize}
    \item In subcase (c), \textsc{EC-ClickComp}
    gives that $\heap \semi \mf{click}_{id}\ \compInst{p}{D}{W}
    \step \mf{click}_{id}\ W \semi \heap$.
    \item Thus, in any subcase we take a step,
    satisfying conclusion \emph{(ii)}.
  \end{itemize}

  \pcase{TC-RunClick}: We have
  \begin{mathpar}
    \inferrule*[right=TC-RunClick] {
      \inferrule* {(1)} {
        \emptyset \semi \Pi \semi \Phi_M \vdash V: \mf{Attr}
      }
    } {
      \emptyset \semi \Pi \semi \Phi_M \vdash \mf{runClick}\ V:
      \tone \bang \mf{Mut}
    }
  \end{mathpar}
  \begin{itemize}
    \item By Lemma \ref{thm:lamhook-canon} (canonical forms) on (1),
    we have two subcases:
    either (a) $V = \mf{attr}(V_1, V_2)$
    or (b) $V = \mf{onClick}\ W$.
    \item In subcase (a), we apply \textsc{EC-RunAttr} to get
    $\heap \semi \mf{runClick\ attr}(V_1, V_2) \step \mf{val}\ () \semi \heap$.
    \item In subcase (v), we apply \textsc{EC-RunClick} to get
    $\heap \semi \mf{runClick\ onClick}\ W \step W\ () \semi \heap$.
    \item Thus, in either subcase we take a step,
    satisfying conclusion \emph{(ii)}.
  \end{itemize}

  \pcase{TC-Update}: We have
  \begin{mathpar}
    \inferrule*[right=TC-Update] {
      \inferrule* {(1)} {
        \emptyset \semi \Pi \semi \Phi_M \vdash V: \mf{ElemInst}
      }
    } {
      \emptyset \semi \Pi \semi \Phi_M \vdash \mf{update}\ V:
      \mf{ElemInst} \bang \mf{Mut}
    }
  \end{mathpar}
  \begin{itemize}
    \item By Lemma \ref{thm:lamhook-canon} (canonical forms) on (1),
    we have three subcases:
    either (a) $V = \mf{textInst}\ W$,
    (b) $V = \mf{tagInst}_{id}(V_1, V_2)$,
    or (c) $V = \compInst{p}{D}{W}$.
    \item In subcase (a), \textsc{EC-UpdateText} gives that
    $\heap \semi \mf{update\ textInst}\ W \step
    \mf{val\ textInst}\ W \semi \heap$.
    \item In subcase (b), we apply \textsc{EC-UpdateTag} to get\\
    $\heap \semi \mf{update\ tagInst}_{id}(V_1, V_2) \step
    \elet{x}{\mf{updates}\ V_2}{\mf{val\ tagInst}_{id}(V_1, x)}
    \semi \heap$.
    \item In subcase (c), we continue derivation from (1) as follows
      (where $\Phi_M = \Phi_1 \Join \Phi_2$):
    \begin{mathpar}
      \inferrule*[right=TV-CompI] {
        \inferrule* {(2)} {\Gamma \semi \Pi \semi \Phi_1 \vdash D: \Pi[p]}
        \and
        \Phi_M[p] = \mf{Succ}
        \and
        \Gamma \semi \Pi \semi \Phi_2 \vdash W: \mf{ElemInst}
      } {
        (1) \quad \Gamma \semi \Pi \semi \Phi_M
        \vdash \compInst{p}{D}{W}: \mf{ElemInst}
      }
    \end{mathpar}
    \item (2) gives that $\Pi[p]$ is defined. Then by premise \emph{(ii)},
    we know that $\heap[p].\mf{rerender}$ is defined.
    We further subcase on the value of $\heap[p].\mf{rerender}$:
    \begin{itemize}
      \item If $\heap[p].\mf{rerender} = \mf{tt}$, we apply
      \textsc{EC-UpdateCompT}
      to get\\ $\heap \semi \mf{update}\ \compInst{p}{D}{W}
      \step \eval{D}{\mf{Succ}}{p}{D} \semi \heap'$
      for some $\heap'$.
      \item If instead $\heap[p].\mf{rerender} = \mf{ff}$,
      we apply \textsc{EC-UpdateCompF} to get\\
      $\heap \semi \mf{update}\ \compInst{p}{D}{W}
      \step \elet{x}{\mf{update}\ W}{
        \mf{val}\ \compInst{p}{D}{x}
      } \semi \heap$.
    \end{itemize}
    \item Thus, in any subcase we take a step,
    satisfying conclusion \emph{(ii)}.
  \end{itemize}

\end{proof}

\begin{restatable}[Progress for $\step^\mf{Init}_p$]{theorem}{lamhookprogbi}
\label{thm:lamhook-progbi}
If
\begin{enumerate}[label=(\roman*)]
\item $\emptyset \semi \Pi \semi \Phi_H \semi \Phi_D \vdash \heap \semi D: \Ell$
  and
\item $p \in \dom(\Pi)$,
\end{enumerate}
then either:
\begin{enumerate}[label=(\roman*)]
\item $D=\mf{val}\ V$ or
\item $\heap \semi D \step^\mf{Init}_p D' \semi \heap'$
for some $D'$ and $\heap'$.
\end{enumerate}
\end{restatable}
\begin{proof}
  By definition of premise \emph{(i)}, we have subpremises
  \begin{enumerate}[label=(i-\alph*)]
  \item $\emptyset \semi \Pi \semi \Phi_D \vdash D: \Ell$ and
  \item $\emptyset \semi \Pi \semi \Phi_H \vDash \heap$.
  \end{enumerate}

  We proceed by mutual induction on the derivation of the body
  typing judgment (i-a) alongside Lemmas \ref{thm:lamhook-progc}
  and \ref{thm:lamhook-progbs}.

  \pcase{TB-Val}: We immediately satisfy conclusion \emph{(i)}.
  \begin{mathpar}
    \inferrule*[right=TB-Val] {
      \emptyset \semi \Pi \semi \Phi_D \vdash V: \mf{Element}
    } {
      \emptyset \semi \Pi \semi \Phi_D \vdash \mf{val}\ V: \emptyset
    }
  \end{mathpar}

  \pcase{TB-Let}: We have
  \begin{mathpar}
    \inferrule*[right=TB-Let] {
      \inferrule* {(1)} {
        \emptyset \semi \Pi \semi \Phi_1 \vdash M: A \bang \mf{Pure}
      }
      \and
      x: A \semi \Pi \semi \Phi_2 \vdash D: \Ell
    } {
      \emptyset \semi \Pi \semi \Phi_1 \Join \Phi_2 \vdash \elet{x}{M}{D}: \Ell
    }
  \end{mathpar}
  \begin{itemize}
    \item By mutual inductive hypothesis (Theorem \ref{thm:lamhook-progc})
    on (1) and premise (i-b), we have two subcases:
    either (a) $M = \mf{val}\ V$
    or (b) $\heap \semi M \step N \semi \heap'$.
    \item In subcase (a), we apply \textsc{EB-Val} to get
    $\heap \semi \elet{x}{\mf{val}\ V}{D}
    \step^\mf{Init}_p D[V/x] \semi \heap$.
    \item In subcase (b), we derive a step as follows:
    \begin{mathpar}
      \inferrule*[right=EB-Let] {
        \inferrule*[right=\textnormal{(b)}] { } {
          \heap \semi M \step N \semi \heap'
        }
      } {
        \heap \semi \elet{x}{M}{D}
        \step^\mf{Init}_p \elet{x}{N}{D} \semi \heap'
      }
    \end{mathpar}
    \item Thus, in either subcase we take a step,
    satisfying conclusion \emph{(ii)}.
  \end{itemize}

  \pcase{TB-State}: We have
  \begin{mathpar}
    \inferrule*[right=TB-State] {
      \emptyset \semi \Pi \semi \Phi_1 \vdash V: A
      \and
      x: A, x_\mf{set}: A \xrightarrow{\mf{Mut}} \tone \semi \Pi \semi \Phi_2
      \vdash D: \Ell
      \and
      \ell \notin \dom(\Ell)
    } {
      \emptyset \semi \Pi \semi \Phi_1 \Join \Phi_2
      \vdash \estate{x}{\ell}{V}{D}: \Ell, \ell: A
    }
  \end{mathpar}
  \begin{itemize}
    \item Say $\heap' = \heap \mid \heap[p].\mf{st}[\ell] \mapsto V$.
    $\heap'$ is only well-defined if $p \in \dom(\heap)$:
    this is given by premise \emph{(ii)}.
    \item We can then apply \textsc{EB-StateI}
    to satisfy conclusion \emph{(ii)}:
    $$\heap \semi \estate{x}{\ell}{V}{D}
    \step^\mf{Init}_p D[V/x, \setter^\ell_p/x_\mf{set}]
    \semi \heap'$$
  \end{itemize}
\end{proof}

\begin{restatable}[Progress for $\step^\mf{Succ}_p$]{theorem}{lamhookprogbs}
\label{thm:lamhook-progbs}
If
\begin{enumerate}[label=(\roman*)]
\item $\emptyset \semi \Pi \semi \Phi \vdash \heap \semi D: \Ell$ and
\item $\Ell \subseteq \Pi[p]$,
\end{enumerate}
then either:
\begin{enumerate}[label=(\roman*)]
\item $D=\mf{val}\ V$ or
\item $\heap \semi D \step^\mf{Succ}_p D' \semi \heap'$
for some $D'$ and $\heap'$.
\end{enumerate}
\end{restatable}
\begin{proof}
  By definition of premise \emph{(i)}, we have subpremises
  \begin{enumerate}[label=(i-\alph*)]
  \item $\emptyset \semi \Pi \semi \Phi_D \vdash D: \Ell$ and
  \item $\emptyset \semi \Pi \semi \Phi_H \vDash \heap$.
  \end{enumerate}

  We proceed by mutual induction on the derivation of the body
  typing judgment (i-a) alongside Lemmas \ref{thm:lamhook-progc}
  and \ref{thm:lamhook-progbi}.

  \textbf{Cases} \textsc{TB-Val} and \textsc{TB-Let}
  are shown analogously to Theorem \ref{thm:lamhook-progbi}
  (progress for $\step^\mf{Init}_p$).

  \pcase{TB-State}: We have
  \begin{mathpar}
    \inferrule*[right=TB-State] {
      \emptyset \semi \Pi \semi \Phi_1 \vdash V: A
      \and
      x: A, x_\mf{set}: A \xrightarrow{\mf{Mut}} \tone \semi \Pi \semi \Phi_2
      \vdash D: \Ell
      \and
      \ell \notin \dom(\Ell)
    } {
      \emptyset \semi \Pi \semi \Phi_1 \Join \Phi_2
      \vdash \estate{x}{\ell}{V}{D}: \Ell, \ell: A
    }
  \end{mathpar}
  \begin{itemize}
    \item We want to apply rule \textsc{EB-StateS}, which says:
    \[\heap \semi \estate{x}{\ell}{V}{D}
    \step^\mf{Succ}_p D[\heap[p].\mf{st}[\ell]/x, \setter^\ell_p/x_\mf{set}]
    \semi \heap \]
    To do so, it suffices to show that $\heap[p].\mf{st}[\ell]$ is well-defined.
    \item Premise \emph{(ii)}, alongside the type of the body,
    tells us that $\ell: A \in \Pi[p]$. Then by premise (i-b),
    we have that $\Gamma \semi \Pi \semi \Phi_3
    \vdash \heap[p].\mf{st}[\ell]: A$,
    i.e. $\heap[p].\mf{st}[\ell]$ is well-defined.
    \item We can thus apply \textsc{EB-StateS}, satisfying
    conclusion \emph{(ii)}, and we are done.
  \end{itemize}
\end{proof}

We may now prove progress for the top-level step judgment
using the above results:

\lamhookprogtop*

\begin{proof}\
  \begin{itemize}
    \item By Theorem \ref{thm:lamhook-progc} (progress for $\step$)
    on the premise, we either have (a) $M = \mf{val}\ V$
    or (b) $\heap \semi M \step N \semi \heap'$.
    \item In case (a), we apply rule \textsc{ET-Event} to get that
    $\heap \semi \mf{val}\ V \step_\top
    \mf{click}_{id}\ V \semi \mf{update}\ V \semi \heap$.
    \item In case (b), we derive a step as follows:
    \begin{mathpar}
      \inferrule*[right=ET-Lift] {
        \inferrule*[right=\textnormal{(b)}] { } {
          \heap \semi M \step N \semi \heap'
        }
      } {
        \heap \semi M \step_\top N \semi \heap'
      }
    \end{mathpar}
  \item Thus, in either case, we take a step and satisfy the conclusion.
  \end{itemize}
\end{proof}

\subsection{Preservation}
\label{appdx:prez}

We begin with a couple of auxiliary lemmas regarding typing contexts.

\freshphi*
\begin{proof}
  We first want to show (c1) that
  $\Phi[p] \in \Phi_1' \Join \Phi_2$ for all $p \in \dom(\Phi)$.
  We case-analyze on the origin of $p$ via the definition of $\Join$:

  \textbf{Case} $p: \mf{Init} \in \Phi_1$:
  We must also have $p \notin \dom(\Phi_2)$.
  By the second premise, we have some $\phi \sqsupseteq \mf{Init}$
  with $p: \phi \in \Phi_1'$.
  Then by definition, $p: \phi \in \Phi_1' \Join \Phi_2$.

  \textbf{Case} $p: \mf{Succ} \in \Phi_1$:
  By the second premise, we have $p: \mf{Succ} \in \Phi_1'$.
  We also must have either $p \notin \dom(\Phi_2)$ or $\Phi_2(p) = \mf{Succ}$.
  In either case, by definition, $p: \mf{Succ} \in \Phi_1' \Join \Phi_2$.

  \textbf{Case} $p \notin \dom(\Phi_1)$:
  We must have $p: \phi \in \Phi_2$.
  Then by premise the second premise, since $p$ is \emph{not} fresh,
  we have $p \notin \dom(\Phi_1')$. Therefore, for any $\phi$,
  we get $p: \phi \in \Phi_1' \Join \Phi_2$.

  We have thus shown (c1). It remains to show (c2) that all $p$ in
  $\dom(\Phi_1' \Join \Phi_2) \setminus \dom(\Phi_1 \Join \Phi_2)$ are fresh:
  this follows from the second premise, which gives us that all $p$ in
  $\dom(\Phi_1') \setminus \dom(\Phi_1)$ are fresh, alongside the fact that
  $\dom(\Phi_1 \Join \Phi_2) = \dom(\Phi_1) \cup \dom(\Phi_2)$.

  The conclusion is thus given by (c1) and (c2), and we are done.
\end{proof}

\begin{lemma}[$\Pi$-order preservation]
\label{thm:lamhook-piprez}
If \begin{enumerate}[label=(\roman*)]
  \item $\Pi \sqsubseteq_{\Phi_1} \Pi'$ and
  \item $\Phi = \Phi_1 \Join \Phi_2$,
\end{enumerate}
then $\Pi \sqsubseteq_\Phi \Pi'$.

\begin{proof}
  By case-analysis on the definition of premise \emph{(i)}.
  If $\Pi = \Pi'$, then the conclusion is trivial.

  Otherwise, for some number of $p$, we have $p: \mf{Init} \in \Phi_1$.
  Consider an arbitrary such $p$: by premise \emph{(i)},
  we have $\Pi[p] \subseteq \Pi'[p]$.
  By definition of $\Join$, we must also have $p: \mf{Init} \in \Phi$.
  It therefore remains that $\Pi \sqsubseteq_{\Phi} \Pi'$ by definition.

  Therefore, in any case, we have $\Pi \sqsubseteq_\Phi \Pi'$.
\end{proof}
\end{lemma}

We prove substitution lemmas as usual, which are fairly standard.
For values and computations, we have as additional premises
that the context merge on $\Phi$ is well-defined.

\begin{lemma}[Substitution, values]
\label{thm:lamhook-substv}
If \begin{enumerate}[label=(\roman*)]
  \item $\Gamma \semi \Pi \semi \Phi_1 \vdash V: A$ and
  \item $\Gamma, x: A \semi \Pi \semi \Phi_2 \vdash W: B$,
\end{enumerate}
then $\Gamma \semi \Pi \semi \Phi_1 \Join \Phi_2 \vdash W[V/x]: V$.

\begin{proof}
  By mutual induction on the typing judgment \emph{(ii)}
  alongside Lemmas \ref{thm:lamhook-substc} and \ref{thm:lamhook-substb}.

\end{proof}
\end{lemma}

\begin{lemma}[Substitution, computations]
\label{thm:lamhook-substc}
If \begin{enumerate}[label=(\roman*)]
  \item $\Gamma \semi \Pi \semi \Phi_1 \vdash V: A$,
  \item $\Gamma, x: A \semi \Pi \semi \Phi_2 \vdash M: B \bang E$, and
  \item $\Phi = \Phi_1 \Join \Phi_2$,
\end{enumerate}
then $\Gamma \semi \Pi \semi \Phi \vdash M[V/x]: B \bang E$.

\begin{proof}
  By mutual induction on the typing judgment \emph{(ii)}
  alongside Lemmas \ref{thm:lamhook-substv} and \ref{thm:lamhook-substb}.

\end{proof}
\end{lemma}

\begin{lemma}[Substitution, bodies]
\label{thm:lamhook-substb}
If \begin{enumerate}[label=(\roman*)]
  \item $\Gamma \semi \Pi \semi \Phi_1 \vdash V: A$,
  \item $\Gamma, x: A \semi \Pi \semi \Phi_2 \vdash D: \Ell$, and
  \item $\Phi = \Phi_1 \Join \Phi_2$,
\end{enumerate}
then $\Gamma \semi \Pi \semi \Phi \vdash D[V/x]: \Ell$.

\begin{proof}
  By mutual induction on the typing judgment \emph{(ii)}
  alongside Lemmas \ref{thm:lamhook-substv} and \ref{thm:lamhook-substc}.

\end{proof}
\end{lemma}

We may now proceed to preservation.
We first prove a separate theorem for $E=\mf{Pure}$,
which is then used in the proofs for bodies.
Mutual induction between the theorems is therefore unnecessary.

\begin{restatable}[Preservation for $\step$, pure]{theorem}{lamhookprezcp}
\label{thm:lamhook-prezcp}
If
\begin{enumerate}[label=(\roman*)]
  \item $\Gamma \semi \Pi \semi \Phi \vdash M: A \bang \mf{Pure}$ and
  \item $\heap \semi M \step N \semi \heap'$ (for any $\heap$),
\end{enumerate}
then
\begin{enumerate}[label=(\roman*)]
  \item $\Gamma \semi \Pi \semi \Phi \vdash N: A \bang \mf{Pure}$ and
  \item $\heap' = \heap$.
\end{enumerate}
\end{restatable}
\begin{proof}
  By induction on the derivation of the computation step judgment \emph{(ii)}.

  \pcase{EC-Let}: We have
  \begin{mathpar}
    \inferrule*[right=EC-Let] {
      \inferrule* {(1)} {\heap \semi M \step M' \semi \heap'}
    } {
      \heap \semi \elet{x}{M}{N} \step \elet{x}{M'}{N} \semi \heap'
    }

    \inferrule*[right=TC-Let] {
      \inferrule* {(2)} {
        \Gamma \semi \Pi \semi \Phi_1 \vdash M: A \bang \mf{Pure}
      }
      \and
      \inferrule* {(3)} {
        \Gamma, x: A \semi \Pi \semi \Phi_2 \vdash N: B \bang \mf{Pure}
      }
    } {
      \Gamma \semi \Pi \semi \Phi_1 \Join \Phi_2 \vdash \elet{x}{M}{N}:
      B \bang \mf{Pure}
    }
  \end{mathpar}
  \begin{itemize}
    \item By inductive hypothesis on (1) and (2),
    we get both that $\heap' = \heap$ (satisfying conclusion \emph{(ii)}),
    and (4) that $\Gamma \semi \Pi \semi \Phi_1 \vdash
    M': A \bang \mf{Pure}$.
    \item We derive conclusion \emph{(i)} as follows:
    \begin{mathpar}
      \inferrule*[right=TC-Let] {
        \inferrule* {(4)} {
          \Gamma \semi \Pi \semi \Phi_1 \vdash M': A \bang \mf{Pure}
        }
        \and
        \inferrule* {(3)} {
          \Gamma, x: A \semi \Pi \semi \Phi_2 \vdash N: B \bang \mf{Pure}
        }
      } {
        \Gamma \semi \Pi \semi \Phi_1 \Join \Phi_2 \vdash \elet{x}{M'}{N}:
        B \bang \mf{Pure}
      }
    \end{mathpar}
  \end{itemize}

  \pcase{EC-Val}: We have
  $\heap \semi \elet{x}{\mf{val}\ V}{M} \step M[V/x] \semi \heap$
  (from which conclusion \emph{(ii)} is immediate) and
  \begin{mathpar}
    \inferrule*[right=EC-Let] {
      \inferrule*[right=EC-Val] {
        \inferrule* {(1)} {
          \Gamma \semi \Pi \semi \Phi_1 \vdash V: A
        }
      } {
        \Gamma \semi \Pi \semi \Phi_1 \vdash \mf{val}\ v: A \bang \mf{Pure}
      }
      \and
      \inferrule* {(2)} {
        \Gamma, x: A \semi \Pi \semi \Phi_2 \vdash M: B \bang \mf{Pure}
      }
    } {
      \Gamma \semi \Pi \semi \Phi_1 \Join \Phi_2 \vdash
      \elet{x}{\mf{val}\ V}{M}: B \bang \mf{Pure}
    }
  \end{mathpar}
  \begin{itemize}
    \item Conclusion \emph{(i)} is given by Lemma \ref{thm:lamhook-substc}
    (substitution) on (1) and (2).
  \end{itemize}

  \pcase{EC-App}: We have
  $\heap \semi (\lambda x.M)\ V \step M[V/x] \semi \heap$
  (from which conclusion \emph{(ii)} is immediate) and
  \begin{mathpar}
    \inferrule*[right=TC-App] {
      \inferrule*[right=TV-Lam] {
        \inferrule* {(1)} {
          \Gamma, x: A \semi \Pi \semi \Phi_1 \vdash M: B \bang \mf{Pure}
        }
        \and
        \mf{Succ}(\Phi_1)
      } {
        \Gamma \semi \Pi \semi \Phi_1 \vdash \lambda x.M:
        A \xrightarrow{\mf{Pure}} B
      }
      \and
      \inferrule* {(2)} {
        \Gamma \semi \Pi \semi \Phi_2 \vdash V: A
      }
    } {
      \Gamma \semi \Pi \semi \Phi_1 \Join \Phi_2 \vdash (\lambda x.M)\ V:
      B \bang \mf{Pure}
    }
  \end{mathpar}
  \begin{itemize}
    \item Conclusion \emph{(i)} is given by Lemma \ref{thm:lamhook-substc}
    (substitution) on (1) and (2).
  \end{itemize}

  \pcase{EC-IfT}: We have
  $\heap \semi \eif{\mf{true}}{M}{N} \step M \semi \heap$
  (from which conclusion \emph{(ii)} is immediate) and
  \begin{mathpar}
    \inferrule*[right=TC-If] {
      \inferrule*[right=TV-True] {
        \inferrule* {(1)} {\mf{Succ}(\Phi_1)}
      } {
        \Gamma \semi \Pi \semi \Phi_1 \vdash \mf{true}: \mf{Bool}
      }
      \
      \inferrule* {(2)} {
        \Gamma \semi \Pi \semi \Phi_2 \vdash M: A \bang \mf{Pure}
      }
      \and
      \Gamma \semi \Pi \semi \Phi_2 \vdash N: A \bang \mf{Pure}
    } {
      \Gamma \semi \Pi \semi \Phi_1 \Join \Phi_2
      \vdash \eif{\mf{true}}{M}{N}: A \bang \mf{Pure}
    }
  \end{mathpar}
  \begin{itemize}
    \item Conclusion \emph{(i)} is given by Lemma \ref{thm:lamhook-phiwknc}
    ($\Phi$-weakening) on (2) using (1).
  \end{itemize}

  \pcase{EC-IfF}: Analogous to the above case.

  \pcase{EC-Fst}: We have
  $\heap \semi \mf{fst}\ (V, W) \step \mf{val}\ V \semi \heap$
  (from which conclusion \emph{(ii)} is immediate) and
  \begin{mathpar}
    \inferrule*[right=TC-Fst] {
      \inferrule*[right=TV-Pair] {
        \inferrule* {(1)} {
          \Gamma \semi \Pi \semi \Phi_1 \vdash V: A
        }
        \and
        \inferrule* {(2)} {
          \Gamma \semi \Pi \semi \Phi_2 \vdash W: B
        }
      } {
        \Gamma \semi \Pi \semi \Phi_1 \Join \Phi_2 \vdash (V, W):
        A \times B
      }
    } {
      \Gamma \semi \Pi \semi \Phi_1 \Join \Phi_2
      \vdash \mf{fst}\ (V, W): A \bang \mf{Pure}
    }
  \end{mathpar}
  \begin{itemize}
    \item By Lemma \ref{thm:lamhook-unrestv} (unrestriction) on (2),
    we get (3) that $\mf{Succ}(\Phi_2)$.
    \item By Lemma \ref{thm:lamhook-phiwknv} ($\Phi$-weakening) on (1) and (3),
    we get (4) that $\Gamma \semi \Pi \semi \Phi_1 \Join \Phi_2
    \vdash V: A$.
    \item We derive conclusion \emph{(i)} as follows:
    \begin{mathpar}
      \inferrule*[right=TC-Val] {
        \inferrule*[right=(4)] { } {
          \Gamma \semi \Pi \semi \Phi_1 \Join \Phi_2 \vdash V: A
        }
      } {
        \Gamma \semi \Pi \semi \Phi_1 \Join \Phi_2 \vdash \mf{val}\ V:
        A \bang \mf{Pure}
      }
    \end{mathpar}
  \end{itemize}

  \pcase{EC-Snd}: Analogous to the above case.

  \pcase{EC-FdN}: We have
  $\heap \semi \mf{fold}\ V_1\ V_2\ [] \step \mf{val}\ V_2 \semi \heap$
  (from which conclusion \emph{(ii)} is immediate) and
  \begin{mathpar}
    \inferrule*[right=TC-Fold] {
      \inferrule* {(1)} {
        \Gamma \semi \Pi \semi \Phi_1 \vdash V_1:
        A \xrightarrow{\mf{Pure}} B \xrightarrow{\mf{Pure}} B
      }
      \and
      \inferrule* {(2)} {
        \Gamma \semi \Pi \semi \Phi_2 \vdash V_2: B
      }
      \and
      \inferrule* {(3)} {
        \Gamma \semi \Pi \semi \Phi_3 \vdash []: \mf{List}\ A
      }
    } {
      \Gamma \semi \Pi \semi \Phi_1 \Join \Phi_2 \Join \Phi_3 \vdash
      \mf{fold}\ V_1\ V_2\ []: B \bang \mf{Pure}
    }
  \end{mathpar}
  \begin{itemize}
    \item By Lemma \ref{thm:lamhook-unrestv} (unrestriction) on (1) and (3),
      we get (4) that $\mf{Succ}(\Phi_1)$ and (5) that $\mf{Succ}(\Phi_3)$,
      respectively.
    \item By Lemma \ref{thm:lamhook-phiwknv} ($\Phi$-weakening) on (2)
      using (4) and (5),
      we get (6) that\\
      $\Gamma \semi \Pi \semi \Phi_1 \Join \Phi_2 \Join \Phi_3 \vdash V_2: B$.
    \item We derive conclusion \emph{(i)} as follows:
    \begin{mathpar}
      \inferrule*[right=TC-Val] {
        \inferrule*[right=(6)] { } {
          \Gamma \semi \Pi \semi \Phi_1 \Join \Phi_2 \Join \Phi_3 \vdash
          V_2: B
        }
      } {
        \Gamma \semi \Pi \semi \Phi_1 \Join \Phi_2 \Join \Phi_3 \vdash
        \mf{val}\ V_2: B \bang \mf{Pure}
      }
    \end{mathpar}
  \end{itemize}

  \pcase{EC-FdC}: We have
  $\heap \semi \mf{fold}\ V_1\ V_2\ (V_h :: V_t)
  \step V_1\ V_h\ (\mf{fold}\ V_1\ V_2\ V_t) \semi \heap$,
  whose right-hand term desguars to
  $\elet{x}{V_1\ V_h}{\elet{y}{\mf{fold}\ V_1\ V_2\ V_t}{x\ y}}$.
  Conclusion \emph{(ii)} is immediate, and we also have:
  \begin{mathpar}
    \inferrule*[right=TC-Fold] {
      (1) \and (2) \and
      \inferrule* {
        \inferrule* {(3)} {
          \Gamma \semi \Pi \semi \Phi_3 \vdash V_h: A
        }
        \and
        \inferrule* {(4)} {
          \Gamma \semi \Pi \semi \Phi_4 \vdash V_t: \mf{List}\ A
        }
      } {
        \Gamma \semi \Pi \semi \Phi_3 \Join \Phi_4
        \vdash V_h :: V_t: \mf{List}\ A
      }
    } {
      \Gamma \semi \Pi \semi \Phi_1 \Join \Phi_2 \Join \Phi_3 \Join \Phi_4
      \vdash
      \mf{fold}\ V_1\ V_2\ (V_h :: V_t): B \bang \mf{Pure}
    }
  \end{mathpar}
  where (1) and (2) are identical to the previous case.
  \begin{itemize}
    \item By Lemma \ref{thm:lamhook-unrestv} (unrestriction) on (1)-(4),
      we get that $\mf{Succ}(\Phi_i)$ for $i \in \{1, 2, 3, 4\}$.
      It follows (5) that $\mf{Succ}(\Phi)$, where $\Phi=\Join_{i=1}^4 \Phi_i$.
    \item We derive conclusion \emph{(i)} by way of
      intermediate derivations (6)-(8) as follows,
      using $\Phi$ intuitionistically by (5) and weakening (1)-(4)
      appropriately:
  \end{itemize}
  \begin{mathpar}
    \inferrule*[right=TC-App] {
      \inferrule*[right=TV-Var] { } {
        \Gamma_6 \semi \Pi \semi \Phi
        \vdash x: B \xrightarrow{\mf{Pure}} B
      }
      \and
      \inferrule*[right=TV-Var] { } {
        \Gamma_6 \semi \Pi \semi \Phi
        \vdash y: B
      }
    } {
      (6) \quad
      \Gamma_6 \semi \Pi \semi \Phi \vdash
      x\ y: B \bang \mf{Pure}
    }

    \inferrule*[right=TC-Fold] {
      \inferrule*[right=(1)] { } {
        \Gamma_7 \semi \Pi \semi \Phi \vdash V_1:
        A \xrightarrow{\mf{Pure}} B \xrightarrow{\mf{Pure}} B
      }
      \
      \inferrule*[right=(2)] { } {
        \Gamma_7 \semi \Pi \semi \Phi \vdash V_2: B
      }
      \
      \inferrule*[right=(4)] { } {
        \Gamma_7 \semi \Pi \semi \Phi \vdash V_t: \mf{List}\ A
      }
    } {
      (7) \quad
      \Gamma_7 \semi \Pi \semi \Phi \vdash
      \mf{fold}\ V_1\ V_2\ V_t: B \bang \mf{Pure}
    }

    \inferrule*[right=TC-Let] {
      [(7),\ \Gamma_7=\Gamma_8]
      \and
      [(6),\ \Gamma_6=\Gamma_8, y: B]
    } {
      (8) \quad
      \Gamma_8 \semi \Pi \semi \Phi \vdash
      \elet{y}{\mf{fold}\ V_1\ V_2\ V_t}{\cdots}: B \bang \mf{Pure}
    }

    \inferrule*[right=TC-Let] {
      \inferrule*[Right=TC-App] {
        \inferrule*[right=(1)] { } {
          \Gamma \semi \Pi \semi \Phi \vdash V_1:
          A \xrightarrow{\mf{Pure}} B \xrightarrow{\mf{Pure}} E
        }
        \
        \inferrule*[right=(3)] { } {
          \Gamma \semi \Pi \semi \Phi \vdash V_h: A
        }
      } {
        \Gamma \semi \Pi \semi \Phi \vdash
        V_1\ V_h: B \xrightarrow{\mf{Pure}} B \bang \mf{Pure}
      }
      \
      [(8),\ \Gamma_8 = \Gamma, x: B \xrightarrow{\mf{Pure}} B]
    } {
      \Gamma \semi \Pi \semi \Phi
      \vdash
      \elet{x}{V_1\ V_h}{\cdots}: B \bang \mf{Pure}
    }
  \end{mathpar}

  \textbf{All other cases} are untypeable with the $\mf{Pure}$ effect,
  and so the theorem holds vacuously.
\end{proof}

\begin{restatable}[Preservation for $\step^\mf{Init}_p$]
  {theorem}{lamhookprezbi}
\label{thm:lamhook-prezbi}
If
\begin{enumerate}[label=(\roman*)]
  \item $\Gamma \semi \Pi \semi \Phi_H \semi \Phi_D \vdash \heap \semi D: \Ell$,
  \item $\Ell \cap \Pi[p] = \emptyset$,
  \item $p \notin \dom(\Phi_H \Join \Phi_D)$, and
  \item $\heap \semi D \step^\mf{Init}_p D' \semi \heap'$,
\end{enumerate}
then for some $\Phi_H'$ and $\Pi' \sqsupseteq_{p: \mf{Init}} \Pi$, we have
\begin{enumerate}[label=(\roman*)]
  \item $\Gamma \semi \Pi' \semi \Phi_H' \semi \Phi_D
    \vdash \heap' \semi D': \Ell'$,
  \item $\Phi_H' \Join \Phi_D = \Phi_H \Join \Phi_D$, and
  \item $\Ell' \uplus \Pi'[p] = \Ell \uplus \Pi[p]$.
\end{enumerate}
\end{restatable}
\begin{proof}
  By definition of premise \emph{(i)}, we have subpremises
  \begin{enumerate}[label=(i-\alph*)]
    \item $\Gamma \semi \Pi \semi \Phi_D \vdash D: \Ell$ and
    \item $\Gamma \semi \Pi \semi \Phi_H \vDash \heap$.
  \end{enumerate}
  Likewise, to show conclusion \emph{(i)} it suffices to pick some
  $\Pi' \sqsupseteq_{p: \mf{Init}} \Pi$ and
  $\Phi_H'$ such that conclusion \emph{(ii)} holds and:
  \begin{enumerate}[label=(i-\alph*)]
    \item $\Gamma \semi \Pi' \semi \Phi_D \vdash D': \Ell'$ and
    \item $\Gamma \semi \Pi' \semi \Phi_H' \vDash \heap'$.
  \end{enumerate}

  Unless otherwise mentioned, we choose $\Phi_H'=\Phi_H$
  such that conclusion \emph{(ii)} is trivial,
  and we choose $\Pi' = \Pi$ such that
  conclusion \emph{(iii)} is trivial by premise \emph{(ii)}
  when $\Ell' = \Ell$.
  Under both of these choices, to show conclusion (i-b), by premise (i-b)
  it suffices to show that $\heap' = \heap$.

  We proceed by induction on the derivation of the
  body step judgment \emph{(iv)}.

  \pcase{EB-Let}: We have
  \begin{mathpar}
    \inferrule*[right=EB-Let] {
      \inferrule* {(1)} {\heap \semi M \step N \semi \heap'}
    } {
      \heap \semi \elet{x}{M}{D}
      \step^\mf{Init}_p \elet{x}{N}{D} \semi \heap'
    }

    \inferrule*[right=TB-Let] {
      \inferrule* {(2)} {
        \Gamma \semi \Pi \semi \Phi_1 \vdash M: A \bang \mf{Pure}
      }
      \and
      \inferrule* {(3)} {\Gamma, x: A \semi \Pi \semi \Phi_2 \vdash D: \Ell}
    } {
      \Gamma \semi \Pi \semi \Phi_1 \Join \Phi_2 \vdash \elet{x}{M}{D}: \Ell
    }
  \end{mathpar}
  \begin{itemize}
    \item By Theorem \ref{thm:lamhook-prezcp} (pure preservation)
    on (1) and (2),
    we get that $\heap' = \heap$, satisfying conclusion (i-b),
    and (4) that $\Gamma \semi \Pi \semi \Phi_1
    \vdash N: A \bang \mf{Pure}$.
    \item We derive conclusion (i-a) as follows:
    \begin{mathpar}
      \inferrule*[right=TB-Let] {
        \inferrule*[right=(4)] { } {
          \Gamma \semi \Pi \semi \Phi_1 \vdash N: A \bang \mf{Pure}
        }
        \and
        \inferrule* {(3)} {\Gamma, x: A \semi \Pi \semi \Phi_2 \vdash D: \Ell}
      } {
        \Gamma \semi \Pi \semi \Phi_1 \Join \Phi_2 \vdash \elet{x}{N}{D}: \Ell
      }
    \end{mathpar}
  \end{itemize}

  \pcase{EB-Val}: We have
  $\heap \semi \elet{x}{\mf{val}\ V}{D}
  \step^\mf{Init}_p D[V/x] \semi \heap$
  (from which conclusion (i-b) is immediate) and
  \begin{mathpar}
  \inferrule*[right=TB-Let] {
    \inferrule*[right=TC-Val] {
      \inferrule* {(1)} {\Gamma \semi \Pi \semi \Phi_1 \vdash V: A}
    } {
      \Gamma \semi \Pi \semi \Phi_1 \vdash \mf{val}\ V: A \bang \mf{Pure}
    }
    \and
    \inferrule* {(2)} {\Gamma, x: A \semi \Pi \semi \Phi_2 \vdash D: \Ell}
  } {
    \Gamma \semi \Pi \semi \Phi_1 \Join \Phi_2 \vdash
    \elet{x}{\mf{val}\ V}{D}: \Ell
  }
  \end{mathpar}
  \begin{itemize}
    \item Conclusion (i-a) is given by Lemma \ref{thm:lamhook-substb}
    (substitution) on (1) and (2).
  \end{itemize}

  \pcase{EB-StateI}: We have
  $\heap \semi \estate{x}{\ell}{V}{D}
  \step^\mf{Init}_p D[V/x, \setter^\ell_p/x_\mf{set}]
  \semi \heap'$,
  where $\heap' = \heap \mid \heap[p].\mf{st}[\ell] \mapsto V$, and
  \begin{mathpar}
  \inferrule*[right=TB-State] {
    \inferrule* {(1)} {\Gamma \semi \Pi \semi \Phi_1 \vdash V: A}
    \and
    \inferrule* {(2)} {
      \Gamma, x: A, x_\mf{set}: A \xrightarrow{\mf{Mut}} \tone \semi \Pi
      \semi \Phi_2 \vdash D: \Ell
    }
    \and
    \inferrule* {(3)} {\ell \notin \dom(\Ell)}
  } {
    \Gamma \semi \Pi \semi \Phi_1 \Join \Phi_2
    \vdash \estate{x}{\ell}{V}{D}: \Ell, \ell: A
  }
  \end{mathpar}
  \begin{itemize}
    \item Choose $\Pi' = \Pi \mid \Pi[p][\ell] = A$.
    By definition, (4) $\Pi \sqsubseteq_{p: \mf{Init}} \Pi'$.
    \item Choose $\Phi_H' = \Phi_H \Join \Phi_1$.
    \begin{itemize}
      \item Since $(\Phi_1 \Join \Phi_2) \Join \Phi_H$ is given as defined,
      we know $\Phi_1 \Join \Phi_H$ is defined.
      \item To show conclusion \emph{(ii)},
      we want to show $(\Phi_1 \Join \Phi_2) \Join (\Phi_H \Join \Phi_1)
      = (\Phi_1 \Join \Phi_2) \Join \Phi_H$.
      \item These are equal by associativity/commutativity of $\Join$
      and by Lemma \ref{thm:lamhook-unrestv} (unrestriction)
      on (1), meaning $\Phi_1 \Join \Phi_1 = \Phi_1$.
    \end{itemize}
    \item We show conclusion (i-a) as follows:
    \begin{itemize}
      \item By premise \emph{(iii)}, we have (5)
        that $p: \mf{Init} \Join \Phi_1$ is defined
        and (6) that $p: \mf{Init} \Join \Phi_2$ is defined.
      \item By Lemma \ref{thm:lamhook-piwknv} ($\Pi$-weakening)
      on (1), (4), and (5),
      we get (7) that $\Gamma \semi \Pi' \semi \Phi_1 \vdash V: A$.
      \item By Lemma \ref{thm:lamhook-piwknb} ($\Pi$-weakening)
      on (2), (4), and (6), we get (8) that\\$\Gamma, x: A, x_\mf{set}:
      A \xrightarrow{\mf{Mut}} \tone \semi \Pi' \semi \Phi_2 \vdash D: \Ell$.
      \item Since $\ell: A \in \Pi'[p]$ by definition and $\mf{Succ}(\emptyset)$
      trivially, we derive the setter type as follows:
      \begin{mathpar}
      \inferrule*[right=TV-Setter] {
        \ell: A \in \Pi'[p]
        \and
        \mf{Succ}(\emptyset)
      } {
        (9) \quad \Gamma \semi \Pi' \semi \emptyset \vdash \setter^\ell_p:
        A \xrightarrow{\mf{Mut}} \tone
      }
      \end{mathpar}
      \item We apply Lemma \ref{thm:lamhook-substb} (substitution) twice to (8),
      using (7) and (9), to get that\\
      $\Gamma \semi \Pi' \semi \Phi_1 \Join \Phi_2 \vdash
      D[V/x, \setter^\ell_p/x_\mf{set}]: \Ell$,
      i.e. conclusion (i-a) holds.
    \end{itemize}
    \item We show conclusion (i-b) as follows:
    \begin{itemize}
      \item We first want (10) that $\dom(\Pi') = \dom(\heap')$.
      \begin{itemize}
        \item By definition of each, it suffices to show that
        $\dom(\Pi) \cup \{p\} = \dom(\heap) \cup \{p\}$.
        \item Premise (i-b) says that $\dom(\Pi) = \dom(\heap)$,
        so (10) holds.
      \end{itemize}
      \item We also want (11) that
      $\dom(\Pi'[p]) = \dom(\heap'[p].\mf{st})$.
      \begin{itemize}
        \item By definition of each, it suffices to show
        $\dom(\Pi[p]) \cup \{\ell\} = \dom(\heap[p].\mf{st}) \cup \{\ell\}$.
        \item Premise (i-b) says that
        $\dom(\Pi[p]) = \dom(\heap[p].\mf{st})$,
        so (11) also holds.
      \end{itemize}
      \item Finally, we want (12) that
      $\Gamma \semi \Pi' \semi \Phi_1
      \vdash \heap'[p].\mf{st}[\ell]: \Pi'[p][\ell]$. This is given by (7).
      \item By premise (i-b) and by definition,
      (10)-(12) are sufficient to show that\\
      $\Gamma \semi \Pi' \semi \Phi_H \Join \Phi_1 \vDash \heap'$,
      i.e. conclusion (i-b) holds.
    \end{itemize}
    \item Finally, we show conclusion \emph{(iii)} as follows:
    \begin{itemize}
      \item We want to show that
      $\Pi'[p] \uplus \Ell = \Pi[p] \uplus \Ell, \ell: A$.
      \item By definition, it suffices to show that
      $\Pi[p], \ell: A \uplus \Ell = \Pi[p] \uplus \Ell, \ell: A$.
      \item It further suffices to show $\Pi[p], \ell: A \uplus \Ell$
      is well-defined, i.e $\Pi[p], \ell: A \cap \Ell = \emptyset$.
      This follows from (3) and premise \emph{(ii)}.
    \end{itemize}
  \end{itemize}

  \pcase{EB-StateS}: Impossible,
  since the step rule requires that $\phi=\mf{Succ}$.

\end{proof}

\begin{restatable}[Preservation for $\step^\mf{Succ}_p$]
  {theorem}{lamhookprezbs}
\label{thm:lamhook-prezbs}
If
\begin{enumerate}[label=(\roman*)]
  \item $\Gamma \semi \Pi \semi \Phi_H \semi \Phi_D \vdash \heap \semi D: \Ell$,
  \item $\Ell \subseteq \Pi[p]$, and
  \item $\heap \semi D \step^\mf{Succ}_p D' \semi \heap'$,
\end{enumerate}
then for some $\Phi_D'$, we have
\begin{enumerate}[label=(\roman*)]
  \item $\Gamma \semi \Pi \semi \Phi_D' \vdash D': \Ell'$,
  \item $\Phi_H \Join \Phi_D' = \Phi_H \Join \Phi_D$,
  \item $\heap' = \heap$, and
  \item $\Ell' \subseteq \Ell$.
\end{enumerate}
\end{restatable}
\begin{proof}
  By induction on the derivation of the body step judgment
  \emph{(iii)}.

  \textbf{Cases} \textsc{EB-Let}, and \textsc{EB-Val}
  are shown analogously to Theorem \ref{thm:lamhook-prezbi}
  (preservation for $\step^\mf{Init}_p$).
  Choosing $\Phi_D' = \Phi_D$ trivializes conclusion \emph{(ii)},
  and we always get $\Ell' = \Ell$ to satisfy conclusion \emph{(iv)}.

  \pcase{EB-StateI}: Impossible,
  since the step rule requires that $\phi=\mf{Init}$.

  \pcase{EB-StateS}: We have
  $\heap \semi \estate{x}{\ell}{V}{D}
  \step^\mf{Succ}_p D[W/x, \setter^\ell_p/x_\mf{set}]
  \semi \heap$ (from which conclusion \emph{(iii)} is immediate),
  where $W = \heap[p].\mf{st}[\ell]$, and
  \begin{mathpar}
  \inferrule*[right=TB-State] {
    \inferrule* {(1)} {
      \Gamma \semi \Pi \semi \Phi_1 \vdash V: A
    }
    \and
    \inferrule* {(2)} {
      \Gamma, x: A, x_\mf{set}: A \xrightarrow{\mf{Mut}} \tone \semi \Pi
      \semi \Phi_2 \vdash D: \Ell
    }
    \and
    \ell \notin \dom(\Ell)
  } {
    \Gamma \semi \Pi \semi \Phi_1 \Join \Phi_2 \vdash
    \estate{x}{\ell}{V}{D}: \Ell, \ell: A
  }
  \end{mathpar}
  \begin{itemize}
    \item By premise \emph{(ii)},
    we know $\Ell, \ell: A \subseteq \Pi[p]$,
    i.e. (3) $\Pi[p][\ell]=A$.
    \item We derive the setter type as follows
    (noting $\mf{Succ}(\emptyset)$ holds trivially):
    \begin{mathpar}
      \inferrule*[right=TV-Setter] {
        \inferrule*[right=(3)] { } {\ell: A \in \Pi[p]}
        \and
        \mf{Succ}(\emptyset)
      } {
        (4) \quad \Gamma \semi \Pi \semi \emptyset \vdash \setter^\ell_p:
        A \xrightarrow{\mf{Mut}} \tone
      }
    \end{mathpar}
    \item By premise (i-b), we have
    that $\Gamma \semi \Pi \semi \Phi_W \vdash W: \Pi[p][\ell]$
    for some $\Phi_3$ with $\Phi_H = \Phi_W \Join \Phi_H'$.
    By (3), this is equivalent to (5) that
    $\Gamma \semi \Pi \semi \Phi_W \vdash W: A$.
    \item Choose $\Phi_D' = \Phi_D \Join \Phi_W$.
      We are given that $\Phi_D \Join (\Phi_W \Join \Phi_H')$
      is defined, so $\Phi_D \Join \Phi_W$ is also defined.
    \item By \ref{thm:lamhook-unrestv} (unrestriction) on (5),
      we have $\mf{Succ}(\Phi_W)$. It follows that
      $(\Phi_D \Join \Phi_W) \Join (\Phi_W \Join \Phi_H')
      = \Phi_D \Join (\Phi_W \Join \Phi_H')$, satisfying conclusion \emph{(ii)}.
    \item Lemma \ref{thm:lamhook-unrestv} (unrestriction) on (1)
      gives that $\mf{Succ}(\Phi_1)$.
      We then apply Lemma \ref{thm:lamhook-phiwknb} ($\Phi$-weakening)
      to (2), getting (6)
      $\Gamma, x: A, x_\mf{set}: A \xrightarrow{\mf{Mut}} \tone
        \semi \Pi \semi \Phi_1 \Join \Phi_2
        \vdash D[W/x,\setter^\ell_p/x_\mf{set}]: \Ell$.
    \item We apply Lemma \ref{thm:lamhook-substb} (substitution) twice to (6),
    using (4) and (5), to obtain that\\
    $\Gamma \semi \Pi \semi \Phi_D'
      \vdash D[W/x,\setter^\ell_p/x_\mf{set}]: \Ell$,
      satisfying conclusion \emph{(i)}.
    \item Conclusion \emph{(iv)} is obvious, and we are done.
  \end{itemize}
\end{proof}

\begin{restatable}[Preservation for $\step$, impure]{theorem}{lamhookprezcm}
\label{thm:lamhook-prezcm}
If
\begin{enumerate}[label=(\roman*)]
  \item $\Gamma \semi \Pi \semi \Phi_H \semi \Phi_M
    \vdash \heap \semi M: A \bang \mf{Mut}$ and
  \item $\heap \semi M \step N \semi \heap'$,
\end{enumerate}
then for some $\Phi_N$ and $\Phi_H'$ such that
$\Phi_N \Join \Phi_H' \sqsupseteq_\mf{fresh} \Phi_M \Join \Phi_H$
and $\Pi' \sqsupseteq_{\Phi_N \Join \Phi_H'} \Pi$, we have
$\Gamma \semi \Pi' \semi \Phi_H' \semi \Phi_N
  \vdash \heap \semi N: A \bang \mf{Mut}$.
\end{restatable}
\begin{proof}
  By definition of premise \emph{(i)}, we have subpremises
  \begin{enumerate}[label=(i-\alph*)]
    \item $\Gamma \semi \Pi \semi \Phi_M \vdash M: A \bang \mf{Mut}$ and
    \item $\Gamma \semi \Pi \semi \Phi_H \vDash \heap$.
  \end{enumerate}
  Likewise, to show the conclusion, it suffices to pick some
  $\Phi_N$, $\Phi_H'$, and $\Pi'$ such that
  \begin{enumerate}[label=(\roman*)]
    \item $\Gamma \semi \Pi' \semi \Phi_N \vdash N: A \bang \mf{Mut}$,
    \item $\Gamma \semi \Pi' \semi \Phi_H' \vDash \heap'$,
    \item $\Phi_M \Join \Phi_H \sqsubseteq_\mf{fresh} \Phi_N \Join \Phi_H'$, and
    \item $\Pi \sqsubseteq_{\Phi_N \Join \Phi_H'} \Pi'$.
  \end{enumerate}

  Unless otherwise mentioned, we choose $\Phi_N = \Phi_M$, $\Phi_H' = \Phi_H$,
  and $\Pi' = \Pi$, such that conclusions (iii)-(iv) are trivial and
  (ii) is trivial by premise (i-b) when $\heap' = \heap$.

  We proceed by induction on the derivation of the
  computation step judgment \emph{(ii)}.

  \pcase{EC-Let}: We have
  \begin{mathpar}
    \inferrule*[right=EC-Let] {
      \inferrule* {(1)} {\heap \semi M \step M' \semi \heap'}
    } {
      \heap \semi \elet{x}{M}{N} \step \elet{x}{M'}{N} \semi \heap'
    }

    \inferrule*[right=TC-Let] {
      \inferrule* {(2)} {
        \Gamma \semi \Pi \semi \Phi_1 \vdash M: A \bang \mf{Mut}
      }
      \and
      \inferrule* {(3)} {
        \Gamma, x: A \semi \Pi \semi \Phi_2 \vdash N: B \bang \mf{Mut}
      }
    } {
      \Gamma \semi \Pi \semi \Phi_1 \Join \Phi_2 \vdash \elet{x}{M}{N}:
      B \bang \mf{Mut}
    }
  \end{mathpar}
  \begin{itemize}
    \item By inductive hypothesis on (1), (2), and premise (i-b),
    we get some $\Phi_1'$, $\Phi_H'$, and $\Pi'$ such that:
    \begin{enumerate}
      \setcounter{enumi}{3}
      \item $\Gamma \semi \Pi' \semi \Phi_1' \vdash M': A \bang \mf{Mut}$,
      \item $\Gamma \semi \Pi' \semi \Phi_H' \vDash \heap'$,
      \item $\Phi_1 \Join \Phi_H \sqsubseteq_\mf{fresh} \Phi_1' \Join \Phi_H'$,
      and \item $\Pi \sqsubseteq_{\Phi_1' \Join \Phi_H'} \Pi'$.
    \end{enumerate}
    \item Choose $\Phi_N = \Phi_1' \Join \Phi_2$ and this $\Phi_H'$ as our own.
    \item By Lemma \ref{thm:lamhook-freshphi} (fresh $\Phi$-extension) using (6)
    and extending by $\Phi_2$, we get that $\Phi_N \Join \Phi_H'$
    is defined and
    $\Phi_M \Join \Phi_H \sqsubseteq_{\mf{fresh}} \Phi_N \Join \Phi_H'$,
    satisfying conclusion (iii).
    \item Choose this $\Pi'$ as our own.
    By Lemma \ref{thm:lamhook-piprez} ($\Pi$-order preservation) using (7)
    and extending by $\Phi_2$,
    we get that $\Pi \sqsubseteq_{\Phi_N \Join \Phi_H} \Pi'$.
    \item (5) gives us conclusion (ii).
    \item By Lemma \ref{thm:lamhook-piwknc} ($\Pi$-weakening)
    on (3) using (7) and the fact that $\Phi_N \Join \Phi_H$ is defined,
    we get (8) that $\Gamma, x: A \semi \Pi' \semi \Phi_2 \vdash N:
    B \bang \mf{Pure}$.
    \item We derive conclusion (i) as follows:
    \begin{mathpar}
      \inferrule*[right=TC-Let] {
        \inferrule*[right=(6)] { } {
          \Gamma \semi \Pi' \semi \Phi_1' \vdash M': A \bang \mf{Pure}
        }
        \and
        \inferrule*[right=(8)] { } {
          \Gamma, x: A \semi \Pi' \semi \Phi_2 \vdash N: B \bang \mf{Pure}
        }
      } {
        \Gamma \semi \Pi \semi \Phi_1' \Join \Phi_2 \vdash \elet{x}{M'}{N}:
        B \bang \mf{Pure}
      }
    \end{mathpar}
  \end{itemize}

  \textbf{Cases}
  \textsc{EC-Val},
  \textsc{EC-App},
  \textsc{EC-IfT},
  \textsc{EC-IfF},
  \textsc{EC-Fst},
  \textsc{EC-Snd},
  \textsc{EC-FdN}, and
  \textsc{EC-FdC} are all shown analogously to Theorem \ref{thm:lamhook-prezcp}
  (pure preservation). The main difference is that rule \textsc{TC-Sub} now
  appears in typing derivations, but this ultimately has no effect
  on what we can derive.
  (These cases are also spelled out in \ref{thm:trans-simc})

  \pcase{EC-AppSet}: We have
  $\heap \semi \setter^\ell_p\ V \step \mf{val}\ () \semi \heap'$ and
  \begin{mathpar}
    \inferrule*[right=TC-App] {
      \inferrule*[right=TV-Setter] {
        \inferrule* {(1)} {\ell: A \in \Pi[p]}
        \and
        \inferrule* {(2)} {\mf{Succ}(\Phi_1)}
      } {
        \Gamma \semi \Pi \semi \Phi_1 \vdash \setter^\ell_p:
        A \xrightarrow{\mf{Mut}} \tone
      }
      \and
      \inferrule* {(3)} {
        \Gamma \semi \Pi \semi \Phi_2 \vdash V: A
      }
    } {
      \Gamma \semi \Pi \semi \Phi_1 \Join \Phi_2 \vdash
      \setter^\ell_p\ V: \tone \bang \mf{Mut}
    }
  \end{mathpar}
  where $\heap' = \heap \mid \heap[p].\mf{st}[\ell] \mapsto V\
  \text{and}\ \heap[p].\mf{rerender} \mapsto \mf{tt}$.
  \begin{itemize}
    \item (1) and premise (i-b) give us (4) that
    $\Gamma \semi \Pi \semi \Phi_3 \vdash \heap[p].\mf{st}[\ell]: A$
    for $\Phi_H = \Phi_3 \Join \Phi_4$.
    \item By Lemma \ref{thm:lamhook-unrestv} (unrestriction) on (3),
      we know (5) that $\mf{Succ}(\Phi_2)$.
    \item Choose $\Phi_N = \Phi_M$, $\Phi_H' = \Phi_H \Join \Phi_2$,
      and $\Pi' = \Pi$, trivializing conclusion (iv).
      By (5) and Lemma \ref{thm:lamhook-unresth} (unrestriction) on $\Phi_H$,
      $\Phi_H'$ is defined.
    \item By (5),
      $(\Phi_1 \Join \Phi_2) \Join (\Phi_H \Join \Phi_2)
      = (\Phi_1 \Join \Phi_2) \Join \Phi_H$,
      satisfying conclusion (iii).
    \item Clearly $\mf{Succ}(\Phi_3)$ by Lemma \ref{thm:lamhook-unresth},
      so we weaken (3) by $\Phi_3$ (Lemma \ref{thm:lamhook-phiwknv})
      to get that\\$\Gamma \semi \Pi \semi \Phi_2 \Join \Phi_3 \vdash V: A$.
      It then follows from premise (i-b) that
      $\Gamma \semi \Pi \semi (\Phi_2 \Join \Phi_3) \Join \Phi_4 \vDash \heap$,
      i.e. conclusion (ii) holds.
    \item We derive conclusion (i) as follows:
    \begin{mathpar}
    \inferrule*[right=TC-Sub] {
      \inferrule*[right=TC-Val] {
        \inferrule*[right=TV-Unit] {
          \inferrule*[right=(2) \textnormal{and} (5)] { } {
            \mf{Succ}(\Phi_1 \Join \Phi_2)
          }
        } {
          \Gamma \semi \Pi \semi \Phi_1 \Join \Phi_2 \vdash (): \tone
        }
      } {
        \Gamma \semi \Pi \semi \Phi_1 \Join \Phi_2 \vdash \mf{val}\ ():
        \tone \bang \mf{Pure}
      }
    } {
      \Gamma \semi \Pi \semi \Phi_1 \Join \Phi_2 \vdash \mf{val}\ ():
      \tone \bang \mf{Mut}
    }
    \end{mathpar}
  \end{itemize}

  \pcase{EC-Eval}, $\phi=\mf{Init}$: We have
  \begin{mathpar}
    \inferrule*[right=EC-Eval] {
      \inferrule* {(1)} {\heap \semi D \step^\mf{Init}_p D' \semi \heap'}
    } {
      \heap \semi \eval{D_0}{\mf{Init}}{p}{D}
      \step
      \eval{D_0}{\mf{Init}}{p}{D'} \semi \heap'
    }

    \inferrule*[right=TC-EvalI] {
      \inferrule* {(2)} {
        \Gamma \semi \Pi \semi \Phi_1 \vdash D_0: \Ell \uplus \Pi[p]
      }
      \and
      \inferrule* {(3)} {\Gamma \semi \Pi \semi \Phi_2 \vdash D: \Ell}
    } {
      \Gamma \semi \Pi \semi \Phi_1 \Join \Phi_2,p: \mf{Init} \vdash
      \eval{D_0}{\mf{init}}{p}{D}: \mf{ElemInst} \bang \mf{Mut}
    }
  \end{mathpar}
  \begin{itemize}
    \item (2) implies (4) that $\Ell \cap \Pi[p] = \emptyset$.
      That $\Phi_1 \Join \Phi_2, p: \mf{Init}$ is defined entails
      (5) $p \notin \dom(\Phi_2)$.
    \item
    By Theorem \ref{thm:lamhook-prezbi}
    ($\step^\mf{Init}_p$-preservation)
    on (1), (3), (4), (5), and premise (i-b),
    we have some $\Phi_H'$ and $\Pi'$ with:
    \begin{enumerate}
      \setcounter{enumi}{5}
      \item $\Gamma \semi \Pi' \semi \Phi_2 \vdash D': \Ell'$,
      \item $\Gamma \semi \Pi' \semi \Phi_H' \vDash \heap'$,
      \item $\Phi_2 \Join \Phi_H' = \Phi_2 \Join \Phi_H$,
      \item $\Pi \sqsubseteq_{p: \mf{Init}} \Pi'$, and
      \item $\Pi'[p] \uplus \Ell' = \Pi[p] \uplus \Ell$.
    \end{enumerate}
    \item We choose this $\Pi'$ and $\Phi_H'$,
      and retain $\Phi_N = \Phi_M = \Phi_1 \Join \Phi_2, p: \mf{Init}$.
    \item We started with $\Phi_1 \Join \Phi_2 \Join \Phi_H, p: \mf{Init}$
      as defined. By (8), we know (11) that\\
      $\Phi_1 \Join \Phi_2 \Join \Phi_H', p: \mf{Init}$ is also defined.
    \item Conclusion (iii) follows from
      Lemma \ref{thm:lamhook-freshphi} (fresh $\Phi$-extension) using (8)
      and extending by $\Phi_1, p: \mf{Init}$, which is valid by (11).
    \item (11) also lets us apply Lemma
      \ref{thm:lamhook-piprez} ($\Pi$-order preservation) on (8),
      extending by\\$\Phi_1 \Join \Phi_2 \Join \Phi_H'$,
      to get that $\Pi \sqsubseteq_{\Phi_M \Join \Phi_H'} \Pi'$,
      satisfying conclusion (iv).
    \item By choice of $\Pi$ and $\Phi$, (7) gives conclusion (ii).
    \item Lemma \ref{thm:lamhook-piwknb} ($\Pi$-weakening) on (2) using (9),
      which is valid by (11), gives us (12) that
      $\Gamma \semi \Pi' \semi \Phi_1 \vdash D_0: \Ell \uplus \Pi[p]$.
    \item By (10), (12) is equivalent to (13) that
      $\Gamma \semi \Pi' \semi \Phi_1 \vdash D_0: \Ell' \uplus \Pi'[p]$.
    \item We derive conclusion (i) as follows:
    \begin{mathpar}
      \inferrule*[right=TC-EvalI] {
        \inferrule*[right=(13)] { } {
          \Gamma \semi \Pi' \semi \Phi_1 \vdash D_0: \Ell' \uplus \Pi'[p]
        }
        \and
        \inferrule*[right=(6)] { } {
          \Gamma \semi \Pi' \semi \Phi_2 \vdash D': \Ell'
        }
      } {
        \Gamma \semi \Pi' \semi \Phi_1 \Join \Phi_2, p: \mf{Init} \vdash
        \eval{D_0}{\mf{Init}}{p}{D'}: \mf{ElemInst} \bang \mf{Mut}
      }
    \end{mathpar}
  \end{itemize}

  \pcase{EC-Eval}, $\phi=\mf{Succ}$: We have
  \begin{mathpar}
    \inferrule*[right=EC-Eval] {
      \inferrule* {(1)} {\heap \semi D \step^\mf{Succ}_p D' \semi \heap'}
    } {
      \heap \semi \eval{D_0}{\mf{Succ}}{p}{D}
      \step
      \eval{D_0}{\mf{Succ}}{p}{D'} \semi \heap'
    }

    \inferrule*[right=TC-EvalS] {
      \inferrule* {(2)} {\Gamma \semi \Pi \semi \Phi_1 \vdash D_0: \Pi[p]}
      \
      \inferrule* {(3)} {\Gamma \semi \Pi \semi \Phi_2 \vdash D: \Ell}
      \
      \inferrule* {(4)} {\Ell \subseteq \Pi[p]}
      \
      \inferrule* {(5)} {(\Phi_1 \Join \Phi_2)[p] = \mf{Succ}}
    } {
      \Gamma \semi \Pi \semi \Phi_1 \Join \Phi_2 \vdash
      \eval{D_0}{\mf{Succ}}{p}{D}: \mf{ElemInst} \bang \mf{Mut}
    }
  \end{mathpar}
  \begin{itemize}
    \item By Theorem \ref{thm:lamhook-prezbs}
    ($\step^\mf{Succ}_p$-preservation)
    on (1), (3), (4), and premise (i-b), we have some $\Phi_2'$ such that
    \begin{enumerate}
      \setcounter{enumi}{5}
      \item $\Gamma \semi \Pi \semi \Phi_2' \vdash D': \Ell'$,
      \item $\Phi_2' \Join \Phi_H = \Phi_2 \Join \Phi_H$,
      \item $\heap' = \heap$, and
      \item $\Ell' \subseteq \Ell$.
    \end{enumerate}
    \item Choose $\Phi_N = \Phi_1 \Join \Phi_2'$ and $\Phi_H' = \Phi_H$.
      By Lemma \ref{thm:lamhook-unrestb} (unrestriction) on (2) and (6),
      both $\mf{Succ}(\Phi_1)$ and $\mf{Succ}(\Phi_2')$, so $\Phi_N$
      is defined. Lemma \ref{thm:lamhook-freshphi} (fresh $\Phi$-extension)
      on (7) then gives us conclusion (iii).
    \item Choose $\Pi' = \Pi$. Conclusion (iv) is trivial,
      and (8) and premise (i-b) give us conclusion (ii).
    \item By (4) and (9), we have (10) that
      $\Ell' \subseteq \Pi[p]$.
    \item We want to show (11) that $(\Phi_1 \Join \Phi_2')[p] = \mf{Succ}$
      by case-analysis on (5):
    \begin{itemize}
      \item If $p: \mf{Succ} \in \Phi_1$,
        then clearly $p: \mf{Succ} \in (\Phi_1 \Join \Phi_2')$.
      \item If not, then we must have $p: \mf{Succ} \in \Phi_2$.
        By definition, $p: \mf{Succ} \in \Phi_2 \Join \Phi_H$.
        By (7), $p: \mf{Succ} \in \Phi_2' \Join \Phi_H$.
        Again by definition, $p: \mf{Succ} \in \Phi_2'$.
        Finally, we get\\$p: \mf{Succ} \in (\Phi_1 \Join \Phi_2')$.
    \end{itemize}
    \item We derive conclusion (i) by \textsc{TC-EvalS} as follows:
  \end{itemize}
  \begin{mathpar}
    \inferrule {
      \inferrule*[right=(2)] { } {
        \Gamma \semi \Pi \semi \Phi_1 \vdash D_0: \Pi[p]
      }
      \
      \inferrule*[right=(7)] { } {
        \Gamma \semi \Pi \semi \Phi_2' \vdash D': \Ell'
      }
      \
      \inferrule*[right=(10)] { } {\Ell' \subseteq \Pi[p]}
      \
      \inferrule*[right=(11)] { } {(\Phi_1 \Join \Phi_2')[p] = \mf{Succ}}
    } {
      \Gamma \semi \Pi \semi \Phi_1 \Join \Phi_2' \vdash
      \eval{D_0}{\mf{Succ}}{p}{D'}: \mf{ElemInst} \bang \mf{Mut}
    }
  \end{mathpar}

  \pcase{EC-EvalDone}, $\phi=\mf{Init}$: We have
  $$\heap \semi \eval{D}{\mf{Init}}{p}{\mf{val}\ V}
  \step \elet{x}{\mf{init}\ V}{\mf{val}\ \compInst{p}{D}{x}} \semi \heap$$
  \begin{mathpar}
    \inferrule*[right=TC-EvalI] {
      \inferrule* {(1)} {
        \Gamma \semi \Pi \semi \Phi_1 \vdash D_0: \emptyset \uplus \Pi[p]
      }
      \and
      \inferrule*[right=TB-Val] {
        \inferrule* {(2)} {
          \Gamma \semi \Pi \semi \Phi_2 \vdash V: \mf{Element}
        }
      } {
        \Gamma \semi \Pi \semi \Phi_2 \vdash \mf{val}\ V: \emptyset
      }
    } {
      \Gamma \semi \Pi \semi \Phi_1 \Join \Phi_2,p: \mf{Init} \vdash
      \eval{D_0}{\mf{Init}}{p}{\mf{val}\ V}: \mf{ElemInst} \bang \mf{Mut}
    }
  \end{mathpar}
  \begin{itemize}
    \item We choose $\Phi_H' = \Phi_H$ and
    $\Phi_N = \Phi_1 \Join \Phi_2, p: \mf{Succ}$.
    Clearly $\Phi_M \Join \Phi_H \sqsubseteq_\mf{fresh} \Phi_N \Join \Phi_H'$,
    satisfying conclusion (iii).
    \item We retain $\Pi' = \Pi$, trivializing conclusion (iv),
      and also letting conclusion (ii) follow easily from premise (i-b).
    \item By Lemmas \ref{thm:lamhook-unrestb} and \ref{thm:lamhook-unrestv}
      on (1) and (2), we get $\mf{Succ}(\Phi_1)$ and $\mf{Succ}(\Phi_2)$,
      respectively. Clearly then (3) $\mf{Succ}(\Phi_N)$.
    \item By Lemma \ref{thm:lamhook-phiwknb} ($\Phi$-weakening) on (1) using (3),
      we get (4) that $\Gamma \semi \Pi \semi \Phi_N \vdash D: \Pi[p]$.
    \item Similarly, by Lemma \ref{thm:lamhook-phiwknv} ($\Phi$-weakening)
      on (2) using (3), we get (5) that\\
      $\Gamma \semi \Pi \semi \Phi_N \vdash V: \mf{Element}$.
    \item We derive conclusion (i) as follows
    by way of subderivation (6),
    using $\Phi_N$ intuitionistically by (3):
    \begin{mathpar}
      \inferrule*[right=TC-Sub] {
        \inferrule*[Right=TC-Val] {
          \inferrule*[Right=TV-CompI] {
            \inferrule*[right=(4)] { } {
              \Gamma_6 \semi \Pi \semi \Phi_N \vdash D: \Pi[p]
            }
            \and
            \inferrule* { } {\Phi_N[p] = \mf{Succ}}
            \and
            \inferrule*[Right=TV-Var] {
              \inferrule*[right=(3)] { } {\mf{Succ}(\Phi_N)}
            } {
              \Gamma_6 \semi \Pi \semi \Phi_N \vdash
              x: \mf{ElemInst}
            }
          } {
            \Gamma_6 \semi \Pi \semi \Phi_N \vdash
            \compInst{p}{D}{x}: \mf{ElemInst}
          }
        } {
          \Gamma_6 \semi \Pi \semi \Phi_N \vdash
          \mf{val}\ \compInst{p}{D}{x}: \mf{ElemInst} \bang \mf{Pure}
        }
      } {
        (6) \quad
        \Gamma_6 \semi \Pi \semi \Phi_N \vdash
        \mf{val}\ \compInst{p}{D}{x}: \mf{ElemInst} \bang \mf{Mut}
      }
      \and
      \inferrule*[right=TC-Let] {
        \inferrule*[right=TC-Init] {
          \inferrule*[right=(5)] { } {
            \Gamma \semi \Pi \semi \Phi_N \vdash V: \mf{Element}
          }
        } {
          \Gamma \semi \Pi \semi \Phi_N \vdash \mf{init}\ V:
          \mf{ElemInst} \bang \mf{Mut}
        }
        \and
        [(6),\ \Gamma_6 = \Gamma, x: \mf{ElemInst}]
      } {
        \Gamma \semi \Pi \semi \Phi_N \vdash \elet{x}{\mf{init}\ V}{
          \mf{val}\ \compInst{p}{D}{x}
        }: \mf{ElemInst} \bang \mf{Mut}
      }
    \end{mathpar}
  \end{itemize}

  \pcase{EC-EvalDone}, $\phi=\mf{Succ}$:
  \[\heap \semi \eval{D}{\mf{Succ}}{p}{\mf{val}\ V}
  \step \elet{x}{\mf{init}\ V}{\mf{val}\ \compInst{p}{D}{x}} \semi \heap\]
  \begin{mathpar}
    \inferrule*[right=TC-EvalS] {
      \inferrule* {(1)} {\Gamma \semi \Pi \semi \Phi_1 \vdash D: \Pi[p]}
      \
      \inferrule*[right=TB-Val] {
        \inferrule* {(2)} {
          \Gamma \semi \Pi \semi \Phi_2 \vdash V: \mf{Element}
        }
      } {
        \Gamma \semi \Pi \semi \Phi_2 \vdash \mf{val}\ V: \emptyset
      }
      \
      \inferrule* {(3)} {(\Phi_1 \Join \Phi_2)[p] = \mf{Succ}}
      \ \cdots
    } {
      \Gamma \semi \Pi \semi \Phi_1 \Join \Phi_2 \vdash
      \eval{D_0}{\mf{Succ}}{p}{D}: \mf{ElemInst} \bang \mf{Mut}
    }
  \end{mathpar}
  \begin{itemize}
    \item We retain $\Phi_N = \Phi_1 \Join \Phi_2$.
      By Lemmas \ref{thm:lamhook-unrestb} and \ref{thm:lamhook-unrestv}
      on (1) and (2), we get $\mf{Succ}(\Phi_1)$ and $\mf{Succ}(\Phi_2)$,
      respectively. Clearly then (4) $\mf{Succ}(\Phi_N)$.
    \item By Lemma \ref{thm:lamhook-phiwknb} ($\Phi$-weakening) on (1)
      using (4), we get (5) that
      $\Gamma \semi \Pi \semi \Phi_N \vdash D: \Pi[p]$.
    \item By Lemma \ref{thm:lamhook-phiwknv} ($\Phi$-weakening) on (2),
    we get (6) $\Gamma \semi \Pi \semi \Phi_N \vdash V: \mf{Element}$.
    \item We derive conclusion (i) as follows
    by way of subderivation (7),
    using $\Phi_N$ intuitionistically by (4):
    \begin{mathpar}
      \inferrule*[right=TC-Sub] {
        \inferrule*[Right=TC-Val] {
          \inferrule*[Right=TV-CompI] {
            \inferrule*[right=(5)] { } {
              \Gamma_7 \semi \Pi \semi \Phi_N \vdash D: \Pi[p]
            }
            \and
            \inferrule* { } {\Phi_N[p] = \mf{Succ}}
            \and
            \inferrule*[Right=TV-Var] {
              \inferrule*[right=(4)] { } {\mf{Succ}(\Phi_D')}
            } {
              \Gamma_7 \semi \Pi \semi \Phi_N \vdash
              x: \mf{ElemInst}
            }
          } {
            \Gamma_7 \semi \Pi \semi \Phi_N \vdash
            \compInst{p}{D}{x}: \mf{ElemInst}
          }
        } {
          \Gamma_7 \semi \Pi \semi \Phi_N \vdash
          \mf{val}\ \compInst{p}{D}{x}: \mf{ElemInst} \bang \mf{Pure}
        }
      } {
        (7) \quad
        \Gamma_7 \semi \Pi \semi \Phi_N \vdash
        \mf{val}\ \compInst{p}{D}{x}: \mf{ElemInst} \bang \mf{Mut}
      }
      \and
      \inferrule*[right=TC-Let] {
        \inferrule*[right=TC-Init] {
          \inferrule*[right=(5)] { } {
            \Gamma \semi \Pi \semi \Phi_N \vdash V: \mf{Element}
          }
        } {
          \Gamma \semi \Pi \semi \Phi_N \vdash \mf{init}\ V:
          \mf{ElemInst} \bang \mf{Mut}
        }
        \and
        [(7),\ \Gamma_7 = \Gamma, x: \mf{ElemInst}]
      } {
        \Gamma \semi \Pi \semi \Phi_N \vdash \elet{x}{\mf{init}\ V}{
          \mf{val}\ \compInst{p}{D}{x}
        }: \mf{ElemInst} \bang \mf{Mut}
      }
    \end{mathpar}
  \end{itemize}

  \pcase{EC-InitText}: We have
  $\heap \semi \mf{init\ text}\ V \step \mf{val\ textInst}\ V \semi \heap$ and
  \begin{mathpar}
    \inferrule*[right=TC-Init] {
      \inferrule*[right=TV-Text] {
        \inferrule* {(1)} {
          \Gamma \semi \Pi \semi \Phi_M \vdash V: \mf{Str}
        }
      } {
        \Gamma \semi \Pi \semi \Phi_M \vdash \mf{text}\ V: \mf{Element}
      }
    } {
      \Gamma \semi \Pi \semi \Phi_M \vdash \mf{init\ text}\ V:
      \mf{ElemInst} \bang \mf{Mut}
    }
  \end{mathpar}
  \begin{itemize}
    \item We derive conclusion (i) as follows:
    \begin{mathpar}
      \inferrule*[right=TC-Sub] {
        \inferrule*[right=TC-Val] {
          \inferrule*[right=TV-TextI] {
            \inferrule*[right=(1)] { } {
              \Gamma \semi \Pi \semi \Phi_M \vdash V: \mf{Str}
            }
          } {
            \Gamma \semi \Pi \semi \Phi_M \vdash \mf{textInst}\ V: \mf{ElemInst}
          }
        } {
          \Gamma \semi \Pi \semi \Phi_M \vdash \mf{val\ textInst}\ V:
          \mf{ElemInst} \bang \mf{Pure}
        }
      } {
        \Gamma \semi \Pi \semi \Phi_M \vdash \mf{val\ textInst}\ V:
        \mf{ElemInst} \bang \mf{Mut}
      }
    \end{mathpar}
  \end{itemize}

  \pcase{EC-InitTag}: We have
  $\heap \semi \mf{init}\ \mf{tag}(V, W)
  \step
  \elet{x}{\mf{inits}\ W}{\mf{val}\ \mf{tagInst}_{id}(V, x)} \semi \heap$
  and
  \begin{mathpar}
  \inferrule*[right=TC-Init] {
    \inferrule*[right=TV-Tag] {
      \inferrule* {(1)} {
        \Gamma \semi \Pi \semi \Phi_1 \vdash V: \mf{List\ Attr}
      }
      \and
      \inferrule* {(2)} {
        \Gamma \semi \Pi \semi \Phi_2 \vdash W: \mf{List\ Element}
      }
    } {
      \Gamma \semi \Pi \semi \Phi_1 \Join \Phi_2 \vdash \mf{tag}(V, W):
      \mf{Element}
    }
  } {
    \Gamma \semi \Pi \semi \Phi_1 \Join \Phi_2 \vdash \mf{init\ tag}(V, W):
    \mf{ElemInst} \bang \mf{Mut}
  }
  \end{mathpar}
  \begin{itemize}
    \item We derive conclusion (i) via subderivation (3)
    as follows, weakening (1) appropriately on $\Gamma$.
  \end{itemize}
  \begin{mathpar}
    \inferrule*[right=TC-Sub] {
      \inferrule*[right=TC-Val] {
        \inferrule*[right=TV-TagI] {
          \inferrule*[right=(1)] { } {
            \Gamma_3 \semi \Pi \semi \Phi_1
            \vdash V: \mf{List\ Attr}
          }
          \
          \inferrule*[Right=TV-Var] { } {
            \Gamma_3 \semi \Pi \semi \emptyset
            \vdash x: \mf{List\ Eleminst}
          }
        } {
          \Gamma_3 \semi \Pi \semi \Phi_1
          \vdash \mf{tagInst}_{id}(V, x): \mf{ElemInst}
        }
      } {
        \Gamma_3 \semi \Pi \semi \Phi_1
        \vdash \mf{val}\ \mf{tagInst}_{id}(V, x):
        \mf{ElemInst} \bang \mf{Pure}
      }
    } {
      (3) \quad \Gamma_3 \semi \Pi \semi \Phi_1
      \vdash \mf{val}\ \mf{tagInst}_{id}(V, x):
      \mf{ElemInst} \bang \mf{Mut}
    }
    \and
    \inferrule*[right=TC-Let] {
      \inferrule*[right=TC-Inits] {
        \inferrule*[right=(2)] { } {
          \Gamma \semi \Pi \semi \Phi_2 \vdash W: \mf{List\ Element}
        }
      } {
        \Gamma \semi \Pi \semi \Phi_2 \vdash \mf{inits}\ W:
        \mf{List\ ElemInst} \bang \mf{Mut}
      }
      \
      [(3),\ \Gamma_3 = \Gamma, x: \mf{List\ ElemInst}]
    } {
      \Gamma \semi \Pi \semi \Phi_1 \Join \Phi_2 \vdash
      \elet{x}{\mf{inits}\ W}{\mf{val}\ \mf{tagInst}_{id}(V, x)}:
      \mf{Element} \bang \mf{Mut}
    }
  \end{mathpar}

 \pcase{EC-InitComp}: We have
  $\heap \semi \mf{init}\ \mf{component}\ \{D\}
  \step \eval{D}{\mf{Init}}{p}{D}
  \semi \heap'$, where $p$ is fresh,\\
  $\heap' = \heap, p \mapsto
  \{\mf{st}: \emptyset, \mf{rerender}: \mf{ff}\}$,
  and
  \begin{mathpar}
  \inferrule*[right=TC-Init] {
    \inferrule*[right=TV-Comp] {
      \inferrule* {(1)} {\Gamma \semi \Pi \semi \Phi_M \vdash D: \Ell}
    } {
      \Gamma \semi \Pi \semi \Phi_M \vdash \mf{component}\ \{D\}: \mf{Element}
    }
  } {
    \Gamma \semi \Pi \semi \Phi_M \vdash \mf{init\ component}\ \{D\}:
    \mf{ElemInst} \bang \mf{Mut}
  }
  \end{mathpar}
  \begin{itemize}
    \item Choose $\Phi_N = \Phi_M, p: \mf{Init}$ and $\Phi_H' = \Phi_H$.
      Since $p$ is fresh, $\Phi_N \Join \Phi_H$ is defined and we have
      $\Phi_M \Join \Phi_H \sqsubseteq_\mf{fresh} \Phi_N \Join \Phi_H$,
      satisfying conclusion (iii).
    \item Choose $\Pi'=\Pi,p \mapsto \emptyset$. By definition,
      we have $\Pi \sqsubseteq_\mf{\Phi_N \Join \Phi_H} \Pi'$,
      satisfying conclusion (iv).
    \item By premise (i-b)
      and definition of each of $\Pi'$ and $\heap'$,
      we have that $\Gamma \semi \Pi' \semi \Phi_H \vDash \heap'$,
      satisfying conclusion (ii).
    \item By Lemma \ref{thm:lamhook-piwknb} ($\Pi$-weakening) on (1),
    using the fact that $p$ is fresh and so $p \notin \dom(\Phi_M)$,
    we get (3) that $\Gamma \semi \Pi' \semi \Phi_M \vdash D: \Ell$.
    \item Since $\Pi'[p]=\emptyset$,
    (3) is equivalent to (4) that
    $\Gamma \semi \Pi' \semi \Phi_M \vdash D: \Ell \uplus \Pi'[p]$.
    \item By Lemma \ref{thm:lamhook-unrestb} on (1),
      we know $\mf{Succ}(\Phi_M)$.
      It follows that (5) $\Phi_M \Join \Phi_M = \Phi_M$.
    \item We derive conclusion (i) as follows, implicitly using (5):
    \begin{mathpar}
      \inferrule*[right=TC-EvalI] {
        \inferrule*[right=(4)] { } {
          \Gamma \semi \Pi' \semi \Phi_M \vdash D: \Ell \uplus \Pi'[p]
        }
        \and
        \inferrule*[right=(3)] { } {
          \Gamma \semi \Pi' \semi \Phi_M \vdash D: \Ell
        }
      } {
        \Gamma \semi \Pi' \semi \Phi_M, p: \mf{Init} \vdash
        \eval{D}{\mf{Init}}{p}{D}: \mf{ElemInst} \bang \mf{Mut}
      }
    \end{mathpar}
  \end{itemize}

  \pcase{EC-ClickText}: We have
  $\heap \semi \mf{click}_{id}\ \mf{textInst}\ V
  \step \mf{val}\ () \semi \heap$ and
  \begin{mathpar}
    \inferrule*[right=TC-Click] {
      \inferrule* {(1)} {
        \Gamma \semi \Pi \semi \Phi_M \vdash \mf{textInst}\ V: \mf{ElemInst}
      }
    } {
      \Gamma \semi \Pi \semi \Phi_M \vdash \mf{click}_{id}\ \mf{textInst}\ V:
      \tone \bang \mf{Mut}
    }
  \end{mathpar}
  \begin{itemize}
    \item By Lemma \ref{thm:lamhook-unrestv} (unrestriction) on (1),
    we get (2) that $\mf{Succ}(\Phi_M)$.
    \item We derive conclusion (i) as follows:
    \begin{mathpar}
      \inferrule*[right=TC-Sub] {
        \inferrule*[right=TC-Val] {
          \inferrule*[right=TV-Unit] {
            \inferrule*[right=(2)] { } {\mf{Succ}(\Phi_M)}
          } {
            \Gamma \semi \Pi \semi \Phi_M \vdash (): \tone
          }
        } {
          \Gamma \semi \Pi \semi \Phi_M
          \vdash \mf{val}\ (): \tone \bang \mf{Pure}
        }
      } {
        \Gamma \semi \Pi \semi \Phi_M
        \vdash \mf{val}\ (): \tone \bang \mf{Mut}
      }
    \end{mathpar}
  \end{itemize}

  \pcase{EC-ClickTagT}: We have
  $\heap \semi \mf{click}_{id}\ \mf{tagInst}_{id}(V, W)
  \step \mf{runClicks}\ V \semi \mf{clicks}_{id}\ W \semi \heap$ and
  \begin{mathpar}
    \inferrule*[right=TC-Click] {
      \inferrule*[right=TV-TagI] {
        \inferrule* {(1)} {
          \Gamma \semi \Pi \semi \Phi_1 \vdash V: \mf{List\ Attr}
        }
        \and
        \inferrule* {(2)} {
          \Gamma \semi \Pi \semi \Phi_2 \vdash W: \mf{List\ ElemInst}
        }
      } {
        \Gamma \semi \Pi \semi \Phi_1 \Join \Phi_2 \vdash
        \mf{tagInst}_{id}(V, W): \mf{ElemInst}
      }
    } {
      \Gamma \semi \Pi \semi \Phi_1 \Join \Phi_2 \vdash
      \mf{click}_{id}\ \mf{tagInst}_{id}(V, W): \tone \bang \mf{Mut}
    }
  \end{mathpar}
  \begin{itemize}
    \item We derive conclusion (i) as follows:
  \end{itemize}
  \begin{mathpar}
    \inferrule*[right=TC-Let] {
      \inferrule*[right=TC-RunClicks] {
        \inferrule*[right=(1)] { } {
          \Gamma \semi \Pi \semi \Phi_1 \vdash V: \mf{List\ Attr}
        }
      } {
        \Gamma \semi \Pi \semi \Phi_1 \vdash \mf{runClicks}\ V:
        \tone \bang \mf{Mut}
      }
      \
      \inferrule*[Right=TC-Clicks] {
        \inferrule*[right=(2)] { } {
          \Gamma \semi \Pi \semi \Phi_2 \vdash W: \mf{List\ ElemInst}
        }
      } {
        \Gamma \semi \Pi \semi \Phi_2 \vdash \mf{clicks}_{id}\ W:
        \tone \bang \mf{Mut}
      }
    } {
      \Gamma \semi \Pi \semi \Phi_1 \Join \Phi_2 \vdash
      \mf{runClicks}\ V \semi \mf{clicks}_{id}\ W: \tone \bang \mf{Mut}
    }
  \end{mathpar}

  \pcase{EC-ClickTagF}: We have
  $\heap \semi \mf{click}_{id}\ \mf{tagInst}_{{id}'}(V, W)
  \step \mf{clicks}_{id}\ W \semi \heap$ and
  \begin{mathpar}
    \inferrule*[right=TC-Click] {
      \inferrule*[right=TV-TagI] {
        \inferrule* {(1)} {
          \Gamma \semi \Pi \semi \Phi_1 \vdash V: \mf{List\ Attr}
        }
        \and
        \inferrule* {(2)} {
          \Gamma \semi \Pi \semi \Phi_2 \vdash W: \mf{List\ ElemInst}
        }
      } {
        \Gamma \semi \Pi \semi \Phi_1 \Join \Phi_2 \vdash
        \mf{tagInst}_{id}(V, W): \mf{ElemInst}
      }
    } {
      \Gamma \semi \Pi \semi \Phi_1 \Join \Phi_2 \vdash
      \mf{click}_{id}\ \mf{tagInst}_{id}(V, W): \tone \bang \mf{Mut}
    }
  \end{mathpar}
  \begin{itemize}
    \item By Lemma \ref{thm:lamhook-unrestv} (unrestriction) on (1),
    we get (3) that $\mf{Succ}(\Phi_1)$.
    \item By Lemma \ref{thm:lamhook-phiwknv} ($\Phi$-weakening) on (2)
    using (3), we get (4) that\\
    $\Gamma \semi \Pi \semi \Phi_1 \Join \Phi_2
    \vdash W: \mf{List\ ElemInst}$.
    \item We derive conclusion (i) as follows:
    \begin{mathpar}
      \inferrule*[right=TC-Clicks] {
        \inferrule*[right=(1)] { } {
          \Gamma \semi \Pi \semi \Phi_1 \Join \Phi_2 \vdash W:
          \mf{List\ ElemInst}
        }
      } {
        \Gamma \semi \Pi \semi \Phi_1 \Join \Phi_2 \vdash
        \mf{clicks}_{id}\ W: \tone \bang \mf{Mut}
      }
    \end{mathpar}
  \end{itemize}

  \pcase{EC-ClickComp}: We have
  $\heap \semi \mf{click}_{id}\ \compInst{p}{D}{V}
  \step \mf{click}_{id}\ V \semi \heap$ and
  \begin{mathpar}
    \inferrule*[right=TC-Click] {
      \inferrule*[Right=TV-CompI] {
        \inferrule* {(1)} {\Gamma \semi \Pi \semi \Phi_1 \vdash D: \Pi[p]}
        \and
        (\Phi_1 \Join \Phi_2)[p] = \mf{Succ}
        \and
        \inferrule* {(1)} {
          \Gamma \semi \Pi \semi \Phi_2 \vdash V: \mf{ElemInst}
        }
      } {
        \Gamma \semi \Pi \semi \Phi_1 \Join \Phi_2
        \vdash \compInst{p}{D}{V}: \mf{ElemInst}
      }
    } {
      \Gamma \semi \Pi \semi \Phi_1 \Join \Phi_2 \vdash
      \mf{click}_{id}\ \compInst{p}{D}{V}: \tone \bang \mf{Mut}
    }
  \end{mathpar}
  \begin{itemize}
    \item By Lemma \ref{thm:lamhook-unrestb} (unrestriction) on (1),
    we get (3) that $\mf{Succ}(\Phi_1)$.
    \item By Lemma \ref{thm:lamhook-phiwknv} ($\Phi$-weakening) on (2)
    using (3), we get (4) that
    $\Gamma \semi \Pi \semi \Phi_1 \Join \Phi_2
    \vdash V: \mf{ElemInst}$.
    \item We derive conclusion (i) as follows:
    \begin{mathpar}
      \inferrule*[right=TC-Click] {
        \inferrule*[right=(4)] { } {
          \Gamma \semi \Pi \semi \Phi_1 \Join \Phi_2 \vdash V: \mf{ElemInst}
        }
      } {
        \Gamma \semi \Pi \semi \Phi_1 \Join \Phi_2 \vdash
        \mf{click}_{id}\ V: \tone \bang \mf{Mut}
      }
    \end{mathpar}
  \end{itemize}

  \pcase{EC-RunAttr}: Analogous to \textsc{EC-ClickText}.

  \pcase{EC-RunClick}: We have
  $\heap \semi \mf{runClick\ onClick}\ V \step V\ () \semi \heap$ and
  \begin{mathpar}
    \inferrule*[right=TC-RunClick] {
      \inferrule*[right=TV-OnClick] {
        \inferrule* {(1)} {
          \Gamma \semi \Pi \semi \Phi_M \vdash V:
          \tone \xrightarrow{\mf{Mut}} \tone
        }
      } {
        \Gamma \semi \Pi \semi \Phi_M \vdash \mf{onClick}\ V: \mf{Attr}
      }
    } {
      \Gamma \semi \Pi \semi \Phi_M \vdash \mf{runClick\ onClick}\ V:
      \tone \bang \mf{Mut}
    }
  \end{mathpar}
  \begin{itemize}
    \item We derive conclusion (i) as follows:
    \begin{mathpar}
      \inferrule*[right=TC-App] {
        \inferrule*[right=(1)] { } {
          \Gamma \semi \Pi \semi \Phi_M \vdash V:
          \tone \xrightarrow{\mf{Mut}} \tone
        }
        \and
        \inferrule*[right=TV-Unit] { } {
          \Gamma \semi \Pi \semi \emptyset \vdash (): \tone
        }
      } {
        \Gamma \semi \Pi \semi \Phi_M \vdash V\ (): \tone \bang \mf{Mut}
      }
    \end{mathpar}
  \end{itemize}

  \textbf{Case} \textsc{EC-UpdateText} and \pcase{EC-UpdateTag} are
  analogous to \textsc{EC-InitText} and \textsc{EC-InitTag} respectively.

  \pcase{EC-UpdateCompT}: We have
  $$\heap \semi \mf{update}\ \compInst{p}{D}{v}
  \step \eval{D}{\mf{Succ}}{p}{D}
  \semi \heap \mid \heap[p].\mf{rerender} \mapsto \mf{ff}$$
  \begin{mathpar}
    \inferrule*[right=TC-Update] {
      \inferrule*[Right=TV-CompI] {
        \inferrule* {(1)} {\Gamma \semi \Pi \semi \Phi_1 \vdash D: \Pi[p]}
        \and
        \inferrule* {(2)} {(\Phi_1 \Join \Phi_2)[p] = \mf{Succ}}
        \and
        \inferrule* {(3)} {\Gamma \semi \Pi \semi \Phi_2 \vdash V: \mf{ElemInst}}
      } {
        \Gamma \semi \Pi \semi \Phi_1 \Join \Phi_2
        \vdash \compInst{p}{D}{V}: \mf{ElemInst}
      }
    } {
      \Gamma \semi \Pi \semi \Phi_1 \Join \Phi_2
      \vdash \mf{update}\ \compInst{p}{D}{V}: \mf{ElemInst} \bang \mf{Mut}
    }
  \end{mathpar}
  \begin{itemize}
    \item We retain $\Phi_N = \Phi_M$, $\Phi_H' = \Phi_H$, and $\Pi' = \Pi$.
      Conclusions (iii) and (iv) are trivial.
      By premise (i-b), we have
      $\Gamma \semi \Pi \semi \Phi_H \vDash
      \heap \mid \heap[p].\mf{rerender} \mapsto \mf{ff}$,
      satisfying conclusion (ii).
    \item By Lemma \ref{thm:lamhook-unrestv} (unrestriction) on (3),
      we must have (4) $\mf{Succ}(\Phi_2)$.
    \item By Lemma \ref{thm:lamhook-phiwknb} ($\Phi$-weakening) on (1)
      using (4), we get (5) that
      $\Gamma \semi \Pi \semi \Phi_M \vdash D: \Pi[p]$.
    \item We derive conclusion (i) as follows, noting that
       $\Phi_M \Join \Phi_M = \Phi_M$ by (4):
    \begin{mathpar}
      \inferrule*[right=TC-EvalS] {
        \inferrule*[right=(5) (x2)] { } {
          \Gamma \semi \Pi \semi \Phi_M \vdash D: \Pi[p]
        }
        \and
        \inferrule* { } {\Pi[p] \subseteq \Pi[p]}
        \and
        \inferrule*[right=(2)] { } {\Phi_M[p] = \mf{Succ}}
      } {
        \Gamma \semi \Pi \semi \Phi_M \vdash
        \eval{D}{\mf{Succ}}{p}{D}: \mf{Element} \bang \mf{Mut}
      }
    \end{mathpar}
  \end{itemize}

  \pcase{EC-UpdateCompF}: Analogous to \textsc{EC-ClickTagF}.
\end{proof}

\lamhookpreztop*

\begin{proof}
  By case-analysis on the derivation of the small-step judgment.
  (We label our premises and conclusions as in the previous proof.)

  \pcase{ET-Lift}: We have
  \begin{mathpar}
  \inferrule*[right=ET-Lift] {
    \inferrule* {(1)} {\heap \semi M \step N \semi \heap'}
  } {
    \heap \semi M \step_\top N \semi \heap'
  }
  \end{mathpar}
  \begin{itemize}
    \item The conclusion follows from Theorem \ref{thm:lamhook-prezcm}
    (preservation for $\step$) on (1) and the typing premise.
  \end{itemize}

  \pcase{ET-Event}: We have $\heap \semi \mf{val}\ V
  \step_\top \mf{click}_{id}\ V \semi \mf{update}\ V \semi \heap$
  and
  \begin{mathpar}
    \inferrule*[right=TC-Sub] {
      \inferrule*[right=TC-Val] {
        \inferrule* {(1)} {
          \Gamma \semi \Pi \semi \Phi_M \vdash V: \mf{ElemInst}
        }
      } {
        \Gamma \semi \Pi \semi \Phi_M \vdash \mf{val}\ V:
        \mf{ElemInst} \bang \mf{Pure}
      }
    } {
      \Gamma \semi \Pi \semi \Phi_M \vdash \mf{val}\ V:
      \mf{ElemInst} \bang \mf{Mut}
    }
  \end{mathpar}
  \begin{itemize}
    \item Choose $\Phi_N = \Phi_M$, $\Phi_H' = \Phi_H$,
      and $\Pi'=\Pi$, such that conclusions (iii) and (iv) are trivial,
      and conclusion (ii) is given by premise (i-a).
    \item By Lemna \ref{thm:lamhook-unrestv} (unrestriction) on (1),
      we get $\mf{Succ}(\Phi_M)$.
      Then by definition, $\Phi_M = \Phi_M \Join \Phi_M$.
    \item We derive conclusion (i) as follows:
  \end{itemize}
  \begin{mathpar}
    \inferrule*[right=TC-Let] {
      \inferrule*[right=TC-Click] {
        \inferrule*[right=(1)] { } {
          \Gamma \semi \Pi \semi \Phi_M \vdash V: \mf{ElemInst}
        }
      } {
        \Gamma \semi \Pi \semi \Phi_M \vdash \mf{click}_{id}\ V:
        \tone \bang \mf{Mut}
      }
      \
      \inferrule*[Right=TC-Update] {
        \inferrule*[right=(1)] { } {
          \Gamma \semi \Pi \semi \Phi_M \vdash V: \mf{ElemInst}
        }
      } {
        \Gamma \semi \Pi \semi \Phi_M \vdash \mf{update}\ V:
        \mf{ElemInst} \bang \mf{Mut}
      }
    } {
      \Gamma \semi \Pi \semi \Phi_M \vdash
      \elet{\_}{\mf{click}_{id}\ V}{\mf{update}\ V}:
      \mf{ElemInst} \bang \mf{Mut}
    }
  \end{mathpar}
\end{proof}

\subsection{Consistency}
\label{appdx:consistency}

To prove consistency, since
$\mf{Init}$-phase evaluation updates the heap,
we must reason about heaps with respect to a partial order:

\begin{definition}[Heap ordering]
Define $\sigma \subseteq \sigma'$ to hold iff $\sigma[\ell] = \sigma'[\ell]$
for all $\ell \in \dom(\sigma)$.

Define $\pi \sqsubseteq \pi'$ to hold iff
$\pi.\mf{st} \subseteq \pi'.\mf{st}$ and
$\pi.\mf{rerender} = \pi'.\mf{rerender}$.

Finally, define $\heap \sqsubseteq \heap'$ to hold iff
$\heap[p] \sqsubseteq \heap'[p]$
for all $p \in \dom(\heap)$.
\end{definition}

We prove that initialization $\step^\mf{Init}_p$
is monotonic with respect to this heap ordering.

\begin{lemma}[Initialization monotonicity, single-step]
\label{thm:lamhook-monoinitss}
If
\begin{enumerate}[label=(\roman*)]
  \item $\Gamma \semi \Pi \semi \Phi_H \semi \Phi_D \vdash \heap \semi D: \Ell$,
  \item $\Ell \cap \Pi[p] = \emptyset$, and
  \item $\heap \semi D \step^\mf{Init}_p D' \semi \heap'$,
\end{enumerate}
then $\heap \sqsubseteq \heap'$.

\begin{proof}
  By definition of premise \emph{(i)}, we have subpremises
  \begin{enumerate}[label=(i-\alph*)]
    \item $\Gamma \semi \Pi \semi \Phi_D \vdash D: \Ell$ and
    \item $\Gamma \semi \Pi \semi \Phi_H \vDash \heap$.
  \end{enumerate}

  We proceed by induction on the derivation of the step premise \emph{(iii)}.
  It often suffices to show $\heap' = \heap$, which entails
  the conclusion by reflexivity of $\sqsubseteq$.

  \pcase{EB-Let}: We have
  \begin{mathpar}
    \inferrule*[right=EB-Let] {
      \inferrule* {(1)} {\heap \semi M \step N \semi \heap'}
    } {
      \heap \semi \elet{x}{M}{D}
      \step^\mf{Init}_p \elet{x}{N}{D} \semi \heap'
    }

    \inferrule*[right=TB-Let] {
      \inferrule* {(2)} {
        \Gamma \semi \Pi \semi \Phi_1 \vdash M: A \bang \mf{Pure}
      }
      \and
      \Gamma, x: A \semi \Pi \semi \Phi_2 \vdash D: \Ell
    } {
      \Gamma \semi \Pi \semi \Phi_1 \Join \Phi_2 \vdash \elet{x}{M}{D}: \Ell
    }
  \end{mathpar}
  \begin{itemize}
    \item By Theorem \ref{thm:lamhook-prezcp} (pure preservation)
    on (1) and (2), we get that $\heap' = \heap$.
  \end{itemize}

  \pcase{EB-Val}: We have
  $\heap \semi \elet{x}{\mf{val}\ V}{D} \step^\mf{Init}_p
  D[V/x] \semi \heap$, i.e. $\heap' = \heap$ and we are done.

  \pcase{EB-StateI}: We have
  $\heap \semi \estate{x}{\ell}{V}{D}
  \step^\mf{Init}_p D[V/x, \setter^\ell_p/x_\mf{set}]
  \semi \heap'$,
  where $\heap' = \heap \mid \heap[p].\mf{st}[\ell] \mapsto V$, and
  \begin{mathpar}
  \inferrule*[right=TB-State] {
    \Gamma \semi \Pi \semi \Phi_1 \vdash V: A
    \and
    \Gamma, x: A, x_\mf{set}: A \xrightarrow{\mf{Mut}} \tone \semi \Pi
    \semi \Phi_2 \vdash D: \Ell
    \and
    \ell \notin \dom(\Ell)
  } {
    \Gamma \semi \Pi \semi \Phi_1 \Join \Phi_2
    \vdash \estate{x}{\ell}{V}{D}: \Ell, \ell: A
  }
  \end{mathpar}
  \begin{itemize}
    \item By definition of $\heap'$ and $\sqsubseteq$,
    to show $\heap \sqsubseteq \heap'$,
    it suffices to show that $\ell \notin \dom(\heap[p].\mf{st})$.
    \item Premise \emph{(ii)} says that
    $(\Ell, \ell: A) \cap \Pi[p] = \emptyset$,
    i.e. (1) $\ell \notin \dom(\Pi[p])$.
    \item By premise (i-b), (1) entails
    $\ell \notin \dom(\heap[p].\mf{st})$,
    which is what we needed to show.
  \end{itemize}

  \pcase{EB-StateS} is impossible,
  since the step rule requires that $\phi=\mf{Init}$.
\end{proof}
\end{lemma}

\begin{lemma}[Initialization monotonicity]
\label{thm:lamhook-monoinit}
If
\begin{enumerate}[label=(\roman*)]
  \item $\Gamma \semi \Pi \semi \Phi_H \semi \Phi_D \vdash \heap \semi D: \Ell$,
  \item $\Ell \cap \Pi[p] = \emptyset$,
  \item $p \notin \dom(\Phi_H \Join \Phi_D)$, and
  \item $\heap \semi D \steprt^\mf{Init}_p\ D' \semi \heap'$,
\end{enumerate}
then $\heap \sqsubseteq \heap'$.

\begin{proof}
  By induction on the length $n$ of the reduction sequence.

  \textbf{Base case} $n=0$: We have $\heap' = \heap$,
  and the conclusion is trivial.

  \textbf{Inductive case} $n+1$: We have a single step
  (1) $\heap_0 \semi D_0 \step^\mf{Init}_p D_1 \semi \heap_1$
  and a sequence (2)
  $\heap_1 \semi D_1 \steprt^\mf{Init}_p\ D_2 \semi \heap_2$
  of length $n$.
  (We want to show $\heap_0 \sqsubseteq \heap_2$.)
  \begin{itemize}
    \item By Lemma \ref{thm:lamhook-monoinitss} (single-step monotonicity)
    on (1) and premises \emph{(i)} and \emph{(ii)},
    we get (3) that $\heap_0 \sqsubseteq \heap_1$.
    \item By Theorem \ref{thm:lamhook-prezbi} (preservation)
    on (1) and premises \emph{(i)}-\emph{(iii)}, we get some $\Phi_H'$
    and $\Pi'$ with:
    \begin{enumerate}
      \setcounter{enumi}{3}
      \item $\Gamma \semi \Pi' \semi \Phi_H' \semi \Phi_D \vdash
        \heap \semi D_1$,
      \item $\Phi_H' \Join \Phi_D = \Phi_H \Join \Phi_D$, and
      \item $\Ell' \uplus \Pi'[p] = \Ell \uplus \Pi[p]$.
    \end{enumerate}
    \item (5) and premise \emph{(iii)} entail (7)
      that $p \notin \dom(\Phi_H' \Join \Phi_D)$.
    \item (6) implies (8) that $\Ell' \cap \Pi'[p] = \emptyset$.
    \item Now by inductive hypothesis on (2), (4), (7), and (8),
      we get (9) that $\heap_1 \sqsubseteq \heap_2$.
    \item The conclusion is given by (3), (9),
      and transitivity of $\sqsubseteq$.
  \end{itemize}
\end{proof}
\end{lemma}

We next show that steps which \emph{do not change} the heap $\heap$
(but \emph{may read} from $\heap$)
are preserved under a greater heap $\heap' \sqsupseteq \heap$.

\begin{lemma}[Heap lifting, computations]
\label{thm:lamhook-liftc}
If
\begin{enumerate}[label=(\roman*)]
  \item $\heap \semi M \step N \semi \heap$ and
  \item $\heap \sqsubseteq \heap'$,
\end{enumerate}
then $\heap' \semi M \step N \semi \heap'$.

\begin{proof}
  By mutual induction on the derivation of the step judgment
  with Theorem \ref{thm:lamhook-liftb}.

  \pcase{EC-Let}: We have
  \begin{mathpar}
    \inferrule*[right=EC-Let] {
      \inferrule* {(1)} {\heap \semi M \step M' \semi \heap}
    } {
      \heap \semi \elet{x}{M}{N} \step \elet{x}{M'}{N} \semi \heap
    }
  \end{mathpar}
  \begin{itemize}
    \item By inductive hypothesis on (1)
    and premise \emph{(ii)}, we get (2) that
    $\heap' \semi M \step M' \semi \heap'$.
    \item We derive the conclusion as follows:
    \begin{mathpar}
      \inferrule*[right=EC-Let] {
        \inferrule*[right=(2)] { } {\heap' \semi M \step M' \semi \heap'}
      } {
        \heap' \semi \elet{x}{M}{N} \step \elet{x}{M'}{N} \semi \heap'
      }
    \end{mathpar}
  \end{itemize}

  \pcase{EC-Eval}: We have
  \begin{mathpar}
    \inferrule*[right=EC-Eval] {
      \inferrule* {(1)} {\heap \semi D \step^\phi_p D' \semi \heap}
    } {
      \heap \semi \eval{D_0}{\phi}{p}{D}
      \step
      \eval{D_0}{\phi}{p}{D'} \semi \heap
    }
  \end{mathpar}
  \begin{itemize}
    \item By mutual inductive hypothesis on (1) and premise \emph{(ii)},
    we get (2) that $\heap' \semi D \step^\phi_p D' \semi \heap'$.
    \item We derive the conclusion as follows:
    \begin{mathpar}
      \inferrule*[right=EC-Eval] {
        \inferrule*[right=(2)] { } {
          \heap' \semi D \step^\phi_p D' \semi \heap'
        }
      } {
        \heap' \semi \eval{D_0}{\phi}{p}{D}
        \step
        \eval{D_0}{\phi}{p}{D'} \semi \heap'
      }
    \end{mathpar}
  \end{itemize}

  \pcase{EC-AppSet}: We have
  $\heap \semi \setter^\ell_p\ V \step \mf{val}\ () \semi \heap''$,
  where $\heap'' = \heap \mid \heap[p].\mf{st}[\ell] \mapsto V$
  and $\heap[p].\mf{rerender} \mapsto \mf{tt}$.
  \begin{itemize}
    \item By the step premise, we must have $\heap'' = \heap$,
    i.e. $\heap[p].\mf{st}[\ell] = V$ and $\heap[p].\mf{rerender} = \mf{tt}$.
    \item By premise \emph{(ii)}, we therefore also have
    $\heap'[p].\mf{st}[\ell] = V$ and $\heap'[p].\mf{rerender} = \mf{tt}$.
    \item We can reapply \textsc{EC-AppSet} using $\heap'$, getting that
    $\heap' \semi \setter^\ell_p\ V \step \mf{val}\ () \semi \heap'$,
    which satisfies the conclusion.
  \end{itemize}

  \pcase{EC-InitComp}: We have
  $\heap \semi \mf{init}\ \mf{component}\ \{D\}
  \step \eval{D}{\mf{Init}}{p}{D}
  \semi \heap''$, where $p$ is fresh and
  $\heap'' = \heap, p \mapsto \{\mf{st}: \emptyset, \mf{rerender}: \mf{ff}\}$.
  \begin{itemize}
    \item To satisfy the step premise, we must have $\heap'' = \heap$. However,
    this is \emph{impossible}, since $p$ must be fresh, i.e.
    $p \notin \dom(\heap)$. Therefore, this case holds vacuously.
  \end{itemize}

  \pcase{EC-Val}: We have
  $\heap \semi \elet{x}{\mf{val}\ V}{M} \step M[V/x] \semi \heap$.
  \begin{itemize}
    \item The conclusion is immediately rederivable
    via the same rule using $\heap'$.
  \end{itemize}

  \textbf{All other cases} can be shown analogously to \textsc{EC-Val}.

\end{proof}
\end{lemma}

\begin{lemma}[Heap lifting, bodies]
\label{thm:lamhook-liftb}
If
\begin{enumerate}[label=(\roman*)]
  \item $\heap \semi D \step^\mf{Succ}_p D' \semi \heap$ and
  \item $\heap \sqsubseteq \heap'$,
\end{enumerate}
then $\heap' \semi D \step^\mf{Succ}_p D' \semi \heap'$.

\begin{proof}
  By mutual induction on the derivation of the step judgment
  with Theorem \ref{thm:lamhook-liftc}.

  \pcase{EB-Let}: We have
  \begin{mathpar}
    \inferrule*[right=EB-Let] {
      \inferrule* {(1)} {\heap \semi M \step N \semi \heap}
    } {
      \heap \semi \elet{x}{M}{D} \step^\mf{Succ}_p
      \elet{x}{N}{D} \semi \heap
    }
  \end{mathpar}
  \begin{itemize}
    \item By mutual inductive hypothesis on (1)
    and premise \emph{(ii)}, we get (2) that
    $\heap' \semi M \step N \semi \heap'$.
    \item We derive the conclusion as follows:
    \begin{mathpar}
      \inferrule*[right=EB-Let] {
        \inferrule*[right=(2)] { } {\heap' \semi M \step N \semi \heap'}
      } {
        \heap' \semi \elet{x}{M}{D} \step^\mf{Succ}_p
        \elet{x}{N}{D} \semi \heap'
      }
    \end{mathpar}
  \end{itemize}

  \pcase{EB-Ret}: We have
  $\heap \semi \elet{x}{\mf{val}\ V}{D}
  \step^\phi_p D[V/x] \semi \heap$.
  \begin{itemize}
    \item The conclusion is immediately rederivable via the same rule.
  \end{itemize}

  \pcase{EB-StateI}: We have
  $\heap \semi \estate{x}{\ell}{V}{D}
  \step^\mf{Init}_p D[V/x, \setter^\ell_p/x_\mf{set}]
  \semi \heap''$,
  where $\heap'' = \heap \mid \heap[p].\mf{st}[\ell] \mapsto V$.
  \begin{itemize}
    \item By the step premise, we also must have $\heap'' = \heap$,
    i.e. $\heap[p].\mf{st}[\ell] = V$.
    \begin{itemize}
      \item Note: this fact means that the configuration is inherently
      ill-typed, since our type system rules out re-initialization.
      Since this lemma does not concern typing,
      we must still show this case (and it is easy to do so.)
    \end{itemize}
    \item By premise \emph{(ii)}, we also have $\heap'[p].\mf{st}[\ell] = V$.
    \item We can reapply \textsc{EB-StateI} using $\heap'$,
    getting that $\heap' \semi \estate{x}{\ell}{V}{D}
    \step^\mf{Init}_p D[V/x, \setter^\ell_p/x_\mf{set}] \semi \heap'$,
    which satisfies the conclusion.
  \end{itemize}

  \pcase{EB-StateS}: We have
  $$\heap \semi \estate{x}{\ell}{V}{D}
  \step^\mf{Succ}_p D[\heap[p].\mf{st}[\ell]/x, \setter^\ell_p/x_\mf{set}]
  \semi \heap$$
  \begin{itemize}
    \item To show the conclusion (via \textsc{EB-StateS}),
    it suffices to show that $\heap'[p].\mf{st}[\ell] = \heap[p].\mf{st}[\ell]$.
    \item This is given by definition of $\heap \sqsubseteq \heap'$,
    and we are done.
  \end{itemize}
\end{proof}
\end{lemma}

We now proceed to our main proofs of component consistency.

\begin{lemma}[Component consistency, single-step]
\label{thm:lamhook-consistencyss}
If
\begin{enumerate}[label=(\roman*)]
\item $\Gamma \semi \Pi \semi \Phi_H \semi \Phi_D \vdash \heap \semi D: \Ell$
and
\item $\heap \semi D \step^{\phi}_p D' \semi \heap'$,
\end{enumerate}
then $\heap' \semi D \step^{\mf{Succ}}_p D' \semi \heap'$.

\begin{proof}
  By definition of premise \emph{(i)}, we have subpremises
  \begin{enumerate}[label=(i-\alph*)]
    \item $\Gamma \semi \Pi \semi \Phi_D \vdash D: \Ell$ and
    \item $\Gamma \semi \Pi \semi \Phi_H \vDash \heap$.
  \end{enumerate}

  We proceed by case-analysis on the derivation of the step judgment.

  \pcase{EB-Let}: We have
  \begin{mathpar}
    \inferrule*[right=EB-Let] {
      \inferrule* {(1)} {\heap \semi M \step N \semi \heap'}
    } {
      \heap \semi \elet{x}{M}{D} \step^\phi_p
      \elet{x}{N}{D} \semi \heap'
    }

    \inferrule*[right=TB-Let] {
      \inferrule* {(2)} {
        \Gamma \semi \Pi \semi \Phi_1 \vdash M: A \bang \mf{Pure}
      }
      \and
      \Gamma, x: A \semi \Pi \semi \Phi_2 \vdash D: \Ell
    } {
      \Gamma \semi \Pi \semi \Phi_1 \Join \Phi_2 \vdash \elet{x}{M}{D}: \Ell
    }
  \end{mathpar}
  \begin{itemize}
    \item By Theorem \ref{thm:lamhook-prezcp} (pure preservation)
    on (1) and (2),
    we get that $\heap' = \heap$.
    \item We derive the conclusion as follows:
    \begin{mathpar}
      \inferrule*[right=EB-Let] {
        \inferrule*[right=(1)] { } {\heap' \semi M \step N \semi \heap'}
      } {
        \heap' \semi \elet{x}{M}{D} \step^\mf{Succ}_p
        \elet{x}{N}{D} \semi \heap'
      }
    \end{mathpar}
  \end{itemize}

  \pcase{EB-Ret}: We have
  $\heap \semi \elet{x}{\mf{val}\ V}{D}
  \step^\phi_p D[V/x] \semi \heap$.
  \begin{itemize}
    \item The conclusion is immediately rederivable at $\phi=\mf{Succ}$
    via the same rule.
  \end{itemize}

  \pcase{EB-StateI}: We have
  $\heap \semi \estate{x}{\ell}{V}{D}
  \step^\mf{Init}_p D[V/x, \setter^\ell_p/x_\mf{set}]
  \semi \heap'$,
  where $\heap' = \heap \mid \heap[p].\mf{st}[\ell] \mapsto V$.
  \begin{itemize}
    \item Since $\heap'[p].\mf{st}[\ell] = V$, we apply \textsc{EB-StateS}
    to get that
    \[\heap' \semi \estate{x}{\ell}{V}{D}
    \step^\mf{Succ}_p D[V/x, \setter^\ell_p/x_\mf{set}] \semi \heap'\]
    which satsisfies the conclusion.
  \end{itemize}

  \pcase{EB-StateS}: We have
  $\heap \semi \estate{x}{\ell}{V}{D} \step^\mf{Succ}_p
  D[W/x, \setter^\ell_p/x_\mf{set}] \semi \heap$,
  where $W = \heap[p].\mf{st}[\ell]$.
  \begin{itemize}
    \item Since $\heap' = \heap$ and $\phi=\mf{Succ}$,
    the conclusion is immediately given.
  \end{itemize}
\end{proof}
\end{lemma}

\consistency*

\begin{proof}
  We label the premises as follows:
  \begin{enumerate}[label=(\roman*)]
    \item $\Gamma \semi \Pi \semi \Phi_H \semi \Phi_D
      \vdash \heap \semi D: \Ell$,
    \item $\Ell \cap \Pi[p] = \emptyset$,
    \item $p \notin \dom(\Phi_H \Join \Phi_D)$, and
    \item $\heap \semi D \steprt^\phi_p\ \mf{val}\ V \semi \heap'$.
  \end{enumerate}

  We proceed by induction on the length $n$ of the reduction sequence.

  \textbf{Base case} $n=0$: We have
  $\heap \semi \mf{val}\ V \steprt^\phi_p\ \mf{val}\ V \semi \heap$.
  \begin{itemize}
    \item We get that $\heap' = \heap$.
    The conclusion is thus trivial by reflexivity of
    $\steprt^\mf{Succ}_p$.
  \end{itemize}

  \textbf{Inductive case} $n+1$: We have a single step
  (1) $\heap_0 \semi D_0 \step^\mf{Init}_p D_1 \semi \heap_1$
  and a sequence (2)
  $\heap_1 \semi D_1 \steprt^\mf{Init}_p\ \mf{val}\ V \semi \heap_2$
  of length $n$.
  (We want to show
  $\heap_2 \semi D_0 \steprt^\mf{Succ}_p \mf{val}\ V \semi \heap_2$.)
  \begin{itemize}
    \item By Theorem \ref{thm:lamhook-prezbi} (preservation)
    on (1) and premises \emph{(i)}-\emph{(iii)},
    we get some $\Phi_H'$ and $\Pi'$ with
    \begin{enumerate}
      \setcounter{enumi}{2}
      \item $\Gamma \semi \Pi' \semi \Phi_H' \semi \Phi_D
        \vdash \heap_1 \semi D_1: \Ell'$,
      \item $\Phi_D \Join \Phi_H' = \Phi_D \Join \Phi_H$, and
      \item $\Ell' \uplus \Pi'[p] = \Ell \uplus \Pi[p]$.
    \end{enumerate}
    \item (4) and premise \emph{(iii)} entail (6)
      that $p \notin \dom(\Phi_H' \Join \Phi_D)$.
    \item (5) implies (7) that $\Ell' \cap \Pi'[p] = \emptyset$.
    \item By inductive hypothesis on (2), (3), (6), and (7),
      we get (8)
      $\heap_2 \semi D_1 \steprt^\mf{Succ}_p \mf{val}\ V \semi \heap_2$.
    \item By Lemma \ref{thm:lamhook-monoinit} (initialization monotonicity)
      on (2), (3), (5), and (6),
      we get (9) that $\heap_1 \sqsubseteq \heap_2$.
    \item By Lemma \ref{thm:lamhook-consistencyss} (single-step consistency)
      on (1) and premise \emph{(i)}, we get (10) that\\
      $\heap_1 \semi D_0 \step^\mf{Succ}_p D_1 \semi \heap_1$.
    \item By Lemma \ref{thm:lamhook-liftb} (heap lifting) on (10) using (9),
    we get (11) that
    $\heap_2 \semi D_0 \step^\mf{Succ}_p D_1 \semi \heap_2$.
    \item We conclude by (11) and (8) that
    $\heap_2 \semi D_0 \steprt^\mf{Succ}_p\ \mf{val}\ V \semi \heap_2$,
    and we are done.
  \end{itemize}
\end{proof}

\section{Translation Proofs}
\label{appdx:trans}

\subsection{Type Preservation}
\label{appdx:transprez}

We begin with a couple of lemmas which let us type the results of
the $\delta$-function:

\begin{lemma}[Flag typing]
\label{thm:trans-tpflag}
If
\begin{enumerate}[label=(\roman*)]
\item $\Gamma \semi \Pi \vDash_p \delta : \Ell$ and
\item $\Pi \semi \emptyset \vDash \R$,
\end{enumerate}
then $\Gamma \semi \bsq{\Pi}\R \vdash \delta(p, \R): \mf{Ref\ Bool}$.

\begin{proof}
  By case-analysis on premise \emph{(i)}, either:
  \begin{enumerate}[label=(\alph*)]
  \item $\mf{flag} \mapsto x \in \delta$
  and $x: \mf{Ref\ Bool} \in \Gamma$, or
  \item $\mf{flag} \notin \dom(\delta)$ and $p \in \dom(\Pi)$.
  \end{enumerate}

  \textbf{Case} (a):
  \begin{itemize}
    \item By definition, $\delta(p, \R)=x$.
    \item From (a), we also have that $x: \mf{Ref\ Bool} \in \Gamma$,
    i.e. (1) $\delta(p, \R): \mf{Ref\ Bool} \in \Gamma$.
    \item We derive the conclusion as follows:
    \begin{mathpar}
      \inferrule*[right=TV-Var] {
        \inferrule*[right=(1)] { } {
          \delta(p, \R): \mf{Ref\ Bool} \in \Gamma
        }
      } {
        \Gamma \semi \bsq{\Pi}\R \vdash \delta(p, \R): \mf{Ref\ Bool}
      }
    \end{mathpar}
  \end{itemize}

  \textbf{Case} (b):
  \begin{itemize}
    \item By definition, $\delta(p, \R)=\R(p)$.
    \item By premise \emph{(ii)}, since $p \in \dom(\Pi)$,
    we have that $\R(p)$ is well-defined. Then by definition of $\bsq{\Pi}\R$,
    we get that $\R(p): \mf{Bool} \in \bsq{\Pi}\R$,
    i.e. (2) $\delta(p, \R): \mf{Bool} \in \bsq{\Pi}\R$.
    \item We derive the conclusion as follows:
    \begin{mathpar}
      \inferrule*[right=TV-Loc] {
        \inferrule*[right=(2)] { } {\delta(p, \R): \mf{Bool} \in \bsq{\Pi}\R}
      } {
        \Gamma \semi \bsq{\Pi}\R \vdash \delta(p, \R): \mf{Ref\ Bool}
      }
    \end{mathpar}
  \end{itemize}
\end{proof}
\end{lemma}

\begin{lemma}[Label typing]
\label{thm:trans-tplab}
If
\begin{enumerate}[label=(\roman*)]
\item $\Gamma \semi \Pi \vDash_p \delta : (\Ell, \ell: A)$ and
\item $\Pi \semi \emptyset \vDash \R$,
\end{enumerate}
then $\Gamma \semi \bsq{\Pi}\R \vdash
\delta(p, \ell, \R): \mf{Ref}\ \bsq{A}$.

\begin{proof}
  By case-analysis on premise \emph{(i)}, either:
  \begin{enumerate}[label=(\alph*)]
    \item $x \mapsto \ell \in \delta$ and $x: \mf{Ref}\ \bsq{A} \in \Gamma$, or
    \item $\ell \notin \dom(\delta)$ and $\ell: A \in \Pi[p]$.
  \end{enumerate}

  \textbf{Case} (a):
  \begin{itemize}
    \item By definition, $\delta(p, \ell, \R)=x$.
    \item From (a), we also have that $x: \mf{Ref}\ \bsq{A} \in \Gamma$,
    i.e. (1) $\delta(p, \ell, \R): \mf{Ref}\ \bsq{A} \in \Gamma$.
    \item We derive the conclusion as follows:
    \begin{mathpar}
      \inferrule*[right=TV-Var] {
        \inferrule*[right=(1)] { } {
          \delta(p, \ell, \R): \mf{Ref}\ \bsq{A} \in \Gamma
        }
      } {
        \Gamma \semi \bsq{\Pi}\R \vdash \delta(p, \ell, \R):
        \mf{Ref}\ \bsq{A}
      }
    \end{mathpar}
  \end{itemize}

  \textbf{Case} (b):
  \begin{itemize}
    \item By definition, $\delta(p, \ell, \R)=\R(p, \ell)$.
    \item By premise \emph{(ii)}, since $\ell: A \in \Pi[p]$,
    we have that $\R(p, \ell)$ is well-defined.
    Then by definition of $\bsq{\Pi}\R$,
    we get that $\R(p, \ell): \bsq{A} \in \bsq{\Pi}\R$,
    i.e. (2) $\delta(p, \ell, \R): \bsq{A} \in \bsq{\Pi}\R$.
    \item We derive the conclusion as follows:
    \begin{mathpar}
      \inferrule*[right=TV-Loc] {
        \inferrule*[right=(2)] { } {
          \delta(p, \ell, \R): \bsq{A} \in \bsq{\Pi}\R
        }
      } {
        \Gamma \semi \bsq{\Pi}\R \vdash \delta(p, \ell, \R):
        \mf{Ref}\ \bsq{A}
      }
    \end{mathpar}
  \end{itemize}

\end{proof}
\end{lemma}

We now proceed to the main proofs of type preservation.

\begin{restatable}[Type preservation, values]{theorem}{transprezv}
\label{thm:trans-prezv}
If
\begin{enumerate}[label=(\roman*)]
\item $\Gamma \semi \Pi \semi \Phi \vdash V: A$ and
\item $\Pi \semi \mf{ids}(V) \vDash \R$,
\end{enumerate}
then $\bsq{\Gamma} \semi \bsq{\Pi}\R \vdash \bsq{V}\R: \bsq{A}$.
\end{restatable}
\begin{proof}
  By mutual induction on the derivation of the value typing judgment \emph{(i)}
  with Theorems \ref{thm:trans-prezc}, \ref{thm:trans-prezbs},
  and \ref{thm:trans-prezbi}

  \pcase{TV-Var}: We have $\bsq{x}\R = x$ and
  \begin{mathpar}
    \inferrule*[right=TV-Var] {
      \inferrule* {(1)} {
        x: A \in \Gamma
      }
      \and \mf{Succ}(\Phi)
    } {
      \Gamma \semi \Pi \semi \Phi \vdash x: A
    }
  \end{mathpar}
  \begin{itemize}
    \item By (1) and definition of $\bsq{\Gamma}$, we have (2) that
    $x: \bsq{A} \in \bsq{\Gamma}$.
    \item We derive the conclusion as follows:
    \begin{mathpar}
      \inferrule*[right=TV-Var] {
        \inferrule*[right=(2)] { } {
          x: \bsq{A} \in \bsq{\Gamma}
        }
      } {
        \bsq{\Gamma} \semi \bsq{\Pi}\R \vdash x: \bsq{A}
      }
    \end{mathpar}
  \end{itemize}

  \pcase{TV-Unit}: $\bsq{()}\R=()$ and $\bsq{\tone}=\tone$,
  so the conclusion is easy:
  \begin{mathpar}
    \inferrule*[right=TV-Unit] {
      \mf{Succ}(\Phi)
    } {
      \Gamma \semi \Pi \semi \Phi \vdash (): \tone
    }

    \inferrule*[right=TV-Unit] { } {
      \bsq{\Gamma} \semi \bsq{\Pi}\R \vdash (): \tone
    }
  \end{mathpar}

  \textbf{Cases} \textsc{TV-True}, \textsc{TV-False},
  \textsc{TV-Str}, and \textsc{TV-Nil} can all be shown analogously
  to the above case.

  \pcase{TV-Pair}: We have \[
  \bsq{(V, W)}\R = (\bsq{V}\R, \bsq{W}\R)
  \quad
  \bsq{A \times B} = \bsq{A} \times \bsq{B}
  \quad
  \mf{ids}((V, W)) = \mf{ids}(V) \cup \mf{ids}(W)
  \]
  \begin{mathpar}
    \inferrule*[right=TV-Pair] {
      \inferrule* {(1)} {\Gamma \semi \Pi \semi \Phi_1 \vdash V: A}
      \and
      \inferrule* {(2)} {\Gamma \semi \Pi \semi \Phi_2 \vdash W: B}
    } {
      \Gamma \semi \Pi \semi \Phi_1 \Join \Phi_2 \vdash
      (V, W): A \times B
    }
  \end{mathpar}
  \begin{itemize}
    \item By premise \emph{(ii)} and definition of $\mf{ids}$,
    we get (3) that $\Pi \semi \mf{ids}(V) \vDash \R$
    and (4) that\\$\Pi \semi \mf{ids}(W) \vDash \R$.
    \item By inductive hypothesis on each of pairs (1)/(3) and (2)/(4),
    we get (5) that\\
    $\bsq{\Gamma} \semi \bsq{\Pi}\R \vdash \bsq{V}\R: \bsq{A}$
    and (6) that
    $\bsq{\Gamma} \semi \bsq{\Pi}\R \vdash \bsq{W}\R: \bsq{B}$,
    respectively.
    \item We derive the conclusion as follows:
    \begin{mathpar}
      \inferrule*[right=TV-Pair] {
        \inferrule*[right=(5)] { } {
          \bsq{\Gamma} \semi \bsq{\Pi}\R \vdash \bsq{V}\R: \bsq{A}
        }
        \and
        \inferrule*[right=(6)] { } {
          \bsq{\Gamma} \semi \bsq{\Pi}\R \vdash \bsq{W}\R: \bsq{B}
        }
      } {
        \bsq{\Gamma} \semi \bsq{\Pi}\R
        \vdash (\bsq{V}\R, \bsq{W}\R): \bsq{A} \times \bsq{B}
      }
    \end{mathpar}
  \end{itemize}

  \textbf{Cases} \textsc{TV-Cons}, \textsc{TV-Attr},
  \textsc{TV-OnClick}, \textsc{TV-Text}, \textsc{TV-Tag},
  \textsc{TV-TextI}, and \textsc{TV-TagI}
  can all be shown analogously
  to the above case.

  \pcase{TV-Lam}: We have
  \[
  \bsq{\lambda x.M}\R = \lambda x.\bsq{M}\R
  \qquad
  \bsq{A \xrightarrow{E} B} = \bsq{A} \rightarrow \bsq{B}
  \qquad
  \mf{ids}(\lambda x.M) = \mf{ids}(M)
  \]
  \begin{mathpar}
    \inferrule*[right=TV-Lam] {
      \inferrule* {(1)} {
        \Gamma, x: A \semi \Pi \semi \Phi \vdash M: B \bang E
      }
      \and
      \mf{Succ}(\Phi)
    } {
      \Gamma \semi \Pi \semi \Phi \vdash \lambda x.M: A \xrightarrow{E} B
    }
  \end{mathpar}
  \begin{itemize}
    \item By mutual inductive hypothesis (Theorem \ref{thm:trans-prezc})
    on (1) and premise \emph{(ii)}, we have (2) that
    $\bsq{\Gamma} x: \bsq{A} \semi \bsq{\Pi}\R
    \vdash \bsq{M}\R: \bsq{B}$.
    \item We derive the conclusion as follows:
    \begin{mathpar}
      \inferrule*[right=TV-Lam] {
        \inferrule*[right=(2)] { } {
          \bsq{\Gamma}, x: \bsq{A} \semi \bsq{\Pi}\R
          \vdash \bsq{M}\R: \bsq{B}
        }
      } {
        \bsq{\Gamma} \semi \bsq{\Pi}\R
        \vdash \lambda x.\bsq{M}\R: \bsq{A} \rightarrow \bsq{B}
      }
    \end{mathpar}
  \end{itemize}

  \pcase{TV-Setter}: We have
  \[
  \bsq{\setter^\ell_p}\R = \lambda x.
    (\R(p, \ell) := x \semi \R(p) := \mf{true})
  \qquad
  \bsq{A \xrightarrow{\mf{Mut}} \tone} = \bsq{A} \rightarrow \tone
  \]
  \begin{mathpar}
  \inferrule*[right=TV-Setter] {
    \inferrule* {(1)} {\ell: A \in \Pi[p]}
    \and \mf{Succ}(\Phi)
  } {
    \Gamma \semi \Pi \semi \Phi \vdash \setter^\ell_p:
    A \xrightarrow{\mf{Mut}} \tone
  }
  \end{mathpar}
  \begin{itemize}
    \item By (1) and premise \emph{(ii)}, we have
    (2) that $\R(p)$ is well-defined
    and (3) that $\R(p, \ell)$ is well-defined.
    \item By definition of $\bsq{\Pi}\R$, from (2) and (3) respectively
    we get (4) that $\R(p): \mf{Bool} \in \bsq{\Pi}\R$
    and (5) that $\R(p, \ell): \bsq{A} \in \bsq{\Pi}\R$,
    noting that $\Pi[p][\ell]=A$ by (1).
    \item We derive the conclusion by way of subderivations (6) and (7)
    as follows:
    \begin{mathpar}
      \inferrule*[right=TC-Set] {
        \inferrule*[right=TV-Loc] {
          \inferrule*[right=(5)] { } {
            \R(p, \ell): \bsq{A} \in \bsq{\Pi}_R
          }
        } {
          \Gamma_6 \semi \bsq{\Pi}\R
          \vdash \R(p, \ell): \mf{Ref}\ \bsq{A}
        }
        \and
        \inferrule*[right=TV-Var] {
          \inferrule* { } {x: \bsq{A} \in \Gamma_6}
        } {
          \Gamma_6 \semi \bsq{\Pi}\R
          \vdash x: \bsq{A}
        }
      } {
        (6) \quad
        \Gamma_6 \semi \bsq{\Pi}\R
        \vdash \R(p, \ell) := x: \tone
      }
      \and
      \inferrule*[right=TC-Set] {
        \inferrule*[right=TV-Loc] {
          \inferrule*[right=(4)] { } {
            \R(p): \mf{Bool} \in \bsq{\Pi}\R
          }
        } {
          \Gamma_7 \semi \bsq{\Pi}\R \vdash \R(p): \mf{Ref\ Bool}
        }
        \and
        \inferrule*[right=TV-True] { } {
          \Gamma_7 \semi \bsq{\Pi}\R \vdash \mf{true}: \mf{Bool}
        }
      } {
        (7) \quad
        \Gamma_7 \semi \bsq{\Pi}\R
        \vdash \R(p) := \mf{true}: \tone
      }
      \and
      \inferrule*[right=TV-Lam] {
        \inferrule*[right=TC-Let] {
          [(6),\ \Gamma_6=\bsq{\Gamma}, x: \bsq{A}]
          \and
          [(7),\ \Gamma_7=\bsq{\Gamma}, x: \bsq{A}]
        } {
          \Gamma, x: \bsq{A} \semi \bsq{\Pi}\R
          \vdash \R(p, \ell) := x \semi \R(p) := \mf{true}: \tone
        }
      } {
        \bsq{\Gamma} \semi \bsq{\Pi}\R
        \vdash \lambda x.(\R(p, \ell) := x \semi \R(p) := \mf{true}):
        \bsq{A} \rightarrow \tone
      }
    \end{mathpar}
  \end{itemize}

  \pcase{TV-Comp}: We have
  \begin{gather*}
    \bsq{\mf{component}\ \{D\}}\R = \mf{component}\left(
    \lambda ().\elet{x}{\mf{ref}(\mf{false})}{
      \itransp{D}{D}{\bot}{\R}{\mf{flag} \mapsto x}
    }
  \right)\\
    \bsq{\mf{Element}} = \mf{Element}\\
    \mf{ids}(\mf{component}\ \{D\}) = \mf{ids}(D)
  \end{gather*}
  \begin{mathpar}
  \inferrule*[right=TV-Comp] {
    \inferrule* {(1)} {\Gamma \semi \Pi \semi \Phi \vdash D: \Ell}
  } {
    \Gamma \semi \Pi \semi \Phi \vdash \mf{component}\ \{D\}: \mf{View}
  }
  \end{mathpar}
  \begin{itemize}
    \item By definition, we have (2) that
    $x: \mf{Ref\ Bool} \semi \Pi \vDash_\bot \mf{flag} \mapsto x: \emptyset$.
    \item By Theorem \ref{thm:trans-prezbi} (initial-phase type preservation)
    on (1), (2), and premise \emph{(ii)}, we get (3) that
    $\bsq{\Gamma}, x: \mf{Ref\ Bool} \semi \bsq{\Pi}\R \vdash
    \itransp{D}{D}{\bot}{\R}{\mf{flag} \mapsto x}:
    A_\mf{comp}$.
    \item We derive the conclusion by way of subderivation (4) as follows:
    \begin{mathpar}
      \inferrule*[right=TC-Ref] {
        \inferrule*[right=TV-False] { } {
          \bsq{\Gamma} \semi \bsq{\Pi}\R \vdash
          \mf{false}: \mf{Bool}
        }
      } {
        (4) \quad
        \bsq{\Gamma} \semi \bsq{\Pi}\R \vdash
        \mf{ref}(\mf{false}): \mf{Ref\ Bool}
      }
      \and
      \inferrule*[right=TV-Comp] {
        \inferrule*[Right=TV-Lam] {
          \inferrule*[Right=TC-Let] {
            (4)
            \and
            \inferrule*[right=(3)] { } {
              \bsq{\Gamma}, x: \mf{Ref\ Bool} \semi \bsq{\Pi}\R \vdash
              \itransp{D}{D}{\bot}{\R}{\mf{flag} \mapsto x}:
              A_\mf{comp}
            }
          } {
            \bsq{\Gamma} \semi \bsq{\Pi}\R \vdash
            \elet{x}{\mf{ref}(\mf{false})}{
              \itransp{D}{D}{\bot}{\R}{\mf{flag} \mapsto x}
            } :A_\mf{comp}
          }
        } {
          \bsq{\Gamma} \semi \bsq{\Pi}\R \vdash
          \lambda ().\elet{x}{\mf{ref}(\mf{false})}{
            \itransp{D}{D}{\bot}{\R}{\mf{flag} \mapsto x}
          }: \tone \rightarrow A_\mf{comp}
        }
      } {
        \bsq{\Gamma} \semi \bsq{\Pi}\R \vdash
        \mf{component}\ \lambda ().\elet{x}{\mf{ref}(\mf{false})}{
          \cdots
        }: \mf{Element}
      }
    \end{mathpar}
  \end{itemize}

  \pcase{TV-CompI}: We have
  \begin{gather*}
    \bsq{\compInst{p}{D}{V}}\R = \mf{compInst}(
    \lambda ().\stransp{D}{\bot}{\R}{\emptyset},
    \R(p),
    \bsq{V}\R
  )\\
  \bsq{\mf{ElemInst}}=\mf{ElemInst}
  \qquad
  \mf{ids}(\compInst{p}{D}{V}) = \mf{ids}(D) \cup \mf{ids}(V)
  \end{gather*}
  \begin{mathpar}
    \inferrule*[right=TV-CompI] {
      \inferrule* {(1)} {\Gamma \semi \Pi \semi \Phi_1 \vdash D: \Pi[p]}
      \and
      (\Phi_1 \Join \Phi_2)[p] = \mf{Succ}
      \and
      \inferrule* {(2)} {\Gamma \semi \Pi \semi \Phi_2 \vdash V: \mf{ElemInst}}
    } {
      \Gamma \semi \Pi \semi \Phi_1 \Join \Phi_2
      \vdash \compInst{p}{D}{V}: \mf{ElemInst}
    }
  \end{mathpar}
  \begin{itemize}
    \item By premise \emph{(ii)} and definition of $\mf{ids}$,
    we have (3) that $\Pi \semi \mf{ids}(D) \vDash \R$
    and (4) that\\$\Pi \semi \mf{ids}(V) \vDash \R$.
    \item By definition, we have (5) that
    $\emptyset \semi \Pi \vDash_p \emptyset: \Pi[p]$.
    \item By mutual inductive hypothesis (Theorem \ref{thm:trans-prezbs})
    on (1), (3) and (5), we get (6) that\\
    $\bsq{\Gamma} \semi \bsq{\Pi}\R \vdash
    \stransp{D}{\bot}{\R}{\emptyset}: \mf{Element}$.
    \item By inductive hypothesis on (2) and (4),
    we get (7) that
    $\bsq{\Gamma} \semi \bsq{\Pi}\R \vdash \bsq{V}\R: \mf{ElemInst}$.
    \item By premise \emph{(ii)}, since $p \in \dom(\Pi)$ (given by (1)),
    we have that $\R(p)$ is well-defined. Then by definition of $\bsq{\Pi}\R$,
    we get (8) that $\R(p): \mf{Bool} \in \bsq{\Pi}\R$.
    \item We derive the conclusion by way of subderivation (9) as follows:
    \begin{mathpar}
      \inferrule*[right=TV-Lam] {
        \inferrule*[right=(6)] { } {
          \bsq{\Gamma} \semi \bsq{\Pi}\R \vdash
          \stransp{D}{\bot}{\R}{\emptyset}:
          \mf{Element}
        }
      } {
        (9) \quad
        \bsq{\Gamma} \semi \bsq{\Pi}\R \vdash
        \lambda ().\stransp{D}{\bot}{\R}{\emptyset}:
        \tone \rightarrow \mf{Element}
      }
      \and
      \inferrule*[right=TV-CompI] {
        (9)
        \and
        \inferrule*[right=TV-Loc] {
          \inferrule*[right=(8)] { } {
            \R(p): \mf{Bool} \in \bsq{\Pi}\R
          }
        } {
          \bsq{\Gamma} \semi \bsq{\Pi}\R \vdash \R(p): \mf{Ref\ Bool}
        }
        \
        \inferrule*[Right=(7)] { } {
          \bsq{\Gamma} \semi \bsq{\Pi}\R \vdash \bsq{V}\R: \mf{ElemInst}
        }
      } {
        \bsq{\Gamma} \semi \bsq{\Pi}\R \vdash \mf{compInst}(
          \lambda ().\stransp{D}{\bot}{\R}{\emptyset},
          \R(p),
          \bsq{V}\R
        ): \mf{ElemInst}
      }
    \end{mathpar}
  \end{itemize}

\end{proof}

\begin{restatable}[Type preservation, computations]{theorem}{transprezc}
\label{thm:trans-prezc}
If
\begin{enumerate}[label=(\roman*)]
\item $\Gamma \semi \Pi \semi \Phi \vdash M: A \bang E$ and
\item $\Pi \semi \mf{ids}(M) \vDash \R$,
\end{enumerate}
then $\bsq{\Gamma} \semi \bsq{\Pi}\R \vdash \bsq{M}\R: \bsq{A}$.
\end{restatable}
\begin{proof}
  By mutual induction on the derivation of the computation
  typing judgment \emph{(i)} with Theorems \ref{thm:trans-prezv},
  \ref{thm:trans-prezbs}, and \ref{thm:trans-prezbi}.

  \pcase{TC-Val}: We have \[
  \bsq{\mf{val}\ V}\R = \mf{val}\ \bsq{V}\R
  \qquad
  \mf{ids}(\mf{val}\ V) = \mf{ids}(V)
  \]
  \begin{mathpar}
    \inferrule*[right=TC-Val] {
      \inferrule* {(1)} {\Gamma \semi \Pi \semi \Phi \vdash V: A}
    } {
      \Gamma \semi \Pi \semi \Phi \vdash \mf{val}\ V: A \bang \mf{Pure}
    }
  \end{mathpar}
  \begin{itemize}
    \item By mutual inductive hypothesis (Theorem \ref{thm:trans-prezv})
    on (1) and premise \emph{(ii)}, we have (2) that
    $\bsq{\Gamma} \semi \bsq{\Pi}\R \vdash \bsq{V}\R: \bsq{A}$.
    \item We derive the conclusion as follows:
    \begin{mathpar}
      \inferrule*[right=TC-Val] {
        \inferrule*[right=(2)] { } {
          \bsq{\Gamma} \semi \bsq{\Pi}\R \vdash \bsq{V}\R: \bsq{A}
        }
      } {
        \bsq{\Gamma} \semi \bsq{\Pi}\R \vdash \mf{val}\ \bsq{V}\R: \bsq{A}
      }
    \end{mathpar}
  \end{itemize}

  \textbf{Cases} \textsc{TC-Fst}, \textsc{TC-Snd}, \textsc{TC-Fold},
  \textsc{TC-Init}, \textsc{TC-Click}, \textsc{TC-RunClick}, and
  \textsc{TC-Update} can all be shown analogously to the above case.

  \pcase{TC-Let}: We have \[
  \bsq{\elet{x}{M}{N}}\R = \elet{x}{\bsq{M}\R}{\bsq{N}\R}
  \qquad
  \mf{ids}(\elet{x}{M}{N}) = \mf{ids}(M) \cup \mf{ids}(N)
  \]
  \begin{mathpar}
    \inferrule*[right=TC-Let] {
      \inferrule* {(1)} {\Gamma \semi \Pi \semi \Phi_1 \vdash M: A \bang E}
      \and
      \inferrule* {(2)} {
        \Gamma, x: A \semi \Pi \semi \Phi_2 \vdash N: B \bang E
      }
    } {
      \Gamma \semi \Pi \semi \Phi_1 \Join \Phi_2 \vdash
      \elet{x}{M}{N}: B \bang E
    }
  \end{mathpar}
  \begin{itemize}
    \item By premise \emph{(ii)} and definition of $\mf{ids}$,
    we have (3) that $\Pi \semi \mf{ids}(M) \vDash \R$
    and (4) that\\$\Pi \semi \mf{ids}(N) \vDash \R$.
    \item By inductive hypothesis on each of pairs (1)/(3) and (2)/(4),
    we get (5) that\\
    $\bsq{\Gamma} \semi \bsq{\Pi}\R \vdash \bsq{M}\R: \bsq{A}$
    and (6) that
    $\bsq{\Gamma}, x: \bsq{A} \semi \bsq{\Pi}\R
    \vdash \bsq{N}\R: \bsq{B}$,
    respectively.
    \item We derive the conclusion as follows:
    \begin{mathpar}
      \inferrule*[right=TC-Let] {
        \inferrule*[right=(5)] { } {
          \bsq{\Gamma} \semi \bsq{\Pi}\R \vdash \bsq{M}\R: \bsq{A}
        }
        \and
        \inferrule*[right=(6)] { } {
          \bsq{\Gamma}, x: \bsq{A} \semi \bsq{\Pi}\R
          \vdash \bsq{N}\R: \bsq{B}
        }
      } {
        \bsq{\Gamma} \semi \bsq{\Pi}\R \vdash
        \elet{x}{\bsq{M}\R}{\bsq{N}\R}: \bsq{B}
      }
    \end{mathpar}
  \end{itemize}

  \pcase{TC-App}: We have \[
  \bsq{V\ W}\R = \bsq{V}\R\ \bsq{W}\R
  \qquad
  \mf{ids}(V\ W) = \mf{ids}(V) \cup \mf{ids}(W)
  \]
  \begin{mathpar}
    \inferrule*[right=TC-App] {
      \inferrule* {(1)} {
        \Gamma \semi \Pi \semi \Phi_1 \vdash V: A \xrightarrow{E} B
      }
      \and
      \inferrule* {(2)} {
        \Gamma \semi \Pi \semi \Phi_2 \vdash W: A
      }
    } {
      \Gamma \semi \Pi \semi \Phi_1 \Join \Phi_2 \vdash V\ W: B \bang E
    }
  \end{mathpar}
  \begin{itemize}
    \item By premise \emph{(ii)} and definition of $\mf{ids}$,
    we have (3) that $\Pi \semi \mf{ids}(V) \vDash \R$
    and (4) that\\$\Pi \semi \mf{ids}(W) \vDash \R$.
    \item By mutual inductive hypothesis (Theorem \ref{thm:trans-prezv})
    on each of pairs (1)/(3) and (2)/(4), we get (5) that
    $\bsq{\Gamma} \semi \bsq{\Pi}\R \vdash \bsq{V}\R:
    \bsq{A} \rightarrow \bsq{B}$
    and (6) that
    $\bsq{\Gamma} \semi \bsq{\Pi}\R \vdash \bsq{W}\R: \bsq{A}$,
    respectively.
    \item We derive the conclusion as follows:
    \begin{mathpar}
      \inferrule*[right=TC-App] {
        \inferrule*[right=(5)] { } {
          \bsq{\Gamma} \semi \bsq{\Pi}\R \vdash \bsq{V}\R:
           \bsq{A} \rightarrow \bsq{B}
        }
        \and
        \inferrule*[right=(6)] { } {
          \bsq{\Gamma} \semi \bsq{\Pi}\R \vdash \bsq{W}\R: \bsq{A}
        }
      } {
        \bsq{\Gamma} \semi \bsq{\Pi}\R \vdash \bsq{V}\R\ \bsq{W}\R: \bsq{B}
      }
    \end{mathpar}
  \end{itemize}

  \pcase{TC-If}: We have
  \begin{gather*}
  \bsq{\eif{V}{M_1}{M_2}}\R = \eif{\bsq{V}\R}{\bsq{M_1}\R}{\bsq{M_2}\R}
  \\
  \mf{ids}(\eif{V}{M_1}{M_2}) =
  \mf{ids}(V) \cup \mf{ids}(M_1) \cup \mf{ids}(M_2)
  \end{gather*}
  \begin{mathpar}
    \inferrule*[right=TC-If] {
      \inferrule* {(1)} {\Gamma \semi \Pi \semi \Phi_1 \vdash V: \mf{Bool}}
      \and
      \inferrule* {(2)} {\Gamma \semi \Pi \semi \Phi_2 \vdash M_1: A \bang E}
      \and
      \inferrule* {(3)} {\Gamma \semi \Pi \semi \Phi_2 \vdash M_2: A \bang E}
    } {
      \Gamma \semi \Pi \semi \Phi_1 \Join \Phi_2 \vdash
      \eif{V}{M_1}{M_2}: A \bang E
    }
  \end{mathpar}
  \begin{itemize}
    \item By premise \emph{(ii)} and definition of $\mf{ids}$,
    we have (4) that $\Pi \semi \mf{ids}(V) \vDash \R$,
    (5) that $\Pi \semi \mf{ids}(M_1) \vDash \R$,
    and (6) that $\Pi \semi \mf{ids}(M_2) \vDash \R$.
    \item By mutual inductive hypothesis (Theorem \ref{thm:trans-prezv})
    on (1) and (4), we get (7) that\\
    $\bsq{\Gamma} \semi \bsq{\Pi}\R \vdash \bsq{V}\R: \mf{Bool}$.
    \item By inductive hypothesis
    on each of pairs (2)/(5) and (3)/(6), we get (8) that\\
    $\bsq{\Gamma} \semi \bsq{\Pi}\R \vdash \bsq{M_1}\R: \bsq{A}$
    and (9) that
    $\bsq{\Gamma} \semi \bsq{\Pi}\R \vdash \bsq{M_2}\R: \bsq{A}$,
    respectively.
    \item We derive the conclusion as follows:
  \end{itemize}
  \begin{mathpar}
    \inferrule*[right=TC-If] {
      \inferrule*[right=(7)] { } {
        \bsq{\Gamma} \semi \bsq{\Pi}\R \vdash \bsq{V}\R: \mf{Bool}
      }
      \
      \inferrule*[right=(8)\ \&\ (9)] { } {
        \bsq{\Gamma} \semi \bsq{\Pi}\R \vdash \bsq{M_i}\R: \bsq{A}
        \quad (i \in \{1, 2\})
      }
    } {
      \bsq{\Gamma} \semi \bsq{\Pi}\R \vdash
      \eif{\bsq{V}\R}{\bsq{M_1}\R}{\bsq{M_2}\R}: \bsq{A}
    }
  \end{mathpar}

  \pcase{TC-Sub}: We have
  \begin{mathpar}
  \inferrule*[right=TC-Sub] {
    \inferrule* {(1)} {\Gamma \semi \Pi \semi \Phi \vdash M: A \bang \mf{Pure}}
  } {
    \Gamma \semi \Pi \semi \Phi \vdash M: A \bang \mf{Mut}
  }
  \end{mathpar}
  \begin{itemize}
    \item The result follows directly from the inductive hypothesis
    applied to (1) and premise \emph{(ii)}.
  \end{itemize}

  \pcase{TC-EvalI}: We have
  \[\bsq{\eval{D_0}{\mf{Init}}{p}{D}}\R =
  \elet{rr, f, x}{\itransp{D}{D_0}{p}{\R}{\emptyset}}{
    \mf{compInst}(rr, f, \mf{init}\ x)
  }\] whose right-hand side desugars to
  \begin{align*}
    &\elet{x_0}{\itransp{D}{D_0}{p}{\R}{\emptyset}}{
      \elet{rr}{\mf{fst}\ x_0}{\elet{x_1}{\mf{snd}\ x_0}{\\
      &\elet{f}{\mf{fst}\ x_1}{\elet{x}{\mf{snd}\ x_1}{
        \elet{x_2}{\mf{init}\ x}{\mf{val\ compInst}(rr, f, x_2)}
      }}
    }}}
  \end{align*}
  where all $x_i$ are fresh (as are $rr$, $f$, and $x$). We also have
  $\mf{ids}(\eval{D_0}{\mf{Init}}{p}{D}) =
  \mf{ids}(D_0) \cup \mf{ids}(D)$ and
  \begin{mathpar}
    \inferrule*[right=TC-EvalI] {
      \inferrule* {(1)} {
        \Gamma \semi \Pi \semi \Phi_1 \vdash D_0: \Ell \uplus \Pi[p]
      }
      \and
      \inferrule* {(2)} {\Gamma \semi \Pi \semi \Phi_2 \vdash D: \Ell}
      \and
      \mf{Succ}(\Phi)
    } {
      \Gamma \semi \Pi \semi \Phi_1 \Join \Phi_2, p: \mf{Init} \vdash
      \eval{D_0}{\mf{Init}}{p}{D}: \mf{ElemInst} \bang \mf{Mut}
    }
  \end{mathpar}
  \begin{itemize}
    \item Call $\Ell_0 = \Ell \uplus \Pi[p]$.
    By definition of $\uplus$, we get (3) that $\Ell_0 \setminus \Ell = \Pi[p]$
    and (4) $\Ell \subseteq \Ell_0$.
    \item By (3) and definition, we get
    (5) $\emptyset \semi \Pi \vDash_p \emptyset: \Ell_0 \setminus \Ell$.
    \item By mutual inductive hypothesis (Theorem \ref{thm:trans-prezbi})
    on (1), (2), (4), (5), and premise \emph{(ii)}, we get (6) that
    $\bsq{\Gamma} \semi \bsq{\Pi}\R \vdash
    \itransp{D}{D_0}{p}{\R}{\emptyset}: A_\mf{comp}$.
    \item We derive the conclusion by way of subderivations (7)-(12) as follows,
    recalling that\\$A_\mf{comp} \triangleq (\tone \rightarrow \mf{Element})
    \times \mf{Ref\ Bool} \times \mf{Element}$,
    and additionally abbreviating\\
    $A'_\mf{comp} = \mf{Ref\ Bool} \times \mf{Element}$:
  \end{itemize}
  \begin{mathpar}
    \inferrule*[right=TC-Val] {
      \inferrule*[Right=TV-CompI] {
        \inferrule* { } {
          \Gamma_7 \semi \bsq{\Pi}\R \vdash rr: \tone \rightarrow \mf{Element}
        }
        \
        \inferrule* { } {
          \Gamma_7 \semi \bsq{\Pi}\R \vdash f: \mf{Ref\ Bool}
        }
        \
        \inferrule*[Right=TV-Var (x3)] { } {
          \Gamma_7 \semi \bsq{\Pi}\R \vdash x_2: \mf{ElemInst}
        }
      } {
        \Gamma_7 \semi \bsq{\Pi}\R \vdash
        \mf{compInst}(rr, f, x_2): \mf{ElemInst}
      }
    } {
      (7) \quad
      \Gamma_7 \semi \bsq{\Pi}\R \vdash
      \mf{val\ compInst}(rr, f, x_2): \mf{ElemInst}
    }
    \and
    \inferrule*[right=TC-Let] {
      \inferrule*[right=TC-Init] {
        \inferrule*[right=TV-Var] { } {
          \Gamma_8 \semi \bsq{\Pi}\R \vdash x: \mf{Element}
        }
      } {
        \Gamma_8 \semi \bsq{\Pi}\R \vdash \mf{init}\ x: \mf{ElemInst}
      }
      \and [(7),\ \Gamma_7=\Gamma_8,x_2: \mf{ElemInst}]
    } {
      (8) \quad
      \Gamma_8 \semi \bsq{\Pi}\R \vdash \elet{x_2}{\mf{init}\ x}{\cdots}:
      \mf{ElemInst}
    }
    \and
    \inferrule*[right=TC-Let] {
      \inferrule*[right=TC-Snd] {
        \inferrule*[right=TV-Var] { } {
          \Gamma_9 \semi \bsq{\Pi}\R \vdash x_1: A'_\mf{comp}
        }
      } {
        \Gamma_9 \semi \bsq{\Pi}\R \vdash \mf{snd}\ x_1: \mf{Element}
      }
      \and [(8),\ \Gamma_8=\Gamma_9,x: \mf{Element}]
    } {
      (9) \quad
      \Gamma_9 \semi \bsq{\Pi}\R \vdash \elet{x}{\mf{snd}\ x_1}{\cdots}:
      \mf{ElemInst}
    }
  \end{mathpar}
  \begin{mathpar}
    \inferrule*[right=TC-Let] {
      \inferrule*[right=TC-Fst] {
        \inferrule*[right=TV-Var] { } {
          \Gamma_{10} \semi \bsq{\Pi}\R \vdash x_1: A'_\mf{comp}
        }
      } {
        \Gamma_{10} \semi \bsq{\Pi}\R \vdash \mf{fst}\ x_1: \mf{Ref\ Bool}
      }
      \and [(9),\ \Gamma_9=\Gamma_{10},f: \mf{Ref\ Bool}]
    } {
      (10) \quad
      \Gamma_{10} \semi \bsq{\Pi}\R \vdash \elet{f}{\mf{fst}\ x_1}{\cdots}:
      \mf{ElemInst}
    }
    \and
    \inferrule*[right=TC-Let] {
      \inferrule*[right=TC-Snd] {
        \inferrule*[right=TV-Var] { } {
          \Gamma_{11} \semi \bsq{\Pi}\R \vdash x_0: A_\mf{comp}
        }
      } {
        \Gamma_{11} \semi \bsq{\Pi}\R \vdash \mf{snd}\ x_0: A'_\mf{comp}
      }
      \and
      [(10),\ \Gamma_{10}=\Gamma_{11}, x_1: A'_\mf{comp}]
    } {
      (11) \quad
      \Gamma_{11} \semi \bsq{\Pi}\R \vdash \elet{x_1}{\mf{snd}\ x_0}{\cdots}:
      \mf{ElemInst}
    }
    \and
    \inferrule*[right=TC-Let] {
      \inferrule*[right=TC-Fst] {
        \inferrule*[right=TV-Var] { } {
          \Gamma_{12} \semi \bsq{\Pi}\R \vdash x_0: A_\mf{comp}
        }
      } {
        \Gamma_{12} \semi \bsq{\Pi}\R \vdash \mf{fst}\ x_0:
        \tone \rightarrow \mf{Element}
      }
      \and [(11),\ \Gamma_{12}=\Gamma_{12},rr: \tone \rightarrow \mf{Element}]
    } {
      (12) \quad
      \Gamma_{12} \semi \bsq{\Pi}\R \vdash \elet{rr}{\mf{fst}\ x_0}{\cdots}:
      \mf{ElemInst}
    }
    \and
    \inferrule*[right=TC-Let] {
      \inferrule*[right=(6)] { } {
        \bsq{\Gamma} \semi \bsq{\Pi}\R \vdash
        \itransp{D}{D_0}{p}{\R}{\emptyset}:
        A_\mf{comp}
      }
      \and [(12),\ \Gamma_{12}=\bsq{\Gamma}, x_0: A_\mf{comp}]
    } {
      \bsq{\Gamma} \semi \bsq{\Pi}\R \vdash
      \elet{x_0}{
        \itransp{D}{D_0}{p}{\R}{\emptyset}
      }{\cdots}: \mf{ElemInst}
    }
  \end{mathpar}

  \pcase{TC-EvalS}: We have
  \[\bsq{\eval{D_0}{\mf{Succ}}{p}{D}}\R =
  \mf{compInst}(
    \lambda ().\stransp{D_0}{p}{\R}{\emptyset},
    \R(p),
    \mf{init}\ \stransp{D}{p}{\R}{\emptyset}
  )\] whose right-hand side desugars to
  \begin{align*}
    &\elet{x_0}{\stransp{D}{p}{\R}{\emptyset}}{
      \elet{x_1}{\mf{init}\ x_0}{\\
        &\mf{val\ compInst}(
          \lambda ().\stransp{D_0}{p}{\R}{\emptyset},
          \R(p),
          x_1
        )
    }}
  \end{align*}
  where all $x_i$ are fresh. We also have
  $\mf{ids}(\eval{D_0}{\mf{Succ}}{p}{D}) =
  \mf{ids}(D_0) \cup \mf{ids}(D)$ and
  \begin{mathpar}
    \inferrule*[right=TC-EvalS] {
      \inferrule* {(1)} {\Gamma \semi \Pi \semi \Phi_1 \vdash D_0: \Pi[p]}
      \
      \inferrule* {(2)} {\Gamma \semi \Pi \semi \Phi_2 \vdash D: \Ell}
      \
      \inferrule* {(3)} {\Ell \subseteq \Pi[p]}
      \
      (\Phi_1 \Join \Phi_2)[p] = \mf{Succ}
    } {
      \Gamma \semi \Pi \semi \Phi_1 \Join \Phi_2 \vdash
      \eval{D_0}{\mf{Succ}}{p}{D}: \mf{ElemInst} \bang \mf{Mut}
    }
  \end{mathpar}
  \begin{itemize}
    \item By premise \emph{(ii)} and definition of $\mf{ids}$,
    we have (4) that $\Pi \semi \mf{ids}(D_0) \vDash \R$
    and (5) that\\$\Pi \semi \mf{ids}(D) \vDash \R$.
    \item By definition, we have (6) that
    $\emptyset \semi \Pi \vDash_p \emptyset: \Pi[p]$.
    \item By (3), we also have (7) that
    $\emptyset \semi \Pi \vDash_p \emptyset: \Ell$.
    \item By mutual inductive hypothesis (Theorem \ref{thm:trans-prezbs})
    on (1), (4) and (6), we get (8) that\\
    $\bsq{\Gamma} \semi \bsq{\Pi}\R \vdash
    \stransp{D_0}{p}{\R}{\emptyset}: \mf{Element}$.
    \item By mutual inductive hypothesis (Theorem \ref{thm:trans-prezbs})
    on (2), (5), and (7), we get (9) that\\
    $\bsq{\Gamma} \semi \bsq{\Pi}\R \vdash
    \stransp{D}{p}{\R}{\emptyset}: \mf{Element}$.
    \item By premise \emph{(ii)}, since $p \in \dom(\Pi)$
    (given e.g. by (3)),
    we have that $\R(p)$ is well-defined. Then by definition of $\bsq{\Pi}\R$,
    we get (10) that $\R(p): \mf{Bool} \in \bsq{\Pi}\R$.
    \item We derive the conclusion by way of subderivations
    (11)-(13) as follows, weakening (6) on $\Gamma$ appropriately:
  \end{itemize}
  \begin{mathpar}
    \inferrule*[right=TV-Lam] {
      \inferrule*[right=(8)] { } {
        \Gamma_{11} \semi \bsq{\Pi}\R \vdash
        \stransp{D_0}{p}{\R}{\emptyset}: \mf{Element}
      }
    } {
      (11) \quad
      \Gamma_{11} \semi \bsq{\Pi}\R \vdash
      \lambda ().\stransp{D_0}{p}{\R}{\emptyset}:
      \tone \rightarrow \mf{Element}
    }
    \and
    \inferrule*[right=TC-Val] {
      \inferrule*[Right=TV-CompI] {
        [(11),\ \Gamma_{11} = \Gamma_{12}]
        \quad
        \inferrule*[right=TV-Loc] {
          \inferrule*[right=(10)] { } {
            \R(p): \mf{Bool} \in \bsq{\Pi}\R
          }
        } {
          \Gamma_{12} \semi \bsq{\Pi}\R \vdash \R(p): \mf{Ref\ Bool}
        }
        \
        \inferrule*[Right=TV-Var] { } {
          \Gamma_{12} \semi \bsq{\Pi}\R \vdash x_1: \mf{ElemInst}
        }
      } {
        \Gamma_{12} \semi \bsq{\Pi}\R \vdash
        \mf{compInst}(\lambda ().\stransp{D_0}{p}{\R}{\emptyset},
            \R(p),
            x_1
        ): \mf{ElemInst}
      }
    } {
      (12) \quad
      \Gamma_{12} \semi \bsq{\Pi}\R \vdash \mf{val\ compInst}(
        \lambda ().\stransp{D_0}{p}{\R}{\emptyset}, \R(p), x_1
      ): \mf{ElemInst}
    }
    \and
    \inferrule*[right=TC-Let] {
      \inferrule*[right=TC-Init] {
        \inferrule*[right=TV-Var] { } {
          \Gamma_{13} \semi \bsq{\Pi}\R \vdash x_0: \mf{Element}
        }
      } {
        \Gamma_{13} \semi \bsq{\Pi}\R \vdash \mf{init}\ x_0: \mf{ElemInst}
      }
      \and
      [(12),\ \Gamma_{12} = \Gamma_{13}, x_1: \mf{ElemInst}]
    } {
      (13) \quad
      \Gamma_{13} \semi \bsq{\Pi}\R \vdash
      \elet{x_1}{\mf{init}\ x_0}{\cdots}: \mf{ElemInst}
    }
    \and
    \inferrule*[right=TC-Let] {
      \inferrule*[right=(7)] { } {
        \bsq{\Gamma} \semi \bsq{\Pi}\R \vdash
        \stransp{D}{p}{\R}{\emptyset}: \mf{Element}
      }
      \and
      [(13),\ \Gamma_{13} = \bsq{\Gamma}, x_0: \mf{Element}]
    } {
      \bsq{\Gamma} \semi \bsq{\Pi}\R \vdash
      \elet{x_0}{\stransp{D}{p}{\R}{\emptyset}}{\cdots}:
      \mf{ElemInst}
    }
  \end{mathpar}
\end{proof}

\begin{restatable}[Type preservation, bodies, successor]{theorem}{transprezbs}
\label{thm:trans-prezbs}
If
\begin{enumerate}[label=(\roman*)]
\item $\Gamma \semi \Pi \semi \Phi \vdash D: \Ell$,
\item $\Pi \semi \mf{ids}(D) \vDash \R$, and
\item $\Gamma' \semi \Pi \vDash_p \delta: \Ell$,
\end{enumerate}
then $\bsq{\Gamma}, \Gamma' \semi \bsq{\Pi}\R
  \vdash \stransp{D}{p}{\R}{\delta}: \mf{Element}$.
\end{restatable}
\begin{proof}
  By mutual induction on the derivation of the body typing judgment \emph{(i)}
  with Theorems \ref{thm:trans-prezv}, \ref{thm:trans-prezc},
  and \ref{thm:trans-prezbi}.

  \pcase{TB-Val}: We have
  $\stransp{\mf{val}\ V}{p}{\R}{\delta} = \mf{val}\ \bsq{V}\R$ and
  \begin{mathpar}
    \inferrule*[right=TB-Val] {
      \inferrule* {(1)} {
        \Gamma \semi \Pi \semi \Phi \vdash V: \mf{Element}
      }
    } {
      \Gamma \semi \Pi \semi \Phi \vdash V: \emptyset
    }
  \end{mathpar}
  \begin{itemize}
    \item By mutual inductive hypothesis (Theorem \ref{thm:trans-prezv})
    on (1) and premise \emph{(ii)}, we have (2) that
    $\bsq{\Gamma} \semi \bsq{\Pi}\R \vdash \bsq{V}\R: \mf{Element}$.
    \item Weakening (2) on $\Gamma$ gives us (3) that
    $\bsq{\Gamma}, \Gamma' \semi \bsq{\Pi}\R \vdash \bsq{V}\R: \mf{Element}$.
    \item We derive the conclusion as follows:
    \begin{mathpar}
      \inferrule*[right=TC-Val] {
        \inferrule*[right=(3)] { } {
          \bsq{\Gamma}, \Gamma' \semi \bsq{\Pi}\R
          \vdash \bsq{V}\R: \mf{Element}
        }
      } {
        \bsq{\Gamma}, \Gamma' \semi \bsq{\Pi}\R \vdash \mf{val}\ \bsq{V}\R:
        \mf{Element}
      }
    \end{mathpar}
  \end{itemize}

  \pcase{TB-Let}: We have
  \begin{gather*}
  \stransp{\elet{x}{M}{D}}{p}{\R}{\delta} =
  \elet{x}{\bsq{M}\R}{\stransp{D}{p}{\R}{\delta}}
  \\
  \mf{ids}(\elet{x}{M}{D}) = \mf{ids}(M) \cup \mf{ids}(D)
  \end{gather*}
  \begin{mathpar}
    \inferrule*[right=TB-Let] {
      \inferrule* {(1)} {
        \Gamma \semi \Pi \semi \Phi_1 \vdash M: A \bang \mf{Pure}
      }
      \and
      \inferrule* {(2)} {\Gamma, x: A \semi \Pi \semi \Phi_2 \vdash D: \Ell}
    } {
      \Gamma \semi \Pi \semi \Phi_1 \Join \Phi_2 \vdash \elet{x}{M}{D}: \Ell
    }
  \end{mathpar}
  \begin{itemize}
    \item By premise \emph{(ii)} and definition of $\mf{ids}$,
    we have (3) that $\Pi \semi \mf{ids}(M) \vDash \R$
    and (4) that\\$\Pi \semi \mf{ids}(D) \vDash \R$.
    \item By mutual inductive hypothesis (Theorem \ref{thm:trans-prezc})
    on (1) and (3), we get (5) that\\
    $\bsq{\Gamma} \semi \bsq{\Pi}\R \vdash \bsq{M}\R: \bsq{A}$.
    \item By inductive hypothesis on (2), (4), and premise \emph{(iii)},
    we get (6) that\\
    $\bsq{\Gamma}, x: \bsq{A} \semi \bsq{\Pi}\R \vdash
    \stransp{D}{p}{\R}{\delta}: \mf{Element}$.
    \item We derive the conclusion as follows:
  \end{itemize}
  \begin{mathpar}
    \inferrule*[right=TC-Let] {
      \inferrule*[right=(5)] { } {
        \bsq{\Gamma} \semi \bsq{\Pi}\R \vdash \bsq{M}\R: \bsq{A}
      }
      \
      \inferrule*[right=(6)] { } {
        \bsq{\Gamma}, x: \bsq{A} \semi \bsq{\Pi}\R \vdash
        \stransp{D}{p}{\R}{\delta}: \mf{Element}
      }
    } {
      \bsq{\Gamma} \semi \bsq{\Pi}\R \vdash
      \elet{x}{\bsq{M}\R}{\stransp{D}{p}{\R}{\delta}}:
      \mf{Element}
    }
  \end{mathpar}

  \pcase{TB-State}: We have
  \begin{gather*}
    \begin{aligned}
      \stransp{\estate{x}{\ell}{V}{D}}{p}{\R}{\delta}\
      &= \elet{x}{\bang \delta(p, \ell, \R)}{\\
        &\quad\ \elet{x_\mf{set}}{\mf{val}\ \lambda y.
          (\delta(p, \ell, \R) := y \semi \delta(p, \R) := \mf{true})
        }{\\
        &\quad\ \stransp{D}{p}{\R}{\delta}
        }
      }
    \end{aligned}\\
    \mf{ids}(\estate{x}{\ell}{V}{D}) = \mf{ids}(V) \cup \mf{ids}(D)
  \end{gather*}
  \begin{mathpar}
    \inferrule*[right=TB-State] {
      \Gamma \semi \Pi \semi \Phi_1 \vdash V: A
      \and
      \inferrule* {(1)} {
        \Gamma, x: A, x_\mf{set}: A \xrightarrow{\mf{Mut}} \tone \semi \Pi
        \semi \Phi_2 \vdash D: \Ell
      }
      \and
      \ell \notin \dom(\Ell)
    } {
      \Gamma \semi \Pi \semi \Phi_1 \Join \Phi_2
      \vdash \estate{x}{\ell}{V}{D}: \Ell, \ell: A
    }
  \end{mathpar}
  \begin{itemize}
    \item By premise \emph{(ii)} and definition of $\mf{ids}$,
    we have (2) that $\Pi \semi \mf{ids}(D) \vDash \R$.
    \item From premise \emph{(ii)} it also follows that (3)
    $\Pi \semi \emptyset \vDash \R$.
    \item Since we have $\Gamma' \semi \Pi \vDash_p \delta: \Ell, \ell: A$
    from premise \emph{(iii)}, by definition we also have (4) that
    $\Gamma' \semi \Pi \vDash_p \delta: \Ell$.
    \item By inductive hypothesis on (1), (2), and (4),
    we get (5) that\\
    $\bsq{\Gamma, x: A, x_\mf{set}: A \xrightarrow{\mf{Mut}} \tone}, \Gamma'
    \semi \bsq{\Pi}\R
    \vdash \stransp{D}{p}{\R}{\delta}: \mf{Element}$.
    \item By definition, (5) is equivalent to (6) that\\
    $\bsq{\Gamma}, \Gamma', x: \bsq{A}, x_\mf{set}: \bsq{A} \rightarrow \tone
    \semi \bsq{\Pi}\R
    \vdash \stransp{D}{p}{\R}{\delta}: \mf{Element}$.
    \item By Lemma \ref{thm:trans-tplab} (label typing) on (3)
    and premise \emph{(iii)}, we get (7) that\\
    $\Gamma' \semi \bsq{\Pi}\R \vdash \delta(p, \ell, \R): \mf{Ref}\ \bsq{A}$.
    \item By \ref{thm:trans-tpflag} (flag typing) on (3)
    and premise \emph{(iii)}, we get (8) that
    $\Gamma' \semi \bsq{\Pi}\R \vdash \delta(p, \R): \mf{Ref\ Bool}$.
    \item We derive the conclusion by way of subderivations (9)-(12)
    as follows, weakening our previous subconclusions on $\Gamma$
    as appropriate:
    \begin{mathpar}
      \inferrule*[right=TC-Set] {
        \inferrule*[right=(7)] { } {
          \Gamma_9 \semi \bsq{\Pi}\R \vdash \delta(p, \ell, \R):
          \mf{Ref}\ \bsq{A}
        }
        \and
        \inferrule*[right=TV-Var] { } {
          \Gamma_9 \semi \bsq{\Pi}\R \vdash y: \bsq{A}
        }
      } {
        (9) \quad
        \Gamma_9 \semi \bsq{\Pi}\R \vdash \delta(p, \ell, \R) := y: \tone
      }
      \and
      \inferrule*[right=TC-Set] {
        \inferrule*[right=(8)] { } {
          \Gamma_{10} \semi \bsq{\Pi}\R \vdash \delta(p, \R): \mf{Ref\ Bool}
        }
        \and
        \inferrule*[right=TV-True] { } {
          \Gamma_{10} \semi \bsq{\Pi}\R \vdash \mf{true}: \mf{Bool}
        }
      } {
        (10) \quad
        \Gamma_{10} \semi \bsq{\Pi}\R \vdash \delta(p, \R) := \mf{true}: \tone
      }
      \and
      \inferrule*[right=TC-Val] {
        \inferrule*[right=TV-Lam] {
          \inferrule*[right=TC-Let] {
            [(9),\ \Gamma_9 = \Gamma_{11}, y: \bsq{A}]
            \and
            [(10),\ \Gamma_{10} = \Gamma_{11}, y: \bsq{A}]
          } {
            \Gamma_{11}, y: \bsq{A} \semi \bsq{\Pi}\R \vdash
            \delta(p, \ell, \R) := y \semi \delta(p, \R) := \mf{true}: \tone
          }
        } {
          \Gamma_{11} \semi \bsq{\Pi}\R \vdash
          \lambda y. (\delta(p, \ell, \R) := y
          \semi \delta(p, \R) := \mf{true}):
          \bsq{A} \rightarrow \tone
        }
      } {
        (11) \quad
        \Gamma_{11} \semi \bsq{\Pi}\R \vdash
        \mf{val}\ \lambda y.
        (\delta(p, \ell, \R) := y \semi \delta(p, \R) := \mf{true}):
        \bsq{A} \rightarrow \tone
      }
      \and
      \inferrule*[right=TC-Let] {
        [(11),\ \Gamma_{11} = \Gamma_{12}]
        \and
        \inferrule*[right=(6)] { } {
          \Gamma_{12}, x_\mf{set}: \bsq{A} \rightarrow \tone \semi \bsq{\Pi}\R
          \vdash \stransp{D}{p}{\R}{\delta}: \mf{Element}
        }
      } {
        (12) \quad
        \Gamma_{12} \semi \bsq{\Pi}\R \vdash
        \elet{x_\mf{set}}{\cdots}{\stransp{D}{p}{\R}{\delta}}:
        \mf{Element}
      }
      \and
      \inferrule*[right=TC-Let] {
        \inferrule*[right=TC-Get] {
          \inferrule*[right=(7)] { } {
            \bsq{\Gamma}, \Gamma' \semi \bsq{\Pi}\R \vdash
            \delta(p, \ell, \R): \mf{Ref}\ \bsq{A}
          }
        } {
          \bsq{\Gamma}, \Gamma' \semi \bsq{\Pi}\R \vdash
          \bang \delta(p, \ell, \R): \bsq{A}
        }
        \
        [(12), \Gamma_{12} = \bsq{\Gamma}, \Gamma', x: \bsq{A}]
      } {
        \bsq{\Gamma}, \Gamma' \semi \bsq{\Pi}\R \vdash
        \elet{x}{\bang \delta(p, \ell, \R)}{\cdots}: \mf{Element}
      }
    \end{mathpar}
  \end{itemize}
\end{proof}

\transprezbi*

\begin{proof}
  By mutual induction on the derivation of the body typing judgment \emph{(i)}
  with Theorems \ref{thm:trans-prezv}, \ref{thm:trans-prezc},
  and \ref{thm:trans-prezbs}.
  (Recall that $A_\mf{comp} \triangleq (\tone \rightarrow \mf{Element})
  \times \mf{Ref\ Bool} \times \mf{Element}$.)

  \pcase{TB-Val}: We have
  \[
  \itransp{\mf{val}\ V}{D_0}{p}{\R}{\delta} =
  \mf{val\ }(
    \lambda ().\stransp{D_0}{p}{\R}{\delta},
    \delta(p, \R), \bsq{V}\R)
  \]
  \begin{mathpar}
    \inferrule*[right=TB-Val] {
      \inferrule* {(1)} {
        \Gamma \semi \Pi \semi \Phi_1 \vdash V: \mf{Element}
      }
    } {
      \Gamma \semi \Pi \semi \Phi_1 \vdash \mf{val}\ V: \emptyset
    }
  \end{mathpar}
  \begin{itemize}
    \item By premise \emph{(iii)} and definition of $\mf{ids}$,
    we have (2) that $\Pi \semi \mf{ids}(V) \vDash \R$,
    (3) that $\Pi \semi \mf{ids}(D_0) \vDash \R$,
    and (4) that $\Pi \semi \emptyset \vDash \R$.
    \item By mutual inductive hypothesis (Theorem \ref{thm:trans-prezv})
    on (1) and (2), we have (5) that\\
    $\bsq{\Gamma} \semi \bsq{\Pi}\R \vdash \bsq{V}\R: \mf{Element}$.
    \item Since we have $\Ell=\emptyset$, premise \emph{(iv)}
    reduces to (6) that $\Gamma' \semi \Pi \vDash_p \delta: \Ell_0$.
    \item By mutual inductive hypothesis (Theorem \ref{thm:trans-prezbs})
    on (3), (6), and premise \emph{(ii)} we have (7) that
    $\bsq{\Gamma}, \Gamma' \semi \bsq{\Pi}\R \vdash
    \stransp{D_0}{p}{\R}{\delta}: \mf{Element}$.
    \item By Lemma \ref{thm:trans-tpflag} (flag typing)
    on (4) and (6), we get (8) that
    $\Gamma' \semi \bsq{\Pi}\R \vdash \delta(p, \R): \mf{Ref\ Bool}$.
    \item We derive the conclusion by way of subderivations (9) and (10)
    as follows, where we weaken our previous subconclusions on $\Gamma$
    as appropriate:
  \end{itemize}
  \begin{mathpar}
    \inferrule*[right=TV-Lam] {
      \inferrule*[right=(4)] { } {
        \bsq{\Gamma}, \Gamma' \semi \bsq{\Pi}\R \vdash
        \stransp{D_0}{p}{\R}{\delta}: \mf{Element}
      }
    } {
      (9) \quad
      \bsq{\Gamma}, \Gamma' \semi \bsq{\Pi}\R \vdash
      \lambda ().\stransp{D_0}{p}{\R}{\delta}:
      \tone \rightarrow \mf{Element}
    }
    \and
    \inferrule*[right=TV-Pair] {
      \inferrule*[right=(8)] { } {
        \bsq{\Gamma}, \Gamma' \semi \bsq{\Pi}\R \vdash
        \delta(p, \R): \mf{Ref\ Bool}
      }
      \and
      \inferrule*[right=(5)] { } {
        \bsq{\Gamma}, \Gamma' \semi \bsq{\Pi}\R \vdash \bsq{V}\R: \mf{Element}
      }
    } {
      (10) \quad
      \bsq{\Gamma}, \Gamma' \semi \bsq{\Pi}\R \vdash
      (\delta(p, \R), \bsq{V}\R): \mf{Ref\ Bool} \times \mf{Element}
    }
    \and
    \inferrule*[right=TC-Val] {
      \inferrule*[Right=TV-Pair] {
        (9)
        \and
        (10)
      } {
        \bsq{\Gamma}, \Gamma' \semi \bsq{\Pi}\R \vdash
        (\lambda ().\stransp{D_0}{p}{\R}{\delta},
        \delta(p, \R), \bsq{V}\R): A_\mf{comp}
      }
    } {
      \bsq{\Gamma}, \Gamma' \semi \bsq{\Pi}\R \vdash
      \mf{val}\ (
        \lambda ().\stransp{D_0}{p}{\R}{\delta},
        \delta(p, \R), \bsq{V}\R): A_\mf{comp}
    }
  \end{mathpar}

  \pcase{TB-Let}: We have
  $\itransp{\elet{x}{M}{D}}{D_0}{p}{\R}{\delta}=
  \elet{x}{\bsq{M}\R}{\itransp{D}{D_0}{p}{\R}{\delta}}$
  and
  \begin{mathpar}
    \inferrule*[right=TB-Let] {
      \inferrule* {(1)} {
        \Gamma \semi \Pi \semi \Phi_1 \vdash M: A \bang \mf{Pure}
      }
      \and
      \inferrule* {(2)} {\Gamma, x: A \semi \Pi \semi \Phi_1' \vdash D: \Ell}
    } {
      \Gamma \semi \Pi \semi \Phi_1 \Join \Phi_1' \vdash \elet{x}{M}{D}: \Ell
    }
  \end{mathpar}
  \begin{itemize}
    \item By premise \emph{(iii)} and definition of $\mf{ids}$,
    we have (3) that $\Pi \semi \mf{ids}(M) \vDash \R$ and
    (4) that\\$\Pi \semi \mf{ids}(D) \vDash \R$.
    \item By mutual inductive hypothesis (Theorem \ref{thm:trans-prezc})
    on (1) and (3), we have (5) that\\
    $\bsq{\Gamma} \semi \bsq{\Pi}\R \vdash \bsq{M}\R: \bsq{A}$.
    \item We weaken premise \emph{(ii)} on $\Gamma$
    to get (6) that $\Gamma, x: A \semi \Pi \vdash D_0: \Ell_0$.
    \item By inductive hypothesis on (2), (4), (6), and premise \emph{(iv)},
    we get (7) that\\
    $\bsq{\Gamma}, \Gamma', x: \bsq{A} \semi \bsq{\Pi}\R \vdash
    \itransp{D}{D_0}{p}{\R}{\delta}: A_\mf{comp}$.
    \item We derive the conclusion as follows, weakening (5)
    on $\Gamma$ as appropriate:
    \begin{mathpar}
      \inferrule[(TC-Let)] {
        \inferrule*[right=(5)] { } {
          \bsq{\Gamma}, \Gamma' \semi \bsq{\Pi}\R \vdash \bsq{M}\R: \bsq{A}
        }
        \
        \inferrule*[Right=(6)] { } {
          \bsq{\Gamma}, \Gamma', x: \bsq{A} \semi \bsq{\Pi}\R \vdash
          \itransp{D}{D_0}{p}{\R}{\delta}: A_\mf{comp}
        }
      } {
        \bsq{\Gamma}, \Gamma' \semi \bsq{\Pi}\R \vdash
        \elet{x}{\bsq{M}\R}{
          \itransp{D}{D_0}{p}{\R}{\delta}
        }: A_\mf{comp}
      }
    \end{mathpar}
  \end{itemize}

  \pcase{TB-State}: We have
  \begin{gather*}
    \begin{aligned}
      \itransp{\estate{x}{\ell}{V}{D}}{D_0}{p}{\R}{\delta}
      &\triangleq \elet{x_\ell}{\mf{ref}(\bsq{V}\R)}{\\
        &\quad\ \elet{x}{\bang x_\ell}{\\
        &\quad\ \elet{x_\mf{set}}{\mf{val}\ \lambda y.
          (x_\ell := y \semi \delta(p, \R) := \mf{true})
        }{\\
        &\quad\ \itransp{D}{D_0}{p}{\R}{\delta, \ell \mapsto x_\ell}
        }}
      }
    \end{aligned}\\
    \mf{ids}(\estate{x}{\ell}{V}{D})
    = \mf{ids}(V) \cup \mf{ids}(D)
  \end{gather*}
  \begin{mathpar}
    \inferrule*[right=TB-State] {
      \inferrule* {(1)} {\Gamma \semi \Pi \semi \Phi_1 \vdash V: A}
      \and
      \inferrule* {(2)} {
        \Gamma, x: A, x_\mf{set}: A \xrightarrow{\mf{Mut}} \tone \semi \Pi
        \semi \Phi_1' \vdash D: \Ell
      }
      \and
      \inferrule* {(3)} {\ell \notin \dom(\Ell)}
    } {
      \Gamma \semi \Pi \semi \Phi_1 \Join \Phi_1'
      \vdash \estate{x}{\ell}{V}{D}: \Ell, \ell: A
    }
  \end{mathpar}
  \begin{itemize}
    \item By premise \emph{(iii)} and definition of $\mf{ids}$,
    we have (4) that $\Pi \semi \mf{ids}(V) \vDash \R$,
    (5) that $\Pi \semi \mf{ids}(D) \vDash \R$,
    and (6) that $\Pi \semi \emptyset \vDash \R$.
    \item By mutual inductive hypothesis (Theorem \ref{thm:trans-prezv})
    on (1) and (4), we have (7) that\\
    $\bsq{\Gamma} \semi \bsq{\Pi}\R \vdash \bsq{V}\R: \bsq{A}$.
    \item Let $\delta' = \delta, \ell \mapsto x_\ell$
    and $\Gamma'' = \Gamma', x_\ell: \mf{Ref}\ \bsq{A}$. We want to show (8)
    that\\$\Gamma'' \semi \Pi \vDash_p \delta': (\Ell_0 \setminus \Ell)$:
    \begin{itemize}
      \item Premise \emph{(ii)} gives us that
      $\Ell, \ell: A \subseteq \Ell_0$.
      (3) says that $\ell \notin \dom(\Ell)$.
      \item Therefore, we must have $\ell: A \in \Ell_0 \setminus \Ell$,
      and clearly $\ell: A \notin \Ell_0 \setminus (\Ell, \ell: A)$.
      \item Then by premise \emph{(iv)} and definition,
      since $\ell \mapsto x_\ell \in \delta'$,
      it suffices to show that $x_\ell: \mf{Ref}\ \bsq{A} \in \Gamma''$.
      This is given by definition of $\Gamma''$.
    \end{itemize}
    \item We weaken premise \emph{(ii)} on $\Gamma$
    to get (9) that $\Gamma, x: A, x_\mf{set}: A \xrightarrow{\mf{Mut}} \tone
    \semi \Pi \vdash D_0: \Ell_0$.
    \item Now by inductive hypothesis on (2), (5), (8), and (9),
    we have (10) that\\
    $\bsq{\Gamma, x: A, x_\mf{set}: A \xrightarrow{\mf{Mut}} \tone}, \Gamma''
    \semi \bsq{\Pi}\R \vdash
    \itransp{D}{D_0}{p}{\R}{\delta'} : A_\mf{comp}$.
    \item By definition, (10) is equivalent to (11) that\\
    $\bsq{\Gamma}, \Gamma', x: \bsq{A}, x_\mf{set}: \bsq{A} \rightarrow \tone,
    x_\ell: \mf{Ref}\ \bsq{A} \semi \bsq{\Pi}\R \vdash
    \itransp{D}{D_0}{p}{\R}{\delta'} : A_\mf{comp}$.
    \item Finally, we apply Lemma \ref{thm:trans-tpflag} (flag typing)
    to (6) and premise \emph{(iv)} to get (12) that\\
    $\Gamma' \semi \bsq{\Pi}\R \vdash \delta(p, \R): \mf{Ref\ Bool}$.
    \item We derive the conclusion by way of subderivations (13)-(17)
    as follows, weakening our previous subconclusions on $\Gamma$
    as appropriate:
    \begin{mathpar}
      \inferrule*[right=TC-Set] {
        \inferrule*[right=TV-Var] { } {
          \Gamma_{13} \semi \bsq{\Pi}\R \vdash x_\ell: \mf{Ref}\ \bsq{A}
        }
        \and
        \inferrule*[right=TV-Var] { } {
          \Gamma_{13} \semi \bsq{\Pi}\R \vdash y: \bsq{A}
        }
      } {
        (13) \quad
        \Gamma_{13} \semi \bsq{\Pi}\R \vdash x_\ell := y: \tone
      }
      \and
      \inferrule*[right=TC-Set] {
        \inferrule*[right=(12)] { } {
          \Gamma_{14} \semi \bsq{\Pi}\R \vdash \delta(p, \R): \mf{Ref\ Bool}
        }
        \and
        \inferrule*[right=TV-True] { } {
          \Gamma_{14} \semi \bsq{\Pi}\R \vdash \mf{true}: \mf{Bool}
        }
      } {
        (14) \quad
        \Gamma_{14} \semi \bsq{\Pi}\R \vdash
        \delta(p, \R) := \mf{true}: \tone
      }
      \and
      \inferrule*[right=TC-Val] {
        \inferrule*[right=TV-Lam] {
          \inferrule*[right=TC-Let] {
            [(13),\ \Gamma_{13} = \Gamma_{15}, y: \bsq{A}]
            \and
            [(14),\ \Gamma_{14} = \Gamma_{15}, y: \bsq{A}]
          } {
            \Gamma_{15}, y: \bsq{A} \semi \bsq{\Pi}\R \vdash
            x_\ell := y \semi \delta(p, \R) := \mf{true}: \tone
          }
        } {
          \Gamma_{15} \semi \bsq{\Pi}\R \vdash
          \lambda y. (x_\ell := y \semi \delta(p, \R) := \mf{true}):
          \bsq{A} \rightarrow \tone
        }
      } {
        (15) \quad
        \Gamma_{15} \semi \bsq{\Pi}\R \vdash
        \mf{val}\ \lambda y.
        (x_\ell := y \semi \delta(p, \R) := \mf{true}):
        \bsq{A} \rightarrow \tone
      }
    \end{mathpar}
    \begin{mathpar}
      \inferrule*[right=TC-Let] {
        [(15),\ \Gamma_{15} = \Gamma_{16}] \quad
        \inferrule*[Right=(11)] { } {
          \Gamma_{16}, x_\mf{set}: \bsq{A} \rightarrow \tone \semi \bsq{\Pi}\R
          \vdash \itransp{D}{D_0}{p}{\R}{\delta'}
          : A_\mf{comp}
        }
      } {
        (16) \quad
        \Gamma_{16} \semi \bsq{\Pi}\R \vdash
        \elet{x_\mf{set}}{\cdots}{
          \itransp{D}{D_0}{p}{\R}{\delta'}
        }: A_\mf{comp}
      }
      \and
      \inferrule*[right=TC-Let] {
        \inferrule*[right=TC-Get] {
          \inferrule*[right=TV-Var] { } {
            \Gamma_{17} \semi \bsq{\Pi}\R \vdash x_\ell: \mf{Ref}\ \bsq{A}
          }
        } {
          \Gamma_{17} \semi \bsq{\Pi}\R \vdash
          \bang x_\ell:\bsq{A}
        }
        \and [(16),\ \Gamma_{16} = \Gamma_{17}, x: \bsq{A}]
      } {
        (17) \quad
        \Gamma_{17} \semi \bsq{\Pi}\R \vdash
        \elet{x}{\bang x_\ell}{\cdots}: A_\mf{comp}
      }
      \and
      \inferrule*[right=TC-Let] {
        \inferrule*[right=TC-Ref] {
          \inferrule*[right=(7)] { } {
            \bsq{\Gamma}, \Delta \semi \bsq{\Pi}\R \vdash \bsq{V}\R:
            \bsq{A}
          }
        } {
          \bsq{\Gamma}, \Delta \semi \bsq{\Pi}\R \vdash
          \mf{ref}(\bsq{V}\R): \mf{Ref}\ \bsq{A}
        }
        \ [(17),\ \Gamma_{17} = \bsq{\Gamma}, \Delta, x_\ell: \mf{Ref}\ \bsq{A}]
      } {
        \bsq{\Gamma}, \Delta \semi \bsq{\Pi}\R \vdash
        \elet{x_\ell}{\mf{ref}(\bsq{V}\R)}{\cdots}: A_\mf{comp}
      }
    \end{mathpar}
  \end{itemize}
\end{proof}

\paragraph{Heap/Store Preservation}

We need one final lemma regarding the appending of well-formed stores
in order to prove type preservation for heaps:

\begin{lemma}[Store appending]
\label{thm:trans-append}
If
\begin{enumerate}[label=(\roman*)]
\item $\Gamma \semi \Sigma_0 \vDash \sigma: \Sigma$,
\item $\Gamma \semi \Sigma_0 \vDash \sigma': \Sigma'$, and
\item $\dom(\sigma) \cap \dom(\sigma') = \emptyset$,
\end{enumerate}
then $\Gamma \semi \Sigma_0 \semi \Phi \vDash \sigma \uplus \sigma':
\Sigma \uplus \Sigma'$.

\begin{proof}
  By induction on the derivation of typing premise \emph{(ii)}.
  At each step, well-formedness is preserved by premise \emph{(iii)}.
\end{proof}
\end{lemma}

\begin{restatable}[Type preservation, state stores]{lemma}{transprezs}
\label{thm:trans-prezs}
If
\begin{enumerate}[label=(\roman*)]
\item $\Gamma \semi \Pi \semi \Phi \vDash \sigma: \Ell$,
\item $\Pi \semi \mf{ids}(\sigma) \vDash \R$, and
\item $\Ell \subseteq \Pi[p]$ (for some $p$),
\end{enumerate}
then $\bsq{\Gamma} \semi \bsq{\Pi}\R \vDash \bsq{\sigma}\R: \bsq{\Ell}$.

\begin{proof}
  By induction on the derivation of the typing premise \emph{(i)}.

  \pcase{TS-Nil}: We have
  \begin{mathpar}
    \inferrule*[right=TS-Nil] {\mf{Succ}(\Phi)} {
      \Gamma \semi \Pi \semi \Phi \vDash \emptyset: \emptyset
    }
  \end{mathpar}
  \begin{itemize}
    \item By definition, $\bsq{\Ell}\R = \emptyset$
    and $\bsq{\sigma}\R=\emptyset$.
    \item We immediately rederive the conclusion via
    the corresponding target-language rule:
    \begin{mathpar}
      \inferrule*[right=TS-Nil] { } {
        \bsq{\Gamma} \semi \bsq{\Pi}\R \vDash \emptyset: \emptyset
      }
    \end{mathpar}
  \end{itemize}

  \pcase{TS-Cons}: We have
  \[\bsq{\sigma, \ell \mapsto V}\R = \bsq{\sigma}\R,
  \R(p, \ell) \mapsto \bsq{V}\R
  \quad
  \bsq{\Ell, \ell: A}\R = \bsq{\Ell}\R, \R(p, \ell): \bsq{A}\]
  \begin{mathpar}
    \inferrule*[right=TS-Cons] {
      \inferrule* {(1)} {\Gamma \semi \Pi \semi \Phi_1 \vdash V: A}
      \and
      \inferrule* {(2)} {\Gamma \semi \Pi \semi \Phi_2 \vDash \sigma: \Ell}
    } {
      \Gamma \semi \Pi \semi \Phi_1 \Join \Phi_2
      \vDash \sigma, \ell \mapsto V: \Ell, \ell: A
    }
  \end{mathpar}
  \begin{itemize}
    \item By premise \emph{(ii)} and definition of $\mf{ids}$,
    we have (3) that $\Pi \semi \mf{ids}(V) \vDash \R$
    and (4) that\\$\Pi \semi \mf{ids}(\sigma) \vDash \R$.
    \item By Lemma \ref{thm:trans-prezv} (value type preservation)
    on (1) and (3), we get (5) that\\$\bsq{\Gamma} \semi \bsq{\Pi}\R
    \vdash \bsq{V}\R: \bsq{A}$.
    \item By inductive hypothesis on (2) and (4), we get (6) that
    $\bsq{\Gamma} \semi \bsq{\Pi}\R \vdash \bsq{\sigma}\R: \bsq{\Ell}\R$.
    \item By premises \emph{(ii)} and \emph{(iii)}, we get that
    $\R(p, \ell)$ is well-defined. By injectivity of $\R$
    and definition of translation, we further get that
    $\R(p, \ell) \notin \dom \bsq{\sigma}\R$
    and $\R(p, \ell) \notin \dom \bsq{\Ell}\R$,
    i.e. $\bsq{\sigma}, \R(p, \ell) \mapsto \bsq{V}\R$
    and $\bsq{\Ell}, \R(p, \ell): \bsq{A}$ are well-formed.
    \item We derive the conclusion as follows:
    \begin{mathpar}
      \inferrule*[right=TS-Cons] {
        \inferrule*[right=(5)] { } {
          \bsq{\Gamma} \semi \bsq{\Pi}\R \vdash \bsq{V}\R: \bsq{A}
        }
        \and
        \inferrule*[right=(6)] { } {
          \bsq{\Gamma} \semi \bsq{\Pi}\R \vdash \bsq{\sigma}\R: \bsq{\Ell}
        }
      } {
        \bsq{\Gamma} \semi \bsq{\Pi}\R \vDash
        \bsq{\sigma}, \R(p, \ell) \mapsto \bsq{V}\R:
        \bsq{\Ell}, (\R(p, \ell): \bsq{A})
      }
    \end{mathpar}
  \end{itemize}
\end{proof}
\end{restatable}

\begin{restatable}[Type preservation, heaps, intermediate]{lemma}{transprezhi}
\label{thm:trans-prezhi}
If
\begin{enumerate}[label=(\roman*)]
\item $\Gamma \semi \Pi_0 \semi \Phi \vDash \heap: \Pi$,
\item $\Pi_0 \semi \mf{ids}(\heap) \vDash \R$, and
\item $\Pi \subseteq \Pi_0$,
\end{enumerate}
then $\bsq{\Gamma} \semi \bsq{\Pi_0}\R \vDash \bsq{\heap}\R: \bsq{\Pi}\R$.
\end{restatable}
\begin{proof}
  By induction on the derivation of the typing premise \emph{(i)}.

  \pcase{TH-Nil}: We have
  \begin{mathpar}
    \inferrule*[right=TH-Nil] {\mf{Succ}(\Phi)} {
      \Gamma \semi \Pi_0 \semi \Phi \vDash \emptyset: \emptyset
    }
  \end{mathpar}
  \begin{itemize}
    \item By definition, $\bsq{\Pi}\R = \emptyset$
    and $\bsq{\heap}\R=\emptyset$.
    \item We immediately rederive the conclusion via
    the corresponding target-language rule:
    \begin{mathpar}
      \inferrule*[right=TH-Nil] { } {
        \bsq{\Gamma} \semi \bsq{\Pi}\R \vDash \emptyset: \emptyset
      }
    \end{mathpar}
  \end{itemize}

  \pcase{TH-Cons}: We have
  \begin{gather*}
    \bsq{\heap, p \mapsto \pi}\R = \bsq{\heap}\R \cup \bsq{\pi.\mf{st}}\R,
    \R(p) \mapsto \bsq{\pi.\mf{rerender}}\\
    \bsq{\Pi, p: \Ell}\R = \bsq{\Pi}\R \cup \bsq{\Ell}\R,
    \R(p): \mf{Bool}
  \end{gather*}
  \begin{mathpar}
    \inferrule*[right=TH-Cons] {
      \inferrule* {(1)} {
        \Gamma \semi \Pi_0 \semi \Phi_1 \vDash \pi.\mf{st}: \Ell
      }
      \and
      \inferrule* {(2)} {
        \Gamma \semi \Pi_0 \semi \Phi_2 \vDash \heap: \Pi
      }
    } {
      \Gamma \semi \Pi_0 \semi \Phi_1 \Join \Phi_2
      \vDash \heap, p \mapsto \pi: \Pi, p: \Ell
    }
  \end{mathpar}
  \begin{itemize}
    \item By premise \emph{(ii)} and definition of $\mf{ids}$,
    we have (3) that $\Pi_0 \semi \mf{ids}(\pi.\mf{st}) \vDash \R$
    and (4) that $\Pi_0 \semi \mf{ids}(\heap) \vDash \R$.
    \item By Lemma \ref{thm:trans-prezs} (state store type preservation)
    on (1) and (3), noting that $\Ell=\Pi[p]=\Pi_0[p]$ by premise \emph{(iii)},
    we have (5) $\bsq{\Gamma} \semi \bsq{\Pi_0}\R \vDash
    \bsq{\pi.\mf{st}}\R: \bsq{\Ell}\R$.
    \item By inductive hypothesis on (2) and (4),
    noting that transitivity of $\sqsubseteq$ satisfies the third premise,
    we get (6) that $\bsq{\Gamma} \semi \bsq{\Pi_0}\R \vDash
    \bsq{\heap}\R: \bsq{\Pi}\R$.
    \item By injectivity of $\R$ and translation definitions,
    we get (7)
    $\dom(\bsq{\heap}\R) \cap \dom(\bsq{\pi.\mf{st}}\R) = \emptyset$.
    \item By Lemma \ref{thm:trans-append} (store appending) on (5)-(7),
    we get (8) that\\$\bsq{\Gamma} \semi \bsq{\Pi_0}\R \vDash
    \bsq{\heap}\R \uplus \bsq{\pi.\mf{st}}\R:
    \bsq{\Pi}\R \uplus \bsq{\Ell}\R$.
    \item By definition, either $\bsq{\pi.\mf{rerender}} = \mf{true}$
    or $\bsq{\pi.\mf{rerender}} = \mf{false}$. In either case, (9)
    $\bsq{\Gamma} \semi \bsq{\Pi_0}\R \vdash \bsq{\mf{rr}}: \mf{Bool}$.
    \item By premise \emph{(ii)}, we have $\R(p)$ is well-defined.
    \item We derive the conclusion as follows, again getting well-formedness by
    injectivity of $\R$, and abbreviating
    $\mf{rr} \triangleq \pi.\mf{rerender}$:
  \end{itemize}
  \begin{mathpar}
    \inferrule*[right=TS-Cons] {
      \inferrule*[right=(9)] { } {
        \bsq{\Gamma} \semi \bsq{\Pi_0}\R \vdash \bsq{\mf{rr}}:
        \mf{Bool}
      }
      \
      \inferrule*[right=(8)] { } {
        \bsq{\Gamma} \semi \bsq{\Pi_0}\R \vDash
        \bsq{\heap}\R \uplus \bsq{\pi.\mf{st}}\R:
        \bsq{\Pi}\R \uplus \bsq{\Ell}\R
      }
    } {
      \bsq{\Gamma} \semi \bsq{\Pi_0}\R \vDash
      \bsq{\heap}\R \uplus \bsq{\pi.\mf{st}}\R,
      \R(p) \mapsto \bsq{\mf{rr}}:
      \bsq{\Pi}\R \uplus \bsq{\Ell}\R, (\R(p): \mf{Bool})
    }
  \end{mathpar}
\end{proof}

\begin{restatable}[Type preservation, heaps]{theorem}{transprezh}
\label{thm:trans-prezh}
If
\begin{enumerate}[label=(\roman*)]
\item $\Gamma \semi \Pi \semi \Phi \vDash \heap$ and
\item $\Pi \semi \mf{ids}(\heap) \vDash \R$,
\end{enumerate}
then $\bsq{\Gamma} \semi \bsq{\Pi}\R \vDash \bsq{\heap}\R$.
\end{restatable}
\begin{proof}
  Corollary of Lemma \ref{thm:trans-prezhi}, choosing $\Pi_0=\Pi$
  and noting that $\Pi \sqsubseteq_\emptyset \Pi$.
\end{proof}

\subsection{Simulation}
\label{appdx:transsim}

We begin with an assortment of ``equivalence lemmas'',
which help us assert the equivalence of terms under
various translation-related operations.

\subsubsection{Substitution}

We prove that substitution commutes with all five term translation functions.
The lemma statements and proofs for values and computations are fairly standard.
Those for bodies additionally require that the target variable
is not in the range of $\delta$: we treat those variables
separately via a series of \emph{$\delta$-elimination} lemmas,
to be presented later.

By the Barendregt variable convention, we are free to assume
in our lemmas that the target variable is not in the bound variables
of the term: we guarantee this by appropriately $\alpha$-converting any terms
which violate this assumption.

We additionally remark that, in the translation
$\itransp{D}{D_0}{p}{\R}{\delta}$,
the body $D_0$ is
\emph{in scope of all bindings of} $D$.
We thus stipulate by convention that
$\mf{bv}(D) \cap \mf{fv}(D_0) = \emptyset$;
this can likewise be achieved by $\alpha$-conversion on $D$.
(This stipulation excludes those variables introduced by $\delta$
in the translation: these $\delta$-variables are the only way to
``pass information'' between the translation phases.)

\begin{lemma}[Substitution, values]
\label{thm:trans-substv}
If \begin{enumerate}[label=(\roman*)]
  \item $\bsq{V}\R$ is well-defined and
  \item $\bsq{W}\R$ is well-defined,
\end{enumerate}
then $\bsq{W[V/x]}\R = \bsq{W}\R[\bsq{V}\R/x]$.

\begin{proof}
  By mutual induction on the structure of $W$
  with Lemmas \ref{thm:trans-substc}, \ref{thm:trans-substbs},
  and \ref{thm:trans-substbi}.
  (Premise \emph{(i)} allows us to safely expand $\bsq{W}\R$.)

  \pcase{$W=x'$}: We subcase on whether $x' = x$.
  \begin{itemize}
    \item If $x' = x$, then we show the conclusion as follows:
    \begin{align*}
      \bsq{x}\R[\bsq{V}\R/x]
      &= x[\bsq{V}\R/x] &&\text{def. of translation}\\
      &= \bsq{V}\R &&\text{def. of substitution}\\
      &= \bsq{x[V/x]}\R &&\text{def. of substitution}
    \end{align*}
  \end{itemize}
  \begin{itemize}
    \item If $x' \neq x$, then we show the conclusion as follows:
    \begin{align*}
      \bsq{x'}\R[\bsq{V}\R/x]
      &= x'[\bsq{V}\R/x] &&\text{def. of translation}\\
      &= x' &&\text{def. of substitution}\\
      &= \bsq{x'}\R &&\text{def. of translation}\\
      &= \bsq{x'[V/x]}\R &&\text{def. of substitution}
    \end{align*}
  \end{itemize}

  \textbf{Cases} $W=()$, $W=\mf{true}$, $W=\mf{false}$, $W=s$, $W=\nil$,
  and $W=\setter^\ell_p$ are all shown analogously to the latter subcase above.

  \pcase{$W=\mf{tagInst}_{id}(W_1, W_2)$}:
  \begin{itemize}
    \item By premise \emph{(i)},
    we have (1) that $\bsq{W_1}\R$ is well-defined and
    (2) that $\bsq{W_2}\R$ is well-defined.
    \item By inductive hypothesis on each of (1) and (2)
    alongside premise \emph{(ii)},
    we get (3) that $\bsq{W_1[V/x]}\R = \bsq{W_1}\R[\bsq{V}\R/x]$
    and (4) that $\bsq{W_2[V/x]}\R = \bsq{W_2}\R[\bsq{V}\R/x]$,
    respectively.
    \item We show the conclusion as follows:
    \begin{align*}
      &\bsq{\mf{tagInst}_{id}(W_1, W_2)}\R[\bsq{V}\R/x]\\
      &= \mf{tagInst}_{id}(\bsq{W_1}\R, \bsq{W_2}\R)[\bsq{V}\R/x]
      &&\text{def. of translation}\\
      &= \mf{tagInst}_{\R({id})}(\bsq{W_1}\R[\bsq{V}\R/x],
        \bsq{W_2}\R[\bsq{V}\R/x])
      &&\text{def. of substitution}\\
      &= \mf{tagInst}_{\R({id})}(\bsq{W_1[V/x]}\R, \bsq{W_2[V/x]}\R)
      &&\text{by (3) and (4)}\\
      &= \bsq{\mf{tagInst}_{id}(W_1[V/x], W_2[V/x])}\R
      &&\text{def. of translation}\\
      &= \bsq{\mf{tagInst}_{id}(W_1, W_2)[V/x]}\R
      &&\text{def. of substitution}
    \end{align*}
  \end{itemize}

  \textbf{Cases} $W=(W_1, W_2)$, $W=W_1 :: W_2$, $W=\mf{text}\ W_1$,
  $W=\mf{tag}(W_1, W_2)$, $W=\mf{attr}(W_1, W_2)$, $W=\mf{onClick}\ W_1$,
  and $W=\mf{textInst}\ W_1$ are all shown analogously to the above case.

  \pcase{$W=\lambda y.M$}:
  \begin{itemize}
    \item By premise \emph{(i)}, we have (1) that $\bsq{M}\R$ is well-defined.
    \item By mutual inductive hypothesis (Lemma \ref{thm:trans-substc})
    on (1) and premise \emph{(ii)}, we get (2) that
    $\bsq{M[V/x]}\R = \bsq{M}\R[\bsq{V}\R/x]$.
    \item By convention, $x \neq x'$, so we show the conclusion as follows:
    \begin{align*}
      \bsq{\lambda y.M}\R[\bsq{V}\R/x]
      &= (\lambda y.\bsq{M}\R)[\bsq{V}\R/x]
      &&\text{def. of translation}\\
      &= \lambda y.(\bsq{M}\R[\bsq{V}\R/x])
      &&\text{def. of substitution}\\
      &= \lambda y.(\bsq{M[V/x]}\R) &&\text{by (4)}\\
      &= \bsq{\lambda y.M[V/x]}\R &&\text{def. of translation}\\
      &= \bsq{(\lambda y.M)[V/x]}\R &&\text{def. of substitution}
    \end{align*}
  \end{itemize}

  \pcase{$W = \mf{component}\ \{D\}$}:
  \begin{itemize}
    \item By premise \emph{(i)}, we have (1) that
    $\itransp{D}{D}{\bot}{\R}{\emptyset}$ is well-defined.
    \item By mutual inductive hypotheisis (Lemma \ref{thm:trans-substbi})
    on (1) and premise \emph{(ii)}, we get (2) that
    $\itransp{D[V/x]}{D[V/x]}{\bot}{\R}{\emptyset} =
    \itransp{D[V/x]}{D[V/x]}{\bot}{\R}{\emptyset}[\bsq{V}\R/x]$.
    \item We show the conclusion as follows:
    \begin{align*}
      &\bsq{\mf{component}\ \{D\}}\R[\bsq{V}\R/x]\\
      &= (\mf{component}\ \lambda().\elet{y}{\mf{ref}(\mf{false})}{\\
        &\quad \itransp{D[V/x]}{D[V/x]}{\bot}{\R}{\emptyset}})[\bsq{V}\R/x]
      &&\text{def. of translation}\\
      &= \mf{component}\ \lambda().\elet{y}{\mf{ref}(\mf{false})}{\\
        &\quad \itransp{D[V/x]}{D[V/x]}{\bot}{\R}{\emptyset}}[\bsq{V}\R/x]
      &&\text{def. of substitution}\\
      &= \mf{component}\ \lambda().\elet{y}{\mf{ref}(\mf{false})}{\\
        &\quad \itransp{D[V/x]}{D[V/x]}{\bot}{\R}{\emptyset}}
      &&\text{by (2)}\\
      &= \bsq{\mf{component}\ \{D[V/x]\}}\R
      &&\text{def. of translation}\\
      &= \bsq{(\mf{component}\ \{D\})[V/x]}\R
      &&\text{def. of substitution}
    \end{align*}
  \end{itemize}

  \pcase{$W=\compInst{p}{D}{W_1}$}:
  \begin{itemize}
    \item By premise \emph{(i)}, we have (1) that
    $\stransp{D}{\bot}{\R}{\emptyset}$ is well-defined and (2) that
    $\bsq{W_1}\R$ is well-defined.
    \item By mutual inductive hypotheisis (Lemma \ref{thm:trans-substbs})
    on (1) and premise \emph{(ii)}, we get (3) that
    $\stransp{D[V/x]}{\bot}{\R}{\emptyset} =
    \stransp{D}{\bot}{\R}{\emptyset}[\bsq{V}\R/x]$.
    \item By inductive hypotheisis on (2) and premise \emph{(ii)},
    we get (4) that\\$\bsq{W_1[V/x]}\R = \bsq{W_1}\R[\bsq{V}\R/x]$.
    \item We show the conclusion as follows:
    \begin{align*}
      &\bsq{\compInst{p}{D}{W_1}}\R[\bsq{V}\R/x]\\
      &= \mf{compInst}(\lambda().\stransp{D}{\bot}{\R}{\emptyset},
        \R(p), \bsq{W_1}\R)[\bsq{V}\R/x]
      &&\text{def. of translation}\\
      &= \mf{compInst}(\lambda().\stransp{D}{\bot}{\R}{\emptyset}
        [\bsq{V}\R/x], \R(p), \bsq{W_1}\R[\bsq{V}\R/x])
      &&\text{def. of substitution}\\
      &= \mf{compInst}(\lambda().\stransp{D[V/x]}{\bot}{\R}{\emptyset},
        \R(p), \bsq{W_1}\R[\bsq{V}\R/x])
      &&\text{by (3)}\\
      &= \mf{compInst}(\lambda().\stransp{D[V/x]}{\bot}{\R}{\emptyset},
        \R(p), \bsq{W_1[V/x]}\R)
      &&\text{by (4)}\\
      &= \bsq{\compInst{p}{D[V/x]}{W_1[V/x]}}\R
      &&\text{def. of translation}\\
      &= \bsq{(\compInst{p}{D}{W_1})[V/x]}\R
      &&\text{def. of substitution}
    \end{align*}
  \end{itemize}
\end{proof}
\end{lemma}

\begin{lemma}[Substitution, computations]
\label{thm:trans-substc}
If \begin{enumerate}[label=(\roman*)]
  \item $\bsq{M}\R$ is well-defined and
  \item $\bsq{V}\R$ is well-defined,
\end{enumerate}
then $\bsq{M[V/x]}\R = \bsq{M}\R[\bsq{V}\R/x]$.

\begin{proof}
  By mutual induction on the structure of $M$
  with Lemmas \ref{thm:trans-substv}, \ref{thm:trans-substbs},
  and \ref{thm:trans-substbi}.

  \pcase{$M = \mf{val}\ W$}:
  \begin{itemize}
    \item By premise \emph{(i)}, (1) $\bsq{W}\R$ is well-defined.
    \item By mutual inductive hypothesis (Lemma \ref{thm:trans-substv})
    on (1) and premise \emph{(ii)}, we get (2)
    that $\bsq{W[V/x]}\R = \bsq{W}\R[\bsq{V}\R/x]$.
    \item We show the conclusion as follows:
    \begin{align*}
      &\bsq{\mf{val}\ W}\R[\bsq{V}\R/x]\\
      &= (\mf{val}\ \bsq{W}\R)[\bsq{V}\R/x]
      &&\text{def. of translation}\\
      &= \mf{val}\ \bsq{W}\R[\bsq{V}\R/x]
      &&\text{def. of substitution}\\
      &= (\mf{val}\ \bsq{W[V/x]}\R)
      &&\text{by (2)}\\
      &= \bsq{\mf{val}\ W[V/x]}\R
      &&\text{def. of translation}\\
      &= \bsq{(\mf{val}\ W)[V/x]}\R
      &&\text{def. of substitution}
    \end{align*}
  \end{itemize}

  \textbf{Cases}
  $M=\mf{fst}\ W$,
  $M=\mf{snd}\ W$,
  $M=\mf{init}\ W$,
  $M=\mf{click}_{id}\ W$,
  $M=\mf{runClick}\ W$, and
  $M=\mf{update}\ W$
  can all be shown analogously to the above case.

  \pcase{$M=\elet{y}{M_1}{M_2}$}:
  \begin{itemize}
    \item By premise \emph{(i)}, we have (1) that $\bsq{M_1}\R$ is well-defined
    and (2) that $\bsq{M_2}\R$ is well-defined.
    \item By inductive hypothesis on each of (1) and (2)
    alongside premise \emph{(ii)}, we get (3) that
    $\bsq{M_1[V/x]}\R = \bsq{M_1}\R[\bsq{V}\R/x]$
    and (4) that $\bsq{M_2[V/x]}\R = \bsq{M_2}\R[\bsq{V}\R/x]$,
    respectively.
    \item By convention, $x \neq y$, so we show the conclusion as follows:
    \begin{align*}
      &\bsq{\elet{y}{M_1}{M_2}}\R[\bsq{V}\R/x]\\
      &= (\elet{y}{\bsq{M_1}\R}{\bsq{M_2}\R})[\bsq{V}\R/x]
      &&\text{def. of translation}\\
      &= \elet{y}{\bsq{M_1}\R[\bsq{V}\R/x]}{\bsq{M_2}\R[\bsq{V}\R/x]}
      &&\text{def. of substitution}\\
      &= \elet{y}{\bsq{M_1[V/x]}\R}{\bsq{M_2[V/x]}\R}
      &&\text{by (3) and (4)}\\
      &= \bsq{\elet{y}{M_1[V/x]}{M_2[V/x]}}\R &&\text{def. of translation}\\
      &= \bsq{(\elet{y}{M_1}{M_2})[V/x]}\R
      &&\text{def. of substitution}
    \end{align*}
  \end{itemize}

  \pcase{$M=V_1\ V_2$}:
  \begin{itemize}
    \item By premise \emph{(i)}, (1) $\bsq{V_1}\R$ is well-defined
    and (2) $\bsq{V_2}\R$ is well-defined.
    \item By mutual inductive hypothesis (Lemma \ref{thm:trans-substv})
    on each of (1) and (2) alongside premise \emph{(ii)}, we get (3) that
    $\bsq{V_1[V/x]}\R = \bsq{V_1}\R[\bsq{V}\R/x]$
    and (4) that $\bsq{V_2[V/x]}\R = \bsq{V_2}\R[\bsq{V}\R/x]$,
    respectively.
    \item We show the conclusion as follows:
    \begin{align*}
      &\bsq{V_1\ V_2}\R[\bsq{V}\R/x]\\
      &= (\bsq{V_1}\R\ \bsq{V_2}\R)[\bsq{V}\R/x]
      &&\text{def. of translation}\\
      &= \bsq{V_1}\R[\bsq{V}\R/x]\ \bsq{V_2}\R[\bsq{V}\R/x]
      &&\text{def. of substitution}\\
      &= \bsq{V_1[V/x]}\R\ \bsq{V_2[V/x]}\R
      &&\text{by (3) and (4)}\\
      &= \bsq{V_1[V/x]\ V_2[V/x]}\R
      &&\text{def. of translation}\\
      &= \bsq{(V_1\ V_2)[V/x]}\R
      &&\text{def. of substitution}
    \end{align*}
  \end{itemize}

  \pcase{$M=\mf{fold}\ V_1\ V_2\ V_3$} can be shown analogously
  to the above case.

  \pcase{$M = \eif{W}{M_1}{M_2}$}:
  \begin{itemize}
    \item By premise \emph{(i)}, we have
    (1) that $\bsq{W}\R$ is well-defined,
    (2) that $\bsq{M_1}\R$ is well-defined, and
    (3) that $\bsq{M_2}\R$ is well-defined.
    \item By mutual inductive hypothesis (Lemma \ref{thm:trans-substv})
    on (1) and premise \emph{(ii)}, we get (4) that
    $\bsq{W[V/x]}\R = \bsq{W}\R[\bsq{V}\R/x]$.
    \item By inductive hypothesis on each of (2) and (3) alongside
    premise \emph{(ii)},
    we get (5) that $\bsq{M_1[V/x]}\R = \bsq{M_1}\R[\bsq{V}\R/x]$
    and (6) that $\bsq{M_2[V/x]}\R = \bsq{M_2}\R[\bsq{V}\R/x]$,
    respectively.
    \item We show the conclusion as follows:
    \begin{align*}
      &\bsq{\eif{W}{M_1}{M_2}}\R[\bsq{V}\R/x]\\
      &= (\eif{\bsq{W}\R}{\bsq{M_1}\R}{\bsq{M_2}\R})[\bsq{V}\R/x]
      &&\text{def. of translation}\\
      &= \eif{\bsq{W}\R[\bsq{V}\R/x]}{\bsq{M_1}\R[\bsq{V}\R/x]}
        {\bsq{M_2}\R[\bsq{V}\R/x]}
      &&\text{def. of substitution}\\
      &= \eif{\bsq{W[V/x]}\R}{\bsq{M_1}\R[\bsq{V}\R/x]}
        {\bsq{M_2}\R[\bsq{V}\R/x]}
      &&\text{by (4)}\\
      &= \eif{\bsq{W[V/x]}\R}{\bsq{M_1[V/x]}\R}{\bsq{M_2[V/x]}\R}
      &&\text{by (5) and (6)}\\
      &= \bsq{\eif{W[V/x]}{M_1[V/x]}{M_2[V/x]}}\R
      &&\text{def. of translation}\\
      &= \bsq{(\eif{W}{M_1}{M_2})[V/x]}\R
      &&\text{def. of substitution}
    \end{align*}
  \end{itemize}

  \pcase{TC-EvalI}:
  \begin{itemize}
    \item By premise \emph{(i)}, we have (1) that
    $\itransp{D}{D_0}{p}{\R}{\emptyset}$ is well-defined.
    \item By mutual inductive hypothesis (Lemma \ref{thm:trans-substbi})
    on (1) and premise \emph{(ii)}, we get (2) that
    $\itransp{D[V/x]}{D_0[V/x]}{p}{\R}{\emptyset} =
    \itransp{D}{D_0}{\bot}{\R}{\emptyset}[\bsq{V}\R/x]$.
    \item We show the conclusion as follows, noting that
    all variables introduced in the translation are fresh (and so $\neq x$):
    \begin{align*}
      &\bsq{\eval{D_0}{\mf{Init}}{p}{D}}\R[\bsq{V}\R/x]\\
      &= (\elet{rr, f, x'}
        {\itransp{D}{D_0}{p}{\R}{\emptyset}}
        {\mf{compInst}(rr, f, \mf{init}\ x')})[\bsq{V}\R/x]
      &&\text{def. of translation}\\
      &= (\elet{rr, f, x'}
        {\itransp{D}{D_0}{p}{\R}{\emptyset}[\bsq{V}\R/x]}
        {\mf{compInst}(rr, f, \mf{init}\ x')})
      &&\text{def. of substitution}\\
      &= (\elet{rr, f, x'}
        {\itransp{D[V/x]}{D_0[V/x]}{p}{\R}{\emptyset}}
        {\mf{compInst}(rr, f, \mf{init}\ x')})
      &&\text{by (2)}\\
      &= \bsq{\eval{D_0[V/x]}{\mf{Init}}{p}{D[V/x]}}\R
      &&\text{def. of translation}\\
      &= \bsq{(\eval{D_0}{\mf{Init}}{p}{D})[V/x]}\R
      &&\text{def. of substitution}
    \end{align*}
  \end{itemize}

  \pcase{TC-EvalS}:
  \begin{itemize}
    \item By premise \emph{(i)}, we have (1) that
    $\stransp{D_0}{p}{\R}{\emptyset}$ is well-defined
    and (2) that $\stransp{D}{p}{\R}{\emptyset}$ is well-defined.
    \item By mutual inductive hypothesis on each of (1) and (2)
    alongside premise \emph{(ii)}, we get (3) that
    $\stransp{D_0[V/x]}{p}{\R}{\emptyset}
    = \stransp{D_0}{p}{\R}{\emptyset}[\bsq{V}\R/x]$
    and (4) that\\$\stransp{D[V/x]}{p}{\R}{\emptyset}
    = \stransp{D}{p}{\R}{\emptyset}[\bsq{V}\R/x]$, respectively.
    \item We show the conclusion as follows:
    \begin{align*}
      &\bsq{\eval{D_0}{\mf{Succ}}{p}{D}}\R[\bsq{V}\R/x]\\
      &= \mf{compInst}(\lambda().\stransp{D_0}{p}{\R}{\emptyset},
        \R(p), \mf{init}\ \stransp{D}{p}{\R}{\emptyset})
        [\bsq{V}\R/x]
      &&\text{def. of translation}\\
      &= \mf{compInst}(\lambda().\stransp{D_0}{p}{\R}{\emptyset}
        [\bsq{V}\R/x],\\
      &\quad \R(p), \mf{init}\ \stransp{D}{p}{\R}{\emptyset}
        [\bsq{V}\R/x])
      &&\text{def. of substitution}\\
      &= \mf{compInst}(\lambda().\stransp{D_0[V/x]}{p}{\R}{\emptyset},
        \R(p), \mf{init}\ \stransp{D[V/x]}{p}{\R}{\emptyset})
      &&\text{by (3) and (4)}\\
      &= \bsq{\eval{D_0[V/x]}{\mf{Succ}}{p}{D[V/x]}}\R
      &&\text{def. of translation}\\
      &= \bsq{(\eval{D_0}{\mf{Succ}}{p}{D})[V/x]}\R
      &&\text{def. of substitution}
    \end{align*}
  \end{itemize}

\end{proof}
\end{lemma}

\begin{lemma}[Substitution, bodies, successor]
\label{thm:trans-substbs}
If \begin{enumerate}[label=(\roman*)]
  \item $\stransp{D}{p}{\R}{\delta}$ is well-defined,
  \item $\bsq{V}\R$ is well-defined, and
  \item $x \notin \mf{img}(\delta)$,
\end{enumerate}
then $\stransp{D[V/x]}{p}{\R}{\delta}
= \stransp{D}{p}{\R}{\delta}[\bsq{V}\R/x]$.

\begin{proof}
  By mutual induction on the structure of $D$
  with Lemmas \ref{thm:trans-substv}, \ref{thm:trans-substc},
  and \ref{thm:trans-substbi}.

  \pcase{$D=\mf{val}\ W$}:
  \begin{itemize}
    \item By premise \emph{(i)}, (1) $\bsq{W}\R$ is well-defined.
    \item By mutual inductive hypothesis (Lemma \ref{thm:trans-substv})
    on (1) and premise \emph{(ii)}, we get (2) that
    $\bsq{W[V/x]}\R = \bsq{W}\R[\bsq{V}\R/x]$.
    \item We show the conclusion as follows:
    \begin{align*}
      &\stransp{\mf{val}\ W}{p}{\R}{\delta}[\bsq{V}\R/x]\\
      &= (\mf{val}\ \bsq{W}\R)[\bsq{V}\R/x] &&\text{def. of translation}\\
      &= \mf{val}\ \bsq{W}\R[\bsq{V}\R/x] &&\text{def. of substitution}\\
      &= \mf{val}\ \bsq{W[V/x]}\R &&\text{by (2)}\\
      &= \stransp{\mf{val}\ W[V/x]}{p}{\R}{\delta}
      &&\text{def. of translation}\\
      &= \stransp{(\mf{val}\ W)[V/x]}{p}{\R}{\delta}
      &&\text{def. of substitution}
    \end{align*}
  \end{itemize}

  \pcase{$D = \elet{y}{M}{D'}$}:
  \begin{itemize}
    \item By premise \emph{(i)}, we have (1) that $\bsq{M}\R$ is well-defined
    and (2) that $\stransp{D'}{p}{\R}{\delta}$ is well-defined.
    \item By mutual inductive hypothesis (Lemma \ref{thm:trans-substc})
    on (1) and premise \emph{(ii)}, we get
    (3) that $\bsq{M[V/x]}\R = \bsq{M}\R[\bsq{V}\R/x]$.
    \item By inductive hypothesis on (2)
    and premises \emph{(ii)} and \emph{(iii)}, we get (4) that\\
    $\stransp{D'[V/x]}{p}{\R}{\delta}
    = \stransp{D'}{p}{\R}{\delta}[\bsq{V}\R/x]$.
    \item By convention, $y \neq x$, so we show the conclusion as follows:
    \begin{align*}
      &\stransp{\elet{y}{M}{D'}}{p}{\R}{\delta}[\bsq{V}\R/x]\\
      &= (\elet{y}{\bsq{M}\R}{\stransp{D'}{p}{\R}{\delta}})
        [\bsq{V}\R/x]
      &&\text{def. of translation}\\
      &= \elet{y}{\bsq{M}\R[\bsq{V}\R/x]}
        {\stransp{D'}{p}{\R}{\delta}[\bsq{V}\R/x]}
      &&\text{def. of substitution}\\
      &= \elet{y}{\bsq{M[V/x]}\R}
        {\stransp{D'}{p}{\R}{\delta}[\bsq{V}\R/x]}
      &&\text{by (3)}\\
      &= \elet{y}{\bsq{M[V/x]}\R}
        {\stransp{D'[V/x]}{p}{\R}{\delta}}
      &&\text{by (4)}\\
      &= \stransp{\elet{y}{M[V/x]}{D'[V/x]}}{p}{\R}{\delta}
      &&\text{def. of translation}\\
      &= \stransp{(\elet{y}{M}{D'})[V/x]}{p}{\R}{\delta}
      &&\text{def. of substitution}
    \end{align*}
  \end{itemize}

  \pcase{$D = \estate{y}{\ell}{W}{D'}$}:
  \begin{itemize}
    \item By premise \emph{(i)}, we have (1) that
    $\stransp{D'}{p}{\R}{\delta}$ is well-defined.
    \item By inductive hypothesis on (1) and premises \emph{(ii)}
    and \emph{(iii)}, we get (2) that\\
    $\stransp{D'[V/x]}{p}{\R}{\delta} =
    \stransp{D'}{p}{\R}{\delta}[\bsq{V}\R/x]$.
    \item Note that by premise \emph{(iii)} we must have (3)
    $\delta(p, \ell, \R)[V/x] = \delta(p, \ell, \R)$, since
    if $\delta(p, \ell, \R) = x'$,
    premise \emph{(iii)} guarantees $x' \neq x$. Analogously,
    (4) $\delta(p, \R)[V/x] = \delta(p, \R)$.
    \item By convention, $x \notin \{y, y_\mf{set}\}$, so
    we show the conclusion as follows:
    \begin{align*}
      &\stransp{\estate{y}{\ell}{W}{D'}}{p}{\R}{\delta}
        [\bsq{V}\R/x]\\
      &= (\mf{let}\ y = \bang \delta(p, \ell, \R)\ \mf{in}\\
      &\quad\ \mf{let}\ y_\mf{set} = \lambda z.(
        \delta(p, \ell, \R) := z \semi \delta(p, \R) := \mf{true})\ \mf{in}
      &&\text{def. of translation}\\
      &\quad\ \stransp{D'}{p}{\R}{\delta}
      )[\bsq{V}\R/x]\\
      &= \mf{let}\ y = \bang \delta(p, \ell, \R)[\bsq{V}\R/x]\ \mf{in}\\
      &\quad\ \mf{let}\ y_\mf{set} = \lambda z.(
        \delta(p, \ell, \R)[\bsq{V}\R/x] := z
        \semi \delta(p, \R)[\bsq{V}\R/x] := \mf{true})\ \mf{in}
      &&\text{def. of substitution}\\
      &\quad\ \stransp{D'}{p}{\R}{\delta}[\bsq{V}\R/x]\\
      &= \mf{let}\ y = \bang \delta(p, \ell, \R)\ \mf{in}\\
      &\quad\ \mf{let}\ y_\mf{set} = \lambda z.(
        \delta(p, \ell, \R) := z
        \semi \delta(p, \R) := \mf{true})\ \mf{in}
      &&\text{by (3) and (4)}\\
      &\quad\ \stransp{D'}{p}{\R}{\delta}[\bsq{V}\R/x]\\
      &= \mf{let}\ y = \bang \delta(p, \ell, \R)\ \mf{in}\\
      &\quad\ \mf{let}\ y_\mf{set} = \lambda z.(
        \delta(p, \ell, \R) := z
        \semi \delta(p, \R) := \mf{true})\ \mf{in}
      &&\text{by (2)}\\
      &\quad\ \stransp{D'[V/x]}{p}{\R}{\delta}\\
      &= \stransp{\estate{y}{\ell}{W[V/x]}{D'[V/x]}}{p}{\R}{\delta}
      &&\text{def. of translation}\\
      &= \stransp{(\estate{y}{\ell}{W}{D'})[V/x]}{p}{\R}{\delta}
      &&\text{def. of substitution}
    \end{align*}
  \end{itemize}

\end{proof}
\end{lemma}

\begin{lemma}[Substitution, bodies, initial]
\label{thm:trans-substbi}
If \begin{enumerate}[label=(\roman*)]
  \item $\itransp{D}{D_0}{p}{\R}{\delta}$ is well-defined,
  \item $\bsq{V}\R$ is well-defined, and
  \item $x \notin \mf{img}(\delta)$,
\end{enumerate}
then $\itransp{D[V/x]}{D_0[V/x]}{p}{\R}{\delta}
= \itransp{D}{D_0}{p}{\R}{\delta}[\bsq{V}\R/x]$.

\begin{proof}
  By mutual induction on the structure of $D$
  with Lemmas \ref{thm:trans-substv}, \ref{thm:trans-substc},
  and \ref{thm:trans-substbs}

  \pcase{$D = \mf{val}\ W$}:
  \begin{itemize}
    \item By premise \emph{(i)}, we have (1) that $\bsq{W}\R$ is well-defined
    and (2) that $\stransp{D_0}{p}{\R}{\delta}$ is well-defined.
    \item By mutual inductive hypothesis (Lemma \ref{thm:trans-substv})
    on (1) and premise \emph{(ii)}, we get
    (3) that $\bsq{W[V/x]}\R = \bsq{W}\R[\bsq{V}\R/x]$.
    \item By mutual inductive hypothesis (Lemma \ref{thm:trans-substbs})
    on (2) and premises \emph{(ii)} and \emph{(iii)}, we get
    (4) that $\stransp{D_0[V/x]}{p}{\R}{\delta} =
    \stransp{D_0}{p}{\R}{\delta}[\bsq{V}\R/x]$.
    \item Note that by premise \emph{(iii)}, we must have (5)
    $\delta(p, \R)[V/x] = \delta(p, \R)$, since if $\delta(p, \R) = x'$,
    premise \emph{(iii)} guarantees $x' \neq x$.
    \item We show the conclusion as follows:
  \end{itemize}
  \begin{align*}
    &\itransp{\mf{val}\ W}{D_0}{p}{\R}{\delta}[\bsq{V}\R/x]\\
    &= \mf{val}\ (\lambda ().\stransp{D_0}{p}{\R}{\delta},
      \delta(p, \R), \bsq{W}\R)[\bsq{V}\R/x]
    &&\text{def. of translation}\\
    &= \mf{val}\ (
      \lambda ().\stransp{D_0}{p}{\R}{\delta}[\bsq{V}\R/x],
      \delta(p, \R)[\bsq{V}\R/x], \bsq{W}\R[\bsq{V}\R/x])
    &&\text{def. of substitution}\\
    &= \mf{val}\ (\lambda ().\stransp{D_0[V/x]}{p}{\R}{\delta},
      \delta(p, \R), \bsq{W[V/x]}\R)
    &&\text{by (5), (6), and (7)}\\
    &= \itransp{\mf{val}\ W[V/x]}{D_0[V/x]}{p}{\R}{\delta}
    &&\text{def. of translation}\\
    &= \itransp{(\mf{val}\ W)[V/x]}{D_0[V/x]}{p}{\R}{\delta}
    &&\text{def. of substitution}
  \end{align*}

  \pcase{$D = \elet{y}{M}{D'}$} can be shown analogously to
  the same case in Lemma \ref{thm:trans-substbs}
  ($\mf{Succ}$-phase substitution).

  \pcase{$D=\estate{y}{\ell}{V}{D'}$}:
  \begin{itemize}
    \item By premise \emph{(i)}, we have (1) that $\bsq{V}\R$ is well-defined
    and (2) that $\itransp{D}{D_0}{p}{\R}{\delta,\ell \mapsto y_\ell}$
    is well-defined.
    \item By mutual inductive hypothesis (Lemma \ref{thm:trans-substv})
    on (1) and premise \emph{(ii)}, we have (3) that
    $\bsq{W[V/x]}\R = \bsq{W}\R[\bsq{V}\R/x]$.
    \item Note that whatever $y_\ell$ we generate in our translation
    will be fresh, and so premise \emph{(iii)} gives us (4) that
    $x \notin \mf{img}(\delta, \ell \mapsto y_\ell)$.
    \item By inductive hypothesis on (2), (4), and premise
    \emph{(ii)}, we get (5) that\\
    $\itransp{D[V/x]}{D_0[V/x]}{p}{\R}{\delta} =
    \itransp{D}{D_0}{p}{\R}{\delta}[\bsq{V}\R/x]$.
    \item Note that by premise \emph{(iii)} we must have
    (6) $\delta(p, \R)[V/x] = \delta(p, \R)$, since if $\delta(p, \R) = x'$,
    premise \emph{(iii)} guarantees $x' \neq x$.
    \item By convention, $x \notin \{y, y_\mf{set}\}$, and also $x \neq y_\ell$
    since $y_\ell$ is translation-fresh. We show the conclusion as follows:
    \begin{align*}
      &\itransp{\estate{y}{\ell}{W}{D'}}{D_0}{p}{\R}{\delta}
        [\bsq{V}\R/x]\\
      &= (\mf{let}\ y_\ell = \mf{ref}(\bsq{W}\R)\ \mf{in}\\
      &\quad\ \mf{let}\ y = \bang y_\ell\ \mf{in}
      &&\text{def. of translation}\\
      &\quad\ \mf{let}\ y_\mf{set} = \lambda z.(
        y_\ell := z \semi \delta(p, \R) := \mf{true})\ \mf{in}\\
      &\quad\ \itransp{D'}{D_0}{p}{\R}{\delta, \ell \mapsto y_\ell}
      )[\bsq{V}\R/x]\\
      &= \mf{let}\ y_\ell = \mf{ref}(\bsq{W}\R[\bsq{V}\R/x])\ \mf{in}\\
      &\quad\ \mf{let}\ y = \bang y_\ell\ \mf{in}
      &&\text{def. of substitution}\\
      &\quad\ \mf{let}\ y_\mf{set} = \lambda z.(
        y_\ell := z \semi \delta(p, \R)[\bsq{V}\R/x] := \mf{true})\ \mf{in}\\
      &\quad\ \itransp{D'}{D_0}{p}{\R}{\delta, \ell \mapsto y_\ell}
        [\bsq{V}\R/x]\\
      &= \mf{let}\ y_\ell = \mf{ref}(\bsq{W[V/x]}\R)\ \mf{in}\\
      &\quad\ \mf{let}\ y = \bang y_\ell\ \mf{in}
      &&\text{by (3), (5), and (6)}\\
      &\quad\ \mf{let}\ y_\mf{set} = \lambda z.(
        y_\ell := z \semi \delta(p, \R) := \mf{true})\ \mf{in}\\
      &\quad\ \itransp{D'[V/x]}{D_0[V/x]}{p}{\R}{\delta, \ell \mapsto y_\ell}\\
      &= \itransp{\estate{y}{\ell}{W[V/x]}{D'[V/x]}}{D_0[V/x]}{p}{\R}{\delta}
      &&\text{def. of translation}\\
      &= \itransp{(\estate{y}{\ell}{W}
        {D'})[V/x]}{(D_0[V/x]}{p}{\R}{\delta}
      &&\text{def. of substitution}
    \end{align*}
  \end{itemize}

\end{proof}
\end{lemma}

\subsubsection[Rho-Extension]{$\R$-Extension}

The following set of ``$\R$-extension lemmas'' state that
a well-defined translation under $\R$ is equivalent to
that translation under extensions of $\R$.
They are fairly easy to prove, but nonetheless important.

\begin{lemma}[$\R$-extension, values]
\label{thm:trans-rextv}
If \begin{enumerate}[label=(\roman*)]
  \item $\bsq{V}\R$ is well-defined and
  \item $\R \subseteq \R'$,
\end{enumerate}
then $\bsq{V}\R = \bsq{V}{\R'}$.
\begin{proof} This can be most formally shown by mutual induction
  over the structure of $V$, but we here provide a direct argument instead.

  By premise \emph{(i)} and definition of well-definedness,
  we have that all instances of $\R({id})$, $\R(p)$, and $\R(p, \ell)$
  which occur in the translation result are well-defined.

  Without loss of generality, take $\R(p)$ as an example, and suppose $\R(p)$
  occurs in the translation result. Then by premise \emph{(i)},
  $\R(p)$ is well-defined.
  Now by premise \emph{(ii)}, since $\R(p)$ is well-defined
  and $\R \subseteq \R'$,
  we must have $\R'(p) = \R(p)$. By analogy, $\R'(p, \ell) = \R(p, \ell)$
  and $\R'({id}) = \R({id})$ for any relevant instances of $\R(p, \ell)$
  and $\R({id})$.

  Choice of $\R$ only affects the translation result by occurrences of
  $\R({id})$, $\R(p)$, and $\R(p, \ell)$, and we have shown that these are
  equivalent to $\R'({id})$, $\R'(p)$, and $\R'(p, \ell)$ respectively.
  Therefore, $\bsq{V}\R = \bsq{V}{\R'}$ and we are done.
\end{proof}
\end{lemma}

The above argument applies to all other $\R$-extension lemmas,
whose proofs we therefore elide.

\begin{lemma}[$\R$-extension, computations]
\label{thm:trans-rextc}
If \begin{enumerate}[label=(\roman*)]
  \item $\bsq{M}\R$ is well-defined and
  \item $\R \subseteq \R'$,
\end{enumerate}
then $\bsq{M}\R = \bsq{M}{\R'}$.
\end{lemma}

\begin{lemma}[$\R$-extension, bodies, successor]
\label{thm:trans-rextbs}
If \begin{enumerate}[label=(\roman*)]
  \item $\stransp{D}{p}{\R}{\delta}$ is well-defined and
  \item $\R \subseteq \R'$,
\end{enumerate}
then
$\stransp{D}{p}{\R}{\delta}
= \stransp{D}{p}{\R'}{\delta}$.
\end{lemma}

\begin{lemma}[$\R$-extension, bodies, initial]
\label{thm:trans-rextbi}
If \begin{enumerate}[label=(\roman*)]
  \item $\itransp{D}{D_0}{p}{\R}{\delta}$ is well-defined and
  \item $\R \subseteq \R'$,
\end{enumerate}
then $\itransp{D}{D_0}{p}{\R}{\delta}
= \itransp{D}{D_0}{p}{\R'}{\delta}$.
\end{lemma}

\begin{lemma}[$\R$-extension, heaps]
\label{thm:trans-rextm}
If \begin{enumerate}[label=(\roman*)]
  \item $\bsq{\heap}\R$ is well-defined and
  \item $\R \subseteq \R'$,
\end{enumerate}
then $\bsq{\heap}\R = \bsq{\heap}{\R'}$.
\end{lemma}

\subsubsection{Specialization}

The ``specialization'' lemmas allow us to move
from static bodies (parameterized by the nonexistent pointer $\bot$)
to dynamic bodies (parameterized by some concrete pointer $p$).

\begin{lemma}[Specialization, successor]
\label{thm:trans-specs}
If $\stransp{D}{\bot}{\R}{\delta}$ is well-defined,
then for any $p$, we have $\stransp{D}{\bot}{\R}{\delta}
= \stransp{D}{p}{\R}{\delta}$.

\begin{proof} This can be most formally shown by induction over the structure
  of $D$, but we here provide a direct argument instead.

  By premise \emph{(i)}, any instances of $\delta(\bot, \ell, \R)$ or
  $\delta(\bot, \R)$ which appear in the translation are well-defined.

  Consider first an arbitrary instance of $\delta(\bot, \ell, \R)$:
  by definition, either $\ell \mapsto x_\ell \in \delta$
  and $\delta(\bot, \ell, \R) = x_\ell$,
  or $\ell \notin \dom(\delta)$ and $\delta(\bot, \ell, \R)=\R(\bot, \ell)$.
  However, $\R(\bot, \ell) = \bot$ by definition of $\R$, and so
  we therefore must have $\ell \mapsto x_\ell \in \delta$.
  Now, by definition, $\delta(p, \ell, \R) = x_\ell$ for any $p$. Thus,
  we get that $\delta(\bot, \ell, \R) = \delta(p, \ell)$ for any $p$.
  By analogy, we also have $\delta(\bot, \R) = \delta(p, \R)$ for any $p$.

  Choice of $\delta$ only affects the translation result by occurrences of
  $\delta(p, \ell, \R)$ and $\delta(p, \R)$, and we
  have shown that these are equivalent for any $p$ given the premise.
  Therefore, $\stransp{D}{\bot}{\R}{\delta}
  = \stransp{D}{p}{\R}{\delta}$ for any $p$ and we are done.

\end{proof}
\end{lemma}

\begin{lemma}[Specialization, initial]
\label{thm:trans-speci}
If $\itransp{D}{D_0}{\bot}{\R}{\delta}$ is well-defined,
then for any $p$, we have $\itransp{D}{D_0}{\bot}{\R}{\delta}
= \itransp{D}{D_0}{p}{\R}{\delta}$.

\begin{proof}
  Follows from an analogous argument to that presented for
  Lemma \ref{thm:trans-speci} (successor-phase specialization).
\end{proof}
\end{lemma}

\subsubsection[Delta-Elimination]{$\delta$-Elimination}

Our final set of lemmas concerns substitution with respect to $\delta$,
which we neglected before in our substitution lemmas.
Since the $\delta$ function always has two possible results,
one of which is a variable,
we can show that \emph{substituting the other result for that variable}
lets us simply remove the corresponding entry from $\delta$.

\begin{lemma}[Flag elimination, successor]
\label{thm:trans-felims}
If $\stransp{D}{p}{\R}{\delta, \mf{flag} \mapsto x}$
is well-defined, then we have that
$\stransp{D}{p}{\R}{\delta, \mf{flag} \mapsto x}[\R(p)/x]
= \stransp{D}{p}{\R}{\delta}[\R(p)/x]$.

\begin{proof} Once again, this can be most formally shown by induction over
  the structure of $D$, but we here provide a direct argument instead.

  Let $\delta' = \delta, \mf{flag} \mapsto x$, and consider
  an arbitrary instance of $\delta'(p, \R)$ in the translation result.
  By our premise, $\delta'(p, \R)$ is well-defined, and then by definition,
  $\delta'(p, \R) = x$. Therefore, $\delta'(p, \R)[\R(p)/x] = \R(p)$.
  Now, by well-formedness of $\delta'$, we must have
  $\mf{flag} \notin \dom(\delta)$, i.e. $\delta(p, \R) = \R(p)$.
  Thus, $\delta'(p, \R)[\R(p)/x] = \delta(p, \R)$.

  Choice of $\delta$ versus $\delta'$ only affects the translation result
  by occurrences of $\delta(p, \R)$, and we have shown that
  $\delta(p, \R)[\R(p)/x]$ is equivalent to $\delta'(p, \R)$.
  Therefore, we conclude that\\
  $\stransp{D}{p}{\R}{\delta'}[\R(p)/x]
  = \stransp{D}{p}{\R}{\delta}[\R(p)/x]$ and we are done.

\end{proof}
\end{lemma}

The above argument applies to all other $\delta$-elimination lemmas,
whose proofs we therefore elide.

\begin{lemma}[Flag elimination, initial]
\label{thm:trans-felimi}
If $\itransp{D}{D_0}{p}{\R}{\delta, \mf{flag} \mapsto x}$ is well-defined,
then we have that $\itransp{D}{D_0}{p}{\R}{\delta, \mf{flag} \mapsto x}[\R(p)/x]
= \itransp{D}{D_0}{p}{\R}{\delta}[\R(p)/x]$.
\end{lemma}

\begin{lemma}[Label elimination, successor]
\label{thm:trans-lelims}
If $\stransp{D}{p}{\R}{\delta, \ell \mapsto x}$
is well-defined, then we have that
$\stransp{D}{p}{\R}{\delta, \ell \mapsto x}[\R(p, \ell)/x]
= \stransp{D}{p}{\R}{\delta}[\R(p, \ell)/x]$.
\end{lemma}

\begin{lemma}[Label elimination, initial]
\label{thm:trans-lelimi}
If $\itransp{D}{D_0}{p}{\R}{\delta, \ell\mapsto x}$ is well-defined,
then we have that $\itransp{D}{D_0}{p}{\R}{\delta, \ell \mapsto x}
  [\R(p, \ell)/x]
=\itransp{D}{D_0}{p}{\R}{\delta}[\R(p, \ell)/x]$.
\end{lemma}

\subsubsection{Simulation}

We state one final lemma about the target language:

\begin{lemma}[Let-transitivity]
\label{thm:trans-let}
If $\sigma \semi M \stept\ M' \semi \sigma'$, then
$\sigma \semi \elet{x}{M}{N} \stept\ \elet{x}{M'}{N} \semi \sigma'$.

\begin{proof}
  By induction on the length of the reduction sequence.
  For each intermediate step in the premise,
  we derive a corresponding step via \textsc{EC-Let}.
\end{proof}
\end{lemma}

As mentioned in Section \ref{sec:transcorrect}, the proofs of simulation
mirror those of $\lamhook$ type preservation.
The novel bits are \blue{typeset in blue}.

\begin{restatable}[Simulation, pure computations]{theorem}{transsimcp}
\label{thm:trans-simcp}
If
\begin{enumerate}[label=(\roman*)]
  \item $\Gamma \semi \Pi \semi \blue{\Phi_H \semi} \Phi_M
    \vdash \blue{\heap \semi} M: A \bang \mf{Pure}$,
  \item \blue{$\Pi \semi \mf{ids}(M) \cup \mf{ids}(\heap) \vDash \R$}, and
  \item $\heap \semi M \step N \semi \heap'$,
\end{enumerate}
then
\begin{enumerate}[label=(\roman*)]
  \item $\Gamma \semi \Pi \semi \Phi_M \vdash N: A \bang \mf{Pure}$,
  \item $\heap' = \heap$,
  \item \blue{$\Pi \semi \mf{ids}(N) \cup \mf{ids}(\heap) \vDash \R$}, and
  \item \blue{$\bsq{\heap}\R \semi \bsq{M}\R
  \stept\ \bsq{N}\R \semi \bsq{\heap}\R$}.
\end{enumerate}
\end{restatable}
\begin{proof}
  By definition of premise \emph{(i)}, we have subpremises
  \begin{enumerate}[label=(i-\alph*)]
  \item $\Gamma \semi \Pi \semi \Phi_M \vdash M: A \bang \mf{Pure}$ and
  \item $\Gamma \semi \Pi \semi \Phi_H \vDash \heap$.
  \end{enumerate}

  \blue{We begin with a few general observations:}
  \blue{\begin{itemize}
    \item By premise \emph{(ii)}, we have
    (ii-a) that $\Pi \semi \mf{ids}(M) \vDash \R$
    and (ii-b) that $\Pi \semi \mf{ids}(\heap) \vDash \R$.
    \item To show conclusion \emph{(iv)} in any case, it is necessary to
    show (iv-a) that $\bsq{\heap}\R$ is well-defined
    and (iv-b) that $\bsq{M}\R$ is well-defined:
    \begin{itemize}
      \item (iv-a) is given by Lemma \ref{thm:trans-prezh}
      (heap preservation)
      on premises (i-b) and (ii-b).
      \item (iv-b) is given by Lemma \ref{thm:trans-prezc}
      (computation preservation)
      on premises (i-a) and (ii-a).
    \end{itemize}
  \end{itemize}}

  We proceed by induction on the derivation
  of the computation step judgment \emph{(iii)}.

  \pcase{EC-Let}: We have
  \begin{mathpar}
    \inferrule*[right=EC-Let] {
      \inferrule* {(1)} {\heap \semi M \step M' \semi \heap'}
    } {
      \heap \semi \elet{x}{M}{N} \step \elet{x}{M'}{N} \semi \heap'
    }

    \inferrule*[right=TC-Let] {
      \inferrule* {(2)} {
        \Gamma \semi \Pi \semi \Phi_1 \vdash M: A \bang \mf{Pure}
      }
      \and
      \inferrule* {(3)} {
        \Gamma, x: A \semi \Pi \semi \Phi_2 \vdash N: B \bang \mf{Pure}
      }
    } {
      \Gamma \semi \Pi \semi \Phi_1 \Join \Phi_2 \vdash \elet{x}{M}{N}:
      B \bang \mf{Pure}
    }
  \end{mathpar}
  \begin{itemize}
    \blue{\item By premise \emph{(ii)} and definition,
    we have (5) that $\Pi \semi \mf{ids}(M) \cup \mf{ids}(\heap) \vDash \R$
    and (6) that $\Pi \semi \mf{ids}(N) \cup \mf{ids}(\heap) \vDash \R$.}
    \item By inductive hypothesis on (1), (2), and premise (i-b), we get
    \begin{enumerate}
      \setcounter{enumi}{6}
      \item $\Gamma \semi \Pi \semi \Phi_1 \vdash M': A \bang \mf{Pure}$,
      \item $\heap' = \heap$, satisfying conclusion \emph{(ii)},
      \blue{
        \item $\Pi \semi \mf{ids}(M') \cup \mf{ids}(\heap') \vDash \R$, and
        \item $\bsq{\heap}\R \semi \bsq{M}\R
          \stept\ \bsq{M'}\R \semi \bsq{\heap}\R$.
      }
    \end{enumerate}
    \blue{
      \item By (5) and (8), we get
      $\Pi \semi \mf{ids}(M') \cup \mf{ids}(N)
      \cup \mf{ids}(\heap) \vDash \R$,
      satisfying conclusion \emph{(iii)}.}
    \item We derive conclusion \emph{(i)} as follows:
    \begin{mathpar}
      \inferrule*[right=TC-Let] {
        \inferrule* {(7)} {
          \Gamma \semi \Pi \semi \Phi_1 \vdash M': A \bang \mf{Pure}
        }
        \and
        \inferrule* {(3)} {
          \Gamma, x: A \semi \Pi \semi \Phi_2 \vdash N: B \bang \mf{Pure}
        }
      } {
        \Gamma \semi \Pi \semi \Phi_1 \Join \Phi_2 \vdash \elet{x}{M'}{N}:
        B \bang \mf{Pure}
      }
    \end{mathpar}
    {\color{blue}
    \item We show conclusion \emph{(iv)} as follows:
      \begin{align*}
        &\bsq{\heap}\R \semi \bsq{\elet{x}{M}{N}}\R\\
        \text{(def. of translation)} \quad
        &= \elet{x}{\bsq{M}\R}{\bsq{N}\R} \semi \bsq{\heap}\R\\
        \text{(Lemma \ref{thm:trans-let}, let-transitivity on (10))} \quad
        &= \elet{x}{\bsq{M'}\R}{\bsq{N}\R} \semi \bsq{\heap}\R\\
        \text{(def. of translation)} \quad
        &= \bsq{\elet{x}{M'}{N}}\R \semi \bsq{\heap}\R
      \end{align*}}
  \end{itemize}

  \pcase{EC-Val}: We have
  $\heap \semi \elet{x}{\mf{val}\ V}{M} \step M[V/x] \semi \heap$
  (from which conclusion \emph{(ii)} is immediate) and
  \begin{mathpar}
    \inferrule*[right=EC-Let] {
      \inferrule*[right=EC-Val] {
        \inferrule* {(1)} {
          \Gamma \semi \Pi \semi \Phi_1 \vdash V: A
        }
      } {
        \Gamma \semi \Pi \semi \Phi_1 \vdash \mf{val}\ V: A \bang \mf{Pure}
      }
      \and
      \inferrule* {(2)} {
        \Gamma, x: A \semi \Pi \semi \Phi_2 \vdash M: B \bang \mf{Pure}
      }
    } {
      \Gamma \semi \Pi \semi \Phi_1 \Join \Phi_2 \vdash
      \elet{x}{\mf{val}\ V}{M}: B \bang \mf{Pure}
    }
  \end{mathpar}
  \begin{itemize}
    \item Conclusion \emph{(i)} is given by Lemma \ref{thm:lamhook-substc}
    (type substitution) on (1) and (2).
    \blue{
      \item Conclusion \emph{(iii)} is given by premise \emph{(ii)}
      and the fact that
      $\mf{ids}(M[V/x]) \subseteq \mf{ids}(M) \cup \mf{ids}(V)$.
      \item We show conclusion \emph{(iv)} as follows:
      \begin{align*}
        &\bsq{\heap}\R \semi \bsq{\elet{x}{\mf{val}\ V}{M}}\R\\
        \text{(def. of translation)} \quad
        &= \elet{x}{\mf{val}\ \bsq{V}\R}{\bsq{M}\R} \semi \bsq{\heap}\R\\
        \textsc{(EC-App)}\quad
        &\step \bsq{M}\R[\bsq{V}\R/x] \semi \bsq{\heap}\R\\
        \text{(Lemma \ref{thm:trans-substc}, substitution)}\quad
        &= \bsq{M[V/x]}\R \semi \bsq{\heap}\R
      \end{align*}}
  \end{itemize}

  \pcase{EC-App}: We have
  $\heap \semi (\lambda x.M)\ V \step M[V/x] \semi \heap$
  (from which conclusion \emph{(ii)} is immediate) and
  \begin{mathpar}
    \inferrule*[right=TC-App] {
      \inferrule*[right=TV-Lam] {
        \inferrule* {(1)} {
          \Gamma, x: A \semi \Pi \semi \Phi_1 \vdash M: B \bang \mf{Pure}
        }
        \and
        \mf{Succ}(\Phi_1)
      } {
        \Gamma \semi \Pi \semi \Phi_1 \vdash \lambda x.M:
        A \xrightarrow{\mf{Pure}} B
      }
      \and
      \inferrule* {(2)} {
        \Gamma \semi \Pi \semi \Phi_2 \vdash V: A
      }
    } {
      \Gamma \semi \Pi \semi \Phi_1 \Join \Phi_2 \vdash (\lambda x.M)\ V:
      B \bang \mf{Pure}
    }
  \end{mathpar}
  \begin{itemize}
    \item Conclusion \emph{(i)} is given by Lemma \ref{thm:lamhook-substc}
    (type substitution) on (1) and (2), noting $\Phi \Join \emptyset = \Phi$.
    \blue{
      \item Conclusion \emph{(iii)} is given by premise \emph{(ii)}
      and the fact that
      $\mf{ids}(M[V/x]) \subseteq \mf{ids}(M) \cup \mf{ids}(V)$.
      \item We show conclusion \emph{(iv)} as follows:
      \begin{align*}
        &\bsq{\heap}\R \semi \bsq{(\lambda x.M)\ \bsq{V}}\R\\
        \text{(def. of translation)}\quad
        &= (\lambda x.\bsq{M}\R)\ \bsq{V}\R \semi \bsq{\heap}\R\\
        \textsc{(EC-App)}\quad
        &\step \bsq{M}\R[\bsq{V}\R/x] \semi \bsq{\heap}\R\\
        \text{(Lemma \ref{thm:trans-substc}, substitution)}\quad
        &= \bsq{M[V/x]}\R \semi \bsq{\heap}\R
      \end{align*}}
  \end{itemize}

  \pcase{EC-IfT}: We have
  $\heap \semi \eif{\mf{true}}{M}{N} \step M \semi \heap$
  (from which conclusion \emph{(ii)} is immediate) and
  \begin{mathpar}
    \inferrule*[right=TC-If] {
      \inferrule*[right=TV-True] {
        \inferrule* {(1)} {\mf{Succ}(\Phi_1)}
      } {
        \Gamma \semi \Pi \semi \Phi_1 \vdash \mf{true}: \mf{Bool}
      }
      \
      \inferrule* {(2)} {
        \Gamma \semi \Pi \semi \Phi_2 \vdash M: A \bang \mf{Pure}
      }
      \quad
      \Gamma \semi \Pi \semi \Phi_2 \vdash N: A \bang \mf{Pure}
    } {
      \Gamma \semi \Pi \semi \Phi
      \vdash \eif{\mf{true}}{M}{N}: A \bang \mf{Pure}
    }
  \end{mathpar}
  \begin{itemize}
    \item Conclusion \emph{(i)} is given by Lemma \ref{thm:lamhook-phiwknc}
    ($\Phi$-weakening) on (2) using (1).
    \blue{
      \item By premise \emph{(ii)}, we have
      $\Pi \semi \mf{ids}(M) \cup \mf{ids}(\heap) \vDash \R$,
      which satisfies conclusion \emph{(iii)}.
      \item We show conclusion \emph{(iv)} as follows:
      \begin{align*}
        &\bsq{\heap}\R \semi \bsq{\eif{\mf{true}}{M}{N}}\R\\
        (\text{def. of translation})
        &= \eif{\mf{true}}{\bsq{M}\R}{\bsq{N}\R} \semi \bsq{\heap}\R\\
        \textsc{(EC-IfT)}\quad
        &\step \bsq{M}\R \semi \bsq{\heap}\R
      \end{align*}}
  \end{itemize}

  \pcase{EC-IfF}: Analogous to the above case.

  \pcase{EC-Fst}: We have
  $\heap \semi \mf{fst}\ (V, W) \step \mf{val}\ V \semi \heap$
  (from which conclusion \emph{(ii)} is immediate) and
  \begin{mathpar}
    \inferrule*[right=TC-Fst] {
      \inferrule*[right=TV-Pair] {
        \inferrule* {(1)} {
          \Gamma \semi \Pi \semi \Phi_1 \vdash V: A
        }
        \and
        \inferrule* {(2)} {
          \Gamma \semi \Pi \semi \Phi_2 \vdash W: B
        }
      } {
        \Gamma \semi \Pi \semi \Phi_1 \Join \Phi_2 \vdash (V, W):
        A \times B
      }
    } {
      \Gamma \semi \Pi \semi \Phi_1 \Join \Phi_2
      \vdash \mf{fst}\ (V, W): A \bang \mf{Pure}
    }
  \end{mathpar}
  \begin{itemize}
    \item By Lemma \ref{thm:lamhook-unrestv} (unrestriction) on (2),
    we get (3) that $\mf{Succ}(\Phi_2)$.
    \item By Lemma \ref{thm:lamhook-phiwknv} ($\Phi$-weakening) on (1) and (3),
    we get (4) that $\Gamma \semi \Pi \semi \Phi_1 \Join \Phi_2
    \vdash V: A$.
    \item We derive conclusion \emph{(i)} as follows:
    \begin{mathpar}
      \inferrule*[right=TC-Val] {
        \inferrule*[right=(4)] { } {
          \emptyset \semi \Pi \semi \Phi_1 \Join \Phi_2 \vdash V: A
        }
      } {
        \emptyset \semi \Pi \semi \Phi_1 \Join \Phi_2 \vdash \mf{val}\ V:
        A \bang \mf{Pure}
      }
    \end{mathpar}
    \blue{
      \item By premise \emph{(ii)}, we have
      $\Pi \semi \mf{ids}(V) \cup \mf{ids}(\heap) \vDash \R$,
      which satisfies conclusion \emph{(iii)}.
      \item We show conclusion \emph{(iv)} as follows:
      \begin{align*}
        &\bsq{\heap}\R \semi \bsq{\mf{fst}\ (V, W)}\R\\
        (\text{def. of translation}) \quad
        &= \mf{fst}\ (\bsq{V}\R, \bsq{W}\R) \semi \bsq{\heap}\R\\
        \textsc{(EC-Fst)}\quad
        &\step \mf{val}\ \bsq{V}\R \semi \bsq{\heap}\R\\
        (\text{def. of translation}) \quad
        &= \bsq{\mf{val}\ V}\R \semi \bsq{\heap}\R
      \end{align*}}
  \end{itemize}

  \pcase{EC-Snd}: Analogous to the above case.

  \pcase{EC-FdN}: We have
  $\heap \semi \mf{fold}\ V_1\ V_2\ [] \step \mf{val}\ V_2 \semi \heap$
  (from which conclusion \emph{(ii)} is immediate) and
  \begin{mathpar}
    \inferrule*[right=TC-Fold] {
      \inferrule* {(1)} {
        \Gamma \semi \Pi \semi \Phi_1 \vdash V_1:
        A \xrightarrow{\mf{Pure}} B \xrightarrow{\mf{Pure}} B
      }
      \and
      \inferrule* {(2)} {
        \Gamma \semi \Pi \semi \Phi_2 \vdash V_2: B
      }
      \and
      \inferrule* {(3)} {
        \Gamma \semi \Pi \semi \Phi_3 \vdash []: \mf{List}\ A
      }
    } {
      \Gamma \semi \Pi \semi \Phi_1 \Join \Phi_2 \Join \Phi_3 \vdash
      \mf{fold}\ V_1\ V_2\ []: B \bang \mf{Pure}
    }
  \end{mathpar}
  \begin{itemize}
    \item By Lemma \ref{thm:lamhook-unrestv} (unrestriction) on (1) and (3),
      we get (4) that $\mf{Succ}(\Phi_1)$ and (5) that $\mf{Succ}(\Phi_3)$,
      respectively.
    \item By Lemma \ref{thm:lamhook-phiwknv} ($\Phi$-weakening) on (2)
      using (4) and (5),
      we get (6) that $\Gamma \semi \Pi \semi \Phi_1 \Join \Phi_2 \Join \Phi_3
      \vdash V_2: B$.
    \item We derive conclusion \emph{(i)} as follows:
    \begin{mathpar}
      \inferrule*[right=TC-Val] {
        \inferrule*[right=(6)] { } {
          \Gamma \semi \Pi \semi \Phi_1 \Join \Phi_2 \Join \Phi_3 \vdash
          V_2: B
        }
      } {
        \Gamma \semi \Pi \semi \Phi_1 \Join \Phi_2 \Join \Phi_3 \vdash
        \mf{val}\ V_2: B \bang \mf{Pure}
      }
    \end{mathpar}
    \blue{
      \item By premise \emph{(ii)}, we have
      $\Pi \semi \mf{ids}(V_2) \cup \mf{ids}(\heap) \vDash \R$,
      which satisfies conclusion \emph{(iii)}.
      \item We show conclusion \emph{(iv)} as follows:
      \begin{align*}
        &\bsq{\heap}\R \semi \bsq{\mf{fold}\ V_1\ V_2\ []}\R\\
        (\text{def. of translation}) \quad
        &= \mf{fold}\ \bsq{V_1}\R\ \bsq{V_2}\R\ []
          \semi \bsq{\heap}\R\\
        \textsc{(EC-FdN)}\quad
        &\step \mf{val}\ \bsq{V_2}\R \semi \bsq{\heap}\R\\
        (\text{def. of translation}) \quad
        &= \bsq{\mf{val}\ V_2}\R \semi \bsq{\heap}\R
      \end{align*}}
  \end{itemize}

  \pcase{EC-FdC}: We have
  $\heap \semi \mf{fold}\ V_1\ V_2\ (V_h :: V_t)
  \step V_1\ V_h\ (\mf{fold}\ V_1\ V_2\ V_t) \semi \heap$,
  whose right-hand term desguars to
  $\elet{x}{V_1\ V_h}{\elet{y}{\mf{fold}\ V_1\ V_2\ V_t}{x\ y}}$.
  Conclusion \emph{(ii)} is immediate, and we also have:
  \begin{mathpar}
    \inferrule*[right=TC-Fold] {
      (1) \and (2) \and
      \inferrule* {
        \inferrule* {(3)} {
          \Gamma \semi \Pi \semi \Phi_3 \vdash V_h: A
        }
        \and
        \inferrule* {(4)} {
          \Gamma \semi \Pi \semi \Phi_4 \vdash V_t: \mf{List}\ A
        }
      } {
        \Gamma \semi \Pi \semi \Phi_3 \Join \Phi_4
        \vdash V_h :: V_t: \mf{List}\ A
      }
    } {
      \Gamma \semi \Pi \semi \Phi_1 \Join \Phi_2 \Join \Phi_3 \Join \Phi_4
      \vdash
      \mf{fold}\ V_1\ V_2\ (V_h :: V_t): B \bang \mf{Pure}
    }
  \end{mathpar}
  where (1) and (2) are identical to the previous case.
  \begin{itemize}
    \item By Lemma \ref{thm:lamhook-unrestv} (unrestriction) on (1)-(4),
      we get that $\mf{Succ}(\Phi_i)$ for $i \in \{1, 2, 3, 4\}$.
      It follows (5) that $\mf{Succ}(\Phi)$, where $\Phi=\Join_{i=1}^4 \Phi_i$.
    \item We derive conclusion \emph{(i)} by way of
      intermediate derivations (6)-(8) as follows,
      using $\Phi$ intuitionistically by (5) and weakening (1)-(4)
      appropriately:
  \end{itemize}
  \begin{mathpar}
    \inferrule*[right=TC-App] {
      \inferrule*[right=TV-Var] { } {
        \Gamma_6 \semi \Pi \semi \Phi
        \vdash x: B \xrightarrow{\mf{Pure}} B
      }
      \and
      \inferrule*[right=TV-Var] { } {
        \Gamma_6 \semi \Pi \semi \Phi
        \vdash y: B
      }
    } {
      (6) \quad
      \Gamma_6 \semi \Pi \semi \Phi \vdash
      x\ y: B \bang \mf{Pure}
    }

    \inferrule*[right=TC-Fold] {
      \inferrule*[right=(1)] { } {
        \Gamma_7 \semi \Pi \semi \Phi \vdash V_1:
        A \xrightarrow{\mf{Pure}} B \xrightarrow{\mf{Pure}} B
      }
      \
      \inferrule*[right=(2)] { } {
        \Gamma_7 \semi \Pi \semi \Phi \vdash V_2: B
      }
      \
      \inferrule*[right=(4)] { } {
        \Gamma_7 \semi \Pi \semi \Phi \vdash V_t: \mf{List}\ A
      }
    } {
      (7) \quad
      \Gamma_7 \semi \Pi \semi \Phi \vdash
      \mf{fold}\ V_1\ V_2\ V_t: B \bang \mf{Pure}
    }

    \inferrule*[right=TC-Let] {
      [(7),\ \Gamma_7=\Gamma_8]
      \and
      [(6),\ \Gamma_6=\Gamma_8, y: B]
    } {
      (8) \quad
      \Gamma_8 \semi \Pi \semi \Phi \vdash
      \elet{y}{\mf{fold}\ V_1\ V_2\ V_t}{\cdots}: B \bang \mf{Pure}
    }

    \inferrule*[right=TC-Let] {
      \inferrule*[Right=TC-App] {
        \inferrule*[right=(1)] { } {
          \Gamma \semi \Pi \semi \Phi \vdash V_1:
          A \xrightarrow{\mf{Pure}} B \xrightarrow{\mf{Pure}} E
        }
        \
        \inferrule*[right=(3)] { } {
          \Gamma \semi \Pi \semi \Phi \vdash V_h: A
        }
      } {
        \Gamma \semi \Pi \semi \Phi \vdash
        V_1\ V_h: B \xrightarrow{\mf{Pure}} B \bang \mf{Pure}
      }
      \
      [(8),\ \Gamma_8 = \Gamma, x: B \xrightarrow{\mf{Pure}} B]
    } {
      \Gamma \semi \Pi \semi \Phi
      \vdash
      \elet{x}{V_1\ V_h}{\cdots}: B \bang \mf{Pure}
    }
  \end{mathpar}
  \begin{itemize}
    \blue{
      \item By premise \emph{(ii)}, we have
      $\Pi \semi \mf{ids}(V_1) \cup \mf{ids}(V_2) \cup \mf{ids}(V_h)
        \cup \mf{ids}(V_t) \cup \mf{ids}(\heap) \vDash \R$,
      which satisfies conclusion \emph{(iii)}.
      \item We show conclusion \emph{(iv)} as follows:
      \begin{align*}
        &\bsq{\heap}\R \semi \bsq{\mf{fold}\ V_1\ V_2\ (V_h :: V_t)}\R\\
        (\text{def. of translation}) \quad
        &= \mf{fold}\ \bsq{V_1}\R\ \bsq{V_2}\R\ (\bsq{V_h}\R :: \bsq{V_t}\R)
          \semi \bsq{\heap}\R\\
        \textsc{(EC-FdC)}\quad
        &\step \bsq{V_1}\R\ \bsq{V_h}\R\
          (\mf{fold}\ \bsq{V_1}\R\ \bsq{V_2}\R\ \bsq{V_t}\R)
          \semi \bsq{\heap}\R\\
        (\text{def. of translation}) \quad
        &= \bsq{V_1\ V_h\ (\mf{fold}\ V_1\ V_2\ V_t)}\R
          \semi \bsq{\heap}\R\\
      \end{align*}}
  \end{itemize}

  \textbf{All other cases} are untypeable with the $\mf{Pure}$ effect,
  and so the theorem holds vacuously.
\end{proof}

\begin{restatable}[Simulation, bodies, successor]{theorem}{transsimbs}
\label{thm:trans-simbs}
If
\begin{enumerate}[label=(\roman*)]
\item $\Gamma \semi \Pi \semi \Phi_H \semi \Phi_D \vdash \heap \semi D: \Ell$,
\item $\Ell \subseteq \Pi[p]$,
\item \blue{$\Pi \semi \mf{ids}(D) \cup \mf{ids}(\heap) \vDash \R$}, and
\item $\heap \semi D \step^\mf{Succ}_p D' \semi \heap'$,
\end{enumerate}
then for some $\Phi_D'$, we have:
\begin{enumerate}[label=(\roman*)]
\item $\Gamma \semi \Pi \semi \Phi_D' \vdash D': \Ell'$,
\item $\Phi_H \Join \Phi_D' = \Phi_H \Join \Phi_D$,
\item $\heap' = \heap$,
\item $\Ell' \subseteq \Ell$,
\item \blue{$\Pi \semi \mf{ids}(D') \cup \mf{ids}(\heap) \vDash \R$}, and
\item \blue{$\bsq{\heap}\R \semi \stransp{D}{p}{\R}{\emptyset}
\stept\
\stransp{D'}{p}{\R}{\emptyset} \semi \bsq{\heap}{\R}$}.
\end{enumerate}
\end{restatable}
\begin{proof}

By definition of premise \emph{(i)}, we have subpremises
\begin{enumerate}[label=(i-\alph*)]
\item $\Gamma \semi \Pi \semi \Phi_M \vdash M: A \bang \mf{Pure}$ and
\item $\Gamma \semi \Pi \semi \Phi_H \vDash \heap$.
\end{enumerate}

\blue{We begin with a few general observations:}
  \blue{\begin{itemize}
    \item By premise \emph{(iii)}, we have
    (iii-a) that $\Pi \semi \mf{ids}(D) \vDash \R$
    and (iii-b) that $\Pi \semi \mf{ids}(\heap) \vDash \R$.
    \item To show conclusion \emph{(vi)} in any case, it is necessary to
    show (vi-a) that $\bsq{\heap}\R$ is well-defined
    and (vi-b) that $\stransp{D}{p}{\R}{\emptyset}$ is well-defined:
    \begin{itemize}
      \item (vi-a) is given by Lemma \ref{thm:trans-prezh}
      (heap preservation)
      on premises (iii-b) and (i-b).
      \item (vi-b) is given by Lemma \ref{thm:trans-prezbs}
      (successor body preservation)
      on premises (i-a), \emph{(ii)} (which gives us that
      $\emptyset \semi \Pi \vDash_p \emptyset: \Ell$), and (iii-a).
    \end{itemize}
  \end{itemize}}

  We proceed by induction on the derivation
  of the body step judgment \emph{(iv)}.

  \pcase{EB-Let}: We have
  \begin{mathpar}
    \inferrule*[right=EB-Let] {
      \inferrule* {(1)} {\heap \semi M \step N \semi \heap'}
    } {
      \heap \semi \elet{x}{M}{D}
      \step^\mf{Init}_p \elet{x}{N}{D} \semi \heap'
    }

    \inferrule*[right=TB-Let] {
      \inferrule* {(2)} {
        \Gamma \semi \Pi \semi \Phi_1 \vdash M: A \bang \mf{Pure}
      }
      \and
      \inferrule* {(3)} {\Gamma, x: A \semi \Pi \semi \Phi_2 \vdash D: \Ell}
    } {
      \Gamma \semi \Pi \semi \Phi_1 \Join \Phi_2 \vdash \elet{x}{M}{D}: \Ell
    }
  \end{mathpar}
  \begin{itemize}
    \blue{\item By premise \emph{(iii)} and definition,
    we have (4) that $\Pi \semi \mf{ids}(M) \cup \mf{ids}(\heap) \vDash \R$
    and (5) that $\Pi \semi \mf{ids}(D) \cup \mf{ids}(\heap) \vDash \R$.}
    \item By Theorem \ref{thm:trans-simcp} (pure simulation) on (1), (2),
    \blue{(4)}, and premise (i-b), we get
    \begin{enumerate}
      \setcounter{enumi}{5}
      \item $\Gamma \semi \Pi \semi \Phi_1 \vdash N: A \bang \mf{Pure}$,
      \item $\heap' = \heap$, satisfying conclusion \emph{(iii)},
      \blue{
        \item $\Pi \semi \mf{ids}(N) \cup \mf{ids}(\heap) \vDash \R$, and
        \item $\bsq{\heap}\R \semi \bsq{M}\R
        \stept\ \bsq{N}\R \semi \bsq{\heap}\R$.}
    \end{enumerate}
    \item We choose $\Phi_D' = \Phi_D$, trivializing conclusion \emph{(ii)}.
      Conclusion \emph{(iv)} is obvious since $\Ell' = \Ell$.
    \blue{\item By (5) and (8), we get $\Pi \semi \mf{ids}(N)
    \cup \mf{ids}(D) \cup \mf{ids}(\heap) \vDash \R$, satisfying
    conclusion \emph{(v)}}.
    \item We derive conclusion \emph{(i)} as follows:
    \begin{mathpar}
      \inferrule*[right=TB-Let] {
        \inferrule*[right=(6)] { } {
          \Gamma \semi \Pi \semi \Phi_1 \vdash N: A \bang \mf{Pure}
        }
        \and
        \inferrule*[right=(3)] { } {
          \Gamma, x: A \semi \Pi \semi \Phi_2 \vdash D: \Ell
        }
      } {
        \Gamma \semi \Pi \semi \Phi_1 \Join \Phi_2 \vdash \elet{x}{N}{D}: \Ell
      }
    \end{mathpar}
    \blue{\item We show conclusion \emph{(vi)} as follows:
      \begin{align*}
        &\bsq{\heap}\R \semi \stransp{\elet{x}{M}{D}}{p}{\R}{\emptyset}\\
        (\text{def. of translation}) \quad
        &= \elet{x}{\bsq{M}\R}{\stransp{D}{p}{\R}{\emptyset}}
          \semi \bsq{\heap}\R\\
        \text{(Lemma \ref{thm:trans-let}, let-transitivity on (9))} \quad
        &\stept\ \elet{x}{\bsq{N}\R}{\stransp{D}{p}{\R}{\emptyset}}
        \semi \bsq{\heap}\R\\
        (\text{def. of translation}) \quad
        &= \stransp{\elet{x}{N}{D}}{p}{\R}{\emptyset} \semi \bsq{\heap}\R
      \end{align*}}
  \end{itemize}

  \pcase{EB-Val}: We have
  $\heap \semi \elet{x}{\mf{val}\ V}{D}
  \step^\mf{Succ}_p D[V/x] \semi \heap$
  (from which conclusion \emph{(iii)} is immediate) and
  \begin{mathpar}
    \inferrule*[right=TB-Let] {
      \inferrule*[right=TC-Val] {
        \inferrule* {(1)} {\Gamma \semi \Pi \semi \Phi_1 \vdash V: A}
      } {
        \Gamma \semi \Pi \semi \Phi_1 \vdash \mf{val}\ V: A \bang \mf{Pure}
      }
      \and
      \inferrule* {(2)} {\Gamma, x: A \semi \Pi \semi \Phi_2 \vdash D: \Ell}
    } {
      \Gamma \semi \Pi \semi \Phi_1 \Join \Phi_2
      \vdash \elet{x}{\mf{val}\ V}{D}: \Ell
    }
  \end{mathpar}
  \begin{itemize}
    \item Choose $\Phi_D' = \Phi_D$, trivializing conclusion \emph{(ii)}.
    \item Conclusion \emph{(i)} is given by Lemma \ref{thm:lamhook-substb}
    (type substitution) on (1) and (2).
    \item Conclusion \emph{(iv)} is obvious since $\Ell' = \Ell$.
    \blue{
      \item Conclusion \emph{(v)} is given by premise \emph{(iii)}
      and the fact that
      $\mf{ids}(D[V/x]) \subseteq \mf{ids}(D) \cup \mf{ids}(V)$.
      \item We show conclusion \emph{(vi)} as follows:
      \begin{align*}
        &\bsq{\heap}\R
          \semi \stransp{\elet{x}{\mf{val}\ V}{D}}{p}{\R}{\emptyset}\\
        (\text{def. of translation}) \quad
        &= \elet{x}{\mf{val}\ \bsq{V}\R}{\stransp{D}{p}{\R}{\emptyset}}
          \semi \bsq{\heap}\R\\
        \textsc{(EC-Val)}\quad
        &\step \stransp{D}{p}{\R}{\emptyset}
          [\bsq{V}\R/x] \semi \bsq{\heap}\R\\
        \text{(Lemma \ref{thm:trans-substbs}, substitution)}\quad
        &= \stransp{D[V/x]}{p}{\R}{\emptyset} \semi \bsq{\heap}\R
      \end{align*}}
  \end{itemize}

  \pcase{EB-StateI}: Impossible,
  since the step rule requires that $\phi=\mf{Init}$.

  \pcase{EB-StateS}: We have
  $\heap \semi \estate{x}{\ell}{V}{D}
  \step^\mf{Succ}_p D[W/x, \setter^\ell_p/x_\mf{set}]
  \semi \heap$ (from which conclusion \emph{(iii)} is immediate),
  where $W = \heap[p].\mf{st}[\ell]$, and
  \begin{mathpar}
  \inferrule*[right=TB-State] {
    \inferrule* {(1)} {
      \Gamma \semi \Pi \semi \Phi_1 \vdash V: A
    }
    \and
    \inferrule* {(2)} {
      \Gamma, x: A, x_\mf{set}: A \xrightarrow{\mf{Mut}} \tone \semi \Pi
      \semi \Phi_2 \vdash D: \Ell
    }
    \and
    \ell \notin \dom(\Ell)
  } {
    \Gamma \semi \Pi \semi \Phi_1 \Join \Phi_2 \vdash
    \estate{x}{\ell}{V}{D}: \Ell, \ell: A
  }
  \end{mathpar}

  \pcase{EB-StateS}: We have
  $\heap \semi \estate{x}{\ell}{V}{D}
  \step^\mf{Succ}_p D[W/x, \setter^\ell_p/x_\mf{set}]
  \semi \heap$ (from which conclusion \emph{(iii)} is immediate),
  where $W = \heap[p].\mf{st}[\ell]$, and
  \begin{mathpar}
  \inferrule*[right=TB-State] {
    \inferrule* {(1)} {
      \Gamma \semi \Pi \semi \Phi_1 \vdash V: A
    }
    \and
    \inferrule* {(2)} {
      \Gamma, x: A, x_\mf{set}: A \xrightarrow{\mf{Mut}} \tone \semi \Pi
      \semi \Phi_2 \vdash D: \Ell
    }
    \and
    \ell \notin \dom(\Ell)
  } {
    \Gamma \semi \Pi \semi \Phi_1 \Join \Phi_2 \vdash
    \estate{x}{\ell}{V}{D}: \Ell, \ell: A
  }
  \end{mathpar}
  \begin{itemize}
    \item By premise \emph{(ii)},
    we know $\Ell, \ell: A \subseteq \Pi[p]$,
    i.e. (3) $\Pi[p][\ell]=A$.
    \item We derive the setter type as follows
    (noting $\mf{Succ}(\emptyset)$ holds trivially):
    \begin{mathpar}
      \inferrule*[right=TV-Setter] {
        \inferrule*[right=(3)] { } {\ell: A \in \Pi[p]}
        \and
        \mf{Succ}(\emptyset)
      } {
        (4) \quad \Gamma \semi \Pi \semi \emptyset \vdash \setter^\ell_p:
        A \xrightarrow{\mf{Mut}} \tone
      }
    \end{mathpar}
    \item By premise (i-b), we have
    that $\Gamma \semi \Pi \semi \Phi_W \vdash W: \Pi[p][\ell]$
    for some $\Phi_3$ with $\Phi_H = \Phi_W \Join \Phi_H'$.
    By (3), this is equivalent to (5) that
    $\Gamma \semi \Pi \semi \Phi_W \vdash W: A$.
    \item Choose $\Phi_D' = \Phi_D \Join \Phi_W$.
      We are given that $\Phi_D \Join (\Phi_W \Join \Phi_H')$
      is defined, so $\Phi_D \Join \Phi_W$ is also defined.
    \item By \ref{thm:lamhook-unrestv} (unrestriction) on (5),
      we have $\mf{Succ}(\Phi_W)$. It follows that
      $(\Phi_D \Join \Phi_W) \Join (\Phi_W \Join \Phi_H')
      = \Phi_D \Join (\Phi_W \Join \Phi_H')$, satisfying conclusion \emph{(ii)}.
    \item Lemma \ref{thm:lamhook-unrestv} (unrestriction) on (1)
      gives that $\mf{Succ}(\Phi_1)$.
      We then apply Lemma \ref{thm:lamhook-phiwknb} ($\Phi$-weakening)
      to (2), getting (6)
      $\Gamma, x: A, x_\mf{set}: A \xrightarrow{\mf{Mut}} \tone
        \semi \Pi \semi \Phi_1 \Join \Phi_2
        \vdash D[W/x,\setter^\ell_p/x_\mf{set}]: \Ell$.
    \item We apply Lemma \ref{thm:lamhook-substb} (substitution) twice to (6),
    using (4) and (5), to obtain that\\
    $\Gamma \semi \Pi \semi \Phi_D'
      \vdash D[W/x,\setter^\ell_p/x_\mf{set}]: \Ell$,
      satisfying conclusion \emph{(i)}.
    \item Conclusion \emph{(iv)} is obvious,
      since $\Ell \subseteq \Ell, \ell: A$.
    \blue{\item Conclusion \emph{(v)} is shown as follows:
      \begin{itemize}
        \item By premise \emph{(iii)} and definition, it suffices to show that\\
        $\mf{ids}(\heap) \cup \mf{ids}(D[W/x, \setter^\ell_p/x_\mf{set}])
        \subseteq \mf{ids}(W) \cup \mf{ids}(D) \cup \mf{ids}(\heap)$.
        \item In general, $\mf{ids}(D[V/x]) \subseteq \mf{ids}(D)
        \cup \mf{ids}(V)$. Therefore, it suffices to show that\\
        $\mf{ids}(\heap) \cup \mf{ids}(D) \cup \mf{ids}(W)
        \cup \mf{ids}(\setter^\ell_p/x_\mf{set})
        \subseteq \mf{ids}(V) \cup \mf{ids}(D) \cup \mf{ids}(\heap)$.
        \item By definition, $\mf{ids}(\setter^\ell_p) = \emptyset$
        and $\mf{ids}(W) \subseteq \mf{ids}(\heap)$. We thus get that\\
        $\mf{ids}(\heap) \cup \mf{ids}(D)
        \subseteq \mf{ids}(V) \cup \mf{ids}(D) \cup \mf{ids}(\heap)$,
        which clearly holds, and so we have shown conclusion \emph{(v)}.
      \end{itemize}
    }
    \blue{
      \item By definition of $\bsq{\heap}\R$, we know
      $(\bsq{\heap}\rho)[\R(p, \ell)] = \bsq{W}\R$.
      \item We show conclusion \emph{(vi)} as follows:
    \begin{align*}
      &\bsq{\heap}\R \semi \stransp{\estate{x}{\ell}{V}{D}}{p}{\R}{\emptyset}\\
      (\text{def. of translation}) \quad =\
      &\mf{let}\ x = \bang \R(p, \ell)\ \mf{in}\\
      &\mf{let}\ x_\mf{set} = \mf{val}\ \lambda y. (
        \R(p, \ell) := y \semi \R(p) := \mf{true})\ \mf{in}\\
      &\stransp{D}{p}{\R}{\emptyset} \semi \bsq{\heap}\R\\
      \textsc{(EC-Get, EC-Val)} \quad \stept\
      &(\mf{let}\ x_\mf{set} = \mf{val}\ \lambda y. (
        \R(p, \ell) := y \semi \R(p) := \mf{true})\ \mf{in}\\
      &\stransp{D}{p}{\R}{\emptyset})[\bsq{W}\R/x]
      \semi \bsq{\heap}\R\\
      (\textsc{EC-Val}\text{, def. of translation}) \quad \step\
      &\stransp{D}{p}{\R}{\emptyset}[\bsq{W}\R/x,
      \bsq{\setter^\ell_p}\R/x_\mf{set}]
      \semi \bsq{\heap}\R\\
      (\text{Lemma \ref{thm:trans-substbs}, substitution}) \quad =\
      &\stransp{D[W/x, \setter^\ell_p/x_\mf{set}]}{p}{\R}{\emptyset}
        \semi \bsq{\heap}\R
    \end{align*}
    }
  \end{itemize}
\end{proof}

\begin{restatable}[Simulation, bodies, initial]{theorem}{transsimbi}
\label{thm:trans-simbi}
If
\begin{enumerate}[label=(\roman*)]
\item $\Gamma \semi \Pi \semi \Phi_H \semi \Phi_D \vdash \heap \semi D: \Ell$,
\item $\Gamma \semi \Pi \semi \Phi_0 \vdash D_0: \Ell \uplus \Pi[p]$,
\item $p \notin \dom(\Phi_H \Join \Phi_D \Join \Phi_0)$,
\item \blue{$\Pi \semi \mf{ids}(D) \cup \mf{ids}(D_0)
  \cup \mf{ids}(\heap) \vDash \R$}, and
\item $\heap \semi D \step^\mf{Init}_p D' \semi \heap'$,
\end{enumerate}
then for some $\Phi_H'$, $\Pi'$, and $\R' \supseteq \R$, we have:
\begin{enumerate}[label=(\roman*)]
\item $\Gamma \semi \Pi' \semi \Phi_H' \semi \Phi_D
  \vdash \heap \semi D: \Ell'$,
\item $\Gamma \semi \Pi' \semi \Phi_0 \vdash D_0: \Ell' \uplus \Pi'[p]$,
\item $\Phi_H' \Join \Phi_D = \Phi_H \Join \Phi_D$,
\item $\Pi \sqsubseteq_{p: \mf{Init}} \Pi'$,
\item \blue{$\Pi' \semi \mf{ids}(D') \cup \mf{ids}(D_0)
  \cup \mf{ids}(\heap') \vDash \R'$}, and
\item \blue{$\bsq{\heap}\R \semi
  \itransp{D}{D_0}{p}{\R}{\emptyset}
  \stept\
  \itransp{D'}{D_0}{p}{\R'}{\emptyset}
  \semi \bsq{\heap'}{\R'}$}.
\end{enumerate}
\begin{proof}

  By definition of premise \emph{(i)}, we have subpremises
  \begin{enumerate}[label=(i-\alph*)]
    \item $\Gamma \semi \Pi \semi \Phi_D \vdash D: \Ell$ and
    \item $\Gamma \semi \Pi \semi \Phi_H \vDash \heap$.
  \end{enumerate}

  Likewise, to show conclusion \emph{(i)} it suffices to pick some
  $\Pi' \sqsupseteq_{p: \mf{Init}} \Pi$ and
  $\Phi_H'$ such that conclusion \emph{(iii)} holds and:
  \begin{enumerate}[label=(i-\alph*)]
    \item $\Gamma \semi \Pi' \semi \Phi_D \vdash D': \Ell'$ and
    \item $\Gamma \semi \Pi' \semi \Phi_H' \vDash \heap'$.
  \end{enumerate}

  \blue{We begin with a few general observations:}
  \blue{\begin{itemize}
    \item By premise \emph{(iv)}, we have
    (iv-a) that $\Pi \semi \mf{ids}(D) \cup \mf{ids}(D_0) \vDash \R$
    and (iv-b) that $\Pi \semi \mf{ids}(\heap) \vDash \R$.
    \item To show conclusion \emph{(vi)} in any case, it is necessary to
    show (vi-a) that $\bsq{\heap}\R$ is well-defined
    and (vi-b) that $\itransp{D}{D_0}{p}{\R}{\emptyset}$
    is well-defined:
    \begin{itemize}
      \item (vi-a) is given by Lemma \ref{thm:trans-prezh}
      (heap preservation)
      on premises (i-b) and (iv-b).
      \item (vi-b) is given by Lemma \ref{thm:trans-prezbi}
      (initial body preservation)
      on premises (i-a), \emph{(ii)}, and (iv-a),
      noting that $\emptyset \semi \Pi \vDash_p \emptyset:
      (\Ell \uplus \Pi[p]) \setminus \Ell$ by definition.
    \end{itemize}
  \end{itemize}}

  Unless otherwise mentioned, we choose $\Phi_H' = \Phi_H$, $\Pi' = \Pi$,
  \blue{and $\R' = \R$},
  such that:
  \begin{itemize}
    \item conclusion (i-b) is given by premise \emph{(i-b)}
      when $\heap' = \heap$,
    \item conclusion \emph{(ii)} is given by premise \emph{(ii)}
      when $\Ell' = \Ell$, and
    \item conclusions \emph{(iii)} and \emph{(iv)} are trivial.
  \end{itemize}

  We proceed by induction on the derivation of the
  body step judgment \emph{(iv)}.

  \pcase{EB-Let}: We have
  \begin{mathpar}
    \inferrule*[right=EB-Let] {
      \inferrule* {(1)} {\heap \semi M \step N \semi \heap'}
    } {
      \heap \semi \elet{x}{M}{D}
      \step^\mf{Init}_p \elet{x}{N}{D} \semi \heap'
    }

    \inferrule*[right=TB-Let] {
      \inferrule* {(2)} {
        \Gamma \semi \Pi \semi \Phi_1 \vdash M: A \bang \mf{Pure}
      }
      \and
      \inferrule* {(3)} {\Gamma, x: A \semi \Pi \semi \Phi_2 \vdash D: \Ell}
    } {
      \Gamma \semi \Pi \semi \Phi_1 \Join \Phi_2 \vdash \elet{x}{M}{D}: \Ell
    }
  \end{mathpar}
  \begin{itemize}
    \blue{\item By premise \emph{(iv)} and definition,
    we have (4) that $\Pi \semi \mf{ids}(M) \cup \mf{ids}(\heap) \vDash \R$
    and (5) that $\Pi \semi \mf{ids}(D) \cup \mf{ids}(\heap) \vDash \R$.}
    \item By Theorem \ref{thm:trans-simcp} (pure simulation) on (1), (2),
    \blue{(4)}, and premise (i-b), we get
    \begin{enumerate}
      \setcounter{enumi}{5}
      \item $\Gamma \semi \Pi \semi \Phi_1 \vdash N: A \bang \mf{Pure}$,
      \item $\heap' = \heap$, satisfying conclusion (i-b),
      \blue{
        \item $\Pi \semi \mf{ids}(N) \cup \mf{ids}(\heap) \vDash \R$, and
        \item $\bsq{\heap}\R \semi \bsq{M}\R
        \stept\ \bsq{N}\R \semi \bsq{\heap}\R$.}
    \end{enumerate}
    \blue{\item By (5) and (8), we get $\Pi \semi \mf{ids}(N)
    \cup \mf{ids}(D) \cup \mf{ids}(\heap) \vDash \R$, satisfying
    conclusion \emph{(v)}}.
    \item We derive conclusion (i-a) as follows:
    \begin{mathpar}
      \inferrule*[right=TB-Let] {
        \inferrule*[right=(4)] { } {
          \Gamma \semi \Pi \semi \Phi_1 \vdash N: A \bang \mf{Pure}
        }
        \and
        \inferrule* {(3)} {\Gamma, x: A \semi \Pi \semi \Phi_2 \vdash D: \Ell}
      } {
        \Gamma \semi \Pi \semi \Phi_1 \Join \Phi_2 \vdash \elet{x}{N}{D}: \Ell
      }
    \end{mathpar}
    \blue{\item We show conclusion \emph{(vi)} as follows:
      \begin{align*}
        &\bsq{\heap}\R \semi \itransp{\elet{x}{M}{D}}{D_0}{p}{\R}{\emptyset}\\
        (\text{def. of translation}) \quad
        &= \elet{x}{\bsq{M}\R}{\itransp{D}{D_0}{p}{\R}{\emptyset}}
          \semi \bsq{\heap}\R\\
        \text{(Lemma \ref{thm:trans-let}, let-transitivity on (9))} \quad
        &\stept\ \elet{x}{\bsq{N}\R}{
          \itransp{D}{D_0}{p}{\R}{\emptyset}} \semi \bsq{\heap}\R\\
        (\text{def. of translation}) \quad
        &= \itransp{\elet{x}{N}{D}}{D_0}{p}{\R}{\emptyset}
          \semi \bsq{\heap}\R
      \end{align*}}
  \end{itemize}

  \pcase{EB-Val}: We have
  $\heap \semi \elet{x}{\mf{val}\ V}{D}
  \step^\mf{Init}_p D[V/x] \semi \heap$
  (from which conclusion (i-b) is immediate) and
  \begin{mathpar}
  \inferrule*[right=TB-Let] {
    \inferrule*[right=TC-Val] {
      \inferrule* {(1)} {\Gamma \semi \Pi \semi \Phi_1 \vdash V: A}
    } {
      \Gamma \semi \Pi \semi \Phi_1 \vdash \mf{val}\ V: A \bang \mf{Pure}
    }
    \and
    \inferrule* {(2)} {\Gamma, x: A \semi \Pi \semi \Phi_2 \vdash D: \Ell}
  } {
    \Gamma \semi \Pi \semi \Phi_1 \Join \Phi_2 \vdash
    \elet{x}{\mf{val}\ V}{D}: \Ell
  }
  \end{mathpar}
  \begin{itemize}
    \item Conclusion (i-a) is given by Lemma \ref{thm:lamhook-substb}
    (type substitution) on (1) and (2).
    \blue{
      \item Conclusion \emph{(v)} is given by premise \emph{(iv)}
      and the fact that
      $\mf{ids}(D[V/x]) \subseteq \mf{ids}(D) \cup \mf{ids}(V)$.
      \item We show conclusion \emph{(vi)} as follows:
      \begin{align*}
        \bsq{\heap}\R \semi \itransp{\elet{x}{\mf{val}\ V}{D}}
          {D_0}{p}{\R}{\emptyset}
        &= \elet{x}{\mf{val}\ \bsq{V}\R}{\itransp{D}{D_0}{p}{\R}{\emptyset}}
          \semi \bsq{\heap}\R\\
        \textsc{(EC-Val)}\quad
        &\step \itransp{D}{D_0}{p}{\R}{\emptyset}
          [\bsq{V}\R/x] \semi \bsq{\heap}\R\\
        \text{(Lemma \ref{thm:trans-substbi}, substitution)}\quad
        &= \itransp{D[V/x]}{D_0[V/x]}{p}{\R}{\emptyset}
          \semi \bsq{\heap}\R\\
        (\text{Convention: } x \notin \mf{fv}(D_0))\quad
        &= \itransp{D[V/x]}{D_0}{p}{\R}{\emptyset}
          \semi \bsq{\heap}\R
      \end{align*}}
  \end{itemize}

  \pcase{EB-StateI}: We have,
  for $\heap' = \heap \mid \heap[p].\mf{st}[\ell] \mapsto V$,
  \begin{gather*}
    \heap \semi \estate{x}{\ell}{V}{D}
    \step^\mf{Init}_p D[V/x, \setter^\ell_p/x_\mf{set}]
    \semi \heap'
  \end{gather*}
  \begin{mathpar}
  \inferrule*[right=TB-State] {
    \inferrule* {(1)} {\Gamma \semi \Pi \semi \emptyset \vdash V: A}
    \and
    \inferrule* {(2)} {
      \Gamma, x: A, x_\mf{set}: A \xrightarrow{\mf{Mut}} \tone \semi \Pi
      \vdash D: \Ell
    }
    \and
    \inferrule* {(3)} {\ell \notin \dom(\Ell)}
  } {
    \Gamma \semi \Pi \vdash \estate{x}{\ell}{V}{D}: \Ell, \ell: A
  }
  \end{mathpar}
  \begin{itemize}
    \item Choose $\Phi_H' = \Phi_H \Join \Phi_1$,
    $\Pi' = \Pi \mid \Pi[p][\ell] \mapsto A$,
    \blue{and $\R' = \R, (p, \ell) \mapsto \ell'$,
    where $\ell'$ is the location generated fresh by the target language.}
    \item By definition, $\Pi \sqsubseteq_{p: \mf{Init}} \Pi'$,
    satisfying conclusion \emph{(iv)}.
    \item We show conclusion \emph{(iii)} as follows:
    \begin{itemize}
      \item Since $(\Phi_1 \Join \Phi_2) \Join \Phi_H$ is given as defined,
      we know $\Phi_1 \Join \Phi_H$ is defined.
      \item To show conclusion \emph{(iii)},
      we want to show $(\Phi_1 \Join \Phi_2) \Join (\Phi_H \Join \Phi_1)
      = (\Phi_1 \Join \Phi_2) \Join \Phi_H$.
      \item These are equal by associativity/commutativity of $\Join$
      and by Lemma \ref{thm:lamhook-unrestv} (unrestriction)
      on (1), meaning $\Phi_1 \Join \Phi_1 = \Phi_1$.
    \end{itemize}
    \item We show conclusion (i-a) as follows:
    \begin{itemize}
      \item By premise \emph{(iii)}, we have (5)
        that $p: \mf{Init} \Join \Phi_1$ is defined
        and (6) that $p: \mf{Init} \Join \Phi_2$ is defined.
      \item By Lemma \ref{thm:lamhook-piwknv} ($\Pi$-weakening)
      on (1), (4), and (5),
      we get (7) that $\Gamma \semi \Pi' \semi \Phi_1 \vdash V: A$.
      \item By Lemma \ref{thm:lamhook-piwknb} ($\Pi$-weakening)
      on (2), (4), and (6), we get (8) that\\$\Gamma, x: A, x_\mf{set}:
      A \xrightarrow{\mf{Mut}} \tone \semi \Pi' \semi \Phi_2 \vdash D: \Ell$.
      \item Since $\ell: A \in \Pi'[p]$ by definition and $\mf{Succ}(\emptyset)$
      trivially, we derive the setter type as follows:
      \begin{mathpar}
      \inferrule*[right=TV-Setter] {
        \ell: A \in \Pi'[p]
        \and
        \mf{Succ}(\emptyset)
      } {
        (9) \quad \Gamma \semi \Pi' \semi \emptyset \vdash \setter^\ell_p:
        A \xrightarrow{\mf{Mut}} \tone
      }
      \end{mathpar}
      \item We apply Lemma \ref{thm:lamhook-substb} (substitution) twice to (8),
      using (7) and (9), to get that\\
      $\Gamma \semi \Pi' \semi \Phi_1 \Join \Phi_2 \vdash
      D[V/x, \setter^\ell_p/x_\mf{set}]: \Ell$,
      i.e. conclusion (i-a) holds.
    \end{itemize}
    \item We show conclusion (i-b) as follows:
    \begin{itemize}
      \item We first want (10) that $\dom(\Pi') = \dom(\heap')$.
      \begin{itemize}
        \item By definition of each, it suffices to show that
        $\dom(\Pi) \cup \{p\} = \dom(\heap) \cup \{p\}$.
        \item Premise (i-b) says that $\dom(\Pi) = \dom(\heap)$,
        so (10) holds.
      \end{itemize}
      \item We also want (11) that
      $\dom(\Pi'[p]) = \dom(\heap'[p].\mf{st})$.
      \begin{itemize}
        \item By definition of each, it suffices to show
        $\dom(\Pi[p]) \cup \{\ell\} = \dom(\heap[p].\mf{st}) \cup \{\ell\}$.
        \item Premise (i-b) says that
        $\dom(\Pi[p]) = \dom(\heap[p].\mf{st})$,
        so (11) also holds.
      \end{itemize}
      \item Finally, we want (12) that
      $\Gamma \semi \Pi' \semi \Phi_1
      \vdash \heap'[p].\mf{st}[\ell]: \Pi'[p][\ell]$. This is given by (7).
      \item By premise (i-b) and by definition,
      (10)-(12) are sufficient to show that\\
      $\Gamma \semi \Pi' \semi \Phi_H \Join \Phi_1 \vDash \heap'$,
      i.e. conclusion (i-b) holds.
    \end{itemize}
    \item We show conclusion \emph{(ii)} as follows:
    \begin{itemize}
      \item By Lemma \ref{thm:lamhook-piwknb} on premise \emph{(ii)},
        using the fact that $p \notin \dom(\Phi_0)$ from premise \emph{(iii)},
        we get that $\Gamma \semi \Pi' \semi \Phi_0 \vdash D_0:
          \Pi[p] \uplus \Ell, \ell: A$.
      \item We want to show that
      $\Pi'[p] \uplus \Ell = \Pi[p] \uplus \Ell, \ell: A$.
      \item By definition, it suffices to show that
      $\Pi[p], \ell: A \uplus \Ell = \Pi[p] \uplus \Ell, \ell: A$.
      \item It further suffices to show $\Pi[p], \ell: A \uplus \Ell$
      is well-defined, i.e $\Pi[p], \ell: A \cap \Ell = \emptyset$.
      This follows from (3) and premise \emph{(ii)}.
    \end{itemize}
    \blue{\item We show conclusion \emph{(v)} as follows:
    \begin{itemize}
      \item Given premise \emph{(iv)} and definition of $\Pi'$ and $\R'$,
      we first need to show that $\R'(p, \ell)$ is well-defined;
      this is given by definition of $\R'$.
      \item Then, it suffices to show that
      $\mf{ids}(D[V/x, \setter^\ell_p/x_\mf{set}]) \cup \mf{ids}(\heap')
      \subseteq
      \mf{ids}(D) \cup \mf{ids}(V) \cup \mf{ids}(\heap)$.
      \item By definition of $\heap'$, we have
      $\mf{ids}(\heap') = \mf{ids}(\heap) \cup \mf{ids}(V)$,
      and by definition we have $\mf{ids}(\setter^\ell_p) = \emptyset$.
      The conclusion follows from the fact that
      $\mf{ids}(D[V/x]) \subseteq \mf{ids}(D) \cup \mf{ids}(V)$.
    \end{itemize}}
    \blue{\item Finally, we show conclusion \emph{(vi)} as follows:
    \begin{itemize}
      \item By Lemmas \ref{thm:trans-rextm}/\ref{thm:trans-rextv}
      ($\R$-extension) and definition of translation, we have (13) that\\
      $\bsq{\heap}\R, \ell' \mapsto \bsq{V}\R =
      \bsq{\heap}{\R'}, \ell' \mapsto \bsq{V}{\R'}$.
      \item Since $\R'(p, \ell) = \ell'$, by definition of translation
      we have that\\$\bsq{\heap}{\R'}, \ell' \mapsto \bsq{V}{\R'} =
      \bsq{\heap \mid \heap[p].\mf{st}[\ell] \mapsto V}{\R'}$.
      Then from (13) we get (14) that
      $\bsq{\heap}\R, \ell' \mapsto \bsq{V}\R = \bsq{\heap'}{\R'}$.
      \item By Lemma \ref{thm:trans-rextv} ($\R$-extension) and definition
      of $\R'$ and translation, we have (15)\\
      $\lambda y.(\ell' := y \semi \R(p) := \mf{true})
      = \bsq{\setter^\ell_p}{\R'}$.
      \item By Lemma \ref{thm:trans-rextbi} ($\R$-extension),
      we have (16) $\itransp{D}{D_0}{p}{\R}{\delta} =
      \itransp{D}{D_0}{p}{\R'}{\delta}$.
      \item By Lemma \ref{thm:trans-lelimi} ($\delta$-elimination), we get that
      \\
      $\itransp{D}{D_0}{p}{\R'}{\ell \mapsto x_\ell}[\R'(p, \ell)/x_\ell] =
      \itransp{D}{D_0}{p}{\R}{\emptyset}
        [\R'(p, \ell)/x_\ell]$.
      Since $x_\ell$ is translation-fresh, the above reduces to (17)\\
      $\itransp{D}{D_0}{p}{\R}{\ell \mapsto x_\ell}[\R'(p, \ell)/x_\ell] =
      \itransp{D}{D_0}{p}{\R'}{\emptyset}$.
      \item By Lemma \ref{thm:trans-substbi} (substitution), we have\\
      $\itransp{D}{D_0}{p}{\R'}{\emptyset}[\bsq{V}{\R'}/x]
      = \itransp{D[V/x]}{D_0[V/x]}{p}{\R'}{\emptyset}$.
      By convention, $x \notin \mf{fv}(D_0)$,
      so $D_0[V/x] = D_0$.
      Applying this principle twice, we get (18)\\
      $\itransp{D}{D_0}{p}{\R'}{\emptyset}
        [\bsq{V}{\R'}/x, \bsq{\setter^\ell_p}{\R'}/x_\mf{set}]
      = \itransp{D[V/x, \setter^\ell/p]}{D_0}{p}{\R'}{\emptyset}$.
      \item We show the conclusion as follows:
      \begin{align*}
        &\bsq{\heap}\R \semi \itransp{\mf{let}\ x, x_\mf{set} =\
          \mf{useState}^\ell\ V\ \mf{in}\ D}
          {D_0}{p}{\R}{\emptyset}\\
        (\text{def. of translation}) \quad
        =\ &\mf{let}\ x_\ell = \mf{ref}(\bsq{V}\R)\ \mf{in}\\
        &\mf{let}\ x = \bang x_\ell\ \mf{in}\\
        &\mf{let}\ x_\mf{set} = \mf{val}\ \lambda y.
          (x_\ell := y \semi \R(p) := \mf{true})\ \mf{in}\\
        &\itransp{D}{D_0}{p}{\R}{\ell \mapsto x_\ell}
        \semi \bsq{\heap}\R\\
        (\textsc{EC-Ref, EC-Val}) \quad \stept\
        &\mf{let}\ x = \bang \ell'\ \mf{in}\\
        (\ell'\ \text{fresh}) \qquad\ \ \
        &\mf{let}\ x_\mf{set} = \mf{val}\ \lambda y.
          (\ell' := y \semi \R(p) := \mf{true})\ \mf{in}\\
        &\itransp{D}{D_0}{p}{\R}{\ell \mapsto x_\ell}
        [\ell'/x_\ell]
        \semi \bsq{\heap}\R, \ell' \mapsto \bsq{V}\R\\
        (\textsc{EC-Get, EC-Val}) \quad \stept\
        &\itransp{D}{D_0}{p}{\R}{\ell \mapsto x_\ell}
        [\ell'/x_\ell, \bsq{V}\R/x, \\
        &\quad\lambda y. (\ell' := y \semi \R(p) := \mf{true})/x_\mf{set}]
        \semi \bsq{\heap}\R, \ell' \mapsto \bsq{V}\R\\
        \text{(14), (15), and (16)} \quad =\
        &\itransp{D}{D_0}{p}{\R'}{\ell \mapsto x_\ell}
        [\R'(p, \ell)/x_\ell, \bsq{V}{\R'}/x,
        \bsq{\setter^\ell_p}{\R'}/x_\mf{set}] \semi \bsq{\heap'}{\R'}\\
        (17) \quad =\
        &\itransp{D}{D_0}{p}{\R'}{\emptyset}
        [\bsq{V}{\R'}/x, \bsq{\setter^\ell_p}{\R'}/x_\mf{set}]
        \semi \bsq{\heap'}{\R'}\\
        (18) \quad =\
        &\itransp{D[V/x, \setter^\ell_p/x_\mf{set}]}{D_0}{p}{\R'}{\emptyset}
        \semi \bsq{\heap'}{\R'}\\
      \end{align*}
    \end{itemize}}
  \end{itemize}

  \pcase{EB-StateS}: Impossible,
  since the step rule requires that $\phi=\mf{Succ}$.

\end{proof}
\end{restatable}

\transsimcm*

\begin{proof}
  By definition of premise \emph{(i)}, we have subpremises
  \begin{enumerate}[label=(i-\alph*)]
    \item $\Gamma \semi \Pi \semi \Phi_M \vdash M: A \bang \mf{Mut}$ and
    \item $\Gamma \semi \Pi \semi \Phi_H \vDash \heap$.
  \end{enumerate}
  Likewise, to show conclusion \emph{(iii)}, it suffices to show
  for our chosen $\Phi_N$, $\Phi_H'$, and $\Pi'$ that
  \begin{enumerate}[label=(iii-\alph*)]
    \item $\Gamma \semi \Pi' \semi \Phi_N \vdash N: A \bang \mf{Mut}$ and
    \item $\Gamma \semi \Pi' \semi \Phi_H' \vDash \heap'$.
  \end{enumerate}

  Unless otherwise mentioned, we choose $\Phi_N = \Phi_M$, $\Phi_H' = \Phi_H$,
  and $\Pi' = \Pi$, such that
  are trivial, and
  \begin{itemize}
    \item conclusions \emph{(i)} and \emph{(ii)} are trivial,
    \item conclusion (iii-b) is trivial by premise (i-b) when $\heap' = \heap$,
      and
    \blue{\item to show conclusion \emph{(iv)} given premise \emph{(ii)},
      it suffices to show $\mf{ids}(N) \cup \mf{ids}(\heap')
      \subseteq \mf{ids}(M) \cup \mf{ids}(\heap)$.}
  \end{itemize}

  \blue{We begin by establishing translation well-definedness:}
  \blue{\begin{itemize}
    \item By premise \emph{(ii)}, we have
    (ii-a) that $\Pi \semi \mf{ids}(M) \vDash \R$
    and (ii-b) that $\Pi \semi \mf{ids}(\heap) \vDash \R$.
    \item To show conclusion \emph{(iv)} in any case, it is necessary to
    show (iv-a) that $\bsq{\heap}\R$ is well-defined
    and (iv-b) that $\bsq{M}\R$ is well-defined:
    \begin{itemize}
      \item (iv-a) is given by Lemma \ref{thm:trans-prezh}
      (heap preservation)
      on premises (i-b) and (ii-b).
      \item (iv-b) is given by Lemma \ref{thm:trans-prezc}
      (computation preservation)
      on premises (i-a) and (ii-a).
    \end{itemize}
  \end{itemize}}

  We proceed by induction on the derivation of the
  computation step judgment \emph{(iii)}.

  \pcase{EC-Let}: We have
  \begin{mathpar}
    \inferrule*[right=EC-Let] {
      \inferrule* {(1)} {\heap \semi M \step M' \semi \heap'}
    } {
      \heap \semi \elet{x}{M}{N} \step \elet{x}{M'}{N} \semi \heap'
    }

    \inferrule*[right=TC-Let] {
      \inferrule* {(2)} {
        \Gamma \semi \Pi \semi \Phi_1 \vdash M: A \bang \mf{Mut}
      }
      \and
      \inferrule* {(3)} {
        \Gamma, x: A \semi \Pi \semi \Phi_2 \vdash N: B \bang \mf{Mut}
      }
    } {
      \Gamma \semi \Pi \semi \Phi_1 \Join \Phi_2 \vdash \elet{x}{M}{N}:
      B \bang \mf{Mut}
    }
  \end{mathpar}
  \begin{itemize}
    \blue{\item By premise \emph{(iii)} and definition,
    we have (4) that $\Pi \semi \mf{ids}(M) \cup \mf{ids}(\heap) \vDash \R$,
    and (5) that $\Pi \semi \mf{ids}(N) \cup \mf{ids}(\heap) \vDash \R$.}
    \item By inductive hypothesis on (1), (2), \blue{(4)}, and premise (i-b),
    we get some $\Phi_1'$, $\Phi_H'$, and $\Pi'$ such that:
    \begin{enumerate}
      \setcounter{enumi}{5}
      \item $\Phi_1 \Join \Phi_H \sqsubseteq_\mf{fresh} \Phi_1' \Join \Phi_H'$,
      \item $\Pi \sqsubseteq_{\Phi_1' \Join \Phi_H'} \Pi'$,
      \item $\Gamma \semi \Pi' \semi \Phi_1' \vdash M': A \bang \mf{Mut}$,
      \item $\Gamma \semi \Pi' \semi \Phi_H' \vDash \heap'$,
      \blue{
        \item $\Pi' \semi \mf{ids}(M') \cup \mf{ids}(\heap') \vDash \R'$, and
        \item $\bsq{\heap}\R \semi \bsq{M}\R \stept\
        \bsq{M'}{\R'} \semi \bsq{\heap'}{\R'}$.
      }
    \end{enumerate}
    \item Choose $\Phi_N = \Phi_1' \Join \Phi_2$ and this $\Phi_H'$ as our own.
    \item By Lemma \ref{thm:lamhook-freshphi} (fresh $\Phi$-extension) using (6)
    and extending by $\Phi_2$, we get that $\Phi_N \Join \Phi_H'$
    is defined and
    $\Phi_M \Join \Phi_H \sqsubseteq_{\mf{fresh}} \Phi_N \Join \Phi_H'$,
    satisfying conclusion \emph{(i)}.
    \item Choose this $\Pi'$ as our own.
    By Lemma \ref{thm:lamhook-piprez} ($\Pi$-order preservation) using (7)
    and extending by $\Phi_2$,
    we get that $\Pi \sqsubseteq_{\Phi_N \Join \Phi_H} \Pi'$,
    satisfying conclusion \emph{(ii)}.
    \blue{\item Choose this $\R'$ as our own. By (10) and (5), we have
    $\Pi' \semi \mf{ids}(\elet{x}{M'}{N}) \cup \mf{ids}(\heap') \vDash \R'$,
    satisfying conclusion \emph{(iv)}.}
    \item (9) gives us conclusion (iii-b).
    \item By Lemma \ref{thm:lamhook-piwknc} ($\Pi$-weakening)
    on (3) using (7) and the fact that $\Phi_N \Join \Phi_H$ is defined,
    we get (12) that $\Gamma, x: A \semi \Pi' \semi \Phi_2 \vdash N:
    B \bang \mf{Pure}$.
    \item We derive conclusion (iii-a) as follows:
    \begin{mathpar}
      \inferrule*[right=TC-Let] {
        \inferrule*[right=(6)] { } {
          \Gamma \semi \Pi' \semi \Phi_1' \vdash M': A \bang \mf{Pure}
        }
        \and
        \inferrule*[right=(8)] { } {
          \Gamma, x: A \semi \Pi' \semi \Phi_2 \vdash N: B \bang \mf{Pure}
        }
      } {
        \Gamma \semi \Pi \semi \Phi_1' \Join \Phi_2 \vdash \elet{x}{M'}{N}:
        B \bang \mf{Pure}
      }
    \end{mathpar}
    \blue{\item We show conclusion \emph{(v)} as follows:
      \begin{align*}
        \bsq{\heap}\R \semi &\bsq{\elet{x}{M}{N}}\R\\
        (\text{def. of translation}) \quad
        =\ &\elet{x}{\bsq{M}\R}{\bsq{N}\R} \semi \bsq{\heap}\R\\
        (\text{Lemma \ref{thm:trans-let}, let-transitivity on (11)}) \quad
        \stept\ &\elet{x}{\bsq{M'}{\R'}}{\bsq{N}\R}
        \semi \bsq{\heap'}{\R'}\\
        (\text{Lemma \ref{thm:trans-rextc}, $\R$-extension}) \quad
        =\ &\elet{x}{\bsq{M'}{\R'}}{\bsq{N}{\R'}} \semi \bsq{\heap'}{\R'}\\
        (\text{def. of translation}) \quad
        =\ &\bsq{\elet{x}{M'}{N}}{\R'} \semi \bsq{\heap'}{\R'}
      \end{align*}
    }
  \end{itemize}

  \pcase{EC-Val}: We have
  $\heap \semi \elet{x}{\mf{val}\ V}{M} \step M[V/x] \semi \heap$ and
  \begin{mathpar}
    \inferrule*[right=TC-Let] {
      \inferrule*[right=TC-Sub] {
        \inferrule*[right=TC-Val] {
          \inferrule* {(1)} {
            \Gamma \semi \Pi \semi \Phi_1 \vdash V: A
          }
        } {
          \Gamma \semi \Pi \semi \Phi_1 \vdash \mf{val}\ V: A \bang \mf{Pure}
        }
      } {
        \Gamma \semi \Pi \semi \Phi_1 \vdash \mf{val}\ V: A \bang \mf{Mut}
      }
      \and
      \inferrule* {(2)} {
        \Gamma \semi \Pi \semi \Phi_2 \vdash M: B \bang \mf{Mut}
      }
    } {
      \Gamma \semi \Pi \semi \Phi_1 \Join \Phi_2 \vdash
      \elet{x}{\mf{val}\ V}{M}: B \bang \mf{Mut}
    }
  \end{mathpar}
  \begin{itemize}
    \blue{\item Conclusion \emph{(iv)} follows from premise \emph{(ii)} and
    the fact that $\mf{ids}(M[V/x]) \subseteq \mf{ids}(M) \cup \mf{ids}(V)$.}
    \item Conclusion (iii-a) follows from Lemma \ref{thm:lamhook-substc}
    (type substitution) on (1) and (2).
    \blue{\item We show conclusion \emph{(v)} as follows:
    \begin{align*}
      \bsq{\heap}\R \semi &\bsq{\elet{x}{\mf{val}\ V}{M}}\R\\
      (\text{def. of translation}) \quad =\
      &\elet{x}{\mf{val}\ \bsq{V}\R}{\bsq{M}\R} \semi \bsq{\heap}\R\\
      (\textsc{EC-Val})\quad \step\
      &\bsq{M}\R[\bsq{V}\R/x] \semi \bsq{\heap}\R\\
      (\text{Lemma \ref{thm:trans-substc}, substitution})\quad =\
      &\bsq{M[V/x]}\R \semi \bsq{\heap}\R
    \end{align*}}
  \end{itemize}

  \pcase{EC-App}: We have
  $\heap \semi (\lambda x.M)\ V \step M[V/x] \semi \heap$
  and one of either of the following derivations by premise \emph{(i)}:
  \begin{mathpar}
    \inferrule*[right=TC-App] {
      \inferrule*[right=TV-Lam] {
        \inferrule* {\text{(1a)}} {
          \Gamma, x: A \semi \Pi \semi \Phi_1 \vdash M: B \bang \mf{Mut}
        }
        \and
        \mf{Succ}(\Phi_1)
      } {
        \Gamma \semi \Pi \semi \Phi_1 \vdash \lambda x.M:
        A \xrightarrow{\mf{Mut}} B
      }
      \and
      \inferrule* {(2)} {
        \Gamma \semi \Pi \semi \Phi_2 \vdash V: A
      }
    } {
      \Gamma \semi \Pi \semi \Phi_1 \Join \Phi_2
      \vdash (\lambda x.M)\ V: B \bang \mf{Mut}
    }

    \inferrule*[right=TC-Sub] {
      \inferrule*[right=TC-App] {
        \inferrule*[right=TV-Lam] {
          \inferrule* {\text{(1b)}} {
            \Gamma, x: A \semi \Pi \semi \Phi_1 \vdash M: B \bang \mf{Pure}
          }
          \and
          \mf{Succ}(\Phi_1)
        } {
          \Gamma \semi \Pi \semi \Phi_1 \vdash \lambda x.M:
          A \xrightarrow{\mf{Pure}} B
        }
        \and
        \inferrule* {(2)} {
          \Gamma \semi \Pi \semi \Phi_2 \vdash v: A
        }
      } {
        \Gamma \semi \Pi \semi \Phi_1 \Join \Phi_2
        \vdash (\lambda x.M)\ V: B \bang \mf{Pure}
      }
    } {
      \Gamma \semi \Pi \semi \Phi_1 \Join \Phi_2
      \vdash (\lambda x.M)\ V: B \bang \mf{Mut}
    }
  \end{mathpar}
  \begin{itemize}
    \item \textsc{TC-Sub} on (1b) gives us (1a), so it suffices
    to show the conclusions given (1a).
    \blue{\item Conclusion \emph{(iv)} follows from premise \emph{(ii)} and
    the fact that $\mf{ids}(M[V/x]) \subseteq \mf{ids}(M) \cup \mf{ids}(V)$.}
    \item Conclusion (iii-a) follows from Lemma \ref{thm:lamhook-substc}
    (type substitution) on (1a) and (2).
    \blue{\item We show conclusion \emph{(v)} as follows:
    \begin{align*}
      \bsq{\heap}\R \semi &\bsq{(\lambda x.M)\ V}\R\\
      (\text{def. of translation}) \quad =\
      &(\lambda x.\bsq{M}\R)\ \bsq{V}\R \semi \bsq{\heap}\R\\
      (\textsc{EC-App})\quad \step\
      &\bsq{M}\R[\bsq{V}\R/x] \semi \bsq{\heap}\R\\
      (\text{Lemma \ref{thm:trans-substc}, substitution})\quad =\
      &\bsq{M[V/x]}\R \semi \bsq{\heap}\R
    \end{align*}}
  \end{itemize}

  \pcase{EC-IfT}: We have
  $\heap \semi \eif{\mf{true}}{M}{N} \step M \semi \heap$
  and one of either of the following derivations by premise \emph{(i)}:
  \begin{mathpar}
    \inferrule*[right=TC-If] {
      \cdots \and
      \inferrule* {\text{(1a)}} {
        \Gamma \semi \Pi \semi \Phi_2 \vdash M: A \bang \mf{Mut}
      }
    } {
      \Gamma \semi \Pi \semi \Phi_1 \Join \Phi_2 \vdash \eif{\mf{true}}{M}{N}:
      A \bang \mf{Mut}
    }

    \inferrule*[right=TC-Sub] {
      \inferrule*[right=TC-If] {
        \cdots \and
        \inferrule* {\text{(1b)}} {
          \Gamma \semi \Pi \semi \Phi_2 \vdash M: A \bang \mf{Pure}
        }
      } {
        \Gamma \semi \Pi \semi \Phi_1 \Join \Phi_2 \vdash \eif{\mf{true}}{M}{N}:
        A \bang \mf{Pure}
      }
    } {
      \Gamma \semi \Pi \semi \Phi_1 \Join \Phi_2 \vdash \eif{\mf{true}}{M}{N}:
      A \bang \mf{Mut}
    }
  \end{mathpar}
  \begin{itemize}
    \item \textsc{TC-Sub} on (1a) gives us (1b), so it suffices
    to show the conclusions given (1b).
    \item Conclusion (iii-a) is given by Lemma \ref{thm:lamhook-phiwknc}
    ($\Phi$-weakening) on (1b).
    \blue{
      \item Conclusion \emph{(iv)} follows from premise \emph{(ii)}
      and the fact that $\mf{ids}(M) \subseteq
      \mf{ids}(M) \cup \mf{ids}(N)$.
      \item We show conclusion \emph{(v)} as follows:
      \begin{align*}
        \bsq{\heap}\R \semi &\bsq{\eif{\mf{true}}{M}{N}}\R \\
        (\text{def. of translation}) \quad =\
        &\eif{\mf{true}}{\bsq{M}\R}{\bsq{N}\R} \semi \bsq{\heap}\R\\
        (\textsc{EC-IfT})\quad \step\
        &\bsq{M}\R \semi \bsq{\heap}\R
      \end{align*}
    }
  \end{itemize}

  \pcase{EC-IfF} is shown analogously to the above case.

  \pcase{EC-Fst}: We have
  $\heap \semi \mf{fst}\ (V, W) \step \mf{val}\ V \semi \heap$ and
  \begin{mathpar}
    \inferrule*[right=TC-Sub] {
      \inferrule*[right=TC-Fst] {
        \inferrule*[right=TV-Pair] {
          \inferrule* {(1)} {
            \Gamma \semi \Pi \semi \Phi_1 \vdash V: A
          }
          \and
          \inferrule* {(2)} {
            \Gamma \semi \Pi \semi \Phi_2 \vdash W: B
          }
        } {
          \Gamma \semi \Pi \semi \Phi_1 \Join \Phi_2 \vdash
          (V, W): A \times B
        }
      } {
        \Gamma \semi \Pi \semi \Phi_1 \Join \Phi_2 \vdash
        \mf{fst}\ (V, W): A \bang \mf{Pure}
      }
    } {
      \Gamma \semi \Pi \semi \Phi_1 \Join \Phi_2 \vdash
      \mf{fst}\ (V, W): A \bang \mf{Mut}
    }
  \end{mathpar}
  \begin{itemize}
    \item By Lemma \ref{thm:lamhook-unrestv} (unrestriction) on (2),
    we get (3) that $\mf{Succ}(\Phi_2)$.
    \item By Lemma \ref{thm:lamhook-phiwknv} ($\Phi$-weakening) on (1) and (3),
    we get (4) that
    $\Gamma \semi \Pi \semi \Phi_1 \Join \Phi_2 \vdash V: A$.
    \item We derive conclusion (iii-a) as follows:
    \begin{mathpar}
      \inferrule*[right=TC-Sub] {
        \inferrule*[right=TC-Val] {
          \inferrule*[right=(4)] { } {
            \Gamma \semi \Pi \semi \Phi_1 \Join \Phi_2 \vdash V: A
          }
        } {
          \Gamma \semi \Pi \semi \Phi_1 \Join \Phi_2 \vdash \mf{val}\ V:
          A \bang \mf{Pure}
        }
      } {
        \Gamma \semi \Pi \semi \Phi_1 \Join \Phi_2 \vdash \mf{val}\ V:
        A \bang \mf{Mut}
      }
    \end{mathpar}
    \blue{
      \item Conclusion \emph{(iv)} follows from premise \emph{(ii)}
      and the fact that $\mf{ids}(V) \subseteq
      \mf{ids}(V) \cup \mf{ids}(W)$.
      \item We show conclusion \emph{(v)} as follows:
      \begin{align*}
        \bsq{\heap}\R \semi &\bsq{\mf{fst}\ (V, W)}\R\\
        (\text{def. of translation}) \quad =\
        &\mf{fst}\ (\bsq{V}\R, \bsq{W}\R) \semi \bsq{\heap}\R\\
        (\textsc{EC-Fst})\quad \step\
        &\mf{val}\ \bsq{V}\R \semi \bsq{\heap}\R\\
        (\text{def. of translation})\quad =\
        &\bsq{\mf{val}\ V}\R \semi \bsq{\heap}\R
      \end{align*}
    }
  \end{itemize}

  \pcase{EC-Snd} is shown analogously to the above case.

  \pcase{EC-FdN}: We have
  $\heap \semi \mf{fold}\ V_1\ V_2\ [] \step \mf{val}\ V_2 \semi \heap$ and
  \begin{mathpar}
    \inferrule*[right=TC-Fold] {
      \inferrule* {(1a)} {
        \Gamma \semi \Pi \semi \Phi_1 \vdash V_1:
        A \xrightarrow{\mf{Mut}} B \xrightarrow{\mf{Mut}} B
      }
      \and
      \inferrule* {(2)} {
        \Gamma \semi \Pi \semi \Phi_2 \vdash V_2: B
      }
      \and
      \inferrule* {(3)} {
        \Gamma \semi \Pi \semi \Phi_3 \vdash []: \mf{List}\ A
      }
    } {
      \Gamma \semi \Pi \semi \Phi_1 \Join \Phi_2 \Join \Phi_3 \vdash
      \mf{fold}\ V_1\ V_2\ []: B \bang \mf{Pure}
    }

    \inferrule*[right=TC-Sub] {
      \inferrule*[Right=TC-Fold] {
        \inferrule* {(1b)} {
          \Gamma \semi \Pi \semi \Phi_1 \vdash V_1:
          A \xrightarrow{\mf{Pure}} B \xrightarrow{\mf{Pure}} B
        }
        \and
        \inferrule* {(2)} {
          \Gamma \semi \Pi \semi \Phi_2 \vdash V_2: B
        }
        \and
        \inferrule* {(3)} {
          \Gamma \semi \Pi \semi \Phi_3 \vdash []: \mf{List}\ A
        }
      } {
        \Gamma \semi \Pi \semi \Phi_1 \Join \Phi_2 \Join \Phi_3 \vdash
        \mf{fold}\ V_1\ V_2\ []: B \bang \mf{Pure}
      }
    } {
      \Gamma \semi \Pi \semi \Phi_1 \Join \Phi_2 \Join \Phi_3 \vdash
      \mf{fold}\ V_1\ V_2\ []: B \bang \mf{Mut}
    }
  \end{mathpar}
  \begin{itemize}
    \item By Lemma \ref{thm:lamhook-unrestv} (unrestriction) on
      either of (1a) or (1b) as well as (3),
      we get (4) that $\mf{Succ}(\Phi_1)$ and (5) that $\mf{Succ}(\Phi_3)$,
      respectively.
    \item By Lemma \ref{thm:lamhook-phiwknv} ($\Phi$-weakening) on (2)
      using (4) and (5),
      we get (6) that $\Gamma \semi \Pi \semi \Phi_1 \Join \Phi_2 \Join \Phi_3
      \vdash V_2: B$.
    \item We derive conclusion (iii-a) as follows:
    \begin{mathpar}
      \inferrule*[right=TC-Val] {
        \inferrule*[right=(6)] { } {
          \Gamma \semi \Pi \semi \Phi_1 \Join \Phi_2 \Join \Phi_3 \vdash
          V_2: B
        }
      } {
        \Gamma \semi \Pi \semi \Phi_1 \Join \Phi_2 \Join \Phi_3 \vdash
        \mf{val}\ V_2: B \bang \mf{Pure}
      }
    \end{mathpar}
    \blue{
      \item By premise \emph{(ii)}, we have
      $\Pi \semi \mf{ids}(V_2) \cup \mf{ids}(\heap) \vDash \R$,
      which satisfies conclusion \emph{(iv)}.
      \item We show conclusion \emph{(v)} as follows:
      \begin{align*}
        &\bsq{\heap}\R \semi \bsq{\mf{fold}\ V_1\ V_2\ []}\R\\
        (\text{def. of translation}) \quad
        &= \mf{fold}\ \bsq{V_1}\R\ \bsq{V_2}\R\ []
          \semi \bsq{\heap}\R\\
        \textsc{(EC-FdN)}\quad
        &\step \mf{val}\ \bsq{V_2}\R \semi \bsq{\heap}\R\\
        (\text{def. of translation}) \quad
        &= \bsq{\mf{val}\ V_2}\R \semi \bsq{\heap}\R
      \end{align*}}
  \end{itemize}

  \pcase{EC-FdC}: We have
  $\heap \semi \mf{fold}\ V_1\ V_2\ (V_h :: V_t)
  \step V_1\ V_h\ (\mf{fold}\ V_1\ V_2\ V_t) \semi \heap$,
  whose right-hand term desguars to
  $\elet{x}{V_1\ V_h}{\elet{y}{\mf{fold}\ V_1\ V_2\ V_t}{x\ y}}$.
  Conclusion \emph{(ii)} is immediate, and we also have either:
  \begin{mathpar}
    \inferrule*[right=TC-Fold] {
      (1a) \and (2) \and
      \inferrule* {
        \inferrule* {(3)} {
          \Gamma \semi \Pi \semi \Phi_3 \vdash V_h: A
        }
        \and
        \inferrule* {(4)} {
          \Gamma \semi \Pi \semi \Phi_4 \vdash V_t: \mf{List}\ A
        }
      } {
        \Gamma \semi \Pi \semi \Phi_3 \Join \Phi_4
        \vdash V_h :: V_t: \mf{List}\ A
      }
    } {
      \Gamma \semi \Pi \semi \Phi_1 \Join \Phi_2 \Join \Phi_3 \Join \Phi_4
      \vdash
      \mf{fold}\ V_1\ V_2\ (V_h :: V_t): B \bang \mf{Mut}
    }
  \end{mathpar}
  \begin{mathpar}
    \inferrule*[right=TC-Sub] {
      \inferrule*[Right=TC-Fold] {
        (1b) \and (2) \and
        \inferrule* {
          \inferrule* {(3)} {
            \Gamma \semi \Pi \semi \Phi_3 \vdash V_h: A
          }
          \and
          \inferrule* {(4)} {
            \Gamma \semi \Pi \semi \Phi_4 \vdash V_t: \mf{List}\ A
          }
        } {
          \Gamma \semi \Pi \semi \Phi_3 \Join \Phi_4
          \vdash V_h :: V_t: \mf{List}\ A
        }
      } {
        \Gamma \semi \Pi \semi \Phi_1 \Join \Phi_2 \Join \Phi_3 \Join \Phi_4
        \vdash
        \mf{fold}\ V_1\ V_2\ (V_h :: V_t): B \bang \mf{Pure}
      }
    } {
      \Gamma \semi \Pi \semi \Phi_1 \Join \Phi_2 \Join \Phi_3 \Join \Phi_4
      \vdash
      \mf{fold}\ V_1\ V_2\ (V_h :: V_t): B \bang \mf{Mut}
    }
  \end{mathpar}
  where (1a), (1b), and (2) are identical to the previous case.
  \begin{itemize}
    \item By Lemma \ref{thm:lamhook-unrestv} (unrestriction) on (1)-(4),
      choosing either of (1a) or (1b) for (1),
      we get that $\mf{Succ}(\Phi_i)$ for $i \in \{1, 2, 3, 4\}$.
      It follows (5) that $\mf{Succ}(\Phi)$, where $\Phi=\Join_{i=1}^4 \Phi_i$.
    \item In case (1a), we derive conclusion (iii-a) by way of
      intermediate derivations (6)-(8) as follows,
      using $\Phi$ intuitionistically by (5) and weakening (1)-(4)
      appropriately:
  \end{itemize}
  \begin{mathpar}
    \inferrule*[right=TC-App] {
      \inferrule*[right=TV-Var] { } {
        \Gamma_6 \semi \Pi \semi \Phi
        \vdash x: B \xrightarrow{\mf{Mut}} B
      }
      \and
      \inferrule*[right=TV-Var] { } {
        \Gamma_6 \semi \Pi \semi \Phi
        \vdash y: B
      }
    } {
      (6) \quad
      \Gamma_6 \semi \Pi \semi \Phi \vdash
      x\ y: B \bang \mf{Mut}
    }

    \inferrule*[right=TC-Fold] {
      \inferrule*[right=(1)] { } {
        \Gamma_7 \semi \Pi \semi \Phi \vdash V_1:
        A \xrightarrow{\mf{Mut}} B \xrightarrow{\mf{Mut}} B
      }
      \
      \inferrule*[right=(2)] { } {
        \Gamma_7 \semi \Pi \semi \Phi \vdash V_2: B
      }
      \
      \inferrule*[right=(4)] { } {
        \Gamma_7 \semi \Pi \semi \Phi \vdash V_t: \mf{List}\ A
      }
    } {
      (7) \quad
      \Gamma_7 \semi \Pi \semi \Phi \vdash
      \mf{fold}\ V_1\ V_2\ V_t: B \bang \mf{Mut}
    }

    \inferrule*[right=TC-Let] {
      [(7),\ \Gamma_7=\Gamma_8]
      \and
      [(6),\ \Gamma_6=\Gamma_8, y: B]
    } {
      (8) \quad
      \Gamma_8 \semi \Pi \semi \Phi \vdash
      \elet{y}{\mf{fold}\ V_1\ V_2\ V_t}{\cdots}: B \bang \mf{Mut}
    }

    \inferrule*[right=TC-Let] {
      \inferrule*[Right=TC-App] {
        \inferrule*[right=(1)] { } {
          \Gamma \semi \Pi \semi \Phi \vdash V_1:
          A \xrightarrow{\mf{Mut}} B \xrightarrow{\mf{Mut}} E
        }
        \
        \inferrule*[right=(3)] { } {
          \Gamma \semi \Pi \semi \Phi \vdash V_h: A
        }
      } {
        \Gamma \semi \Pi \semi \Phi \vdash
        V_1\ V_h: B \xrightarrow{\mf{Mut}} B \bang \mf{Mut}
      }
      \
      [(8),\ \Gamma_8 = \Gamma, x: B \xrightarrow{\mf{Mut}} B]
    } {
      \Gamma \semi \Pi \semi \Phi
      \vdash
      \elet{x}{V_1\ V_h}{\cdots}: B \bang \mf{Mut}
    }
  \end{mathpar}
  \begin{itemize}
    \item In case (1b), we derive conclusion (iii-a) by way of
      intermediate derivations (9)-(11) as follows,
      using $\Phi$ intuitionistically by (5) and weakening (1)-(4)
      appropriately:
  \end{itemize}
  \begin{mathpar}
    \inferrule*[right=TC-App] {
      \inferrule*[right=TV-Var] { } {
        \Gamma_9 \semi \Pi \semi \Phi
        \vdash x: B \xrightarrow{\mf{Pure}} B
      }
      \and
      \inferrule*[right=TV-Var] { } {
        \Gamma_9 \semi \Pi \semi \Phi
        \vdash y: B
      }
    } {
      (9) \quad
      \Gamma_9 \semi \Pi \semi \Phi \vdash
      x\ y: B \bang \mf{Pure}
    }

    \inferrule*[right=TC-Fold] {
      \inferrule*[right=(1)] { } {
        \Gamma_{10} \semi \Pi \semi \Phi \vdash V_1:
        A \xrightarrow{\mf{Pure}} B \xrightarrow{\mf{Pure}} B
      }
      \
      \inferrule*[right=(2)] { } {
        \Gamma_{10} \semi \Pi \semi \Phi \vdash V_2: B
      }
      \
      \inferrule*[right=(4)] { } {
        \Gamma_{10} \semi \Pi \semi \Phi \vdash V_t: \mf{List}\ A
      }
    } {
      (10) \quad
      \Gamma_{10} \semi \Pi \semi \Phi \vdash
      \mf{fold}\ V_1\ V_2\ V_t: B \bang \mf{Pure}
    }

    \inferrule*[right=TC-Let] {
      [(10),\ \Gamma_{10}=\Gamma_{11}]
      \and
      [(9),\ \Gamma_9=\Gamma_{11}, y: B]
    } {
      (11) \quad
      \Gamma_{11} \semi \Pi \semi \Phi \vdash
      \elet{y}{\mf{fold}\ V_1\ V_2\ V_t}{\cdots}: B \bang \mf{Pure}
    }

    \inferrule*[right=TC-Sub] {
      \inferrule*[Right=TC-Let] {
        \inferrule*[Right=TC-App] {
          \inferrule*[right=(1)] { } {
            \Gamma \semi \Pi \semi \Phi \vdash V_1:
            A \xrightarrow{\mf{Pure}} B \xrightarrow{\mf{Pure}} E
          }
          \
          \inferrule*[right=(3)] { } {
            \Gamma \semi \Pi \semi \Phi \vdash V_h: A
          }
        } {
          \Gamma \semi \Pi \semi \Phi \vdash
          V_1\ V_h: B \xrightarrow{\mf{Pure}} B \bang \mf{Pure}
        }
        \
        [(11),\ \Gamma_{11} = \Gamma, x: B \xrightarrow{\mf{Pure}} B]
      } {
        \Gamma \semi \Pi \semi \Phi
        \vdash
        \elet{x}{V_1\ V_h}{\cdots}: B \bang \mf{Pure}
      }
    } {
      \Gamma \semi \Pi \semi \Phi
      \vdash
      \elet{x}{V_1\ V_h}{\cdots}: B \bang \mf{Mut}
    }
  \end{mathpar}
  \begin{itemize}
    \blue{
      \item By premise \emph{(ii)}, we have
      $\Pi \semi \mf{ids}(V_1) \cup \mf{ids}(V_2) \cup \mf{ids}(V_h)
        \cup \mf{ids}(V_t) \cup \mf{ids}(\heap) \vDash \R$,
      which satisfies conclusion \emph{(iv)}.
      \item We show conclusion \emph{(v)} as follows:
      \begin{align*}
        &\bsq{\heap}\R \semi \bsq{\mf{fold}\ V_1\ V_2\ (V_h :: V_t)}\R\\
        (\text{def. of translation}) \quad
        &= \mf{fold}\ \bsq{V_1}\R\ \bsq{V_2}\R\ (\bsq{V_h}\R :: \bsq{V_t}\R)
          \semi \bsq{\heap}\R\\
        \textsc{(EC-FdC)}\quad
        &\step \bsq{V_1}\R\ \bsq{V_h}\R\
          (\mf{fold}\ \bsq{V_1}\R\ \bsq{V_2}\R\ \bsq{V_t}\R)
          \semi \bsq{\heap}\R\\
        (\text{def. of translation}) \quad
        &= \bsq{V_1\ V_h\ (\mf{fold}\ V_1\ V_2\ V_t)}\R
          \semi \bsq{\heap}\R\\
      \end{align*}}
  \end{itemize}

  \pcase{EC-AppSet}: We have
  $\heap \semi \setter^\ell_p\ V \step \mf{val}\ () \semi \heap'$ for
  $\heap' = \heap \mid \heap[p].\mf{st}[\ell] \mapsto V$
  and $\heap[p].\mf{rerender} \mapsto \mf{tt}$, and
  \begin{mathpar}
    \inferrule*[right=TC-App] {
      \inferrule*[right=TV-Setter] {
        \inferrule* {(1)} {\ell: A \in \Pi[p]}
        \and
        \inferrule* {(2)} {\mf{Succ}(\Phi_1)}
      } {
        \Gamma \semi \Pi \semi \Phi_1 \vdash \setter^\ell_p:
        A \xrightarrow{\mf{Mut}} \tone
      }
      \and
      \inferrule* {(3)} {
        \Gamma \semi \Pi \semi \Phi_2 \vdash V: A
      }
    } {
      \Gamma \semi \Pi \semi \Phi_1 \Join \Phi_2 \vdash
      \setter^\ell_p\ V: \tone \bang \mf{Mut}
    }
  \end{mathpar}
  \begin{itemize}
    \item (1) and premise (i-b) give us (4) that
    $\Gamma \semi \Pi \semi \Phi_3 \vdash \heap[p].\mf{st}[\ell]: A$
    for $\Phi_H = \Phi_3 \Join \Phi_4$.
    \item By Lemma \ref{thm:lamhook-unrestv} (unrestriction) on (3),
      we know (5) that $\mf{Succ}(\Phi_2)$.
    \item Choose $\Phi_N = \Phi_M$, $\Phi_H' = \Phi_H \Join \Phi_2$,
      and $\Pi' = \Pi$, trivializing conclusion \emph{(ii)}.
      By (5) and Lemma \ref{thm:lamhook-unresth} (unrestriction) on $\Phi_H$,
      $\Phi_H'$ is defined.
    \item By (5),
      $(\Phi_1 \Join \Phi_2) \Join (\Phi_H \Join \Phi_2)
      = (\Phi_1 \Join \Phi_2) \Join \Phi_H$,
      satisfying conclusion \emph{(i)}.
    \item Clearly $\mf{Succ}(\Phi_3)$ by Lemma \ref{thm:lamhook-unresth},
      so we weaken (3) by $\Phi_3$ (Lemma \ref{thm:lamhook-phiwknv})
      to get that\\$\Gamma \semi \Pi \semi \Phi_2 \Join \Phi_3 \vdash V: A$.
      It then follows from premise (i-b) that
      $\Gamma \semi \Pi \semi (\Phi_2 \Join \Phi_3) \Join \Phi_4 \vDash \heap$,
      i.e. conclusion (iii-b) holds.
    \item We derive conclusion (iii-a) as follows:
    \begin{mathpar}
    \inferrule*[right=TC-Sub] {
      \inferrule*[right=TC-Val] {
        \inferrule*[right=TV-Unit] {
          \inferrule*[right=(2) \textnormal{and} (5)] { } {
            \mf{Succ}(\Phi_1 \Join \Phi_2)
          }
        } {
          \Gamma \semi \Pi \semi \Phi_1 \Join \Phi_2 \vdash (): \tone
        }
      } {
        \Gamma \semi \Pi \semi \Phi_1 \Join \Phi_2 \vdash \mf{val}\ ():
        \tone \bang \mf{Pure}
      }
    } {
      \Gamma \semi \Pi \semi \Phi_1 \Join \Phi_2 \vdash \mf{val}\ ():
      \tone \bang \mf{Mut}
    }
    \end{mathpar}
    \blue{
      \item We have $\mf{ids}(\mf{val}\ ())=\emptyset$, and clearly
      $\mf{ids}(\heap') \subseteq \mf{ids}(\heap) \cup \mf{ids}(V)$,
      so conclusion \emph{(iv)} is given by premise \emph{(ii)}.
      \item We show conclusion \emph{(v)} as follows:
      \begin{align*}
        \bsq{\heap}\R \semi &\bsq{\setter^\ell_p\ V}\R\\
        (\text{def. of translation}) \quad =\
        &(\lambda x.\R(p, \ell) := x \semi \R(p) := \mf{true})\ \bsq{V}\R
        \semi \bsq{\heap}\R\\
        (\textsc{EC-App})\quad \step\
        &\R(p, \ell) := \bsq{V}\R \semi \R(p) := \mf{true}
        \semi \bsq{\heap}\R\\
        (\textsc{EC-Set, EC-Val})\quad \stept\
        &\R(p) := \mf{true} \semi
        \bsq{\heap}\R[\R(p, \ell) \mapsto \bsq{V}\R]\\
        (\textsc{EC-Set})\quad \step\
        &\mf{val}\ () \semi \bsq{\heap}\R
        [\R(p, \ell) \mapsto \bsq{V}\R,
        \R(p) \mapsto \mf{true}]\\
        (\text{def. of translation})\quad =\
        &\bsq{\mf{val}\ ()}\R \semi \bsq{\heap'}\R
      \end{align*}
    }
  \end{itemize}

  \pcase{EC-Eval}, $\phi=\mf{Init}$: We have

  \begin{mathpar}
    \inferrule*[right=EC-Eval] {
      \inferrule* {(1)} {\heap \semi D \step^\mf{Init}_p D' \semi \heap'}
    } {
      \heap \semi \eval{D_0}{\mf{Init}}{p}{D}
      \step
      \eval{D_0}{\mf{Init}}{p}{D'} \semi \heap'
    }

    \inferrule*[right=TC-EvalI] {
      \inferrule* {(2)} {
        \Gamma \semi \Pi \semi \Phi_1 \vdash D_0: \Ell \uplus \Pi[p]
      }
      \and
      \inferrule* {(3)} {\Gamma \semi \Pi \semi \Phi_2 \vdash D: \Ell}
    } {
      \Gamma \semi \Pi \semi \Phi_1 \Join \Phi_2,p: \mf{Init} \vdash
      \eval{D_0}{\mf{init}}{p}{D}: \mf{ElemInst} \bang \mf{Mut}
    }
  \end{mathpar}
  \begin{itemize}
    \item We are given that $\Phi_H \Join \Phi_1 \Join \Phi_2, p: \mf{Init}$
      is defined, which entails
      (4) $p \notin \dom(\Phi_H \Join \Phi_1 \Join \Phi_2)$.
    \item By Theorem \ref{thm:trans-simbi} ($\step^\mf{Init}_p$-simulation)
    on (1)-(4) and premises (i-b) and \emph{(ii)},
    we have some $\Phi_H'$, $\Pi'$, \blue{and $\R' \supseteq \R$} with:
    \begin{enumerate}
      \setcounter{enumi}{4}
      \item $\Phi_H \Join \Phi_2 = \Phi_H' \Join \Phi_2$,
      \item $\Pi \sqsubseteq_{p: \mf{Init}} \Pi'$,
      \item $\Gamma \semi \Pi' \semi \Phi_2 \vdash D': \Ell'$,
      \item $\Gamma \semi \Pi' \semi \Phi_1 \vdash D_0: \Ell' \uplus \Pi'[p]$,
      \item $\Gamma \semi \Pi' \semi \Phi_H' \vDash \heap'$,
      \blue{
        \item $\Pi'\ \semi \mf{ids}(D') \cup \mf{ids}(D_0)
        \cup \mf{ids}(\heap') \vDash \R'$, and
        \item $\bsq{\heap}\R \semi \itransp{D}{D_0}{p}{\R}{\emptyset}
          \stept\ \itransp{D'}{D_0}{p}{\R'}{\emptyset}
          \semi \bsq{\heap'}{\R'}$.
      }
    \end{enumerate}
    \item We choose this $\Pi'$, $\Phi_H'$, \blue{and $\R'$},
      and retain $\Phi_N = \Phi_M = \Phi_1 \Join \Phi_2, p: \mf{Init}$.
    \item We started with $\Phi_1 \Join \Phi_2 \Join \Phi_H, p: \mf{Init}$
      as defined. By (5), we know (12) that\\
      $\Phi_1 \Join \Phi_2 \Join \Phi_H', p: \mf{Init}$ is also defined.
    \item Conclusion \emph{(i)} follows from
      Lemma \ref{thm:lamhook-freshphi} (fresh $\Phi$-extension) using (5)
      and extending by $\Phi_1, p: \mf{Init}$, which is valid by (12).
    \item (12) also lets us apply Lemma
      \ref{thm:lamhook-piprez} ($\Pi$-order preservation) on (6),
      extending by\\$\Phi_1 \Join \Phi_2 \Join \Phi_H'$,
      to get that $\Pi \sqsubseteq_{\Phi_M \Join \Phi_H'} \Pi'$,
      satisfying conclusion \emph{(ii)}.
    \item By choice of $\Pi$ and $\Phi$, (9) gives conclusion (iii-b).
    \item Lemma \ref{thm:lamhook-piwknb} ($\Pi$-weakening) on (2) using (6),
      which is valid by (12), gives us (13) that
      $\Gamma \semi \Pi' \semi \Phi_1 \vdash D_0: \Ell \uplus \Pi[p]$.
    \item By (10), (13) is equivalent to (14) that
      $\Gamma \semi \Pi' \semi \Phi_1 \vdash D_0: \Ell' \uplus \Pi'[p]$.
    \item We derive conclusion (iii-a) as follows:
    \begin{mathpar}
      \inferrule*[right=TC-EvalI] {
        \inferrule*[right=(14)] { } {
          \Gamma \semi \Pi' \semi \Phi_1 \vdash D_0: \Ell' \uplus \Pi'[p]
        }
        \and
        \inferrule*[right=(6)] { } {
          \Gamma \semi \Pi' \semi \Phi_2 \vdash D': \Ell'
        }
      } {
        \Gamma \semi \Pi' \semi \Phi_1 \Join \Phi_2, p: \mf{Init} \vdash
        \eval{D_0}{\mf{Init}}{p}{D'}: \mf{ElemInst} \bang \mf{Mut}
      }
    \end{mathpar}
    \blue{\item Conclusion \emph{(iv)} is satisfied by (10).
    \item We show conclusion \emph{(v)} as follows:
    \begin{align*}
      \bsq{\heap}\R \semi &\bsq{\eval{D_0}{\mf{Init}}{p}{D}}\R\\
      (\text{def. of translation}) \quad =\
      &\mf{let}\ rr, f, x = \itransp{D}{D_0}{p}{\R}{\emptyset}\\
      &\mf{in}\ \mf{compInst}(rr, f, \mf{init}\ x) \semi \bsq{\heap}\R\\
      (\text{Lemma \ref{thm:trans-let}, let-trans. on (11)})
      \ \stept\
      &\mf{let}\ rr, f, x = \itransp{D'}{D_0}{p}{\R'}{\emptyset} \\
      &\mf{in}\ \mf{compInst}(rr, f, \mf{init}\ x) \semi \bsq{\heap'}{\R'}\\
      (\text{def. of translation}) \quad =\
      & \bsq{\eval{D_0}{\mf{Init}}{p}{D'}}{\R'}
      \semi \bsq{\heap'}{\R'}
    \end{align*}}
  \end{itemize}


  \pcase{EC-Eval}, $\phi=\mf{Succ}$: We have
  \begin{mathpar}
    \inferrule*[right=EC-Eval] {
      \inferrule* {(1)} {\heap \semi D \step^\mf{Succ}_p D' \semi \heap'}
    } {
      \heap \semi \eval{D_0}{\mf{Succ}}{p}{D}
      \step
      \eval{D_0}{\mf{Succ}}{p}{D'} \semi \heap'
    }

    \inferrule*[right=TC-EvalS] {
      \inferrule* {(2)} {\Gamma \semi \Pi \semi \Phi_1 \vdash D_0: \Pi[p]}
      \
      \inferrule* {(3)} {\Gamma \semi \Pi \semi \Phi_2 \vdash D: \Ell}
      \
      \inferrule* {(4)} {\Ell \subseteq \Pi[p]}
      \
      \inferrule* {(5)} {(\Phi_1 \Join \Phi_2)[p] = \mf{Succ}}
    } {
      \Gamma \semi \Pi \semi \Phi_1 \Join \Phi_2 \vdash
      \eval{D_0}{\mf{Succ}}{p}{D}: \mf{ElemInst} \bang \mf{Mut}
    }
  \end{mathpar}
  \begin{itemize}
    \blue{\item By premise \emph{(iii)}, we have (6) that
    $\Pi \semi \mf{ids}(D) \cup \mf{ids}(\heap) \vDash \R$.}
    \item By Theorem \ref{thm:trans-simbs}
    ($\step^\mf{Succ}_p$-simulation)
    on (1), (3), (4), \blue{(6)} and premise (i-b), we have some $\Phi_2'$
    such that
    \begin{enumerate}
      \setcounter{enumi}{5}
      \item $\Phi_2' \Join \Phi_H = \Phi_2 \Join \Phi_H$,
      \item $\Gamma \semi \Pi \semi \Phi_2' \vdash D': \Ell'$,
      \item $\heap' = \heap$,
      \item $\Ell' \subseteq \Ell$,
      \blue{
        \item $\Pi \semi \mf{ids}(D') \cup \mf{ids}(\heap) \vDash \R$, and
        \item $\bsq{\heap}\R \semi \stransp{D}{p}{\R}{\emptyset}
        \stept\ \stransp{D'}{p}{\R}{\emptyset}  \semi \bsq{\heap}\R$.
      }
    \end{enumerate}
    \item Choose $\Phi_N = \Phi_1 \Join \Phi_2'$ and $\Phi_H' = \Phi_H$.
      By Lemma \ref{thm:lamhook-unrestb} (unrestriction) on (2) and (7),
      both $\mf{Succ}(\Phi_1)$ and $\mf{Succ}(\Phi_2')$, so $\Phi_N$
      is defined. Lemma \ref{thm:lamhook-freshphi} (fresh $\Phi$-extension)
      on (6) then gives us conclusion (i).
    \item Choose $\Pi' = \Pi$. Conclusion (ii) is trivial,
      and (8) and premise (i-b) give us conclusion (iii-b).
    \item By (4) and (9), we have (12) that
      $\Ell' \subseteq \Pi[p]$.
    \item We want to show (13) that $(\Phi_1 \Join \Phi_2')[p] = \mf{Succ}$
      by case-analysis on (5):
    \begin{itemize}
      \item If $p: \mf{Succ} \in \Phi_1$,
        then clearly $p: \mf{Succ} \in (\Phi_1 \Join \Phi_2')$.
      \item If not, then we must have $p: \mf{Succ} \in \Phi_2$.
        By definition, $p: \mf{Succ} \in \Phi_2 \Join \Phi_H$.
        By (6), $p: \mf{Succ} \in \Phi_2' \Join \Phi_H$.
        Again by definition, $p: \mf{Succ} \in \Phi_2'$.
        Finally, we get\\$p: \mf{Succ} \in (\Phi_1 \Join \Phi_2')$.
    \end{itemize}
    \item We derive conclusion (iii-b) by \textsc{TC-EvalS} as follows:
  \end{itemize}
  \begin{mathpar}
    \inferrule {
      \inferrule*[right=(2)] { } {
        \Gamma \semi \Pi \semi \Phi_1 \vdash D_0: \Pi[p]
      }
      \
      \inferrule*[right=(7)] { } {
        \Gamma \semi \Pi \semi \Phi_2' \vdash D': \Ell'
      }
      \
      \inferrule*[right=(12)] { } {\Ell' \subseteq \Pi[p]}
      \
      \inferrule*[right=(13)] { } {(\Phi_1 \Join \Phi_2')[p] = \mf{Succ}}
    } {
      \Gamma \semi \Pi \semi \Phi_1 \Join \Phi_2' \vdash
      \eval{D_0}{\mf{Succ}}{p}{D'}: \mf{ElemInst} \bang \mf{Mut}
    }
  \end{mathpar}
  \begin{itemize}
    \blue{\item Conclusion \emph{(iv)}
      is given by (10) and premise \emph{(ii)}.
    \item We show conclusion \emph{(v)} as follows:
    \begin{align*}
      \bsq{\heap}\R \semi &\bsq{\eval{D_0}{\mf{Succ}}{p}{D}}\R\\
      (\text{def. of translation}) \quad =\
      &\elet{x}{\stransp{D}{p}{\R}{\emptyset}}{\\&\mf{compInst}(
        \stransp{D_0}{p}{\R}{\emptyset} ,
        \R(p),
        \mf{init}\ x
      )} \semi \bsq{\heap}\R\\
      (\text{Lemma \ref{thm:trans-let}, let-trans. on (11)})
      \ \stept\
      &\elet{x}{\stransp{D'}{p}{\R}{\emptyset}}{\\&\mf{compInst}(
        \stransp{D_0}{p}{\R}{\emptyset} ,
        \R(p),
        \mf{init}\ x
      )} \semi \bsq{\heap}\R\\
      (\text{def. of translation}) \quad =\
      & \bsq{\eval{D_0}{\mf{Succ}}{p}{D'}}{\R}
      \semi \bsq{\heap}{\R}\\
    \end{align*}}
  \end{itemize}

  \pcase{EC-EvalDone}, $\phi=\mf{Init}$: We have
  $$\heap \semi \eval{D}{\mf{Init}}{p}{\mf{val}\ V}
  \step \elet{x}{\mf{init}\ V}{\mf{val}\ \compInst{p}{D}{x}} \semi \heap$$
  \begin{mathpar}
    \inferrule*[right=TC-EvalI] {
      \inferrule* {(1)} {
        \Gamma \semi \Pi \semi \Phi_1 \vdash D_0: \emptyset \uplus \Pi[p]
      }
      \and
      \inferrule*[right=TB-Val] {
        \inferrule* {(2)} {
          \Gamma \semi \Pi \semi \Phi_2 \vdash V: \mf{Element}
        }
      } {
        \Gamma \semi \Pi \semi \Phi_2 \vdash \mf{val}\ V: \emptyset
      }
    } {
      \Gamma \semi \Pi \semi \Phi_1 \Join \Phi_2,p: \mf{Init} \vdash
      \eval{D_0}{\mf{Init}}{p}{\mf{val}\ V}: \mf{ElemInst} \bang \mf{Mut}
    }
  \end{mathpar}
  \begin{itemize}
    \item We choose $\Phi_H' = \Phi_H$ and
    $\Phi_N = \Phi_1 \Join \Phi_2, p: \mf{Succ}$.
    Clearly $\Phi_M \Join \Phi_H \sqsubseteq_\mf{fresh} \Phi_N \Join \Phi_H'$,
    satisfying conclusion \emph{(i)}.
    \item We retain $\Pi' = \Pi$, trivializing conclusion \emph{(ii)},
      and also letting conclusion (iii-b) follow easily from premise (i-b).
    \blue{\item We also retain $\R' = \R$. Clearly
    $\mf{ids}(M) = \mf{ids}(N)$, so premise \emph{(ii)} gives
    conclusion \emph{(iv)}.}
    \item By Lemmas \ref{thm:lamhook-unrestb} and \ref{thm:lamhook-unrestv}
      on (1) and (2), we get $\mf{Succ}(\Phi_1)$ and $\mf{Succ}(\Phi_2)$,
      respectively. Clearly then (3) $\mf{Succ}(\Phi_N)$.
    \item By Lemma \ref{thm:lamhook-phiwknb} ($\Phi$-weakening)
      on (1) using (3),
      we get (4) that $\Gamma \semi \Pi \semi \Phi_N \vdash D: \Pi[p]$.
    \item Similarly, by Lemma \ref{thm:lamhook-phiwknv} ($\Phi$-weakening)
      on (2) using (3), we get (5) that\\
      $\Gamma \semi \Pi \semi \Phi_N \vdash V: \mf{Element}$.
    \item We derive conclusion (iii-a) as follows
    by way of subderivation (6),
    using $\Phi_N$ intuitionistically by (3):
    \begin{mathpar}
      \inferrule*[right=TC-Sub] {
        \inferrule*[Right=TC-Val] {
          \inferrule*[Right=TV-CompI] {
            \inferrule*[right=(4)] { } {
              \Gamma_6 \semi \Pi \semi \Phi_N \vdash D: \Pi[p]
            }
            \and
            \inferrule* { } {\Phi_N[p] = \mf{Succ}}
            \and
            \inferrule*[Right=TV-Var] {
              \inferrule*[right=(3)] { } {\mf{Succ}(\Phi_N)}
            } {
              \Gamma_6 \semi \Pi \semi \Phi_N \vdash
              x: \mf{ElemInst}
            }
          } {
            \Gamma_6 \semi \Pi \semi \Phi_N \vdash
            \compInst{p}{D}{x}: \mf{ElemInst}
          }
        } {
          \Gamma_6 \semi \Pi \semi \Phi_N \vdash
          \mf{val}\ \compInst{p}{D}{x}: \mf{ElemInst} \bang \mf{Pure}
        }
      } {
        (6) \quad
        \Gamma_6 \semi \Pi \semi \Phi_N \vdash
        \mf{val}\ \compInst{p}{D}{x}: \mf{ElemInst} \bang \mf{Mut}
      }
      \and
      \inferrule*[right=TC-Let] {
        \inferrule*[right=TC-Init] {
          \inferrule*[right=(5)] { } {
            \Gamma \semi \Pi \semi \Phi_N \vdash V: \mf{Element}
          }
        } {
          \Gamma \semi \Pi \semi \Phi_N \vdash \mf{init}\ V:
          \mf{ElemInst} \bang \mf{Mut}
        }
        \and
        [(6),\ \Gamma_6 = \Gamma, x: \mf{ElemInst}]
      } {
        \Gamma \semi \Pi \semi \Phi_N \vdash \elet{x}{\mf{init}\ V}{
          \mf{val}\ \compInst{p}{D}{x}
        }: \mf{ElemInst} \bang \mf{Mut}
      }
    \end{mathpar}
    \blue{\item We show conclusion \emph{(v)} as follows:
    \begin{align*}
      \bsq{\heap}\R \semi &\bsq{\eval{D}{\mf{Init}}{p}{\mf{val}\ V}}\R\\
      \text{(def. of translation)} \quad =\
      &\mf{let}\ rr, f, x =
        (\lambda().\stransp{D}{p}{\R}{\emptyset}, \R(p), \bsq{V}\R)\\
      &\mf{in}\ \mf{compInst}(rr, f, \mf{init}\ x) \semi \bsq{\heap}\R\\
      (\textsc{EC-Val/Fst/Snd}) \quad \stept\
      &\mf{compInst}(\lambda().\stransp{D}{p}{\R}{\emptyset} ,
      \R(p), \mf{init}\ \bsq{V}\R) \semi \bsq{\heap}\R\\
      (\text{desugar}) \quad =\
      &\elet{x}{\mf{init}\ \bsq{V}\R}{
        \mf{val\ compInst}(\lambda().\stransp{D}{p}{\R}{\emptyset} ,
        \R(p), x)} \semi \bsq{\heap}\R\\
      (\text{def. of translation}) \quad =\
      &\bsq{\elet{x}{\mf{init}\ V}{\mf{val}\ \compInst{p}{D}{x}}}\R
      \semi \bsq{\heap}\R
    \end{align*}}
  \end{itemize}

  \pcase{EC-EvalDone}, $\phi=\mf{Succ}$:
  \[\heap \semi \eval{D}{\mf{Succ}}{p}{\mf{val}\ V}
  \step \elet{x}{\mf{init}\ V}{\mf{val}\ \compInst{p}{D}{x}} \semi \heap\]
  \begin{mathpar}
    \inferrule*[right=TC-EvalS] {
      \inferrule* {(1)} {\Gamma \semi \Pi \semi \Phi_1 \vdash D: \Pi[p]}
      \
      \inferrule*[right=TB-Val] {
        \inferrule* {(2)} {
          \Gamma \semi \Pi \semi \Phi_2 \vdash V: \mf{Element}
        }
      } {
        \Gamma \semi \Pi \semi \Phi_2 \vdash \mf{val}\ V: \emptyset
      }
      \
      \inferrule* {(3)} {(\Phi_1 \Join \Phi_2)[p] = \mf{Succ}}
      \ \cdots
    } {
      \Gamma \semi \Pi \semi \Phi_1 \Join \Phi_2 \vdash
      \eval{D_0}{\mf{Succ}}{p}{D}: \mf{ElemInst} \bang \mf{Mut}
    }
  \end{mathpar}
  \begin{itemize}
    \item We retain $\Phi_N = \Phi_1 \Join \Phi_2$.
      By Lemmas \ref{thm:lamhook-unrestb} and \ref{thm:lamhook-unrestv}
      on (1) and (2), we get $\mf{Succ}(\Phi_1)$ and $\mf{Succ}(\Phi_2)$,
      respectively. Clearly then (4) $\mf{Succ}(\Phi_N)$.
    \item By Lemma \ref{thm:lamhook-phiwknb} ($\Phi$-weakening) on (1)
      using (4), we get (5) that
      $\Gamma \semi \Pi \semi \Phi_N \vdash D: \Pi[p]$.
    \item By Lemma \ref{thm:lamhook-phiwknv} ($\Phi$-weakening) on (2),
    we get (6) $\Gamma \semi \Pi \semi \Phi_N \vdash V: \mf{Element}$.
    \item We derive conclusion (iii-a) as follows
    by way of subderivation (7),
    using $\Phi_N$ intuitionistically by (4):
    \begin{mathpar}
      \inferrule*[right=TC-Sub] {
        \inferrule*[Right=TC-Val] {
          \inferrule*[Right=TV-CompI] {
            \inferrule*[right=(5)] { } {
              \Gamma_7 \semi \Pi \semi \Phi_N \vdash D: \Pi[p]
            }
            \and
            \inferrule* { } {\Phi_N[p] = \mf{Succ}}
            \and
            \inferrule*[Right=TV-Var] {
              \inferrule*[right=(4)] { } {\mf{Succ}(\Phi_D')}
            } {
              \Gamma_7 \semi \Pi \semi \Phi_N \vdash
              x: \mf{ElemInst}
            }
          } {
            \Gamma_7 \semi \Pi \semi \Phi_N \vdash
            \compInst{p}{D}{x}: \mf{ElemInst}
          }
        } {
          \Gamma_7 \semi \Pi \semi \Phi_N \vdash
          \mf{val}\ \compInst{p}{D}{x}: \mf{ElemInst} \bang \mf{Pure}
        }
      } {
        (7) \quad
        \Gamma_7 \semi \Pi \semi \Phi_N \vdash
        \mf{val}\ \compInst{p}{D}{x}: \mf{ElemInst} \bang \mf{Mut}
      }
      \and
      \inferrule*[right=TC-Let] {
        \inferrule*[right=TC-Init] {
          \inferrule*[right=(5)] { } {
            \Gamma \semi \Pi \semi \Phi_N \vdash V: \mf{Element}
          }
        } {
          \Gamma \semi \Pi \semi \Phi_N \vdash \mf{init}\ V:
          \mf{ElemInst} \bang \mf{Mut}
        }
        \and
        [(7),\ \Gamma_7 = \Gamma, x: \mf{ElemInst}]
      } {
        \Gamma \semi \Pi \semi \Phi_N \vdash \elet{x}{\mf{init}\ V}{
          \mf{val}\ \compInst{p}{D}{x}
        }: \mf{ElemInst} \bang \mf{Mut}
      }
    \end{mathpar}
    \blue{\item We show conclusion \emph{(v)} as follows:
    \begin{align*}
      \bsq{\heap}\R \semi &\bsq{\eval{D}{\mf{Succ}}{p}{\mf{val}\ V}}\R\\
      (\text{def. of translation}) \quad =\
      &\elet{y}{\mf{val}\ \bsq{V}\R}{\mf{compInst}(
        \lambda().\stransp{D}{p}{\R}{\emptyset} ,
        \R(p), \mf{init}\ y)} \semi \bsq{\heap}\R\\
      (\textsc{EC-Val}, \text{desugar}) \quad \stept\
      &\elet{x}{\mf{init}\ \bsq{V}\R}{\mf{val\ compInst}(
        \lambda().\stransp{D}{p}{\R}{\emptyset} , \R(p), x)}
        \semi \bsq{\heap}\R\\
      (\text{def. of translation}) \quad =\
      &\bsq{\elet{x}{\mf{init}\ V}{\mf{val}\ \compInst{p}{D}{x}}}\R
      \semi \bsq{\heap}\R
    \end{align*}}
  \end{itemize}

  \pcase{EC-InitText}: We have
  $\heap \semi \mf{init\ text}\ V \step \mf{val\ textInst}\ V \semi \heap$ and
  \begin{mathpar}
    \inferrule*[right=TC-Init] {
      \inferrule*[right=TV-Text] {
        \inferrule* {(1)} {
          \Gamma \semi \Pi \semi \Phi_M \vdash V: \mf{Str}
        }
      } {
        \Gamma \semi \Pi \semi \Phi_M \vdash \mf{text}\ V: \mf{Element}
      }
    } {
      \Gamma \semi \Pi \semi \Phi_M \vdash \mf{init\ text}\ V:
      \mf{ElemInst} \bang \mf{Mut}
    }
  \end{mathpar}
  \begin{itemize}
    \item We derive conclusion (iii-a) as follows:
    \begin{mathpar}
      \inferrule*[right=TC-Sub] {
        \inferrule*[right=TC-Val] {
          \inferrule*[right=TV-TextI] {
            \inferrule*[right=(1)] { } {
              \Gamma \semi \Pi \semi \Phi_M \vdash V: \mf{Str}
            }
          } {
            \Gamma \semi \Pi \semi \Phi_M \vdash \mf{textInst}\ V: \mf{ElemInst}
          }
        } {
          \Gamma \semi \Pi \semi \Phi_M \vdash \mf{val\ textInst}\ V:
          \mf{ElemInst} \bang \mf{Pure}
        }
      } {
        \Gamma \semi \Pi \semi \Phi_M \vdash \mf{val\ textInst}\ V:
        \mf{ElemInst} \bang \mf{Mut}
      }
    \end{mathpar}
    \blue{\item We show conclusion \emph{(v)} as follows:
    \begin{align*}
      \bsq{\heap}\R \semi &\bsq{\mf{init\ text}\ V}\R\\
      (\text{def. of translation}) \quad =\
      &\mf{init\ text}\ \bsq{V}\R \semi \bsq{\heap}\R\\
      (\textsc{EC-InitText}) \quad \step\
      &\mf{val\ textInst}\ \bsq{V}\R \semi \bsq{\heap}\R\\
      (\text{def. of translation}) \quad =\
      &\bsq{\mf{val\ textInst}\ V}\R \semi \bsq{\heap}\R\\
    \end{align*}}
  \end{itemize}

  \pcase{EC-InitTag}: We have
  $\heap \semi \mf{init}\ \mf{tag}(V, W)
  \step
  \elet{x}{\mf{inits}\ W}{\mf{val}\ \mf{tagInst}_{id}(V, x)} \semi \heap$
  and
  \begin{mathpar}
  \inferrule*[right=TC-Init] {
    \inferrule*[right=TV-Tag] {
      \inferrule* {(1)} {
        \Gamma \semi \Pi \semi \Phi_1 \vdash V: \mf{List\ Attr}
      }
      \and
      \inferrule* {(2)} {
        \Gamma \semi \Pi \semi \Phi_2 \vdash W: \mf{List\ Element}
      }
    } {
      \Gamma \semi \Pi \semi \Phi_1 \Join \Phi_2 \vdash \mf{tag}(V, W):
      \mf{Element}
    }
  } {
    \Gamma \semi \Pi \semi \Phi_1 \Join \Phi_2 \vdash \mf{init\ tag}(V, W):
    \mf{ElemInst} \bang \mf{Mut}
  }
  \end{mathpar}
  \begin{itemize}
    \blue{
      \item We choose $\R' = \R, {id} \mapsto {id}'$,
      where ${id}'$ is the identifier
      generated by the target language semantics.
      \item It follows from premise \emph{(ii)} that
      $\Pi \semi \mf{ids}(\elet{x}{\mf{inits}\ W}
      {\mf{val}\ \mf{tagInst}_{id}(V, x)}) \cup \mf{ids}(\heap) \vDash \R'$:
      in particular, we have ${id} \in \dom(\R')$.
      Therefore, conclusion \emph{(iv)} holds.
    }
    \item We derive conclusion (iii-a) via subderivation (3)
    as follows, weakening (1) appropriately on $\Gamma$.
    \begin{mathpar}
      \inferrule*[right=TC-Sub] {
        \inferrule*[right=TC-Val] {
          \inferrule*[Right=TV-TagI] {
            \inferrule*[right=(1)] { } {
              \Gamma_3 \semi \Pi \semi \Phi_1
              \vdash V: \mf{List\ Attr}
            }
            \
            \inferrule*[Right=TV-Var] { } {
              \Gamma_3 \semi \Pi \semi \emptyset
              \vdash x: \mf{List\ ElemInst}
            }
          } {
            \Gamma_3 \semi \Pi \semi \Phi_1
            \vdash \mf{tagInst}_{id}(V, x): \mf{ElemInst}
          }
        } {
          \Gamma_3 \semi \Pi \semi \Phi_1
          \vdash \mf{val\ tagInst}_{id}(V, x): \mf{ElemInst} \bang \mf{Pure}
        }
      } {
        (3) \quad \Gamma_3 \semi \Pi \semi \Phi_1
        \vdash \mf{val\ tagInst}_{id}(V, x): \mf{ElemInst} \bang \mf{Mut}
      }
      \and
      \inferrule*[right=TC-Let] {
        \inferrule*[right=TC-Inits] {
          \inferrule*[Right=(2)] { } {
            \Gamma \semi \Pi \semi \Phi_2 \vdash W: \mf{List\ Element}
          }
        } {
          \Gamma \semi \Pi \semi \Phi_2 \vdash \mf{inits}\ W:
          \mf{List\ ElemInst} \bang \mf{Mut}
        }
        [(3),\ \Gamma_3 = \Gamma, x: \mf{List\ ElemInst}]
      } {
        \Gamma \semi \Pi \semi \Phi_1 \Join \Phi_2 \vdash
        \elet{x}{\mf{inits}\ W}{\mf{val}\ \mf{tagInst}_{id}(V, x)}:
        \mf{ElemInst} \bang \mf{Mut}
      }
    \end{mathpar}
    \blue{\item We show conclusion \emph{(v)} as follows:
    \begin{align*}
      \bsq{\heap}\R \semi &\bsq{\mf{init\ tag}(V, W)}\R\\
      (\text{def. of translation}) \quad =\
      &\mf{init\ tag}(\bsq{V}\R, \bsq{W}\R) \semi \bsq{\heap}\R\\
      (\textsc{EC-InitTag}) \quad \step\
      &\elet{x}{\mf{inits}\ \bsq{W}\R}{
        \mf{val\ tagInst}_{{id}'}(\bsq{V}\R, x)
      } \semi \bsq{\heap}\R\\
      (\text{Lemmas \ref{thm:trans-rextv}/\ref{thm:trans-rextm},
      $\R$-ext.}) \quad =\
      &\elet{x}{\mf{inits}\ \bsq{W}{\R'}}{
        \mf{val\ tagInst}_{\R'({id})}(\bsq{V}{\R'}, x)
      } \semi \bsq{\heap}{\R'}\\
      (\text{def. of translation}) \quad =\
      &\bsq{\elet{x}{\mf{inits}\ W}{\mf{val\ tagInst}_{id}(V, x)}}{\R'}
      \semi \bsq{\heap}{\R'}
    \end{align*}}
  \end{itemize}

  \pcase{EC-InitComp}: We have
  $\heap \semi \mf{init}\ \mf{component}\ \{D\}
  \step \eval{D}{\mf{Init}}{p}{D}
  \semi \heap'$, where $p$ is fresh,\\
  $\heap' = \heap, p \mapsto
  \{\mf{st}: \emptyset, \mf{rerender}: \mf{ff}\}$,
  and
  \begin{mathpar}
  \inferrule*[right=TC-Init] {
    \inferrule*[right=TV-Comp] {
      \inferrule* {(1)} {\Gamma \semi \Pi \semi \Phi_M \vdash D: \Ell}
    } {
      \Gamma \semi \Pi \semi \Phi_M \vdash \mf{component}\ \{D\}: \mf{Element}
    }
  } {
    \Gamma \semi \Pi \semi \Phi_M \vdash \mf{init\ component}\ \{D\}:
    \mf{ElemInst} \bang \mf{Mut}
  }
  \end{mathpar}
  \begin{itemize}
    \item Choose $\Phi_N = \Phi_M, p: \mf{Init}$ and $\Phi_H' = \Phi_H$.
      Since $p$ is fresh, $\Phi_N \Join \Phi_H$ is defined and we have
      $\Phi_M \Join \Phi_H \sqsubseteq_\mf{fresh} \Phi_N \Join \Phi_H$,
      satisfying conclusion (i).
    \item Choose $\Pi'=\Pi,p \mapsto \emptyset$. By definition,
      we have $\Pi \sqsubseteq_\mf{\Phi_N \Join \Phi_H} \Pi'$,
      satisfying conclusion (ii).
    \blue{\item Choose $\R' = \R, p \mapsto \ell$, where $\ell$
      is the location generated by the target-language semantics.
      Since $\Pi'[p] = \emptyset$, $\R'$ satisfies $\Pi'$,
      and since $\mf{ids}(M) = \mf{ids}(D) = \mf{ids}(N)$,
      conclusion \emph{(iv)} holds.}
    \item By premise (i-b)
      and definition of each of $\Pi'$ and $\heap'$,
      we have that $\Gamma \semi \Pi' \semi \Phi_H \vDash \heap'$,
      satisfying conclusion (iii-b).
    \item By Lemma \ref{thm:lamhook-piwknb} ($\Pi$-weakening) on (1),
    using the fact that $p$ is fresh and so $p \notin \dom(\Phi_M)$,
    we get (3) that $\Gamma \semi \Pi' \semi \Phi_M \vdash D: \Ell$.
    \item Since $\Pi'[p]=\emptyset$,
    (3) is equivalent to (4) that
    $\Gamma \semi \Pi' \semi \Phi_M \vdash D: \Ell \uplus \Pi'[p]$.
    \item By Lemma \ref{thm:lamhook-unrestb} on (1),
      we know $\mf{Succ}(\Phi_M)$.
      It follows that (5) $\Phi_M \Join \Phi_M = \Phi_M$.
    \item We derive conclusion (iii-a) as follows, implicitly using (5):
    \begin{mathpar}
      \inferrule*[right=TC-EvalI] {
        \inferrule*[right=(4)] { } {
          \Gamma \semi \Pi' \semi \Phi_M \vdash D: \Ell \uplus \Pi'[p]
        }
        \and
        \inferrule*[right=(3)] { } {
          \Gamma \semi \Pi' \semi \Phi_M \vdash D: \Ell
        }
      } {
        \Gamma \semi \Pi' \semi \Phi_M, p: \mf{Init} \vdash
        \eval{D}{\mf{Init}}{p}{D}: \mf{ElemInst} \bang \mf{Mut}
      }
    \end{mathpar}
    \blue{
      \item By Lemma \ref{thm:trans-felimi} ($\delta$-elimination),
      we have\\
      $\itransp{D}{D}{p}{\R'}{\mf{flag} \mapsto x}[\R'(p)/x]
      = \itransp{D}{D}{p}{\R'}{\emptyset}[\R'(p)/x]$.
      Given that $x \notin \mf{fv}(D)$,
      we have (6) that
      $\itransp{D}{D}{p}{\R'}{\mf{flag} \mapsto x}[\R'(p)/x]
      = \itransp{D}{D}{p}{\R'}{\emptyset}$.
      \item We show conclusion \emph{(v)} as follows:
    \begin{align*}
      \bsq{\heap}\R \semi &\bsq{\mf{init}\ \mf{component}\ \{D\}}\R\\
      \text{(def. of translation)} \quad =\
      &\mf{init}\ \mf{component}\ \lambda().
      \elet{x}{\mf{ref}(\mf{false})}{\\
        &\ \itransp{D}{D}{p}{\R}{\mf{flag} \mapsto x}
      } \semi \bsq{\heap}\R\\
      (\textsc{EC-InitComp, EC-App}) \quad \stept\
      &\mf{let}\ rr, f, y = (\elet{x}{\mf{ref}(\mf{false})}{
        \itransp{D}{D}{p}{\R}{\mf{flag} \mapsto x}
      })\\
      &\mf{in}\ \mf{inst}(rr, f, \mf{init}\ y) \semi \bsq{\heap}\R\\
      (\textsc{EC-Ref, EC-Val}) \quad \stept\
      &\mf{let}\ rr, f, y =
        \itransp{D}{D}{p}{\R}{\mf{flag} \mapsto x}
        [\R'(p)/x]\\
      &\mf{in}\ \mf{inst}(rr, f, \mf{init}\ y)
      \semi \bsq{\heap}\R[\R'(p) \mapsto \mf{false}]\\
      (\text{Lemmas \ref{thm:trans-rextbi}/\ref{thm:trans-rextm}, $\R$-ext.})
      \quad =\
      &\mf{let}\ rr, f, y =
        \itransp{D}{D}{p}{\R'}{\mf{flag} \mapsto x}
        [\R'(p)/x]\\
      &\mf{in}\ \mf{inst}(rr, f, \mf{init}\ y)
      \semi \bsq{\heap}{\R'}[\R'(p) \mapsto \mf{false}]\\
      (\text{by (6)})
      \quad =\
      &\mf{let}\ rr, f, y =
        \itransp{D}{D}{p}{\R'}{\emptyset}\
        \mf{in}\ \mf{inst}(rr, f, \mf{init}\ y) \semi\\
      &\bsq{\heap}{\R'}[\R'(p) \mapsto \mf{false}]\\
      (\text{def. of translation})
      \quad =\
      &\bsq{\eval{D}{\mf{Init}}{p}{D}}{\R'}
      \semi \bsq{\heap'}{\R'}\\
    \end{align*}
    }
  \end{itemize}

  \pcase{EC-ClickText}: We have
  $\heap \semi \mf{click}_{id}\
  \mf{textInst}\ V \step \mf{val}\ () \semi \heap$ and
  \begin{mathpar}
    \inferrule*[right=TC-Click] {
      \inferrule* {(1)} {
        \Gamma \semi \Pi \semi \Phi_M \vdash \mf{textInst}\ v: \mf{ElemInst}
      }
    } {
      \Gamma \semi \Pi \semi \Phi_M \vdash \mf{click}_{id}\ \mf{textInst}\ V:
      \tone \bang \mf{Mut}
    }
  \end{mathpar}
  \begin{itemize}
    \item By Lemma \ref{thm:lamhook-unrestv} (unrestriction) on (1),
    we get (2) that $\mf{Succ}(\Phi_M)$.
    \item We derive conclusion (iii-a) as follows:
    \begin{mathpar}
      \inferrule*[right=TC-Sub] {
        \inferrule*[right=TC-Val] {
          \inferrule*[right=TV-Unit] {
            \inferrule*[right=(2)] { } {\mf{Succ}(\Phi_M)}
          } {
            \Gamma \semi \Pi \semi \Phi_M \vdash (): \tone
          }
        } {
          \Gamma \semi \Pi \semi \Phi_M \vdash \mf{val}\ (): \tone \bang \mf{Pure}
        }
      } {
        \Gamma \semi \Pi \semi \Phi_M \vdash \mf{val}\ (): \tone \bang \mf{Mut}
      }
    \end{mathpar}
    \blue{\item We show conclusion \emph{(v)} as follows:
    \begin{align*}
      \bsq{\heap}\R \semi &\bsq{\mf{click}_{id}\ \mf{textInst}\ V}\R\\
      (\text{def. of translation}) \quad =\
      &\mf{click}_{\R({id})}\ \mf{textInst}\ \bsq{V}\R
      \semi \bsq{\heap}\R\\
      (\textsc{EC-ClickText}) \quad \step\
      &\mf{val}\ () \semi \bsq{\heap}\R\\
      (\text{def. of translation}) \quad =\
      &\bsq{\mf{val}\ ()}\R \semi \bsq{\heap}\R
    \end{align*}}
  \end{itemize}

  \pcase{EC-ClickTagT}: We have
  $\heap \semi \mf{click}_{id}\ \mf{tagInst}_{id}(V, W)
  \step \mf{runClicks}\ V \semi \mf{clicks}_{id}\ W \semi \heap$ and
  \begin{mathpar}
    \inferrule*[right=TC-Click] {
      \inferrule*[right=TV-TagI] {
        \inferrule* {(1)} {
          \Gamma \semi \Pi \semi \Phi_1 \vdash V: \mf{List\ Attr}
        }
        \and
        \inferrule* {(2)} {
          \Gamma \semi \Pi \semi \Phi_2 \vdash W: \mf{List\ ElemInst}
        }
      } {
        \Gamma \semi \Pi \semi \Phi_1 \Join \Phi_2 \vdash
        \mf{tagInst}_{id}(V, W): \mf{ElemInst}
      }
    } {
      \Gamma \semi \Pi \semi \Phi_1 \Join \Phi_2 \vdash
      \mf{click}_{id}\ \mf{tagInst}_{id}(V, W): \tone \bang \mf{Mut}
    }
  \end{mathpar}
  \begin{itemize}
    \item We derive conclusion (iii-a) as follows:
  \end{itemize}
  \begin{mathpar}
    \inferrule*[right=TC-Let] {
      \inferrule*[right=TC-RunClicks] {
        \inferrule*[right=(1)] { } {
          \Gamma \semi \Pi \semi \Phi_1 \vdash V: \mf{List\ Attr}
        }
      } {
        \Gamma \semi \Pi \semi \Phi_1 \vdash \mf{runClicks}\ V:
        \tone \bang \mf{Mut}
      }
      \
      \inferrule*[Right=TC-Clicks] {
        \inferrule*[right=(2)] { } {
          \Gamma \semi \Pi \semi \Phi_2 \vdash V: \mf{List\ ElemInst}
        }
      } {
        \Gamma \semi \Pi \semi \Phi_2 \vdash \mf{clicks}_{id}\ W:
        \tone \bang \mf{Mut}
      }
    } {
      \Gamma \semi \Pi \semi \Phi_1 \Join \Phi_2 \vdash
      \mf{runClicks}\ V \semi \mf{clicks}_{id}\ W: \tone \bang \mf{Mut}
    }
  \end{mathpar}
  \begin{itemize}
    \blue{\item We show conclusion \emph{(v)} as follows:
    \begin{align*}
      \bsq{\heap}\R \semi &\bsq{\mf{click}_{id}\
      \mf{tagInst}_{id}(V, W)}\R\\
      (\text{def. of translation}) \quad =\
      &\mf{click}_{\R({id})}\ \mf{tagInst}_{\R({id})}
        (\bsq{V}\R, \bsq{W}\R) \semi \bsq{\heap}\R\\
      (\textsc{EC-ClickTagT}) \quad \step\
      &\mf{runClicks}\ \bsq{V}\R \semi
        \mf{clicks}_{\R({id})}\ \bsq{W}\R \semi \bsq{\heap}\R\\
      (\text{def. of translation}) \quad =\
      &\bsq{\mf{runClicks}\ V \semi \mf{clicks}_{id}\ W}\R
      \semi \bsq{\heap}\R
    \end{align*}}
  \end{itemize}

  \pcase{EC-ClickTagF}: We have
  $\heap \semi \mf{click}_{id}\ \mf{tagInst}_{{id}'}(V, W)
  \step \mf{clicks}_{id}\ W \semi \heap$ and
  \begin{mathpar}
    \inferrule*[right=TC-Click] {
      \inferrule*[right=TV-TagI] {
        \inferrule* {(1)} {
          \Gamma \semi \Pi \semi \Phi_1 \vdash V: \mf{List\ Attr}
        }
        \and
        \inferrule* {(2)} {
          \Gamma \semi \Pi \semi \Phi_2 \vdash W: \mf{List\ ElemInst}
        }
      } {
        \Gamma \semi \Pi \semi \Phi_1 \Join \Phi_2 \vdash
        \mf{tagInst}_{{id}'}(V, W): \mf{ElemInst}
      }
    } {
      \Gamma \semi \Pi \semi \Phi_1 \Join \Phi_2 \vdash
      \mf{click}_{id}\ \mf{tagInst}_{{id}'}(V, W): \tone \bang \mf{Mut}
    }
  \end{mathpar}
  \begin{itemize}
    \item By Lemma \ref{thm:lamhook-unrestv} (unrestriction) on (1),
    we get (3) that $\mf{Succ}(\Phi_1)$.
    \item By Lemma \ref{thm:lamhook-phiwknv} ($\Phi$-weakening) on (2)
    using (3), we get (4) that\\
    $\Gamma \semi \Pi \semi \Phi_1 \Join \Phi_2 \vdash W: \mf{List\ ElemInst}$.
    \item We derive conclusion (iii-a) as follows:
    \begin{mathpar}
      \inferrule*[right=TC-Clicks] {
        \inferrule*[right=(1)] { } {
          \Gamma \semi \Pi \semi \Phi_1 \Join \Phi_2 \vdash W:
          \mf{List\ ElemInst}
        }
      } {
        \Gamma \semi \Pi \semi \Phi_1 \Join \Phi_2 \vdash
        \mf{clicks}_{id}\ W: \tone \bang \mf{Mut}
      }
    \end{mathpar}
    \blue{
      \item By injectivity of $\R$, we have (5) $\R({id}) \neq \R({id}')$.
      \item We show conclusion \emph{(v)} as follows:
    \begin{align*}
      \bsq{\heap}\R \semi &\bsq{\mf{click}_{id}\
      \mf{tagInst}_{{id}'}(V, W)}\R\\
      (\text{def. of translation}) \quad =\
      &\mf{click}_{\R({id})}\ \mf{tagInst}_{\R({id}')}
        (\bsq{V}\R, \bsq{W}\R) \semi \bsq{\heap}\R\\
      (\textsc{EC-ClickTagF}\text{, by (5)}) \quad \step\
      &\mf{clicks}_{\R({id})}\ \bsq{W}\R \semi \bsq{\heap}\R\\
      (\text{def. of translation}) \quad =\
      &\bsq{\mf{clicks}_{id}\ W}\R \semi \bsq{\heap}\R
    \end{align*}}
  \end{itemize}

  \pcase{EC-ClickComp}: We have
  $\heap \semi \mf{click}_{id}\ \compInst{p}{D}{V}
  \step \mf{click}_{id}\ V \semi \heap$ and
  \begin{mathpar}
    \inferrule*[right=TC-Click] {
      \inferrule*[Right=TV-CompI] {
        \inferrule* {(1)} {\Gamma \semi \Pi \semi \Phi_1 \vdash D: \Pi[p]}
        \and
        (\Phi_1 \Join \Phi_2)[p] = \mf{Succ}
        \and
        \inferrule* {(1)} {
          \Gamma \semi \Pi \semi \Phi_2 \vdash V: \mf{ElemInst}
        }
      } {
        \Gamma \semi \Pi \semi \Phi_1 \Join \Phi_2
        \vdash \compInst{p}{D}{V}: \mf{ElemInst}
      }
    } {
      \Gamma \semi \Pi \semi \Phi_1 \Join \Phi_2 \vdash
      \mf{click}_{id}\ \compInst{p}{D}{V}: \tone \bang \mf{Mut}
    }
  \end{mathpar}
  \begin{itemize}
    \item By Lemma \ref{thm:lamhook-unrestb} (unrestriction) on (1),
    we get (3) that $\mf{Succ}(\Phi_1)$.
    \item By Lemma \ref{thm:lamhook-phiwknv} ($\Phi$-weakening) on (2)
    using (3), we get (4) that
    $\Gamma \semi \Pi \semi \Phi_1 \Join \Phi_2
    \vdash V: \mf{ElemInst}$.
    \item We derive conclusion (iii-a) as follows:
    \begin{mathpar}
      \inferrule*[right=TC-Click] {
        \inferrule*[right=(4)] { } {
          \Gamma \semi \Pi \semi \Phi_1 \Join \Phi_2 \vdash V: \mf{ElemInst}
        }
      } {
        \Gamma \semi \Pi \semi \Phi_1 \Join \Phi_2 \vdash
        \mf{click}_{id}\ V: \tone \bang \mf{Mut}
      }
    \end{mathpar}
    \blue{\item We show conclusion \emph{(v)} as follows:
    \begin{align*}
      \bsq{\heap}\R \semi &\bsq{\mf{click}_{id}\ \compInst{p}{D}{V}}\R\\
      (\text{def. of translation}) \quad =\
      &\mf{click}_{\R({id})}\ \mf{inst}(
        \lambda ().\stransp{D}{p}{\R}{\emptyset},
        \R(p), \bsq{V}\R) \semi \bsq{\heap}\R\\
      (\textsc{EC-ClickComp}) \quad \step\
      &\mf{click}_{\R({id})}\ \bsq{V}\R \semi \bsq{\heap}\R\\
      (\text{def. of translation}) \quad =\
      &\bsq{\mf{click}_{id}\ V}\R \semi \bsq{\heap}\R
    \end{align*}}
  \end{itemize}

  \pcase{EC-RunAttr}: Analogous to \textsc{EC-ClickText}.

  \pcase{EC-RunClick}: We have
  $\heap \semi \mf{runClick\ onClick}\ V \step V\ () \semi \heap$ and
  \begin{mathpar}
    \inferrule*[right=TC-RunClick] {
      \inferrule*[right=TV-OnClick] {
        \inferrule* {(1)} {
          \Gamma \semi \Pi \semi \Phi_M \vdash V:
          \tone \xrightarrow{\mf{Mut}} \tone
        }
      } {
        \Gamma \semi \Pi \semi \Phi_M \vdash \mf{onClick}\ V: \mf{Attr}
      }
    } {
      \Gamma \semi \Pi \semi \Phi_M \vdash \mf{runClick\ onClick}\ V:
      \tone \bang \mf{Mut}
    }
  \end{mathpar}
  \begin{itemize}
    \item We derive conclusion (iii-a) as follows:
    \begin{mathpar}
      \inferrule*[right=TC-App] {
        \inferrule*[right=(1)] { } {
          \Gamma \semi \Pi \semi \Phi_M \vdash V:
          \tone \xrightarrow{\mf{Mut}} \tone
        }
        \and
        \inferrule*[right=TV-Unit] { } {
          \Gamma \semi \Pi \semi \emptyset \vdash (): \tone
        }
      } {
        \Gamma \semi \Pi \semi \Phi_M \vdash V\ (): \tone \bang \mf{Mut}
      }
    \end{mathpar}
    \blue{\item We show conclusion \emph{(v)} as follows:
    \begin{align*}
      \bsq{\heap}\R \semi &\bsq{\mf{runClick\ onClick}\ V}\R\\
      (\text{def. of translation}) \quad =\
      &\mf{runClick\ onClick}\ \bsq{V}\R \semi \bsq{\heap}\R\\
      (\textsc{EC-RunClick}) \quad \step\
      &\bsq{V}\R\ () \semi \bsq{\heap}\R\\
      (\text{def. of translation}) \quad =\
      &\bsq{V\ ()}\R \semi \bsq{\heap}\R
    \end{align*}}
  \end{itemize}

  \textbf{Case} \textsc{EC-UpdateText} and \pcase{EC-UpdateTag} are
  analogous to \textsc{EC-InitText} and \textsc{EC-InitTag} respectively.

  \pcase{EC-UpdateCompT}: We have,
  for $\heap' = \heap \mid \heap[p].\mf{rerender} \mapsto \mf{ff}$,
  \begin{mathpar}
    \inferrule*[right=EC-UpdateCompT] {
      \inferrule* {(1)} {\heap[p].\mf{rerender} = \mf{tt}}
    } {
      \heap \semi \mf{update}\ \compInst{p}{D}{v}
      \step \eval{D}{\mf{Succ}}{p}{D} \semi \heap'
    }

    \inferrule*[right=TC-Update] {
      \inferrule*[Right=TV-CompI] {
        \inferrule* {(2)} {\Gamma \semi \Pi \semi \Phi_1 \vdash D: \Pi[p]}
        \and
        \inferrule* {(3)} {(\Phi_1 \Join \Phi_2)[p] = \mf{Succ}}
        \and
        \inferrule* {(4)} {
          \Gamma \semi \Pi \semi \Phi_2 \vdash V: \mf{ElemInst}
        }
      } {
        \Gamma \semi \Pi \semi \Phi_1 \Join \Phi_2
        \vdash \compInst{p}{D}{V}: \mf{ElemInst}
      }
    } {
      \Gamma \semi \Pi \semi \Phi_1 \Join \Phi_2
      \vdash \mf{update}\ \compInst{p}{D}{V}: \mf{ElemInst} \bang \mf{Mut}
    }
  \end{mathpar}
  \begin{itemize}
    \item We retain $\Phi_N = \Phi_M$, $\Phi_H' = \Phi_H$, $\Pi' = \Pi$,
      \blue{and $\R' = \R$}.
      Conclusions (i) and (ii) are trivial.
      By premise (i-b), we have
      $\Gamma \semi \Pi \semi \Phi_H \vDash \heap'$,
      satisfying conclusion (iii-b).
    \blue{\item We have
      $\mf{ids}(\eval{D}{\mf{Succ}}{p}{D}) = \mf{ids}(D)
      \subseteq \mf{ids}(D) \cup \mf{ids}(V)
      = \mf{ids}(\mf{update}\ \compInst{p}{D}{V})$,
      so conclusion \emph{(iv)} is given by premise \emph{(ii)}.}
    \item By Lemma \ref{thm:lamhook-unrestv} (unrestriction) on (3),
      we must have (5) $\mf{Succ}(\Phi_2)$.
    \item By Lemma \ref{thm:lamhook-phiwknb} ($\Phi$-weakening) on (1)
      using (5), we get (6) that
      $\Gamma \semi \Pi \semi \Phi_M \vdash D: \Pi[p]$.
    \item We derive conclusion (i) as follows, noting that
       $\Phi_M \Join \Phi_M = \Phi_M$ by (5):
    \begin{mathpar}
      \inferrule*[right=TC-EvalS] {
        \inferrule*[right=(6) (x2)] { } {
          \Gamma \semi \Pi \semi \Phi_M \vdash D: \Pi[p]
        }
        \and
        \inferrule* { } {\Pi[p] \subseteq \Pi[p]}
        \and
        \inferrule*[right=(2)] { } {\Phi_M[p] = \mf{Succ}}
      } {
        \Gamma \semi \Pi \semi \Phi_M \vdash
        \eval{D}{\mf{Succ}}{p}{D}: \mf{Element} \bang \mf{Mut}
      }
    \end{mathpar}
    \blue{
      \item By (1) and definition of translation, we have (7)
      that $\bsq{\heap}\R[\R(p)] = \mf{true}$.
      \item We show conclusion \emph{(v)} as follows:
      \begin{align*}
      \bsq{\heap}\R \semi &\bsq{\mf{update}\ \compInst{p}{D}{V}}\R\\
      \text{(def. of translation)} =\
      &\mf{update}\ \mf{compInst}(
        \stransp{D}{p}{\R}{\emptyset}
        \R(p), \bsq{V}\R
      )\\
      (\textsc{EC-UpdateComp}) \step\
      &\mf{let}\ x = \mf{if} \bang \R(p)\\
      &\ \ \mf{then}\ \R(p) := \mf{false}
        \semi \mf{init}\
        ((\lambda ().\stransp{D}{p}{\R}{\emptyset})\ ())\\
      &\ \ \mf{else}\ \mf{update}\ \bsq{V}\R\\
      & \mf{in\ compInst}(\lambda ().\stransp{D}{p}{\R}{\emptyset},
        \R(p), x) \semi \bsq{\heap}\R\\
      (\textsc{EC-Get}, (7), \textsc{EC-IfT}) \stept\
      &\mf{let}\ x = \R(p) := \mf{false}
        \semi \mf{init}\
        ((\lambda ().\stransp{D}{p}{\R}{\emptyset})\ ())\\
      & \mf{in\ compInst}(\lambda ().\stransp{D}{p}{\R}{\emptyset},
        \R(p), x) \semi \bsq{\heap}\R\\
      (\textsc{EC-Set}, \textsc{EC-Val}) \stept\
      &\mf{let}\ x = \mf{init}\
        (\lambda ().\stransp{D}{p}{\R}{\emptyset}\ ())\\
      & \mf{in\ inst}(\lambda ().\stransp{D}{p}{\R}{\emptyset},
        \R(p), x) \semi \bsq{\heap}\R[\R(p) \mapsto \mf{false}]\\
      (\textsc{EC-App}) \step\
      &\mf{let}\ x = \mf{init}\
        (\stransp{D}{p}{\R}{\emptyset})\
        \mf{in\ compInst}(\lambda ().\stransp{D}{p}{\R}{\emptyset},
        \R(p), x) \semi\\
      &\bsq{\heap}\R[\R(p) \mapsto \mf{false}]\\
      (\text{def. of translation}) =\
      &\bsq{\eval{D}{\mf{Succ}}{p}{D}}\R
      \semi \bsq{\heap'}\R
    \end{align*}}
  \end{itemize}

  \pcase{EC-UpdateCompF}: Analogous to \textsc{EC-ClickTagF}.
\end{proof}


\end{document}